\documentclass[a4paper,UKenglish,cleveref, autoref, thm-restate]{lipics-v2021}
\nolinenumbers
\usepackage{hyperref}
\usepackage{xcolor}
\usepackage{xspace}
\usepackage{amsmath,amsthm,amssymb}
\usepackage{bm}
\usepackage{stmaryrd}
\usepackage{cancel}
\usepackage{hyperref}
\usepackage{adjustbox}

\usepackage{bbm} %% mathbb for lowercase

\usepackage{cleveref}
\usepackage{cmll}
\usepackage{tikz-cd}
\usepackage[strict]{changepage}

\usepackage{thmtools, thm-restate} %% for restating thms in appendix
\newcommand{\proofnote}[1]{\textcolor{red}{{\small \textsc{[#1]}}}}

\usepackage{enumitem}
\setlistdepth{9}
\newlist{xenumerate}{enumerate}{12}
\setlist[xenumerate,1]{label*=\arabic*.}
\setlist[xenumerate,2]{label*=\arabic*.}
\setlist[xenumerate,3]{label*=\arabic*.}
\setlist[xenumerate,4]{label*=\arabic*.}
\setlist[xenumerate,5]{label*=\arabic*.}
\setlist[xenumerate,6]{label*=\arabic*.}
\setlist[xenumerate,7]{label*=\arabic*.}
\setlist[xenumerate,8]{label*=\arabic*.}
\setlist[xenumerate,9]{label*=\arabic*.}
\setlist[xenumerate,10]{label*=\arabic*.}
\setlist[xenumerate,11]{label*=\arabic*.}
\setlist[xenumerate,12]{label*=\arabic*.}

\colorlet{darkgreen}{green!60!black}
\definecolor{deepcarrotorange}{rgb}{0.91, 0.41, 0.17}

\newdimen\proofrulebreadth \proofrulebreadth=.05em
\newdimen\proofdotseparation \proofdotseparation=1.25ex
\newdimen\proofrulebaseline \proofrulebaseline=2ex
\newcount\proofdotnumber \proofdotnumber=3
\let\then\relax
\def\hfi{\hskip0pt plus.0001fil}
\mathchardef\squigto="3A3B
\newif\ifinsideprooftree\insideprooftreefalse
\newif\ifonleftofproofrule\onleftofproofrulefalse
\newif\ifproofdots\proofdotsfalse
\newif\ifdoubleproof\doubleprooffalse
\let\wereinproofbit\relax
\newdimen\shortenproofleft
\newdimen\shortenproofright
\newdimen\proofbelowshift
\newbox\proofabove
\newbox\proofbelow
\newbox\proofrulename
\def\shiftproofbelow{\let\next\relax\afterassignment\setshiftproofbelow\dimen0 }
\def\shiftproofbelowneg{\def\next{\multiply\dimen0 by-1 }%
\afterassignment\setshiftproofbelow\dimen0 }
\def\setshiftproofbelow{\next\proofbelowshift=\dimen0 }
\def\setproofrulebreadth{\proofrulebreadth}

\def\prooftree{% NESTED ZERO (\ifonleftofproofrule)
\ifnum  \lastpenalty=1
\then   \unpenalty
\else   \onleftofproofrulefalse
\fi
\ifonleftofproofrule
\else   \ifinsideprooftree
        \then   \hskip.5em plus1fil
        \fi
\fi
\bgroup% NESTED ONE (\proofbelow, \proofrulename, \proofabove,
\setbox\proofbelow=\hbox{}\setbox\proofrulename=\hbox{}%
\let\justifies\proofover\let\leadsto\proofoverdots\let\Justifies\proofoverdbl
\let\using\proofusing\let\[\prooftree
\ifinsideprooftree\let\]\endprooftree\fi
\proofdotsfalse\doubleprooffalse
\let\thickness\setproofrulebreadth
\let\shiftright\shiftproofbelow \let\shift\shiftproofbelow
\let\shiftleft\shiftproofbelowneg
\let\ifwasinsideprooftree\ifinsideprooftree
\insideprooftreetrue
\setbox\proofabove=\hbox\bgroup$\displaystyle % NESTED TWO
\let\wereinproofbit\prooftree
\shortenproofleft=0pt \shortenproofright=0pt \proofbelowshift=0pt
\onleftofproofruletrue\penalty1
}

\def\eproofbit{% NESTED TWO
\ifx    \wereinproofbit\prooftree
\then   \ifcase \lastpenalty
        \then   \shortenproofright=0pt  % 0: some other object, no indentation
        \or     \unpenalty\hfil         % 1: empty hypotheses, just glue
        \or     \unpenalty\unskip       % 2: just had a tree, remove glue
        \else   \shortenproofright=0pt  % eh?
        \fi
\fi
\global\dimen0=\shortenproofleft
\global\dimen1=\shortenproofright
\global\dimen2=\proofrulebreadth
\global\dimen3=\proofbelowshift
\global\dimen4=\proofdotseparation
\global\count255=\proofdotnumber
$\egroup  % NESTED ONE
\shortenproofleft=\dimen0
\shortenproofright=\dimen1
\proofrulebreadth=\dimen2
\proofbelowshift=\dimen3
\proofdotseparation=\dimen4
\proofdotnumber=\count255
}

\def\proofover{% NESTED TWO
\eproofbit % NESTED ONE
\setbox\proofbelow=\hbox\bgroup % NESTED TWO
\let\wereinproofbit\proofover
$\displaystyle
}%
\def\proofoverdbl{% NESTED TWO
\eproofbit % NESTED ONE
\doubleprooftrue
\setbox\proofbelow=\hbox\bgroup % NESTED TWO
\let\wereinproofbit\proofoverdbl
$\displaystyle
}%
\def\proofoverdots{% NESTED TWO
\eproofbit % NESTED ONE
\proofdotstrue
\setbox\proofbelow=\hbox\bgroup % NESTED TWO
\let\wereinproofbit\proofoverdots
$\displaystyle
}%
\def\proofusing{% NESTED TWO
\eproofbit % NESTED ONE
\setbox\proofrulename=\hbox\bgroup % NESTED TWO
\let\wereinproofbit\proofusing
\kern0.3em$
}

\def\endprooftree{% NESTED TWO
\eproofbit % NESTED ONE
  \dimen5 =0pt% spread of hypotheses
\dimen0=\wd\proofabove \advance\dimen0-\shortenproofleft
\advance\dimen0-\shortenproofright
\dimen1=.5\dimen0 \advance\dimen1-.5\wd\proofbelow
\dimen4=\dimen1
\advance\dimen1\proofbelowshift \advance\dimen4-\proofbelowshift
\ifdim  \dimen1<0pt
\then   \advance\shortenproofleft\dimen1
        \advance\dimen0-\dimen1
        \dimen1=0pt
        \ifdim  \shortenproofleft<0pt
        \then   \setbox\proofabove=\hbox{%
                        \kern-\shortenproofleft\unhbox\proofabove}%
                \shortenproofleft=0pt
        \fi
\fi
\ifdim  \dimen4<0pt
\then   \advance\shortenproofright\dimen4
        \advance\dimen0-\dimen4
        \dimen4=0pt
\fi
\ifdim  \shortenproofright<\wd\proofrulename
\then   \shortenproofright=\wd\proofrulename
\fi
\dimen2=\shortenproofleft \advance\dimen2 by\dimen1
\dimen3=\shortenproofright\advance\dimen3 by\dimen4
\ifproofdots
\then
        \dimen6=\shortenproofleft \advance\dimen6 .5\dimen0
        \setbox1=\vbox to\proofdotseparation{\vss\hbox{$\cdot$}\vss}%
        \setbox0=\hbox{%
                \advance\dimen6-.5\wd1
                \kern\dimen6
                $\vcenter to\proofdotnumber\proofdotseparation
                        {\leaders\box1\vfill}$%
                \unhbox\proofrulename}%
\else   \dimen6=\fontdimen22\the\textfont2 % height of maths axis
        \dimen7=\dimen6
        \advance\dimen6by.5\proofrulebreadth
        \advance\dimen7by-.5\proofrulebreadth
        \setbox0=\hbox{%
                \kern\shortenproofleft
                \ifdoubleproof
                \then   \hbox to\dimen0{%
                        $\mathsurround0pt\mathord=\mkern-6mu%
                        \cleaders\hbox{$\mkern-2mu=\mkern-2mu$}\hfill
                        \mkern-6mu\mathord=$}%
                \else   \vrule height\dimen6 depth-\dimen7 width\dimen0
                \fi
                \unhbox\proofrulename}%
        \ht0=\dimen6 \dp0=-\dimen7
\fi
\let\doll\relax
\ifwasinsideprooftree
\then   \let\VBOX\vbox
\else   \ifmmode\else$\let\doll=$\fi
        \let\VBOX\vcenter
\fi
\VBOX   {\baselineskip\proofrulebaseline \lineskip.2ex
        \expandafter\lineskiplimit\ifproofdots0ex\else-0.6ex\fi
        \hbox   spread\dimen5   {\hfi\unhbox\proofabove\hfi}%
        \hbox{\box0}%
        \hbox   {\kern\dimen2 \box\proofbelow}}\doll%
\global\dimen2=\dimen2
\global\dimen3=\dimen3
\egroup % NESTED ZERO
\ifonleftofproofrule
\then   \shortenproofleft=\dimen2
\fi
\shortenproofright=\dimen3
\onleftofproofrulefalse
\ifinsideprooftree
\then   \hskip.5em plus 1fil \penalty2
\fi
}

\newenvironment{bprooftree}
 {\begin{adjustbox}{raise=\depth}\begin{prooftree}}
     {\end{prooftree}\end{adjustbox}}
 
\newcommand{\emptyPremise}{\vphantom{{}^@}}
\newcommand{\indrulename}[1]{\textsc{#1}}
\newcommand{\indrule}[3]{
\ensuremath{
\begin{array}{c}
  \prooftree #2
    \justifies #3
    \thickness=0.05em
    \using \indrulename{#1}
  \endprooftree
\end{array}}}

\newcommand{\indih}[1]{
  \begin{array}{c}
    \text{\ih}
  \\
    #1
  \end{array}
}

\newcommand{\indruledbl}[3]{
\ensuremath{
\begin{array}{c}
  \prooftree #2
    \Justifies #3
    \thickness=0.05em
    \using \indrulename{#1}
  \endprooftree
\end{array}}}

\newcommand{\stkout}[1]{\ifmmode{\text{\sout{\ensuremath{#1}}}}\else{\sout{#1}}\fi}

\renewcommand{\emptyset}{\varnothing}
\newcommand{\Nat}{\mathbb{N}}
\newcommand{\NotNow}[1]{}

\newcommand{\HS}{\quad}
\newcommand{\ST}{\ |\ }
\newcommand{\ie}{{\em i.e.}\xspace}
\newcommand{\eg}{{\em e.g.}\xspace}
\newcommand{\Eg}{{\em E.g.}\xspace}
\newcommand{\etal}{et al.\xspace}
\newcommand{\cf}{{\em cf.}\xspace}
\newcommand{\ih}{IH\xspace}
\newcommand{\set}[1]{\{#1\}}
\newcommand{\eqdef}{:=}%\,\mathrel{\overset{\mathrm{def}}{=}}\,}

\newcommand{\fv}[1]{\mathsf{fv}(#1)}
\newcommand{\flv}[1]{\mathsf{flv}(#1)}

\newcommand{\defn}[1]{\textbf{#1}}

\newcommand{\llem}[1]{\label{lemma:#1}}
\newcommand{\rlem}[1]{Lem.~\ref{lemma:#1}}
\newcommand{\ldef}[1]{\label{def:#1}}
\newcommand{\rdef}[1]{Def.~\ref{def:#1}}
\newcommand{\lprop}[1]{\label{prop:#1}}

\newcommand{\lthm}[1]{\label{thm:#1}}
\newcommand{\rthm}[1]{Thm.~\ref{thm:#1}}
\newcommand{\lremark}[1]{\label{remark:#1}}

\newcommand{\lsec}[1]{\label{section:#1}}

\renewcommand{\arg}{(\cdot)}

\newcommand{\MLL}{\textsf{\textup{MLL}}\xspace}
\newcommand{\IMLL}{\textsf{\textup{IMLL}}\xspace}
\newcommand{\MELL}{\textsf{\textup{MELL}}\xspace}
\newcommand{\LL}{\textsf{\textup{LL}}\xspace}
\newcommand{\CalcMLL}{\ensuremath{\lambda_{\MLL}}\xspace}
\newcommand{\CalcMELL}{\ensuremath{\lambda_{\MELL}}\xspace}

\newcommand{\pretome}{\to_{\mathsf{\CalcMELL}}^{-}}
\newcommand{\tome}{\to_{\mathsf{\CalcMELL}}}
\newcommand{\rttome}{\twoheadrightarrow_{\mathsf{\CalcMELL}}}
\newcommand{\rtpretome}{\twoheadrightarrow_{\mathsf{\CalcMELL}}^{-}}
\newcommand{\toax}[1]{\mapsto_{#1}}

\newcommand{\eqme}{\doteq_{\mathsf{\CalcMELL}}}

\newcommand{\redrulename}[1]{\mathsf{#1}}
\newcommand{\raxlamL}{\redrulename{{\beta}{\lambda}{L}}}
\newcommand{\raxlamR}{\redrulename{{\beta}{\lambda}{R}}}
\newcommand{\raxparL}{\redrulename{{\beta}{\parr}{L}}}
\newcommand{\raxparR}{\redrulename{{\beta}{\parr}{R}}}
\newcommand{\raxtensor}{\redrulename{{\beta}{\tensor}}} 
\newcommand{\raxofc}{\redrulename{{\beta}{\ofc{}}}}
\newcommand{\raxwhy}{\redrulename{{\beta}{\why{}}}}
\newcommand{\raxpartensor}{\redrulename{{\parr}{\tensor}}}
\newcommand{\raxtensorpar}{\redrulename{{\tensor}{\parr}}}
\newcommand{\raxofcwhy}{\redrulename{{\ofc{}}{\why{}}}}
\newcommand{\raxwhyofc}{\redrulename{{\why{}}{\ofc{}}}}

\newcommand{\jull}[1]{\vdash{#1}}

\newcommand{\jum}[3]{#1\vdash#2:#3}
\newcommand{\djum}[4]{#1;#2\vdash#3:#4}
\newcommand{\djumnt}[3]{#1;#2\vdash#3}

\newcommand{\emptytenv}{\cdot}
\newcommand{\tenv}{\Gamma}
\newcommand{\tenvtwo}{\Gamma'}
\newcommand{\tenvthree}{\Gamma''}

\newcommand{\emptyutenv}{\cdot}
\newcommand{\utenv}{\Delta}

\newcommand{\lvar}{a}
\newcommand{\lvartwo}{b}
\newcommand{\lvarthree}{c}
\newcommand{\lvarfour}{d}

\newcommand{\uvar}{u}
\newcommand{\uvartwo}{v}
\newcommand{\uvarthree}{w}

\newcommand{\var}{x}
\newcommand{\vartwo}{y}

\newcommand{\tm}{t}
\newcommand{\tmtwo}{s}
\newcommand{\tmthree}{r}
\newcommand{\tmfour}{p}

\newcommand{\sub}[2]{\textcolor{blue}{\bm{\{}}#1:=#2\textcolor{blue}{\bm{\}}}}
\renewcommand{\cos}[2]{\textcolor{red}{\bm{\{}}#1\backslash\!\!\backslash#2\textcolor{red}{\bm{\}}}}

\newcommand{\btyp}{\alpha}
\newcommand{\btyptwo}{\beta}

\newcommand{\nbtyp}{\overline{\btyp}}
\newcommand{\nbtyptwo}{\overline{\btyptwo}}

\newcommand{\typ}{A}
\newcommand{\typtwo}{B}
\newcommand{\typthree}{C}
\newcommand{\typfour}{D}

\newcommand*{\boldone}{\text{\usefont{U}{bbold}{m}{n}1}}

\newcommand{\bott}{{\perp}}
\newcommand{\one}{\boldone}
\newcommand{\tensor}{\otimes}
\newcommand{\limp}{\multimap}
\newcommand{\lneg}[1]{#1^{\perp}}
\newcommand{\ofc}[1]{{!}{#1}}
\newcommand{\why}[1]{{?}#1}

\newcommand{\unit}{\star}
\newcommand{\lam}[2]{\lambda#1.\,#2}
\newcommand{\pair}[2]{\langle#1,#2\rangle}
\newcommand{\epair}[3]{\textcolor{red}{\bm{[}}\pair{#1}{#2}:=#3\textcolor{red}{\bm{]}}}
\newcommand{\casepair}[4]{{#4}\epair{#2}{#3}{#1}}
\newcommand{\ap}[2]{#1\,\textcolor{blue}{\bullet}\,#2}
\newcommand{\invap}[2]{#1\textcolor{red}{\blacktriangledown}#2}
\newcommand{\iofctwo}[2]{\ofc{#1.#2}}

\newcommand{\ione}{\star}
\newcommand{\eone}[1]{\textcolor{red}{\bm{[}}\star:=#1\textcolor{red}{\bm{]}}}

\newcommand{\rullOne}{\indrulename{l-$\one$}}
\newcommand{\rullBott}{\indrulename{l-$\bott$}}
\newcommand{\rullAx}{\indrulename{l-ax}}
\newcommand{\rullCut}{\indrulename{l-cut}}
\newcommand{\rullTensor}{\indrulename{l-$\tensor$}}
\newcommand{\rullLimp}{\indrulename{l-$\limp$}}
\newcommand{\rullPar}{\indrulename{l-$\parr$}}
\newcommand{\rullP}{\indrulename{l-!p}}
\newcommand{\rullW}{\indrulename{l-?w}}
\newcommand{\rullD}{\indrulename{l-?d}}
\newcommand{\rullC}{\indrulename{l-?c}}
\newcommand{\rulmAx}{\indrulename{m-ax}}
\newcommand{\rulmITensor}{\indrulename{m-i$\tensor$}}
\newcommand{\rulmETensor}{\indrulename{m-e$\tensor$}}
\newcommand{\rulmILimp}{\indrulename{m-i$\limp$}}
\newcommand{\rulmELimpOne}{\indrulename{m-e$\limp_1$}}
\newcommand{\rulmELimpTwo}{\indrulename{m-e$\limp_2$}}
\newcommand{\rulmIPar}{\indrulename{m-i$\parr$}}
\newcommand{\rulmEParOne}{\indrulename{m-e$\parr_1$}}
\newcommand{\rulmEParTwo}{\indrulename{m-e$\parr_2$}}
\newcommand{\rulmIOne}{\indrulename{m-i$\one$}}
\newcommand{\rulmEOne}{\indrulename{m-e$\one$}}
\newcommand{\rulmAxU}{\indrulename{m-uax}}
\newcommand{\rulmIOfc}{\indrulename{m-i$\ofc{}$}}
\newcommand{\rulmEOfc}{\indrulename{m-e$\ofc{}$}}
\newcommand{\rulmEWhy}{\indrulename{m-e$\why{}$}}

\newcommand{\rulmSub}{\indrulename{m-sub}}
\newcommand{\rulmUSub}{\indrulename{m-Usub}}
\newcommand{\rulmCos}{\indrulename{m-cos}}

\newcommand{\size}[1]{\mathsf{sz}(#1)}

\newcommand{\cceq}{\equiv}
\newcommand{\cceqalt}{\cceq^{\star}}

\newcommand{\seq}[1]{\overline{#1}}

\newcommand{\cctxhole}{\Box}
\newcommand{\cctx}{\mathsf{L}}
\newcommand{\cctxtwo}{\mathsf{K}}

\newcommand{\of}[2]{{#1}\bm{\langle}{#2}\bm{\rangle}}
\newcommand{\dom}[1]{\mathsf{dom}(#1)}

\newcommand{\ctxhole}{\Box}
\newcommand{\ctx}{\mathsf{C}}

\newcommand{\ruleCCEqOneLUnitForBott}{\mathsf{Lunit\_for\_\ibott{}{}}}

\newcommand{\ruleCCCosOneSimplLeft} {\mathsf{symm}}

\newcommand{\ruleCCEqBetaOne} {\beta\mathsf{\ione}}
\newcommand{\ruleCCEqCosAndSub} {\mathsf{split}}

\newcommand{\ruleCCEqCtx}{\mathsf{Prop}}
\newcommand{\ruleCCEqiBottSym}{\mathsf{Symm\_for\_\ibott{}{}}}
\newcommand{\ruleCCEqEOneSymm}{\mathsf{\mathsf{Symm}\ione}}

\newcommand{\ruleCCEqApDualToParElimTwo}{\mathsf{E\parr_1/E\parr_2}}
\newcommand{\ruleCCEqPairDualToParElimTwo}{\mathsf{I\tensor/E\parr_2}}

\newcommand{\ruleCCEqBottDualToEone}{\mathsf{I\perp/E\ione}}
\newcommand{\ruleCCEqIwhyDualToEofc}{\mathsf{I\why/E\ofc{}}}

\newcommand{\ruleCCEqIParDualToEPair}{\mathsf{I\parr/E\tensor}}

\newcommand{\ruleCCEqInvapDualToPair}{\mathsf{E\parr_2/I\tensor}}
\newcommand{\ruleCCEqInvapDualToAp}{\mathsf{E\parr_2/I\parr}}

\newcommand{\ruleCCEqIofcDualToEwhy}{\mathsf{I\ofc{}/E\why{}}}

\newcommand{\ruleCCEqEPairDualToAp}{\mathsf{E\tensor/E\parr_1}}

\newcommand{\rulmIBott}{\indrulename{m-i$\bott$}}

\newcommand{\rulmIWhy}{\indrulename{m-i$\why{}$}}
\newcommand{\iwhy}[2]{{?}#1.#2}
\newcommand{\eofc}[2]{\textcolor{red}{\bm{[}}\ofc{#1}:=#2\textcolor{red}{\bm{]}}}

\newcommand{\ewhy}[3]{{#1\textcolor{darkgreen}{\blacktriangleright}^{#2}#3}}

\newcommand{\ibott}[2]{#1\lightning#2}

\newcommand{\ipar}[3]{{\parr}(#1,#2).#3}

\newcommand{\linvap}[3]{#2\textcolor{red}{\blacktriangledown}\,#3}
\newcommand{\llam}[3]{\lambda #2.\,#3}

\newcommand{\lione}[1]{\star}

\newcommand{\patt}{\mathtt{p}}

\newcommand{\pelim}[2]{\textcolor{red}{\bm{[}}#1:=#2\textcolor{red}{\bm{]}}}

\newcommand{\ltm}{t}
\newcommand{\ltmtwo}{s}
\newcommand{\ltmthree}{r}
\newcommand{\ltmfour}{p}

\newcommand{\lab}{\alpha}
\newcommand{\labtwo}{\beta}

\newcommand{\rtto}{\twoheadrightarrow}
\newcommand{\lto}[1]{\to}
\newcommand{\lrtto}[1]{\rtto}

\newcommand{\step}{\sigma}
\newcommand{\steptwo}{\rho}

\newcommand{\mlnode}[1]{\begin{array}{@{}l@{}}#1\end{array}}

\newcommand{\fctx}{\textsf{S}}  % shallow
\newcommand{\dctx}{\textsf{P}}  % deep
\newcommand{\dctxtwo}{\textsf{Q}}  % deep
\newcommand{\rttofrommod}[1]{\leftrightarrow_{#1}^{*}}

\newcommand{\imp}{\supset}
\newcommand{\qtra}[1]{{#1}^{\mathsf{Q}}}
\newcommand{\qtrat}[2]{{#1}^{\mathsf{Q}}_{#2}}
\newcommand{\qtratwo}[1]{{#1}^{\mathsf{\underline{Q}}}}
\newcommand{\ttra}[1]{{#1}^{\mathsf{T}}}  % T translation for types
\newcommand{\ttrat}[2]{{#1}^{\mathsf{T}}_{#2}} % T translation for terms
\newcommand{\pjum}[4]{#1\vdash#2:#3\mid#4}

\newcommand{\emptynenv}{\cdot}
\newcommand{\nenv}{\Sigma}
\newcommand{\nenvtwo}{\Sigma'}

\newcommand{\ptm}{M}
\newcommand{\ptmtwo}{N}
\newcommand{\ptmthree}{O}
\newcommand{\ptmfour}{P}
\newcommand{\ptmfive}{Q}

\newcommand{\pval}{V}

\newcommand{\betav}{\beta_{V}}
\newcommand{\muv}{\mu_{V}}
\newcommand{\muvprime}{\mu_{V}'}

\newcommand{\mvar}{\alpha}

\newcommand{\mval}{v}
\newcommand{\menv}{e}
\newcommand{\mcmd}{c}
\newcommand{\mapp}[2]{#1\cdot#2}
\newcommand{\mmu}[2]{\mu{#1}.\,#2}
\newcommand{\mabs}[2]{\langle#1 \mid #2\rangle}
\newcommand{\mjcmd}[2]{#1\vdash#2}
\newcommand{\mjenv}[3]{#1\mid#2\vdash#3}
\newcommand{\mjval}[3]{#1\vdash#2\mid#3}
\newcommand{\mctx}{\Gamma}
\newcommand{\mctxtwo}{\Delta}

\newcommand{\nvar}{a}
\newcommand{\nvartwo}{b}

\newcommand{\pap}[2]{#1\,\,#2}
\newcommand{\pmu}[2]{\mu {#1}.#2}
\newcommand{\pname}[2]{[#1]#2}

\newcommand{\rulpAx}{\indrulename{p-ax}}
\newcommand{\rulpLam}{\indrulename{p-lam}}
\newcommand{\rulpApp}{\indrulename{p-app}}

\newcommand{\rulpName}{\indrulename{p-name}}
\newcommand{\rulpMu}{\indrulename{p-mu}}

\newcommand{\CalcParigot}{\lambda\mu}
\newcommand{\CalcParigotVal}{\lambda\muv}

\newcommand{\psub}[2]{\textcolor{darkgreen}{\bm{\{}}#1{\lhd}#2\textcolor{darkgreen}{\bm{\}}}}
\newcommand{\psubtwo}[2]{\textcolor{darkgreen}{\bm{\{}}#1{\lhd^\star}#2\textcolor{darkgreen}{\bm{\}}}}
\newcommand{\rensub}[2]{\textcolor{deepcarrotorange}{\bm{\{}}#1{\curvearrowleft}#2\textcolor{deepcarrotorange}{\bm{\}}}}
\newcommand{\tomu}{\to_{\mu}}

\newcommand{\htra}[1]{{#1}^{\mathsf{H}}}

\newcommand{\muDCLL}{\mu\textsf{DCLL}}
\newcommand{\ilam}[2]{\Lambda#1.\,#2}
\newcommand{\iap}[2]{#1\,\textcolor{blue}{{@}}\,#2}
\newcommand{\hjum}[5]{#1;#2\vdash#3:#4\mid#5}

\newcommand{\rulhlAx}{\indrulename{lin-ax}}
\newcommand{\rulhiAx}{\indrulename{int-ax}}
\newcommand{\rulhlLam}{\indrulename{$\limp$-I}}
\newcommand{\rulhiLam}{\indrulename{$\imp$-I}}
\newcommand{\rulhlApp}{\indrulename{$\limp$-E}}
\newcommand{\rulhiApp}{\indrulename{$\imp$-E}}
\newcommand{\rulhName}{\indrulename{$\bott$-I}}
\newcommand{\rulhMu}{\indrulename{$\bott$-E}}

\newcommand{\heq}{\doteq}

\newcommand{\vark}{k}
\newcommand{\varktwo}{k'}
\newcommand{\varkthree}{k''}

\newcommand*{\typtms}[1]{{\cal T}_{#1}}
\newcommand*{\iftyp}[1]{\vartriangleleft_{#1}}

\newcommand*{\llneg}[1]{\lneg{\lneg{ #1 }{}}}
\newcommand*{\lllneg}[1]{\lneg{\lneg{ \lneg{#1}{} }{}}}

\newcommand*{\SNtms}[1]{\textrm{SN}_{#1}}

\newcommand*{\tmsone}{X}
\newcommand*{\tmstwo}{Y}

\newcommand*{\subsj}[2]{\models #1, #2}   % substitution compatible with judgement
\newcommand{\subone}{\sigma}
\newcommand{\subtwo}{\nu}
\newcommand{\subthree}{\tau}

\newcommand*{\redcand}[1]{\llbracket #1 \rrbracket}
\newcommand{\tml}[1]{|#1|}

\newcommand{\proofScaleFactor}{0.9}

\title{A Classical Linear $\lambda$-Calculus based on Contraposition}

\author{Pablo Barenbaum}{Universidad Nacional de Quilmes (CONICET), Argentina \and Universidad de Buenos Aires, Argentina}{pbarenbaum@dc.uba.ar}{https://orcid.org/0009-0003-2494-3345}{}%TODO mandatory, please use full name; only 1 author per \author macro; first two parameters are mandatory, other parameters can be empty. Please provide at least the name of the affiliation and the country. The full address is optional. Use additional curly braces to indicate the correct name splitting when the last name consists of multiple name parts.

\author{Eduardo Bonelli}{Stevens Institute of Technology, United States}{ebonelli@stevens.edu}{0000-0003-1856-2856}{}

\author{Leopoldo Lerena}{Universidad de Buenos Aires (CONICET), Argentina}{leolerena@gmail.com}{[orcid]}{}

\authorrunning{J. Open Access and J.\,R. Public} %TODO mandatory. First: Use abbreviated first/middle names. Second (only in severe cases): Use first author plus 'et al.'

\Copyright{Jane Open Access and Joan R. Public} %TODO mandatory, please use full first names. LIPIcs license is "CC-BY";  http://creativecommons.org/licenses/by/3.0/


 %TODO mandatory; please add comma-separated list of keywords

\category{} %optional, e.g. invited paper

\relatedversion{} %optional, e.g. full version hosted on arXiv, HAL, or other respository/website
\EventEditors{John Q. Open and Joan R. Access}
\EventNoEds{2}
\EventLongTitle{42nd Conference on Very Important Topics (CVIT 2016)}
\EventShortTitle{CVIT 2016}
\EventAcronym{CVIT}
\EventYear{2016}
\EventDate{December 24--27, 2016}
\EventLocation{Little Whinging, United Kingdom}
\EventLogo{}
\SeriesVolume{42}
\ArticleNo{23}
\begin{document}

\maketitle

\begin{abstract}
  We present a novel linear $\lambda$-calculus for Classical
  Multiplicative Exponential Linear Logic (\MELL) along the lines of
  the propositions-as-types paradigm. Starting from the standard term
  assignment for Intuitionistic Multiplicative Linear Logic (\IMLL),
  we observe that if we incorporate linear negation, its involutive
  nature implies that both $\typ\limp\typtwo$ and
  $\lneg\typtwo\limp\lneg\typ$ should have the same proofs. The
  introduction of a linear modus tollens rule, stating that from 
  $\lneg\typtwo\limp\lneg\typ$ and $\typ$ we may conclude $\typtwo$,
  allows one to recover classical \MLL. Furthermore,
  a term assignment for this elimination rule,{the study of proof
  normalization in a $\lambda$-calculus with this elimination rule}
  prompts us to define the novel notion of
  contra-substitution $\tm\cos{\lvar}{\tmtwo}$. Introduced alongside
  linear substitution, contra-substitution denotes the term that
  results from ``grabbing'' the unique occurrence of $\lvar$ in $\tm$
  and ``pulling'' from it, in order to turn the term $\tm$ inside out
  (much like a sock) and then replacing $a$ with $\tmtwo$. We
  call the one-sided natural deduction presentation of classical \MLL,
  the \CalcMLL-calculus. Guided by the behavior of contra-substitution
  in the presence of the exponentials, we extend it to a similar
  presentation for \MELL. We prove that this calculus is sound and
  complete with respect to \MELL and that it satisfies the standard
  properties of a typed programming language: subject reduction,
  confluence and strong normalization. Moreover, we show that several
  well-known term assignments for classical logic can be encoded in
  $\CalcMELL$. These include Parigot's
  $\CalcParigot$~\mbox{\cite{DBLP:conf/lpar/Parigot92}} both via
  Danos, Joinet and Schellinx's T and
  Q-Translations~\mbox{\cite{Danos_Joinet_Schellinx_1995}},
  Curien and Herbelin's
  $\overline{\lambda}\mu\tilde{\mu}$-calculus~\mbox{\cite{DBLP:conf/icfp/CurienH00}},
  and Hasegawa's $\muDCLL$~\mbox{\cite{DBLP:journals/mscs/Hasegawa05,DBLP:conf/csl/Hasegawa02}}.
\end{abstract}

%\tableofcontents

\section{Introduction}

\newcommand{\CL}{\textsf{\textup{CL}}\xspace}
\newcommand{\ILL}{\textsf{\textup{ILL}}\xspace}
\newcommand{\DILL}{\textsf{\textup{DILL}}\xspace}
\newcommand{\LSP}{\ensuremath{\mathsf{LSP}}}
\newcommand{\LCSP}{\ensuremath{\mathsf{LCSP}}}
\newcommand{\outputf}{\ensuremath{{\cal O}}}
\newcommand{\inputf}{\ensuremath{{\cal I}}}

Linear Logic~($\LL$)~\cite{DBLP:journals/tcs/Girard87}
proposes a \emph{resource conscious} approach to logic,
in which formulae cannot be arbitrarily duplicated or erased.
It incorporates two modalities, called the \emph{exponential modalities},
that recover the ability to duplicate and erase formulae in a controlled way.
While the \emph{of-course} modality $\ofc{\typ}$
represents the \emph{obligation} to duplicate and erase the formula $\typ$
as many times as required,
the \emph{why-not} modality $\why{\typ}$
represents the \emph{right} to duplicate and erase the formula $\typ$
as many times as desired. 
This makes $\LL$ a suitable language to model resource-sensitive phenomena
such as concurrency, memory management, and computational complexity.
To explore the computational consequences of these ideas, one can turn to
the \emph{proposition-as-types paradigm}, which allows to interpret
propositions as types, proofs as programs,
and proof normalization as program evaluation.
The canonical example 
is the correspondence between minimal logic
and simply typed $\lambda$-calculus.
The proposition-as-types paradigm has been extended to many other logical systems,
and in particular to Linear Logic and various of its fragments.

\subparagraph*{Tensions Between Classical Symmetries and Lambda Calculi.} 
Linear Logic is known to exhibit numerous symmetries.
This is perhaps most evident in \emph{sequent calculus} presentations, % of $\LL$,
in which the left rules for a logical connective (\eg ``$\ofc{}$'')
are the mirror image of the right rules for the dual connective
(\eg ``$\why{}$'').
These kinds of symmetries are also present in Classical Logic~($\CL$),
but notably not in Intuitionistic Logic.
For instance, in $\CL$ and $\LL$,
affirming a proposition $\typ$ as a thesis corresponds
exactly to denying $\typ$ as a hypothesis---the
key behind the classical principle of \emph{reductio ad absurdum}.
Another related symmetry is expressed in the classical law of \emph{contraposition},
which states that $\typ\to\typtwo$ is equivalent to $\neg{\typtwo}\to\neg{\typ}$.
One essential aspect of this symmetry is that one would expect to be able
to define negation as an involutive connective, so that $\neg\neg\typ = \typ$
becomes a strict equality, which in particular means that a \emph{cut} between
a proof of $\typ$ and one of its negation $\neg{\typ}$ does not prioritize
either side. Compare this with the intuitionistic case, in which
$\neg{\typ}$ corresponds to a function type $\typ\to\bot$, but an arbitrary
type $\typtwo$ cannot be understood as a function type in general.
%~\cite[pg.29]{DBLP:journals/tcs/Abramsky93}. 

When a logic is presented in (single-conclusion) \emph{natural deduction} style,
the computational formalism one obtains is a variant of the
$\lambda$-calculus.
% This is desirable given the extensive body of
% results on the syntax and semantics of the lambda calculus that have been
% accumulated since its introduction by Church in the late 1930s. 
It is not obvious how to formulate systems in 
this style that exhibit classical symmetries
and still enjoy good computational properties such as canonicity.
For example, in the case of classical propositional logic,
Barbanera and Berardi's classical $\lambda$-calculus~\cite{DBLP:journals/iandc/BarbaneraB96}
is symmetric but non-confluent,
while Parigot's $\lambda\mu$-calculus~\cite{DBLP:conf/lpar/Parigot92}
is confluent but non-symmetric (\eg negation is not involutive\footnote{Not only is $\neg\neg\typ \neq \typ$, but $\neg\neg\typ \imp \typ$ is not provable~\cite[Cor.1]{DBLP:conf/icalp/AriolaH03}.}).
In the case of $\LL$, it is not clear how to formulate appropriate
natural deduction rules for logical connectives
that involve sequents with more than one formula on the right.
For example, the right introduction rule for multiplicative conjunction derives
$\tenv\vdash\typ\parr\typtwo$ from $\tenv\vdash\typ,\typtwo$,
and the right contraction rule for the why-not modality derives
$\tenv\vdash\why{\typ}$ from $\tenv\vdash\why{\typ},\why{\typ}$.

\subparagraph*{Intuitionistic and Sequent-Based Systems.} 
The difficulty to harmonize classical symmetries and lambda calculi
is perhaps the reason why %, apart from a few exceptions,
most of the existing literature that studies
$\LL$ from the point of view of the propositions-as-types paradigm
takes one of two routes.
On one hand, some systems such as $\ILL$ or $\DILL$~\cite{BarberPlotkin:1997,Barber:PhDThesis:1997}
restrict $\LL$ to the \textbf{intuitionistic} fragment.
Intuitionistic formulations of $\LL$ sidestep the aforementioned difficulties
by restricting the possible ``shapes'' of formulae and sequents,
in such a way that a judgement involves a number of \emph{input} formulae and a
single distinguished \emph{output} formula.
For instance, in $\ILL$, the multiplicative disjunction is removed in favor of
the (more restricted) linear implication, and the why-not modality
is removed altogether.
These restrictions allow to formulate well-behaved
intuitionistic linear lambda calculi, in which
the proof of a sequent $\typ_1,\hdots,\typ_n \vdash \typtwo$
is understood as a single sequential program taking $n$ inputs and
producing one output.
However, this is at the expense of losing the classical symmetries,
and ruling out a large class of formulae and proofs,
leaving the underlying computational mechanisms unexplored.

On the other hand, there are systems corresponding to unrestricted
(\ie classical) $\LL$, which are usually derived from \textbf{sequent calculus}
%\edu{Todo lo que sigue vale con DN con multiples formulas a la derecha, como el caso del paper Par Means Parallel de Aschieri y Genco }
presentations of $\LL$.
As already mentioned, sequent-based presentations are symmetric,
and cut elimination, in this setting, exhibits good properties such as strong
normalization and confluence.
Formulae do not play a distinguished input or output role: the proof of a sequent $\vdash \typ_1,\hdots,\typ_n$ is understood
as the parallel composition of $n$ interacting processes~\cite{DBLP:journals/tcs/Abramsky93,Wadler-PropositionsAsSessions,WADLER_2014}. Calculi based on these principles can thus be seen
as concurrent systems, more akin to process calculi than to
the $\lambda$-calculus.
Among these systems, one also finds graphical formalisms like \textbf{proof-nets}.
They are intentionally designed to abstract away the permutative rules
of sequent calculus. This allows reasoning at a higher level,
but at the same time it makes it difficult to reason axiomatically: for
example, to formalize proof-nets in a proof
assistant one needs to choose concrete terms as representatives of proof-nets,
and the need for explicit permutation rules reappears.

This leaves open the question of whether a single-conclusion
natural deduction system for (classical) \LL can be designed that
retains the classical symmetries and desirable computational properties.
Such a system, in which proofs can still be understood as ``functional''
programs, could form the basis of linear functional programming languages
and proof assistants.

\subparagraph*{Towards a Classical Linear $\lambda$-Calculus.}
Our starting point is a natural deduction system
with \emph{multiplicative conjunction} ($\tensor$)
and \emph{linear implication} ($\limp$) as the only connectives,
\ie corresponding to \emph{Multiplicative Linear Logic} ($\MLL$).
An involutive negation operator $\lneg{(\cdot)}$ can be defined essentially by the equation
$\lneg{(\typ\limp\typtwo)} = \typ\tensor\lneg{\typtwo}$,
and its dual.
The elimination rule for linear implication in the \emph{intuitionistic}
fragment of $\MLL$ is the linear \emph{modus ponens} rule, which states that
from $\typ \limp \typtwo$ and $\typ$ one may conclude $\typtwo$.
A first observation is that the only missing piece to recover
\emph{classical} $\MLL$ is to add the linear \emph{modus tollens} rule,
which states that
from $\typ \limp \typtwo$ and $\lneg{\typtwo}$ one may conclude $\lneg{\typ}$.
This results in an inference system that
---from the strictly logical point of view---
turns out to prove all and only the valid sequents of $\MLL$.
Also, it raises the question of how to provide a computational interpretation
for it, that is, to devise a proof normalization procedure.
The key to define the computational interpretation are the standard operation
of \emph{substitution} and, to the best of our knowledge, a new operation
we call \emph{contra-substitution}.

  \subparagraph*{The Linear Contra-Substitution Principle.}
\emph{Modus ponens} corresponds to
\emph{application}: a linear function $\tm : \typ \limp \typtwo$
may be applied to an argument $\tmtwo : \typ$
to yield a result $\ap{\tm}{\tmtwo} : \typtwo$.
From the computational point of view, 
a \emph{redex} formed by the interaction between a $\lambda$-abstraction
and an application can be normalized by means of the usual
$\beta$-rule, performing the simplification
$\ap{(\lam{\lvar}{\tm})}{\tmtwo} \rightsquigarrow \tm\sub{\lvar}{\tmtwo}$.
The right-hand side, $\tm\sub{\lvar}{\tmtwo}$, denotes the \emph{linear substitution}
of the (unique) free occurrence of $\lvar$ in $\tm$ by $\tmtwo$.
The \emph{modus tollens} rule corresponds to a comparatively less familiar
operation we dub \emph{contra-application}:
a linear function $\tm : \typ \limp \typtwo$
may be ``contra-applied'' to a term $\tmtwo : \lneg{\typtwo}$
to yield a result $\invap{\tm}{\tmtwo} : \lneg{\typ}$.
A \emph{redex} formed by the interaction between a $\lambda$-abstraction
and a contra-application can be normalized by means of a 
new rule that performs the simplification
$\invap{(\lam{\lvar}{\tm})}{\tmtwo} \rightsquigarrow \tm\cos{\lvar}{\tmtwo}$.
The right hand side, $\tm\cos{\lvar}{\tmtwo}$ corresponds to a non-standard
operation we dub \emph{contra-substitution}.
Intuitively, $\tm\cos{\lvar}{\tmtwo}$ is the expression that results
from ``grabbing'' the \emph{unique} occurrence of $\lvar$ in $\tm$ and ``pulling''
from it, in order to turn the term $\tm$ inside out.
This operation is defined by structural induction on $\tm$. For example,
if $\pair{\tm}{\lvar}:\typ\tensor\typtwo$ is a tensor pair,
and $\tmtwo : \typ\limp\lneg{\typtwo}$,
then $\pair{\tm}{\lvar}\cos{\lvar}{\tmtwo}$
turns out to produce the term $\ap{\tmtwo}{\tm}:\lneg{\typtwo}$ where $\tmtwo$
is applied to $\tm$.
Uniqueness of the occurrence of $\lvar$ (\ie linearity)
is crucial to be able to define this notion.

\subparagraph*{Summary of contributions.} We explore a propositions-as-types correspondence for a one-sided natural deduction presentation of \emph{Multiplicative Exponential Linear Logic} ($\MELL$) based on the novel notion of \emph{contra-substitution}, where proofs are modeled as functional expressions rather than processes.  The resulting \CalcMELL-calculus provides term witnesses for all proofs in \MELL, is strongly normalizing and confluent. It makes use of a notion of structural equivalence that is a strong bisimulation and we believe has an interest of its own. Several well-known term assignments for classical logic are shown to be simulated in \CalcMELL including Parigot's $\CalcParigot$~\mbox{\cite{DBLP:conf/lpar/Parigot92}}, Hasegawa's $\muDCLL$~\mbox{\cite{DBLP:journals/mscs/Hasegawa05,DBLP:conf/csl/Hasegawa02}} (in the appendix), and Curien and Herbelin's $\overline{\lambda}\mu\tilde{\mu}$-calculus~\mbox{\cite{DBLP:conf/icfp/CurienH00}}. Some proof skeletons and additional definitions are available in the appendix. See~\cite{contraposition_long} for full proofs.

%%% Local Variables:
%%% mode: latex
%%% TeX-master: "main"
%%% End:

\section{Preliminaries}

This section recalls \MLL and \MELL (with units) via the usual one-sided sequent presentation.

\medskip

\subparagraph*{Multiplicative Linear Logic.}
We assume given a denumerable set
of \emph{atomic formulae} $\btyp,\btyptwo,\hdots$
each with its corresponding negative version $\nbtyp,\nbtyptwo,\hdots$.
The set of \defn{$\MLL$-formulae} is given by:
\[
  \typ,\typtwo,\hdots ::=
           \btyp
      \mid \nbtyp
      \mid \typ\tensor\typtwo
      \mid \typ\limp\typtwo
\]
The logical connectives $\tensor$ and $\limp$ are called \emph{multiplicative}.
Linear negation is the involutive operator $\lneg{\arg}$ defined by:
\[
  \lneg{\btyp} \eqdef \nbtyp
  \HS\HS\HS
  \lneg{\nbtyp} \eqdef \btyp
  \HS\HS\HS
  \lneg{(\typ\tensor\typtwo)}
  \eqdef
  \typ\limp\lneg{\typtwo}
  \HS\HS\HS
  \lneg{(\typ\limp\typtwo)}
  \eqdef
  \typ\tensor\lneg{\typtwo}
\]
In \MLL, we take linear implication ($\limp$) as a primitive connective
  rather than multiplicative disjunction ($\parr$), which is the more conventional
  choice.
  This is just a minor presentational point, since $\limp$ and $\parr$
  are interdefinable; indeed, one can define $\typ\parr\typtwo \eqdef \lneg{\typ}\limp\typtwo$.

Judgements in $\MLL$ are of the form $\jull{\tenv}$
where $\tenv$ is a finite \emph{multiset} of \MLL formulae.
Note that working with multisets avoids the need of an explicit exchange rule.
We write $\tenv,\tenvtwo$ for the multiset union of finite multisets $\tenv$ and $\tenvtwo$.
Valid $\MLL$ judgement are defined inductively by the following rules.

\begin{definition}[Valid $\MLL$ judgement]
\[
  \indrule{\rullAx}{
    \emptyPremise
  }{
    \jull{\typ,\lneg{\typ}}
  }
  \indrule{\rullCut}{
    \jull{\tenv,\typ}
    \HS
    \jull{\tenvtwo,\lneg{\typ}}
  }{
    \jull{\tenv,\tenvtwo}
  }
  \indrule{\rullTensor}{
    \jull{\tenv,\typ}
    \HS
    \jull{\tenvtwo,\typtwo}
  }{
    \jull{\tenv,\tenvtwo,\typ\tensor\typtwo}
  }
  \indrule{\rullLimp}{
    \jull{\tenv,\lneg{\typ},\typtwo}
  }{
    \jull{\tenv,\typ\limp\typtwo}
  }
\]
\end{definition}
\medskip

%\subsection{\MELL}

\subparagraph*{Multiplicative Exponential Linear Logic.}
%We assume given a denumerable set
%of \emph{atomic formulae} $\btyp,\btyptwo,\hdots$
%each with its corresponding negative version $\nbtyp,\nbtyptwo,\hdots$.
The set of \defn{\MELL-formulae} is given by:
\[
  \typ,\typtwo,\hdots ::=
  \btyp
      \mid \nbtyp
      \mid \typ\tensor\typtwo
      \mid \typ\parr\typtwo
      \mid \one
      \mid \bott
      \mid \ofc{\typ}
      \mid \why{\typ}
\]
The connectives $\tensor$ and $\parr$,
as well as the units $\one$ and $\bott$ are called \emph{multiplicative},
while the modalities~$\ofc{}$ and $\why{}$ are called \emph{exponential}.
Linear negation is the involutive operator $\lneg{\arg}$ defined by:
\[
  \begin{array}{r@{\,\,}c@{\,\,}l@{\HS\HS}r@{\,\,}c@{\,\,}l@{\HS\HS}r@{\,\,}c@{\,\,}l@{\HS\HS}r@{\,\,}c@{\,\,}l}
    \lneg{\btyp} & \eqdef & \nbtyp
  &
    \lneg{\nbtyp} & \eqdef & \btyp
  &
    \lneg{(\typ\tensor\typtwo)}
    & \eqdef &
    \lneg{\typ}\parr\lneg{\typtwo}
  &
    \lneg{(\typ\parr\typtwo)}
    & \eqdef &
    \lneg{\typ}\tensor\lneg{\typtwo}
  \\
    \lneg{(\ofc{\typ})}
    & \eqdef &
    \why{\lneg{\typ}}
  &
    \lneg{(\why{\typ})}
    & \eqdef &
    \ofc{\lneg{\typ}}
  &
    \lneg{\one}
    & \eqdef & \bott
    &
      \lneg{\bott}
      & \eqdef & \one
  \end{array}
\]
Unlike in \MLL, we take multiplicative disjunction ($\parr$) as the primitive
connective in \MELL. Again, this is just a minor point to improve presentation,
and in \MELL one can define $\typ\limp\typtwo \eqdef \lneg{\typ}\parr\typtwo$.
Judgements are of the form $\jull{\tenv}$
where $\tenv$ is a finite multiset of \MELL formulae. 
Valid $\MELL$ judgement are defined inductively by the following rules.

\begin{definition}[Valid $\MELL$ judgement]
% We recall the definition of $\MELL$
% in one-sided sequent calculus style.
\[
  \indrule{\rullAx}{
    \emptyPremise
  }{
    \jull{\typ,\lneg{\typ}}
  }
  \indrule{\rullCut}{
    \jull{\tenv,\typ}
    \HS
    \jull{\tenvtwo,\lneg{\typ}}
  }{
    \jull{\tenv,\tenvtwo}
  }
% \]
% \[
  \indrule{\rullTensor}{
    \jull{\tenv,\typ}
    \HS
    \jull{\tenvtwo,\typtwo}
  }{
    \jull{\tenv,\tenvtwo,\typ\tensor\typtwo}
  }
  % \indrule{\rullLimp}{
  %   \jull{\tenv,\lneg{\typ},\typtwo}
  % }{
  %   \jull{\tenv,\typ\limp\typtwo}
  % }
  \indrule{\rullPar}{
    \jull{\tenv,\typ,\typtwo}
  }{
    \jull{\tenv,\typ\parr\typtwo}
  }
\]
\[
   \indrule{\rullOne}{
     \emptyPremise
  }{
     \jull{\one}
   }
      \indrule{\rullBott}{
     \jull{\tenv}
  }{
     \jull{\tenv,\bott}
  } 
%   \]
% \[
  \indrule{\rullP}{
    \jull{\why{\tenv},\typ}
  }{
    \jull{\why{\tenv},\ofc{\typ}}
  }
  \indrule{\rullW}{
    \jull{\tenv}
  }{
    \jull{\tenv,\why{\typ}}
  }
  \indrule{\rullD}{
    \jull{\tenv,\typ}
  }{
    \jull{\tenv,\why{\typ}}
  }
  \indrule{\rullC}{
    \jull{\tenv,\why{\typ},\why{\typ}}
  }{
    \jull{\tenv,\why{\typ}}
  }
\]
\end{definition}

%%% Local Variables:
%%% mode: latex
%%% TeX-master: "main"
%%% End:

\section{A Contraposition-Based Calculus for \MLL}
This section presents the \defn{\CalcMLL-calculus}, our propositions-as-types interpretation of \MLL.
The set of \defn{terms} and \defn{case contexts} are given by:
\[
  \begin{array}{llrll}
    (\text{Terms})
    & \tm, \tmtwo, \hdots
    & ::=  & \lvar \,\mid\, \pair{\tm}{\tmtwo}\,\mid\, \casepair{\tm}{\lvar}{\lvartwo}{\tmtwo}\,\mid\, \lam{\lvar}{\tm}\,\mid\, \ap{\tm}{\tmtwo}\,\mid\, \invap{\tm}{\tmtwo} \\ 
    (\text{Case contexts})
    & \cctx
    & ::=  & \cctxhole \,\mid\, \casepair{\tm}{\lvar}{\lvartwo}{\cctx}
  \end{array}
\]
Terms include standard constructs of linear $\lambda$-calculus: linear variables $\lvar, \lvartwo, \hdots$, ranging over a countably infinite set; tensor introduction $\pair{\tm}{\tmtwo}$ and elimination $\casepair{\tmtwo}{\lvar}{\lvartwo}{\tm}$; lambda abstraction $\lam{\lvar}{\tm}$ and application $\ap{\tm}{\tmtwo}$. Moreover, they include a novel constructor: the \defn{contra-application} $\invap{\tm}{\tmtwo}$.
From the logical point of view, application and contra-application both correspond
to eliminations of the implication.
As mentioned before, while application corresponds to modus ponens, contra-application corresponds to modus tollens.
Free and bound occurrences of variables are defined as expected,
where $\casepair{\tmtwo}{\lvar}{\lvartwo}{\tm}$
binds free occurrences of $\lvar,\lvartwo$ in $\tm$,
and $\lam{\lvar}{\tm}$ binds free occurrences of $\lvar$ in $\tm$.
Terms are considered up to $\alpha$-renaming of bound variables.

Case contexts are sequences of tensor eliminations with a subterm hole ($\cctxhole$) that can be filled by any given \CalcMLL term.  
We write $\of{\cctx}{\tm}$ (or also $\tm\cctx$) for the term resulting from replacing the unique occurrence of the hole in $\cctx$ with $\tm$,
possibly capturing free variables in $\tm$ in the process.
We shall use case contexts for reduction at a distance.

Typing judgements are of the form $\jum{\tenv}{\tm}{\typ}$
where $\typ$ is a \MLL formula
and $\tenv$ is a \emph{typing environment},
a partial function mapping variables to \MLL formulae,
written $\tenv = (\lvar_1:\typ_1,\hdots,\lvar_n:\typ_n)$
and assumed to be of finite domain.
Valid judgements are given by:

\begin{definition}[Valid \CalcMLL typing judgements]
\quad\\
\scalebox{\proofScaleFactor}{\begin{minipage}{\textwidth}
    \[
  \indrule{\rulmAx}{
    \emptyPremise
  }{
    \jum{\lvar:\typ}{\lvar}{\typ}
  }
  \indrule{\rulmITensor}{
    \jum{\tenv}{\tm}{\typ}
    \HS
    \jum{\tenvtwo}{\tmtwo}{\typtwo}
  }{
    \jum{\tenv,\tenvtwo}{\pair{\tm}{\tmtwo}}{\typ\tensor\typtwo}
  }
  \indrule{\rulmETensor}{
    \jum{\tenv}{\tm}{\typ\tensor\typtwo}
    \HS
    \jum{\tenvtwo,\lvar:\typ,\lvartwo:\typtwo}{\tmtwo}{\typthree}
  }{
    \jum{\tenv,\tenvtwo}{\casepair{\tm}{\lvar}{\lvartwo}{\tmtwo}}{\typthree}
  }
\]
\[
  \indrule{\rulmILimp}{
    \jum{\tenv,\lvar:\typ}{\tm}{\typtwo}
  }{
    \jum{\tenv}{\lam{\lvar}{\tm}}{\typ\limp\typtwo}
  }
  \indrule{\rulmELimpOne}{
    \jum{\tenv}{\tm}{\typ\limp\typtwo}
    \HS
    \jum{\tenvtwo}{\tmtwo}{\typ}
  }{
    \jum{\tenv,\tenvtwo}{\ap{\tm}{\tmtwo}}{\typtwo}
  }
  \indrule{\rulmELimpTwo}{
    \jum{\tenv}{\tm}{\typ\limp\typtwo}
    \HS
    \jum{\tenvtwo}{\tmtwo}{\lneg{\typtwo}}
  }{
    \jum{\tenv,\tenvtwo}{\invap{\tm}{\tmtwo}}{\lneg{\typ}}
  }
\]
\end{minipage}}
\end{definition}

Rules are presented in introduction/elimination pairs style of natural deduction, where the most noteworthy point is that linear implication admits two different elimination rules.
The first, \rulmELimpOne{} (modus ponens) corresponds to the standard application of terms from $\lambda$-calculus.
The second, \rulmELimpTwo{} (modus tollens) is the typing rule of our new construct, contra-application. 
This latter is the only strictly classical rule in our typing system.
The following result justifies that the \CalcMLL calculus corresponds
to \MLL from the point of view of propositions-as-types.
The statement of the lemma relies on an abuse of notation:
if $\tenv$ is a multiset of formulae $\tenv = (\typ_1,\hdots,\typ_n)$,
we write $\jum{\lneg{\tenv}}{\tm}{\typtwo}$
to mean $\jum{\lvar_1:\lneg{\typ_1},\hdots,\lvar_n:\lneg{\typ_n}}{\tm}{\typtwo}$
where $\lvar_1,\hdots,\lvar_n$ is a set of distinct linear variables.

\begin{lemma}[{name=Soundness and Completeness of \CalcMLL~\proofnote{Proof on pg.~\pageref{compl_sound_MLL:proof}}, restate=[name=Completeness and Soundness of  \CalcMLL]CompletenessSoundnessMLL}]
\hfill
\llem{compl_sound_MLL}
\begin{itemize}
\item
  \textbf{\textup{Soundness.}}
  If $\jum{\tenv}{\tm}{\typ}$ holds in \CalcMLL
  then $\jull{\lneg{\tenv},\typ}$ holds in \MLL.
\item
  \textbf{\textup{Completeness.}}
  If $\jull{\tenv_0}$ holds in \MLL
  and $\tenv,\typ$ is any permutation of $\tenv_0$,
  there exists a term $\tm$ such that
  $\jum{\lneg{\tenv}}{\tm}{\typ}$
  holds in \CalcMLL.
\end{itemize}
\end{lemma}
Note that if a sequent $\jull{\typ_1,\hdots,\typ_n}$ holds in \MLL,
completeness allows to select \emph{any} of the $\typ_i$ as the thesis,
leaving the negation of the remaining formulae
($\lneg{\typ_1},\hdots,\lneg{\typ_{i-1}},\lneg{\typ_{i+1}},\hdots,\lneg{\typ_n}$)
as the hypotheses.
For example, the valid \MLL sequent $\jull{\lneg{\typ},\lneg{\typtwo},\typ\tensor\typtwo}$
has three possible ``readings'' in $\CalcMLL$, all of them valid.
Completeness simultaneously ensures the existence of three terms $\tm_1$, $\tm_2$, $\tm_3$ such that
$\jum{\lvar:\typtwo,\lvartwo:\typ\limp\lneg{\typtwo}}{\tm_1}{\lneg{\typ}}$
and $\jum{\lvar:\typ,\lvartwo:\typ\limp\lneg{\typtwo}}{\tm_2}{\lneg{\typtwo}}$
and $\jum{\lvar:\typ,\lvartwo:\typtwo}{\tm_3}{\typ\tensor\typtwo}$.
Moreover, the proof of completeness is constructive, so
from the derivation in \MLL the terms $\tm_1$, $\tm_2$, and $\tm_3$
can be effectively recovered.

\begin{remark}
Completeness fails if $\rulmELimpTwo$ is absent.
%Consider the following proof in $\MLL$:
%\scalebox{\proofScaleFactor}{\begin{minipage}{\textwidth}
%\[
%  \indrule{}
%  {
%  \indrule{}
%  {
%     \jull{\typ,\lneg{\typ}}
%     \quad
%     \indrule{}
%     {\jull{\typthree,\lneg{\typthree}}}
%     {\jull{\typthree \limp \typthree}}
%   }
%   {\jull{\typ\tensor(\typthree \limp \typthree),\lneg{\typ}}}
%   \quad
%   \indrule{}
%   {\jull{\typtwo,\lneg{\typtwo}}}
%   {\jull{\typtwo\limp\typtwo}}
% }
% {\jull{\typ\tensor(\typthree \limp \typthree),\lneg{\typ}\tensor(\typtwo\limp\typtwo)}}
%\]
%\end{minipage}}
%\\
%\noindent
For example, is easy to see that the sequent
$\jull{\typ\tensor(\typthree \limp \typthree),\lneg{\typ}\tensor(\typtwo\limp\typtwo)}$
is valid in \MLL.
However, there is no term $\tm$ and variable $\lvar$ such that
$\jum{\lvar:\typ\limp(\typthree \tensor \lneg{\typthree})}{\tm}{\lneg{\typ}\tensor(\typtwo\limp\typtwo)}$ or
$\jum{\lvar:\lneg{\typ}\limp(\typtwo\tensor\typtwo)}{\tm}{\typ\tensor(\typthree \limp \typthree)}$ are derivable in \CalcMLL if $\rulmELimpTwo$ is absent.
Both are derivable in \CalcMLL. For example:
\scalebox{\proofScaleFactor}{\begin{minipage}{\textwidth}
\[
    \indrule{\rulmITensor}{
  \indrule{\rulmELimpTwo}{
    \jum{\lvar:\typ\limp(\typthree \tensor \lneg{\typthree})}{\lvar}{\typ\limp(\typthree \tensor \lneg{\typthree})}
    \HS
    \jum{\emptytenv}{\lam{\lvarthree}{\lvarthree}}{\typthree\limp\typthree}
  }{
    \jum{\lvar:\typ\limp(\typthree \tensor \lneg{\typthree})}{\invap{\lvar}{(\lam{\lvarthree}{\lvarthree})}}{\lneg{\typ}}
  }
  \HS
  \jum{\emptytenv}{\lam{\lvartwo}{\lvartwo}}{\typtwo\limp\typtwo}
  }
   {
    \jum{\lvar:\typ\limp(\typthree \tensor \lneg{\typthree})}{\pair{\invap{\lvar}{(\lam{\lvarthree}{\lvarthree})}}{\lam{\lvartwo}{\lvartwo}}}{\lneg{\typ}\tensor(\typtwo\limp\typtwo)}
  }
\]
\end{minipage}}
\end{remark}

\subparagraph*{Linear terms.}
A term is \emph{linear} if each free variable occurs exactly once, and there is exactly one occurrence of each bound variable inside the scope of its binder.
Typable terms are linear:
%Specifically,
If $\jum{\tenv}{\tm}{\typtwo}$, then $\tm$ is linear
and the variables in $\tenv$ are exactly the free variables of $\tm$.

\subparagraph*{Substitution and contra-substitution.}
Reduction semantics for the \CalcMLL-calculus relies on two substitution operations.
First, \emph{linear substitution} $\tm \sub{\lvar}{\tmtwo}$ denotes the standard capture-avoiding substitution of a linear variable $\lvar$ by a term $\tmtwo$ in $\tm$.
This operation is compatible with typing in the sense that if $ \jum{\tenv,\lvar:\typ}{\tm}{\typtwo}$ and $\jum{\tenvtwo}{\tmtwo}{\typ}$ hold in \CalcMLL, then so does $\jum{\tenv,\tenvtwo}{\tm\sub{\lvar}{\tmtwo}}{\typtwo}$.
% the following rule is admissible\footnote{Meaning that provability of all its premises implies provability of its conclusion.}:
%   \begin{equation*}
%        \indrule{\rulmSub}{
%         \jum{\tenv,\lvar:\typ}{\tm}{\typtwo}
%         \HS
%         \jum{\tenvtwo}{\tmtwo}{\typ}
%       }{
%         \jum{\tenv,\tenvtwo}{\tm\sub{\lvar}{\tmtwo}}{\typtwo}
%       }
%   \end{equation*}
Second, we introduce \emph{contra-substitution}, written $\tm \cos{\lvar}{\tmtwo}$.
%\begin{figure}
\begin{definition}[Contra-substitution for \CalcMLL]
\ldef{contrasub:for:MLL}
Let $\tm$ be a linear term and let $\lvar$ be a free linear variable
such that $\lvar \in \fv{\tm}$.
Let us write ``$\bm{\ast}_i$'' to abbreviate the condition $\lvar \in \fv{\tm_i}$.
The contra-substitution operation $\tm\cos{\lvar}{\tmtwo}$
is defined by induction on the structure of $\tm$:
  \[
    \begin{array}{r@{\,\,}c@{\,\,}l@{\HS}r@{\,\,}c@{\,\,}l}
    \multicolumn{6}{l}{
      \lvar\cos{\lvar}{\tmtwo}
      \eqdef
      \tmtwo
      \hspace{1cm}
      (\tm_1\epair{\lvartwo}{\lvarthree}{\tm_2})\cos{\lvar}{\tmtwo}
      \eqdef
      \begin{cases}
        \tm_1\cos{\lvar}{\tmtwo}\epair{\lvartwo}{\lvarthree}{\tm_2}
        & \text{if $\bm{\ast}_1$, $\lvar \notin\set{\lvartwo,\lvarthree}$}
      \\
        \tm_2\cos{\lvar}{\lam{\lvartwo}{\tm_1\cos{\lvarthree}{\tmtwo}}}
        & \text{if $\bm{\ast}_2$}
      \end{cases}
    }
    \\
      (\lam{\lvartwo}{\tm'})\cos{\lvar}{\tmtwo}
    & \eqdef &
      \casepair{\tmtwo}{\lvartwo}{\lvarthree}{\tm'\cos{\lvar}{\lvarthree}}
    &
      \pair{\tm_1}{\tm_2}\cos{\lvar}{\tmtwo}
    & \eqdef &
      \begin{cases}
        \tm_1\cos{\lvar}{\invap{\tmtwo}{\tm_2}}
        & \text{if $\bm{\ast}_1$}
      \\
        \tm_2\cos{\lvar}{\ap{\tmtwo}{\tm_1}}
        & \text{if $\bm{\ast}_2$}
      \end{cases}
    \\
      (\ap{\tm_1}{\tm_2})\cos{\lvar}{\tmtwo}
    & \eqdef &
      \begin{cases}
        \tm_1\cos{\lvar}{\pair{\tm_2}{\tmtwo}}
        & \text{if $\bm{\ast}_1$}
      \\
        \tm_2\cos{\lvar}{\invap{\tm_1}{\tmtwo}}
        & \text{if $\bm{\ast}_2$}
      \end{cases}
    &
      (\invap{\tm_1}{\tm_2})\cos{\lvar}{\tmtwo}
    & \eqdef &
      \begin{cases}
        \tm_1\cos{\lvar}{\pair{\tmtwo}{\tm_2}}
        & \text{if $\bm{\ast}_1$}
      \\
        \tm_2\cos{\lvar}{\ap{\tm_1}{\tmtwo}}
        & \text{if $\bm{\ast}_2$}
      \end{cases}
    \end{array}
  \]
\end{definition}
%\caption{Contra-substitution for \CalcMLL.}\label{contrasub:for:MLL}
%\end{figure}
% \end{definition}
For example, $(\lam{\lvartwo}{\ap{\lvar}{\lvartwo}}) \cos{\lvar}{\tmtwo} = (\ap{\lvar}{\lvartwo}) \cos{\lvar}{\lvarthree} \epair{\lvartwo}{\lvarthree}{\tmtwo} = \lvar \cos{\lvar}{\pair{\lvartwo}{\lvarthree}} \epair{\lvartwo}{\lvarthree}{\tmtwo} = \pair{\lvartwo}{\lvarthree} \epair{\lvartwo}{\lvarthree}{\tmtwo}$.
Contra-substitution too is compatible with typing:

\begin{lemma}[{name=Type Compatibility of Contra-substitution for MLL~\proofnote{Proof on pg.~\pageref{contrasubstitution_lemma_MLL:proof}}, restate=[name=Type Compatibility of Contrasubstitution for MLL]ContrasubstitutionLemmaMLL}]
\llem{contrasubstitution_lemma_MLL}
If $\jum{\tenv,\lvar:\typ}{\tm}{\typtwo}$ and $\jum{\tenvtwo}{\tmtwo}{\lneg{\typtwo}}$ hold in \CalcMLL, then so does $ \jum{\tenv,\tenvtwo}{\tm\cos{\lvar}{\tmtwo}}{\lneg{\typ}}$.
\end{lemma}

\subparagraph*{Reduction in $\CalcMLL$.}
The reduction relation $\to$ in the $\CalcMLL$-calculus is defined over
\textbf{linear terms}, as the contextual closure of the three following reduction rules:
\[
  \begin{array}{llll}
    \ap{(\lam{\lvar}{\tm})\cctx}{\tmtwo}
  & \toax{\raxlamL} &
    \tm\sub{\lvar}{\tmtwo}\cctx
  \\
    \invap{(\lam{\lvar}{\tm})\cctx}{\tmtwo}
  & \toax{\raxlamR} &
    \tm\cos{\lvar}{\tmtwo}\cctx
  \\
    \casepair{\pair{\tmtwo}{\tmthree}\cctx}{\lvar}{\lvartwo}{\tm}
  & \toax{\raxtensor} &
    \tm\sub{\lvar}{\tmtwo}\sub{\lvartwo}{\tmthree}\cctx
    & \text{($\lvartwo \notin \fv{\tmtwo}$)}
  \end{array}
\]
Note that the restriction over linear terms is necessary
for the right-hand side of the rule $\toax{\raxlamR}$ to be well-defined, as contra-substitution is
not defined for arbitrary terms.
Due to the presence of the case contexts, these rules are said to
operate ``at a distance''. For example,
$\ap{(\lam{\lvar}{\tm})\epair{\lvartwo}{\lvarthree}{\tmthree}}{\tmtwo}$ is a redex,
even though the $\lambda$-abstraction and the application are not directly
in ``contact'' with each other.

\subparagraph*{Strong normalization and confluence.}
We define the \emph{size} of a \emph{linear} term $\tm$, denoted $\size{\tm}$, as the number of term constructors in $\tm$.
Assuming that $\lvar\in\fv{\tm}$, this measure satisfies the following two properties: 
$\size{\tm \sub{\lvar}{\tmtwo}} = \size{\tm} + \size{\tmtwo}$ and $\size{\tm \cos{\lvar}{\tmtwo}} = \size{\tm} + \size{\tmtwo}$,
due to linearity.
It is then easy to see that every reduction step strictly decreases the size of the term;
hence the $\CalcMLL$-calculus is strongly normalizing.
By Newman's Lemma, confluence is then reduced to local confluence (\ie the weak Church--Rosser property),
which can be proved through a systematic case analysis of all coinitial steps.

\begin{proposition}
  $\CalcMLL$ is strongly normalizing and confluent.
\end{proposition}
%%% Local Variables:
%%% mode: latex
%%% TeX-master: "main"
%%% End:

\section{A Contraposition-Based Calculus for \MELL}

This section presents $\CalcMELL$, our propositions-as-types calculus based on contra-substitution for \MELL.
The set of \defn{terms} and \defn{positive elimination contexts} are given by:

%\begin{definition}[The $\CalcMELL$-Calculus]
% \ldef{calcMELL_terms}  Let $\lvar,\lvartwo,\ldots$ range over a countable set of \defn{linear variables} and $\uvar,\uvartwo,\ldots$ over a countable  set of \defn{unrestricted variables}. Terms are given by:
  \[
    \begin{array}{llrll}
       (\text{Terms}) & \tm,\tmtwo,\ldots & ::= &   \lvar\,|\,\uvar\,|\,\pair{\tm}{\tmtwo}\,|\,\tm\epair{\lvar}{\lvartwo}{\tmtwo}\,|\,\ipar{\lvar}{\lvartwo}{\tm}\,|\,\ap{\tm}{\tmtwo}\,|\,\invap{\tm}{\tmtwo} \\
          & & \mid & \iofctwo{\lvar}{\tm}\,|\,\tm\eofc{\uvar}{\tmtwo}\,|\,\iwhy{\uvar}{\tm}\,|\,\ewhy{\tm}{\lvar}{\tmtwo}\,|\,\ibott{\tm}{\tmtwo}\,|\,\ione\,|\,\tm\eone{\tmtwo} \\
%     \cctx & ::= & \ctxhole \,|\, \cctx\epair{\lvar}{\lvartwo}{\tm}\,|\,\cctx\eofc{\uvar}{\tm} \,|\, \cctx\ewhynew{\lvar}{\tm}{\tmtwo} \,|\,\cctx \eone{\tm}
                    (\text{Positive elim. contexts}) & \cctx & ::= & \ctxhole \,|\, \cctx\epair{\lvar}{\lvartwo}{\tm}\,|\,\cctx\eofc{\uvar}{\tm} \,|\, \cctx \eone{\tm}
    \end{array}
    \]
%\end{definition}
Terms include: linear variables $\lvar,\lvartwo,\ldots$ ranging over a countably infinite set; unrestricted variables $\uvar,\uvartwo,\ldots$ ranging over  a countably infinite set; tensor introduction $\pair{\tm}{\tmtwo}$ and elimination $\tm\epair{\lvar}{\lvartwo}{\tmtwo}$; par introduction $\ipar{\lvar}{\lvartwo}{\tm}$; application $\ap{\tm}{\tmtwo}$; contra-application $\invap{\tm}{\tmtwo}$; of-course introduction $\iofctwo{\lvar}{\tm}$ and elimination $\tm\eofc{\uvar}{\tmtwo}$; why-not introduction $\iwhy{\uvar}{\tm}$ and elimination $\ewhy{\tm}{\lvar}{\tmtwo}$; contradiction $\ibott{\tm}{\tmtwo}$; unit introduction $\ione$ and elimination $\tm\eone{\tmtwo}$. Free and bound occurrences of variables are defined as expected,
where $\tm\epair{\lvar}{\lvartwo}{\tmtwo}$ and $\ipar{\lvar}{\lvartwo}{\tm}$
binds free occurrences of $\lvar,\lvartwo$ in $\tm$;
$\iofctwo{\lvar}{\tm}$ and $\ewhy{\tm}{\lvar}{\tmtwo}$ bind free occurrences of $\lvar$ in $\tm$;
and $\tm\eofc{\uvar}{\tmtwo}$ and $\iwhy{\uvar}{\tm}$ bind free occurrences of $\uvar$ in $\tm$. Terms are considered up to $\alpha$-renaming of bound variables.
% \emph{Positive eliminator contexts} are given by:
%   \[
%     \begin{array}{lll}
%     \cctx & ::= & \ctxhole \,|\, \cctx\epair{\lvar}{\lvartwo}{\tm}\,|\,\cctx\eofc{\uvar}{\tm} \,|\, \cctx \eone{\tm}
%     \end{array}
%   \]
We call expressions of the form
$\epair{\lvar}{\lvartwo}{\tm}$, $\eofc{\uvar}{\tm}$, and $\eone{\tm}$ 
\defn{positive eliminators}\footnote{This terminology stems from polarized linear logic.}.
Positive eliminators are collectively denoted by $\pelim{\patt}{\tm}$, where $\patt$ is called a \emph{pattern},
%We write $\pelim{\patt}{\tm}$ to
%abbreviate any positive eliminator and call $\patt$ its \defn{pattern} and
and
$\fv{\patt}$ is defined as: $\fv{\pair{\lvar}{\lvartwo}}=\{\lvar,\lvartwo\}$, $\fv{\ofc{\uvar}}=\{\uvar\}$,
and $\fv{\ione}=\emptyset$.
%Similarly for $\flv{\patt}$. For example,
%$\flv{\pair{\lvar}{\lvartwo}}=\{\lvar,\lvartwo\}$ and $\flv{\ofc{\uvar}}=\emptyset$. This is extended to positive eliminators by
%defining $\fpv{\pelim{\patt}{\tm}}\eqdef \fv{\patt}$.

\defn{Judgements} have the form $ \djum{\utenv}{\tenv}{\tm}{\typ}$, where $\typ$ is a \MELL formula, $\utenv$ (resp. $\tenv$) is an \defn{unrestricted (resp. linear) typing environment},
\ie a partial function mapping unrestricted (resp. linear) variables to types.
% , and is a \defn{linear typing environment}, a partial function mapping linear variables to types.
%If $\dom{\tenv}\cap\dom{\tenvtwo}=\emptyset$, then
%we write $\tenv,\tenvtwo$ for the linear typing environment
%$(\tenv,\tenvtwo)(\lvar)$ = $\tenv(\lvar)$, if $\lvar\in\dom{\tenv}$,
%and $(\tenv,\tenvtwo)(\lvar)$ = $\tenvtwo(\lvar)$, if
%$\lvar\in\dom{\tenvtwo}$; otherwise, $\tenv,\tenvtwo$ is
%undefined.
Derivable judgements are defined in Fig.~\ref{typing:rules:CalcMELL}.
% $\negtyp$ stands for $\why{\typthree}$ or $\bott$.
%\begin{definition}
%\ldef{calcMELL_typing_rules}
%A term $\tm$ is \defn{typable}, if there exists a typing judgement $ \djum{\utenv}{\tenv}{\tm}{\typ}$ that is derivable by the following typing rules:
%\end{definition}

\begin{figure}
\scalebox{\proofScaleFactor}{\begin{minipage}{\textwidth}
\[
\begin{array}{c}
  \indrule{\rulmAx}{
    \emptyPremise
  }{
    \djum{\utenv}{\lvar:\typ}{\lvar}{\typ}
  }
  \indrule{\rulmAxU}{
    \emptyPremise
  }{
    \djum{\utenv, \uvar:\typ}{\emptytenv}{\uvar}{\typ}
  }
  \indrule{\rulmITensor}{
    \djum{\utenv}{\tenv}{\tm}{\typ}
    \HS
    \djum{\utenv}{\tenvtwo}{\tmtwo}{\typtwo}
  }{
    \djum{\utenv}{\tenv,\tenvtwo}{\pair{\tm}{\tmtwo}}{\typ\tensor\typtwo}
  }
  \\
  \\
  \indrule{\rulmETensor}{
    \djum{\utenv}{\tenv}{\tmtwo}{\typ\tensor\typtwo}
    \HS
    \djum{\utenv}{\tenvtwo,\lvar:\typ,\lvartwo:\typtwo}{\tm}{\typthree}
  }{
    \djum{\utenv}{\tenv,\tenvtwo}{\casepair{\tmtwo}{\lvar}{\lvartwo}{\tm}}{\typthree}
  }
  \indrule{\rulmIPar}{
    \djum{\utenv}{\tenv,\lvar:\typ,\lvartwo:\typtwo}{\tm}{\bott}
  }{
    \djum{\utenv}{\tenv}{\ipar{\lvar}{\lvartwo}{\tm}}{\lneg{\typ}\parr\lneg{\typtwo}}
  }
  \\
  \\
  \indrule{\rulmEParOne}{
    \djum{\utenv}{\tenv}{\tm}{\typ\parr\typtwo}
    \HS
    \djum{\utenv}{\tenvtwo}{\tmtwo}{\lneg{\typ}}
  }{
    \djum{\utenv}{\tenv,\tenvtwo}{\ap{\tm}{\tmtwo}}{\typtwo}
  }
  \quad
  \indrule{\rulmEParTwo}{
    \djum{\utenv}{\tenv}{\tm}{\typ\parr\typtwo}
    \HS
    \djum{\utenv}{\tenvtwo}{\tmtwo}{\lneg{\typtwo}}
  }{
    \djum{\utenv}{\tenv,\tenvtwo}{\invap{\tm}{\tmtwo}}{\typ}
  }
  \\
  \\
  \indrule{\rulmIOfc}
  {\djum{\utenv}{\lvar:\typ}{\tm}{\bott}}
  {\djum{\utenv}{\emptytenv}{\iofctwo{\lvar}{\tm}}{\ofc{\lneg{\typ}}}}
  \indrule{\rulmEOfc}
  {\djum{\utenv}{\tenv}{\tmtwo}{\ofc{\typ}}
    \HS
    \djum{\utenv,\uvar:\typ}{\tenvtwo}{\tm}{\typtwo}
    }
  {\djum{\utenv}{\tenv,\tenvtwo}{\tm\eofc{\uvar}{\tmtwo}}{\typtwo}}
  \\
  \\
  \indrule{\rulmIWhy}{
    \djum{\utenv,\uvar:\lneg{\typ}}{\tenv}{\tm}{\bott}
  }{
    \djum{\utenv}{\tenv}{\iwhy{\uvar}{\tm}}{\why{\typ}}
  }
  \quad
      \indrule{\rulmEWhy}{
    \djum{\utenv}{\tenv}{\tmtwo}{\why{\typ}}
    \HS
  \djum{\utenv}{\lvar:\typ}{\tm}{\bott}
  }{
  \djum{\utenv}{\tenv}{\ewhy{\tm}{\lvar}{\tmtwo}}{\bott}
  }
  \\
  \\
    \indrule{\rulmIBott}{
    \djum{\utenv}{\tenv}{\tm}{\typtwo}
    \HS
  \djum{\utenv}{\tenv'}{\tmtwo}{\lneg{\typtwo}}
  }{
    \djum{\utenv}{\tenv,\tenv'}{\ibott{\tm}{\tmtwo}}{\bott}
  }
  \quad
  \indrule{\rulmIOne}{
    \emptyPremise
  }{
    \djum{\utenv}{\emptytenv}{\ione}{\one}
  }
    \indrule{\rulmEOne}{
    \djum{\utenv}{\tenv}{\tmtwo}{\one}
    \HS
    \djum{\utenv}{\tenvtwo}{\tm}{\typ}
  }{
    \djum{\utenv}{\tenv,\tenvtwo}{\tm\eone{\tmtwo}}{\typ}
  }
\end{array}
\]
\end{minipage}}
\caption{Typing rules for $\CalcMELL$.}\label{typing:rules:CalcMELL}
\end{figure}

\subparagraph*{Discussion.}
Rules are presented in introduction/elimination pairs, as is standard in natural deduction. In particular, bottom has no introduction rules. Our presentation satisfies Local Intrinsic Harmony~\cite{DBLP:journals/jacm/DaviesP01,DBLP:journals/jphil/FrancezD12}: local soundness (elimination rules are not ``too strong'') and local completenes (elimination rules are not ``too weak''). If $\lam{\lvar}{\tm}$ is shorthand for $\ipar{\lvar}{\lvartwo}{(\ibott{\tm}{\lvartwo})}$, then
% $\indrulename{\rulmILimp}$ of~\rdef{calcMLL}  is admissible\footnote{Meaning that provability of all its premises implies provability of its conclusion.} (likewise, when $\ofc{\tm}$ is shorthand for $\iofctwo{\lvar}{(\ibott{\tm}{\lvar})}$, we can derive $\indrulename{\rulmIOfc'}$; also $\indrulename{\rulmEWhyTwo}$ can be generalized to $\indrulename{\rulmEWhyTwo}$):
% %\begin{remark}
% %\lremark{derivable:rules:for:ofc:and:lam}
% %the following rules are admissible
% \[
%   \begin{array}{ccc}
%   \indrule{\rulmILimp}{
%     \djum{\utenv}{\tenv,\lvar:\typ}{\tm}{\typtwo}
%   }{
%     \djum{\utenv}{\tenv}{\lam{\lvar}{\tm}}{\typ\limp\typtwo}
%     }
%     &
%             \indrule{\rulmIOfc'}
%   {\djum{\utenv}{\emptytenv}{\tm}{\typ}}
%       {\djum{\utenv}{\emptytenv}{\ofc{\tm}}{\ofc{\typ}}}
%       &
%            \indrule{\rulmEWhyTwo}{
%     \djum{\utenv}{\tenv}{\tmtwo}{\why{\typ}}
%     \HS
%   \djum{\utenv}{\lvar:\typ}{\tm}{\why{\typtwo}}
%   }{
%   \djum{\utenv}{\tenv}{\iwhy{\uvar}{(\ewhy{(\ibott{\tm}{\ofc{\uvar}})}{\lvar}{\tmtwo})}}{\why{\typtwo}}
%   }
%     \end{array}
% \]
it is possible to add the rule in which $\djum{\utenv}{\tenv}{\lam{\lvar}{\tm}}{\lneg{\typ}}$ is inferred from $\djum{\utenv}{\tenv,\lvar:\typ}{\tm}{\bott}$
  % \[
  %             \indrule{}{
  % \djum{\utenv}{\tenv,\lvar:\typ}{\tm}{\bott}
  % }{
  %   \djum{\utenv}{\tenv}{\lam{\lvar}{\tm}}{\lneg{\typ}}
  % }
  % \]
and drop all of $\rulmETensor$, $\rulmEParOne$, $\rulmEParTwo$, $\rulmEOfc$, $\rulmEWhy$ and $\rulmEOne$ (which become derivable), in the style of~\cite{DBLP:journals/iandc/BarbaneraB96} for classical propositional logic. However, the reduction semantics introduces critical pairs that lead to failure of confluence\footnote{\Eg, if $\uvar:\bott\parr\bott$ and $\uvartwo:\bott\parr\one$, then $  \ibott{(\lam{\lvar}{\ap{\uvar}{\lvar}})}{(\lam{\lvartwo}{\invap{\uvartwo}{\lvartwo}})}$ reduces to both $\ap{\uvar}{(\lam{\lvartwo}{\invap{\uvartwo}{\lvartwo}})}$ and $\invap{\uvartwo}{(\lam{\lvar}{\ap{\uvar}{\lvar}})}$.}, so we abstain from doing so.

There are many alternative ways to formulate natural deduction for \MELL.
%For example, our system departs from the standard choice for \emph{promotion} (of-course introduction),
%which derives the conclusion $\djum{\utenv}{\emptytenv}{\ofc{\tm}}{\ofc{\typ}}$
%from the premise $\djum{\utenv}{\emptytenv}{\tm}{\typ}$.
Although these alternatives may be logically equivalent
(\ie they prove all and only valid \MELL sequents), they exhibit different
computational behaviors.
The rules in $\CalcMELL$ have been carefully chosen to be able to define a
well-behaved contra-substitution operation.

\begin{lemma}[{name=Soundness and Completeness of $\CalcMELL$~\proofnote{Proof on pg.~\pageref{soundness_and_completeness_of_calc_mell:proof}}, restate=[name=Logical soundness/completeness of $\CalcMELL$]SoundnessAndCompletenessOfCalcMELL}]
%\begin{lemma}[Soundness]
  \llem{soundness_and_completeness_of_calc_MELL}
  \quad
  \begin{itemize}
\item \textbf{\textup{Soundness.}}
  If $\djum{\utenv}{\tenv}{\tm}{\typ}$ holds in $\CalcMELL$, then $\jull{\why{\lneg{\utenv}},\lneg{\tenv},A}$ holds in \MELL.
\item \textbf{\textup{Completeness.}}
If $\jull{\tenv_0}$ holds in \MELL and $\tenv,\typ$ is any permutation of $\tenv_0$, there exists a term $\tm$ such that $\djum{\emptyutenv}{\lneg{\tenv}}{\tm}{\typ}$ holds in $\CalcMELL$.
\end{itemize}

\end{lemma}

% \begin{lemma}[{name=Completeness of $\CalcMELL$~\proofnote{Proof on pg.~\pageref{completeness_of_calc_mell:proof}}, restate=[name=Completeness of $\CalcMELL$]CompletenessOfCalcMELL}]
% \llem{completeness_of_calc_mell}
% %\begin{lemma}[Completeness]
% Suppose $\jull{\tenv_0}$ and $\tenv,\typ$ is any permutation of $\tenv_0$. Then $\djum{\emptyutenv}{\seq{\lvar}:\lneg{\tenv}}{\tm}{\typ}$, for some term $\tm$ and linear variables $\seq{\lvar}$.
% \end{lemma}

\subparagraph*{Substitution and contra-substitution.} 
Reduction semantics for the $\CalcMELL$-calculus
%(\cf Sec.~\ref{sec:CalcMELL:reduction})
relies on three notions of
substitution:
\defn{linear substitution} ($\tm\sub{\lvar}{\tmtwo}$),
\defn{unrestricted substitution} ($\tm\sub{\uvar}{\tmtwo}$)
and \defn{contra-substitution} ($\tm\cos{\lvar}{\tmtwo}$).
Linear (resp. unrestricted) substitution is just capture-avoiding substitution
of a linear (resp. unrestricted) variable by a term. 
% in the sense that 
% the following rules are admissible in \CalcMELL:
% \[
%   \begin{array}{cc}
%   \indrule{\rulmSub}{
%     \djum{\utenv}{\tenv,\lvar:\typ}{\tm}{\typtwo}
%     \HS
%     \djum{\utenv}{\tenvtwo}{\tmtwo}{\typ}
%   }{
%     \djum{\utenv}{\tenv,\tenvtwo}{\tm\sub{\lvar}{\tmtwo}}{\typtwo}
%     }
%     &
%       \indrule{\rulmUSub}{
%     \djum{\utenv, \uvar:\typ}{\tenv}{\tm}{\typtwo}
%     \HS
%     \djum{\utenv}{\emptytenv}{\tmtwo}{\typ}
%   }{
%     \djum{\utenv}{\tenv}{\tm\sub{\uvar}{\tmtwo}}{\typtwo}
%   }
%     \end{array}
% \]
% is defined on \emph{linear terms}, a strict superset of typable terms.
%\begin{definition}[Linear term]
% \end{definition}
To define contra-substitution, we need to adapt the notion of \emph{linear term}
to the $\CalcMELL$-calculus.
A $\CalcMELL$ term is \defn{linear} if each of the following hold:
% \begin{enumerate}
% \item
  1) each free \emph{linear} variable occurs exactly once; 
%\item
  2) there is exactly one occurrence of each bound \emph{linear}
  variable in the scope of its binder;
%\item
  3) subterms of the form $\iofctwo{\lvar}{\tm}$ have no free linear variables; and
%\item
  4) in subterms of the form $\ewhy{\tm}{\lvar}{\tmtwo}$,
  only $\lvar$ occurs as a free linear variable in $\tm$.
%\end{enumerate}
%  Let $\tm$ be a linear term such that $\lvar \in \fv{\tm}$.

\begin{definition}[Contra-substitution]
\ldef{cos}
The operation $\tm\cos{\lvar}{\tmtwo}$ is defined
over linear terms
by dropping the clauses of \rdef{contrasub:for:MLL} for lambda abstraction and tensor elimination and adding the ones
of Fig.~\ref{contrasubstitution:for:CalcMELL}. 
%Note that $(\iofctwo{\lvartwo}{\tm})\cos{\lvar}{\tmtwo}$
%is undefined because $\tm$ has no free linear variables. Similarly, in 
%$\ewhy{\tm_1}{\lvartwo}{\tm_2}$, there are no free linear variables in
%$\tm_1$ (except $\lvartwo$ itself) and $\ione$ has no free linear
%variables.
\end{definition}
% Note also that contra-substitution preserves linearity (as does linear substitution): 
%if $\set{\tm,\tmtwo}$ is a set of linear terms such that $\lvar \in \fv{\tm}$, then 
%\begin{enumerate}
%\item $\tm\sub{\lvar}{\tmtwo}$ is a linear term and
%      $\fv{\tm\sub{\lvar}{\tmtwo}}
%       = \fv{\tm} \setminus \set{\lvar} \cup \fv{\tmtwo}$.
%\item
%  $\tm\cos{\lvar}{\tmtwo}$ is a linear term and
%      $\fv{\tm\cos{\lvar}{\tmtwo}}
%       = \fv{\tm} \setminus \set{\lvar} \cup \fv{\tmtwo}$.
%     \end{enumerate}
%\end{definition}
% Informally, $\tm\cos{\lvar}{\tmtwo}$ denotes the term that results
%   from ``grabbing'' the \emph{unique} occurrence of $\lvar$ in $\tm$ and ``pulling''
%   from it, in order to turn the term $\tm$ inside out, very much like
%   a sock, and then replacing $a$ with $\tmtwo$.  
For example:
  $(\ap{(\ipar{\lvarthree}{\lvarfour}{\tmtwo})\eone{\lvar}}{\tm})\cos{\lvar}{\ione}
   = (\ipar{\lvarthree}{\lvarfour}{\tmtwo})\eone{\lvar}\cos{\lvar}{\pair{\tm}{\ione}} 
   = \lvar\cos{\lvar}{\ibott{\ipar{\lvarthree}{\lvarfour}{\tmtwo}}{\pair{\tm}{\ione}}}
   = \ibott{\ipar{\lvarthree}{\lvarfour}{\tmtwo}}{\pair{\tm}{\ione}}$.
Contra-substitution is compatible with typing:   
\begin{lemma}[{name=Type Compatibility of Contra-substitution~\proofnote{Proof on pg.~\pageref{contrasubstitution_lemma:proof}}, restate=[name=Type Compatibility of Contrasubstitution]ContrasubstitutionLemma}]
%\begin{lemma}[Contrasubstitution lemma]
 \llem{contrasubstitution_lemma}
 If $ \djum{\utenv}{\tenv,\lvar:\typ}{\tm}{\typtwo}$ and $\djum{\utenv}{\tenvtwo}{\tmtwo}{\lneg{\typtwo}} $ hold in \CalcMELL, then so does $\djum{\utenv}{\tenv,\tenvtwo}{\tm\cos{\lvar}{\tmtwo}}{\lneg{\typ}}$.
% \[
%   \indrule{\rulmCos}{
%     \djum{\utenv}{\tenv,\lvar:\typ}{\tm}{\typtwo}
%     \HS
%     \djum{\utenv}{\tenvtwo}{\tmtwo}{\lneg{\typtwo}}
%   }{
%     \djum{\utenv}{\tenv,\tenvtwo}{\tm\cos{\lvar}{\tmtwo}}{\lneg{\typ}}
%   }
% \]
\end{lemma}

\begin{figure}
  \[
  \begin{array}{rcl}
  %   \lvar\cos{\lvar}{\tmtwo}
  % & \eqdef &
  %   \tmtwo
  % \\
  %   \pair{\tm_1}{\tm_2}\cos{\lvar}{\tmtwo}
  % & \eqdef &
  %   \begin{cases}
  %     \tm_1\cos{\lvar}{\invap{\tmtwo}{\tm_2}}
  %     & \text{if $\lvar \in \fv{\tm_1}$}
  %   \\
  %     \tm_2\cos{\lvar}{\ap{\tmtwo}{\tm_1}}
  %     & \text{if $\lvar \in \fv{\tm_2}$}
  %   \end{cases}
  % \\
    (\tm_1\epair{\lvartwo}{\lvarthree}{\tm_2})\cos{\lvar}{\tmtwo}
  & \eqdef &
    \begin{cases}
      \tm_1\cos{\lvar}{\tmtwo}\epair{\lvartwo}{\lvarthree}{\tm_2}
      & \text{if $\lvar \in \fv{\tm_1}, \lvar \notin\set{\lvartwo,\lvarthree}$}
    % \\
    % \cancel{\tm_2\cos{\lvar}{\lam{\lvartwo}{\tm_1\cos{\lvarthree}{\tmtwo}}}}
    \\
    \tm_2\cos{\lvar}{\ipar{\lvartwo}{\lvarthree}{(\ibott{\tm_1}{\tmtwo})}}
      & \text{if $\lvar \in \fv{\tm_2}$}
    \end{cases}
    \\
        (\ipar{\lvartwo}{\lvarthree}{\tm})\cos{\lvar}{\tmtwo}
  & \eqdef &
      \tm\cos{\lvar}{\ione}\epair{\lvartwo}{\lvarthree}{\tmtwo}
  %   (\ap{\tm_1}{\tm_2})\cos{\lvar}{\tmtwo}
  % & \eqdef &
  %   \begin{cases}
  %     \tm_1\cos{\lvar}{\pair{\tm_2}{\tmtwo}}
  %     & \text{if $\lvar \in \fv{\tm_1}$}
  %   \\
  %     \tm_2\cos{\lvar}{\invap{\tm_1}{\tmtwo}}
  %     & \text{if $\lvar \in \fv{\tm_2}$}
  %   \end{cases}
  % \\
  %   (\invap{\tm_1}{\tm_2})\cos{\lvar}{\tmtwo}
  % & \eqdef &
  %   \begin{cases}
  %     \tm_1\cos{\lvar}{\pair{\tmtwo}{\tm_2}}
  %     & \text{if $\lvar \in \fv{\tm_1}$}
  %   \\
  %     \tm_2\cos{\lvar}{\ap{\tm_1}{\tmtwo}}
  %     & \text{if $\lvar \in \fv{\tm_2}$}
  %   \end{cases}
  % \\
  % \\
               \\
    (\iofctwo{\lvartwo}{\tm})\cos{\lvar}{\tmtwo}
    & \eqdef &      \text{(impossible, as $\iofctwo{\lvartwo}{\tm}$ has no free linear variables)}
    \\
    \tm_1\eofc{\uvar}{\tm_2}\cos{\lvar}{\tmtwo}
  & \eqdef &
    \begin{cases}
      \tm_1\cos{\lvar}{\tmtwo}\eofc{\uvar}{\tm_2}
      & \text{if $\lvar \in \fv{\tm_1}$}
      \\
      \tm_2\cos{\lvar}{\iwhy{\uvar}{\ibott{\tm_1}{\tmtwo}}}
      & \text{if $\lvar \in \fv{\tm_2}$}
      \\
    \end{cases}
    \\
    (\iwhy{\uvar}{\tm})\cos{\lvar}{\tmtwo}
  & \eqdef &
      \tm\cos{\lvar}{\ione}\eofc{\uvar}{\tmtwo}
    \\
    (\ewhy{\tm_1}{\lvartwo}{\tm_2})\cos{\lvar}{\tmtwo}
  & \eqdef &
    \begin{cases}
      \text{(impossible as $\flv{\tm_1}=\{\lvartwo\}$ and $\lvar \neq \lvartwo$)} 
      & \text{if $\lvar \in \fv{\tm_1}$ }
    \\
    \tm_2\cos{\lvar}{\iofctwo{\lvartwo}{\tm_1} }\eone{\tmtwo} & \text{if $\lvar \in \fv{\tm_2}$}\\
\end{cases} 
      \\
        (\ibott{\tm_1}{\tm_2})\cos{\lvar}{\tmtwo}
  & \eqdef &
    \begin{cases}
      \tm_1\cos{\lvar}{\tm_2 }\eone{\tmtwo}
      & \text{if $\lvar \in \fv{\tm_1}$}
    \\
     \tm_2\cos{\lvar}{\tm_1 }\eone{\tmtwo}
    & \text{if $\lvar \in \fv{\tm_2}$}
  \end{cases}
    \\
    \ione\cos{\lvar}{\tmtwo}
  & \eqdef &
             \text{(impossible, as $\ione$ has no free linear variables)}
    \\
    \tm_1\eone{\tm_2}\cos{\lvar}{\tmtwo}
    & \eqdef &
               \begin{cases}
                \tm_1\cos{\lvar}{\tmtwo}\eone{\tm_2}
                & \text{if $\lvar\in\fv{\tm_1}$}\\
                \tm_2\cos{\lvar}{\ibott{\tm_1}{\tmtwo}}
                 & \text{if $\lvar\in\fv{\tm_2}$}
                 \end{cases}
  \end{array}
\]
\caption{Contra-substitution for \CalcMELL: we take the clauses of \rdef{contrasub:for:MLL}, drop those of lambda abstraction and tensor elimination, and add the ones above.}\label{contrasubstitution:for:CalcMELL}
\end{figure}

\subparagraph*{Reduction Semantics.}%\label{sec:CalcMELL:reduction}
Reduction in \CalcMELL is defined modulo an equivalence relation on terms called \emph{structural equivalence}. We begin by presenting pre-reduction (\rdef{prereduction_in_calcMELL}), in which structural equivalence is absent, then structural equivalence (\rdef{cceq}) and finally reduction proper (\rdef{calcMELL_reduction}).

% \begin{definition}[Contexts]
% \ldef{contexts_and_eliminator_contexts}
% \emph{Contexts} ($\ctx,\ctxtwo,\ldots$) are terms with one occurrence of a fixed variable of the form $\ctxhole$, the ``hole'' of the context.
% % We say \emph{$\cctx$ is \whyfreename} if none of the eliminators in $\cctx$ is of the form $\ewhy{}{\lvar}{\tm}$.
% \end{definition}

% One interesting use of  contrasubstitution is to change the focus of the computation to an arbitrary term that has been tagged by $\ibott{}{}$-composing it with a fresh name. More precisely, the following holds:

% The reduction rules are as follows:
\begin{definition}[Pre-Reduction in \CalcMELL]
\ldef{prereduction_in_calcMELL}
Denoted $\pretome$, it is given by the contextual closure of the following \emph{reduction axioms}, where both sides are assumed to be typed with the same types in the same contexts.
\begin{center}
  $
  \begin{array}{llll}
    \ap{(\ipar{\lvar}{\lvartwo}{\tm})\cctx}{\tmthree}
  & \toax{\raxparL} &
            \tm\sub{\lvar}{\tmthree}\cos{\lvartwo}{\ione}\cctx
            &  \text{($\lvartwo \notin \fv{\tmthree}$)}
    \\
    \invap{(\ipar{\lvar}{\lvartwo}{\tm})\cctx}{\tmthree}
  & \toax{\raxparR} &
            \tm\sub{\lvartwo}{\tmthree}\cos{\lvar}{\ione}\cctx
            & \text{($\lvar \notin \fv{\tmthree}$)}
  \\
    \tm\epair{\lvar}{\lvartwo}{\pair{\tmtwo}{\tmthree}\cctx}
  & \toax{\raxtensor} &
    \tm\sub{\lvar}{\tmtwo}\sub{\lvartwo}{\tmthree}\cctx
    & \text{($\lvartwo \notin \fv{\tmtwo}$)}
  \\
    \tm\eofc{\uvar}{(\iofctwo{\lvar}{\tmtwo})\cctx}
  & \toax{\raxofc} &
    \tm\sub{\uvar}{\tmtwo\cos{\lvar}{\ione}}\cctx
            \\
    \ewhy{\tm}{\lvar}{(\iwhy{\uvar}{\tmtwo})\cctx}
    & \toax{\raxwhy} &
              \tmtwo\sub{\uvar}{\tm\cos{\lvar}{\ione}}\cctx
% \\
%               \cancel{\tm_1\ewhynew{\lvar}{(\iwhy{\uvar}{\tmtwo})\cctx}{\tm_3}}
%     & \tome &
%               \cancel{\tmtwo\sub{\uvar}{\tm_1\cos{\lvar}{\tm_3}}}\cctx
    \\
                            \ibott{(\ipar{\lvar}{\lvartwo}{\tm})\cctx}{\pair{\tmtwo}{\tmthree}\cctxtwo}
   & \toax{\raxpartensor} &
              \tm\sub{\lvar}{\tmtwo}\sub{\lvartwo}{\tmthree}\cctx\cctxtwo
    & \text{($\lvartwo \notin \fv{\tmtwo}$)}
    \\
                  \ibott {\pair{\tmtwo}{\tmthree}\cctx}{(\ipar{\lvar}{\lvartwo}{\tm})\cctxtwo}
    & \toax{\raxtensorpar} &
              \tm\sub{\lvar}{\tmtwo}\sub{\lvartwo}{\tmthree}\cctx\cctxtwo
              & \text{($\lvartwo \notin \fv{\tmtwo}$)}
    \\
        \ibott{(\iofctwo{\lvar}{\tm})\cctx}{(\iwhy{\uvar}{\tmthree}) \cctxtwo}
   & \toax{\raxofcwhy} &
    \tmthree\sub{\uvar}{\tm\cos{\lvar}{\ione}}\cctx\cctxtwo
    \\
    \ibott {(\iwhy{\uvar}{\tmthree}) \cctx}{(\iofctwo{\lvar}{\tm})\cctxtwo}
     & \toax{\raxwhyofc} &
              \tmthree\sub{\uvar}{\tm\cos{\lvar}{\ione}}\cctx\cctxtwo

  \end{array}
  $
  \end{center}
\end{definition}

\subparagraph*{Discussion.}
The first five reduction axioms correspond to introduction/elimination pairs.
The remaining ones are akin to principal cut elimination cases, and
are required for confluence (though they are not sufficient, as discussed below).
The order in which substitutions and contra-substitutions are presented on the right-hand side of ${\raxparL}$, $\raxparR$, $\raxtensor$, $\raxpartensor$, and $\raxtensorpar$ is unimportant since $ \tm\sub{\lvar}{\tmthree}\cos{\lvartwo}{\ione}= \tm \cos{\lvartwo}{\ione}\sub{\lvar}{\tmthree}$, given that $\lvartwo\notin\fv{\tmthree}$ %(\cf~\rlem{mell_sub_contra})
and also $\tm\sub{\lvar}{\tmtwo}\sub{\lvartwo}{\tmthree}= \tm \sub{\lvartwo}{\tmthree}\sub{\lvar}{\tmtwo}$, since $\lvartwo \notin \fv{\tmtwo}$. %(\cf~\rlem{substitution_lemma}).
Pre-reduction is not confluent even in the presence of the last four reduction axioms. Consider $\tm=\overline{\invap{(\ipar{\lvar}{\lvartwo}{\underline{\ap{(\ipar{\lvarthree}{\lvarfour}{\tmtwo})\eone{\lvar}}{\tmfour}}})}{\tmthree}}$, where, say, $\lvartwo\in\fv{\tmtwo}$. It may be reduced in two ways:

  % % \begin{acenter}{-3cm}
  \begin{center}
    \begin{tikzcd}[column sep=1em,row sep=-1em]
      &    \invap{(\ipar{\lvar}{\lvartwo}{\tmtwo\sub{\lvarthree}{\tmfour}\cos{\lvarfour}{\ione}\eone{\lvar}})}{\tmthree}\arrow[rd,dotted, "\raxparR"'] & \\
      
      \tm
      \arrow[ru, "\raxparL"]\arrow[rd, "\raxparR"'] & &       \mlnode{
        % \tmtwo\sub{\lvarthree}{\tmfour}\cos{\lvarfour}{\ione}\eone{\lvar}\sub{\lvartwo}{\tmthree}\cos{\lvar}{\ione}
        % \\
        % \tmtwo\sub{\lvarthree}{\tm}\cos{\lvarfour}{\ione}\sub{\lvartwo}{\tmthree}\eone{\lvar}\cos{\lvar}{\ione}
        % \\
        \ibott{\tmtwo\sub{\lvarthree}{\tmfour}\cos{\lvarfour}{\ione}\sub{\lvartwo}{\tmthree}}{\ione}
        \\        
%        \tmtwo\sub{\lvarthree}{\tmfour}\cos{\lvarfour}{\ione}\sub{\lvartwo}{\tmthree}
 %       \\
%        \tmtwo\sub{\lvarthree}{\tmfour}\sub{\lvartwo}{\tmthree}\cos{\lvarfour}{\ione}
        \\
        \tmtwo\sub{\lvartwo}{\tmthree}\sub{\lvarthree}{\tmfour}\sub{\lvarfour}{\ione}
        }
        \\         
      & 
        \ibott{\ipar{\lvarthree}{\lvarfour}{\tmtwo\sub{\lvartwo}{\tmthree}}}{\pair{\tmfour}{\ione}}
        \arrow[ru,dotted, "\raxpartensor"']
        & \\

    \end{tikzcd}
  \end{center}
  % % \end{acenter}
The terms $
        \ibott{\tmtwo\sub{\lvarthree}{\tmfour}\cos{\lvarfour}{\ione}\sub{\lvartwo}{\tmthree}}{\ione}$
        and
        $\tmtwo\sub{\lvartwo}{\tmfour}\sub{\lvarthree}{\tmfour}\sub{\lvarfour}{\ione}$
        are distinct. For a term $\tm$ of type $\bott$, equating
        $\ibott{\tm}{\ione}$ and $\tm$ would be reasonable. We next
        introduce a notion of equivalence of terms that introduces
        this and other equations.  As it turns out, this notion of equivalence is a strong bisimulation with respect to pre-reduction.

%%% Local Variables:
%%% mode: latex
%%% TeX-master: "main"
%%% End:
 
\subsection{Structural equivalence}

% \begin{definition}[Surface Contexts]
%   \ldef{surface_contexts}
% Surface contexts are given by:
% \[
%   \begin{array}{lll}
%     \fctx & ::= & \ipar{\lvar}{\lvartwo}{\Box}\,|\, \pair{\Box}{\tm} \,|\,\pair{\tm}{\Box}\,|\, \ap{\Box}{\tm} \,|\,\ap{\tm}{\Box}\,|\,\invap{\Box}{\tm} \,|\,\invap{\tm}{\Box}\,|\,
%                    \ibott{\Box}{\tm} \,|\,\ibott{\tm}{\Box}\,|\, \iwhy{\uvar}{\Box}\,|\,\Box\pelim{\patt}{\tm} \,|\,\tm\pelim{\patt}{\Box}\,|\,\ewhy{\tm}{\lvar}{\Box}\\
% %         & & 
%   \end{array}
% \]
% Note that $\ewhy{\Box}{\lvar}{\tm}$, and  $\iofctwo{\lvar}{\Box}$ are excluded from the set of surface contexts. A \emph{term $\tmthree$ is free for a surface context $\fctx$} if $\of{\fctx}{\tmthree}$ does not bind free occurrences of variables of $\tmthree$.
% \end{definition}

\begin{definition}[Structural Equivalence]
  \ldef{cceq}
Surface contexts $\fctx$ are any of the following expressions:
\[
  \ipar{\lvar}{\lvartwo}{\Box}
  \HS\HS
  \pair{\Box}{\tm}
  \HS\HS
  \pair{\tm}{\Box}
  \HS\HS
  \ap{\Box}{\tm}
  \HS\HS
  \ap{\tm}{\Box}
  \HS\HS
  \invap{\Box}{\tm}
  \HS\HS
  \invap{\tm}{\Box}
\]
\[
  \ibott{\Box}{\tm}
  \HS\HS
  \ibott{\tm}{\Box}
  \HS\HS
  \iwhy{\uvar}{\Box}
  \HS\HS
  \Box\pelim{\patt}{\tm}
  \HS\HS
  \tm\pelim{\patt}{\Box}
  \HS\HS
  \ewhy{\tm}{\lvar}{\Box}
\]
Note that $\ewhy{\Box}{\lvar}{\tm}$, and  $\iofctwo{\lvar}{\Box}$ are excluded from the set of surface contexts. A \emph{term $\tmthree$ is free for a surface context $\fctx$} if $\of{\fctx}{\tmthree}$ does not bind free occurrences of variables of $\tmthree$.
\emph{Structural equivalence} is defined as the reflexive, symmetric, transitive and contextual closure of the following axioms. In every equation, both sides are assumed to be typed with the same types in the same contexts:
  \[
    \begin{array}{llll}
      \of{\fctx}{\tm}\pelim{\patt}{\tmthree} & \cceq_{\ruleCCEqCtx} & \of{\fctx}{\tm \pelim{\patt}{\tmthree}} & (\fv{\patt}\cap\fv{\fctx}=\emptyset \mbox{ and }\tmthree \mbox{ is free for }\fctx)
      \\ 
      \tm\eone{\ione} & \cceq_{\ruleCCEqBetaOne} & \tm
      \\      
      \tm\eone{\tmtwo}  & \cceq_{\ruleCCEqEOneSymm} & \tmtwo\eone{\tm}
      \\
      \ibott{\tm\cos{\lvar}{\ione}}{\tmthree} &
                                                \cceq_{\ruleCCEqCosAndSub} & \tm\sub{\lvar}{\tmthree}
      \\ 
    \end{array}
  \]
\end{definition}

\subparagraph*{Discussion.} The first equation allows positive
eliminators to be propagated freely. Though note that $\ruleCCEqCtx$
does not propagate eliminators under the of-course constructor
$\iofctwo{\lvar}{\tm}$ or into the left argument of a why-not eliminator $\ewhy{\tm}{\lvar}{\tmtwo}$
(which are the non-linear positions). The second equation corresponds the introduction/elimination roundabout on $\one$; it is not well-behaved as a reduction rule and thus is modeled as an equation. The third equation allows terms $\tm$ and $\tmtwo$ both of which have type $\one$ to be commuted freely in their placement in an eliminator for this type. In fact, $\tm\eone{\tmtwo}  \cceq \ibott{\tm}{\tmtwo}$ holds whenever $\tm:\bott$.
The fourth equation states that any term of type $\bott$ may be
understood as stemming from a contradiction. It is crucial in relating
contra-substitution and substitution. In particular, if $\tm:\bott$,
then $\tm\cos{\lvar}{\tmthree}  \cceq \tm\sub{\lvar}{\tmthree}$. 
We emphasize that these equations apply to typed terms. For example, $\ruleCCEqEOneSymm$ should be read as: if  $\djum{\utenv}{\tenv}{\tmtwo}{\one}$ and 
    $\djum{\utenv}{\tenvtwo}{\tm}{\one}$, then $\djum{\utenv}{\tenv,\tenvtwo}{\tm\eone{\tmtwo}\cceq \tmtwo\eone{\tm}}{\one}$.
%and $\ruleCCEqCosAndSub$
%as samples using type judgements:
% \[
%     \indrule{$\ruleCCEqEOneSymm$}{
%     \djum{\utenv}{\tenv}{\tmtwo}{\one}
%     \HS
%     \djum{\utenv}{\tenvtwo}{\tm}{\one}
%   }{
%     \djum{\utenv}{\tenv,\tenvtwo}{\tm\eone{\tmtwo}\cceq \tmtwo\eone{\tm}}{\one}
%   }
%   \]
%   \[
%     \indrule{$\ruleCCEqCosAndSub$}{
%       \djum{\utenv}{\tenv,\lvar:\typ}{\tm}{\bott}
%       \HS
%       \djum{\utenv}{\emptytenv}{\ione}{\one}
%        \HS
%        \djum{\utenv}{\tenv'}{\tmthree}{\typ}
%   }{
%     \djum{\utenv}{\tenv,\tenv'}{\ibott{\tm\cos{\lvar}{\ione}}{\tmthree}\cceq \tm\sub{\lvar}{\tmthree}}{\bott}
%   }
%   \]
  
% %%%%%%%%%%%%
% \subparagraph*{Substitution, Contrasubstitution and Structural Equivalence.}
% %%%%%%%%%%%%

% Contrasubstitution is compatible with equivalence. In other words, $\tm\cceq\tmtwo$ and $\tmthree\cceq\tmfour$ imply $\tm\cos{\lvar}{\tmthree}\cceq\tmtwo\cos{\lvar}{\tmfour}$. We first prove one part of this result, since it is required for various results that follow.  The full result is proved later (\cf \rlem{cos_compatible_with_equivalence}).

Structural equivalence is consistent as a theory due to the fact that $\tm\cceq\tmtwo$ implies $\fv{\tm}=\fv{\tmtwo}$ and hence $\lvar\cceq\lvartwo$ is not derivable, for arbitrary linear variables $\lvar$ and $\lvartwo$ such that $\lvar\neq\lvartwo$.  Also, all our notions of substitution are compatible with structural equivalence.  In particular, if $\tmtwo\cceq\tmthree$ then both $\tm\cos{\lvar}{\tmtwo}\cceq\tm\cos{\lvar}{\tmthree}$ and
$\tmtwo\cos{\lvar}{\tm}\cceq\tmthree\cos{\lvar}{\tm}$.
The proof of this fact relies on a series of auxiliary results involving all three notions of substitution and their interaction with $\cceq$. We highlight two of these which we find worthy of mention since they relate linear substitution and contra-substitution.
Let $\set{\tm,\tmtwo,\tmthree}$ be a set of linear terms
such that $\lvar \in \fv{\tm}$
and let $\lvartwo \neq \lvar$
be such that $\lvartwo \in \fv{\tm}\cup\fv{\tmtwo}$.
Then:
\begin{enumerate}
\item
  If $\lvartwo \in \fv{\tm}$, then
    $\tm\cos{\lvar}{\tmtwo}\cos{\lvartwo}{\tmthree}
    \cceq
    \tm\cos{\lvartwo}{\tmtwo}\sub{\lvar}{\tmthree}$
  
\item
  If $\lvartwo \in \fv{\tmtwo}$, then
  $
    \tm\cos{\lvar}{\tmtwo}\cos{\lvartwo}{\tmthree}
    \cceq
    \tmtwo\cos{\lvartwo}{\tm\sub{\lvar}{\tmthree}}
    $
\item If $\lvartwo \in \fv{\tmtwo}$, then
  $
    \tmtwo\sub{\lvartwo}{\tm}\cos{\lvar}{\tmthree}
    \cceq \tm\cos{\lvar}{\tmtwo\cos{\lvartwo}{\tmthree}}
  $
\end{enumerate}

%These equations are used in some of the proofs of the results that follow. For example, $\ruleCCCosOneSimplLeftGen$ is required for \rlem{cosubst_compat_with_reduction} (contra-substitution is compatible with reduction).
The main result in this section is that structural equivalence is a \textbf{strong bisimulation} for pre-reduction. 
The importance of this result is that it allows to lift the notion of
reduction on terms to equivalence classes of terms modulo $\cceq$.

\begin{theorem}[{name=Strong bisimulation~\proofnote{Proof on pg.~\pageref{cceqalt_is_a_strong_bisimulation:proof}}, restate=[name=Strong bisimulation]StrongBisimulation}]
%\begin{lemma}[Soundness]
  \lthm{cceqalt_is_a_strong_bisimulation}
    Let $\tm, \tmtwo$ be two typed linear terms. 
    If $\tm \cceq \tmtwo$ and $\tm \pretome \tm'$ 
    then there is a linear term $\tmtwo'$ such that
    $\tmtwo \pretome \tmtwo'$ and $\tm' \cceq \tmtwo'$.
\end{theorem}

The proof of~\rthm{cceqalt_is_a_strong_bisimulation} 
(and other results involving this notion of equivalence) is complicated by the
presence of contra-substitution and substitution in the last equation $\ruleCCEqCosAndSub$. 
The proof is developed around an
alternative characterization of structural equivalence where
contra-substitution and substitution are absent from its axioms
(\cf~\rdef{alt_equivalence} in the appendix). 
The following lemma exhibits some of the most interesting derived
equations (\cf~\rlem{cceq_derived} in the appendix):

\begin{lemma}[Some derived equations in $\cceq$]
    \[
      \begin{array}{cc}
        \begin{array}{r@{\,\,\,\,}l@{\,\,\,\,}ll}
      \ibott{\ione}{\tm} & \cceq & \tm &   \text{(if $\tm:\bott$)}
      \\
      \ibott{\tm}{\tmtwo} & \cceq  & \ibott{\tmtwo}{\tm} &
      \\
     \tm\cos{\lvar}{\ione}\eone{\tmthree} &  \cceq &
                                                 \tm\cos{\lvar}{\tmthree}
                                                        & 
      \\
      \ibott{\pair{\tmthree}{\tm}}{\tmtwo} & \cceq & \ibott{\tmthree}{(\invap{\tmtwo}{\tm})}
      \\
      \ibott{\pair{\tm}{\tmthree}}{\tmtwo} & \cceq &  \ibott{\tmthree}{(\ap{\tmtwo}{\tm})}
      \\
      \ibott{(\ap{\tmtwo}{\tm})}{\tmthree} & \cceq & \ibott{\tm}{(\invap{\tmtwo}{\tmthree})}
      \\
      \ibott{(\ipar{\lvar}{\lvartwo}{\tmtwo})}{\tm} & \cceq & \tmtwo \epair{\lvar}{\lvartwo}{\tm}
%       \\
%       \ibott{\iwhy{\uvar}{\tm}}{\tmthree} & \cceq & \tm\eofc{\uvar}{\tmthree}
%       \\
%        \ibott{(\iofctwo{\lvar}{\tm})\eone{\tmfour}}{\tmthree} & \cceq & \ibott{\tmfour}{\ewhy{\tm}{\lvar}{\tmthree}} & 
% \\
%        \ibott{(\iofctwo{\lvar}{\tm})}{\tmthree} & \cceq & \ewhy{\tm}{\lvar}{\tmthree} &
% \\
%       \ibott{(\ibott{\tmtwo}{\tm})}{\tmthree}   & \cceq & \ibott{\tmtwo}{\tm\eone{\tmthree}}
%       \\
%       \tm\eone{\tmtwo} & \cceq & \ibott{\tm}{\tmtwo} & (\tm:\bott)
%             \\
%       \ione\eone{\tm} & \cceq & \tm &    
%       \\
%       \tm\cos{\lvar}{\tmthree}  & \cceq & \tm\sub{\lvar}{\tmthree} & (\tm:\bott)
        \end{array}
        &
        \begin{array}{r@{\,\,\,\,}l@{\,\,\,\,}ll}
     %  \ibott{\ione}{\tm} & \cceq & \tm &   (\tm:\bott)
     %  \\
     %  \ibott{\tm}{\tmtwo} & \cceq  & \ibott{\tmtwo}{\tm}
     %  &

     %  \\
     % \tm\cos{\lvar}{\ione}\eone{\tmthree} &  \cceq &
     %                                             \tm\cos{\lvar}{\tmthree}
     %                                                    & 
     %  \\
     %  \ibott{\pair{\tmthree}{\tm}}{\tmtwo} & \cceq & \ibott{\tmthree}{(\invap{\tmtwo}{\tm})}
     %  \\
     %  \ibott{\pair{\tm}{\tmthree}}{\tmtwo} & \cceq &  \ibott{\tmthree}{(\ap{\tmtwo}{\tm})}
     %  \\
     %  \ibott{(\ap{\tmtwo}{\tm})}{\tmthree} & \cceq & \ibott{\tm}{\invap{\tmtwo}{\tmthree}}
     %  \\
     %  \ibott{(\ipar{\lvar}{\lvartwo}{\tmtwo})}{\tm} & \cceq & \tmtwo \epair{\lvar}{\lvartwo}{\tm}
     %  \\
      \ibott{(\iwhy{\uvar}{\tm})}{\tmthree} & \cceq & \tm\eofc{\uvar}{\tmthree}
      \\
       \ibott{(\iofctwo{\lvar}{\tm})\eone{\tmfour}}{\tmthree} & \cceq & \ibott{\tmfour}{(\ewhy{\tm}{\lvar}{\tmthree})} & 
\\
       \ibott{(\iofctwo{\lvar}{\tm})}{\tmthree} & \cceq & \ewhy{\tm}{\lvar}{\tmthree} &
\\
      \ibott{(\ibott{\tmtwo}{\tm})}{\tmthree}   & \cceq & \ibott{\tmtwo}{\tm\eone{\tmthree}}
      \\
      \tm\eone{\tmtwo} & \cceq & \ibott{\tm}{\tmtwo} & \text{(if $\tm:\bott$)}
            \\
      \ione\eone{\tm} & \cceq & \tm &    
      \\
      \tm\cos{\lvar}{\tmthree}  & \cceq & \tm\sub{\lvar}{\tmthree} & \text{(if $\tm:\bott$)}
        \end{array}
      \end{array}      
  \]
\end{lemma}

We now define reduction for the $\CalcMELL$-calculus.

\begin{definition}[$\CalcMELL$-Reduction]
  \ldef{calcMELL_reduction}
  Reduction in $\CalcMELL$ is defined as pre-reduction modulo structural equivalence: $\tm \tome{} \tmtwo$, iff there exist terms $\tm'$ and $\tmtwo'$ such that $\tm\cceq\tm'\pretome\tmtwo'\cceq\tmtwo$.
\end{definition}

Reduction in $\CalcMELL$ enjoys the \textbf{subject reduction} property:

\begin{lemma}[{name=\proofnote{Proof on pg.~\pageref{subject_reduction:proof}}, restate=[name=Subject Reduction]SubjectReduction}]
\llem{subject_reduction}
$\djum{\utenv}{\tenv}{\tm}{\typ} $ and $\tm  \tome \tmtwo$ implies $\djum{\utenv}{\tenv}{\tmtwo}{\typ}$.
\end{lemma}

%%% Local Variables:
%%% mode: latex
%%% TeX-master: "main"
%%% End:
 
\subsection{Strong Normalization and Confluence}

We prove strong normalization for $\CalcMELL$ using the technique of reducibility candidates.

\subparagraph*{Preliminaries.}
%\begin{definition}
    Let
    $
     \typtms{\typ} \eqdef \set{ \tm \ST \exists\utenv,\tenv.\,\,\jum{\utenv;  \tenv}{\tm}{\typ} }
    $
    denote the set of terms of type $\typ$, and let
    $\SNtms{\typ} \eqdef \set{ \tm \in \typtms{\typ} \ST \tm \ \text{is SN}}$
    be the subset of those terms that are strongly normalizing
    (with respect to \emph{pre-reduction} $\pretome$).
    If a term $\tm$ is SN,
    let $\tml{\tm} \in \Nat$ be the length of the longest reduction sequence starting from $\tm$
    (which is well-defined due to K\"onig's Lemma, because $\tome$ is finitely branching).
    If $\typ$ is a type and $\tmsone \subseteq \typtms{\typ}$,
    we write $\tm \iftyp{\typ} \tmsone$
    to mean that $\tm\in\typtms{\typ}$ implies $\tm \in \tmsone$.
    If $\tmsone \subseteq \typtms{\typ}$,
    we define its \emph{orthogonal} by
      $\lneg{\tmsone} = \set{  \tm \in \typtms{\lneg{\typ}} \ST \forall \tmtwo \in \tmsone.\ \ibott{\tm}{\tmtwo} \iftyp{\bot} \SNtms{\bot}  }$.
%\end{definition}
    The following (typical) properties of $\lneg{(\cdot)}$ hold.
    If $\tmsone, \tmstwo \subseteq \typtms{\typ}$ then: 1) $\tmsone \subseteq \tmstwo$ implies $\lneg{\tmstwo} \subseteq \lneg{\tmsone}$; 2) $\tmsone \subseteq \llneg{\tmsone}$; and 3) $\lneg{\tmsone} = \lllneg{\tmsone}$.
Given sets of typed terms $\tmsone \subseteq \typtms{\typ}$ and $\tmstwo \subseteq \typtms{\typtwo}$ we define:
\[
  \begin{array}{lll}
    \tmsone \tensor \tmstwo & \eqdef & \set{ \pair{\tm}{\tmtwo} \in \typtms{\typ \tensor \typtwo} \ST \tm \in \tmsone \land \tmtwo \in \tmstwo   } \\
% \]
% Given a set of typed terms $\tmsone \subseteq \typtms{\typ}$ we define 
% \[
    \why{\tmsone} & \eqdef & \set{ \iwhy{\uvar}{\tm} \in \typtms{\why{\lneg{\typ}}} \ST
    \forall \tmtwo \in \tmsone, \ \tm\sub{\uvar}{\tmtwo} \iftyp{\bot} \SNtms{\bot}   }
    \end{array}
\]
%We are now ready to define the reducibility candidates.
%\begin{definition}
\subparagraph*{Reducibility candidates.}
    To each type $\typ$ we associate a set of terms $\redcand{\typ} \subseteq \typtms{\typ}$ called the \emph{reducibility candidates} of the type $\typ$:
%    We construct them inductively on the type.
\[
  \begin{array}{r@{\,\,}c@{\,\,}l@{\HS}r@{\,\,}c@{\,\,}l@{\HS}r@{\,\,}c@{\,\,}l@{\HS}r@{\,\,}c@{\,\,}l}
    \redcand{\btyp} & \eqdef & \SNtms{\btyp}
    &
    \redcand{\lneg{\btyp}} & \eqdef & \SNtms{\nbtyp}
    &
    \redcand{\typ \tensor \typtwo}  & \eqdef & \llneg{(\redcand{\typ} \tensor \redcand{\typtwo})}
    &
    \redcand{\typ \parr  \typtwo} & \eqdef & \lneg{(\redcand{\lneg{\typ}} \tensor \redcand{\lneg{\typtwo}})}
  \\
    \redcand{\bot} & \eqdef & \SNtms{\bot}
    &
    \redcand{\one} & \eqdef & \SNtms{\one}
    &
    \redcand{\why{\typ}} & \eqdef & \llneg{(\why{\redcand{\lneg{\typ}}})}
    &
    \redcand{\ofc{\typ}} & \eqdef & \lneg{(\why{\redcand{\typ}})}
  \end{array}
\]
Given a linear typing environment $\tenv$, an unrestricted typing environment $\utenv$, and a substitution $\subone$, we say $\subone$ is \emph{compatible} with $\tenv,\utenv$,  written $\subone \subsj{\utenv}{\tenv}$, if:
1) for all $\lvar:\typ \in \Gamma$ we have that $\subone(\lvar) \in \redcand{\typ}$; and
2) for all $\uvar:\typtwo \in \Delta$ we have that $\subone(\uvar) \in \redcand{\typtwo}$. 
  
\begin{theorem}[{name=Adequacy~\proofnote{Proof on pg.~\pageref{adequacy:proof}}, restate=[name=Adequacy]Adequacy}]
%\begin{theorem}[Adequacy]
\lthm{adequacy}
For every valid judgment $\jum{\utenv;\tenv}{\tm}{\typ}$
and every $\subone$ such that $\subone \subsj{\utenv}{\tenv}$
we have that $\tm^{\subone} \iftyp{\typ} \redcand{\typ}$.    
\end{theorem}
Taking in particular the identity substitution, we 
obtain \textbf{strong normalization} as a corollary.

%%% Local Variables:
%%% mode: latex
%%% TeX-master: "main"
%%% End:

%\subsection{Confluence} 

\subparagraph*{Confluence.} Given that $\CalcMELL$ enjoys strong normalization, it suffices to prove that it is weak Church--Rosser to obtain confluence. This can be proved through a tedious case analysis.
%An immediate corollary of strong normalization of $\tome$  and weak Church-Rosser of $\tome$ is confluence.

\begin{proposition}
$\tome$ is confluent.
\end{proposition}

Reduction on untyped terms is not confluent. Consider the term $\overline{\invap{(\lam{\lvar}{\underline{(\ap{(\lam{\lvartwo}{(\ewhy{\pair{\lvartwo}{\lvarthree}}{\lvarthree}{\lvar})})}{\tm})}})}{\tmthree}}$. It is not typable since $\lam{\lvartwo}{\pair{\lvartwo}{\lvarthree}}$ does not have type $\bott$. It may be reduced in two different ways, as indicated by the over and underlining. However, the reducts are not joinable.

%%%%%%%%%%%%%%%
\[
  \begin{tikzcd}
    \overline{\invap{(\lam{\lvar}{\underline{(\ap{(\lam{\lvartwo}{(\ewhy{\pair{\lvartwo}{\lvarthree}}{\lvarthree}{\lvar})})}{\tm})}})}{\tmthree}}
    \arrow[r, "\raxparL"]
    \arrow[d, "\raxparR"']
  &
    \linvap{\labtwo'}{\llam{\beta}{\lvar}{\ewhy{\pair{\tm}{\lvarthree}}{\lvarthree}{\lvar}}}{\tmthree}
    \arrow[d,dotted, "\raxparR"']
  \\
    \ap{(\lam{\lvartwo}{(\ewhy{\pair{\lvartwo}{\lvarthree}}{\lvarthree}{\lvar})})}{\tm})\cos{\lvar}{\tmthree} =
    (\iofctwo{\lvarthree}{(\llam{\lab}{\lvartwo}{\pair{\lvartwo}{\lvarthree}})})\eone{\pair{\tm}{\tmthree}}
    \arrow[r,dotted, "/" anchor=center] 
  &
    (\iofctwo{\lvarthree}{\pair{\tm}{\lvarthree}})\eone{\tmthree}
  \end{tikzcd}
\]

%%% Local Variables:
%%% mode: latex
%%% TeX-master: "main"
%%% End:

\section{Translations of classical $\lambda$-calculi into $\CalcMELL$}

The $\CalcMELL$-calculus is capable of simulating Parigot's
$\CalcParigot$ and Curien and Herbelin's $\overline{\lambda}\mu\tilde{\mu}$,
as well as Hasegawa's $\muDCLL$ (in the appendix due to lack of space).
We briefly recall the first two of these calculi below,
and introduce appropriate translations to the $\CalcMELL$-calculus.

\subsection{Simulation of Parigot's $\CalcParigot$ via the T and Q translations}

We briefly recall Parigot's $\CalcParigot$-calculus~\cite{DBLP:conf/lpar/Parigot92},
obtained via the propositions-as-types paradigm from the term assignment of classical
propositional logic in natural deduction with multiple conclusions.
Let $\var,\vartwo,\ldots$ (resp.  $\nvar,\nvartwo,\ldots$) range over a countable set of \defn{variables} (resp. \defn{continuation variables}). Types and terms of $\CalcParigot$ are defined as follows:
\[
  \begin{array}{l@{\,\,}r@{\hspace{.2cm}}c@{\hspace{.2cm}}l@{\hspace{1cm}}l@{\,\,}r@{\hspace{.2cm}}c@{\hspace{.2cm}}l}
    \text{(Types)} &
    \typ,\typtwo & ::= & \bott\,|\, \btyp\,|\,\typ\imp\typtwo &
    \text{(Terms)} &
    \ptm,\ptmtwo & ::= & \var\,|\,\lam{\var}{\ptm}\,|\,\ptm\,\ptmtwo\,|\, \pname{\nvar}{\ptm}\,|\, \pmu{\nvar}{\ptm}
  \end{array}
\]
\defn{Judgements} are of the form $ \pjum{\tenv}{\ptm}{\typ}{\nenv}$, where $\tenv$ is a partial function mapping variables to types, and $\nenv$ is a partial function mapping continuation
variables to types. Typing rules are given by:

\scalebox{\proofScaleFactor}{\begin{minipage}{\textwidth}
  \[
  \indrule{\rulpAx}
  {}
  {\pjum{\tenv,\var:\typ}{\var}{\typ}{\nenv}}
  \hspace{0.2em}
  \indrule{\rulpLam}{
    \pjum{\tenv,\var:\typ}{\ptm}{\typtwo}{\nenv}
  }{
    \pjum{\tenv}{\lam{\var}{\ptm}}{\typ\imp\typtwo}{\nenv}
  }
  \indrule{\rulpName}{
    \pjum{\tenv}{\ptm}{\typ}{\nenv,\nvar:\typ}
  }{
    \pjum{\tenv}{\pname{\nvar}{\ptm}}{\bott}{\nenv,\nvar:\typ}
  }
  %     \indrule{\rulpPA}{
  %   \pjum{\tenv}{\ptm}{\typtwo}{\nenv,\nvar:\typ}
  % }{
  %   \pjum{\tenv}{\mu{\nvar}{\pname{\nvartwo}{\ptm}}}{\typ}{(\nenv,\nvartwo:\typtwo)\setminus \nvar:\typ}
  %   }
\]
\[
  \indrule{\rulpApp}{
    \pjum{\tenv}{\ptm}{\typ\imp\typtwo}{\nenv}
    \HS
    \pjum{\tenv}{\ptmtwo}{\typ}{\nenv}
  }{
    \pjum{\tenv}{\pap{\ptm}{\ptmtwo}}{\typtwo}{\nenv}
  }
  \indrule{\rulpMu}{
    \pjum{\tenv}{\ptm}{\bott}{\nenv,\nvar:\typ}
  }{
    \pjum{\tenv}{\pmu{\nvar}{\ptm}}{\typ}{\nenv}
  }
  %   \\
  %   \\
  % \indrule{\rulpCL}{
  %   \pjum{\tenv,\var:\typ,\vartwo:\typ}{\ptm}{\typthree}{\nenv}
  % }{
  %   \pjum{\tenv, \varthree:\typ}{\ptm\sub{\var,\vartwo}{\varthree}}{\typthree}{\nenv}
  %   }
  %   \quad
  %    \indrule{\rulpCR}{
  %   \pjum{\tenv}{\ptm}{\typthree}{\nenv, \nvar:\typ,\nvartwo:\typ}
  % }{
  %   \pjum{\tenv}{\ptm\sub{\nvar,\nvartwo}{\nvarthree}}{\typthree}{\nenv, \nvarthree:\typ}
  %   }
  %   \\
  %   \\
  %  \indrule{\rulpWL}{
  %   \pjum{\tenv}{\ptm}{\typ}{\nenv}
  % }{
  %   \pjum{\tenv, \var:\typtwo}{\ptm}{\typ}{\nenv}
  %   }
  %   \quad
  %    \indrule{\rulpWR}{
  %   \pjum{\tenv}{\ptm}{\typ}{\nenv}
  % }{
  %   \pjum{\tenv}{\ptm}{\typ}{\nenv, \nvar:\typtwo}
  %   }   
\]
\end{minipage}}
\bigskip

\noindent We write $\ptm\psub{\nvar}{\ptmtwo}$ for \defn{structural substitution}, which replaces subterms of $\ptm$ of the form $\pname{\nvar}{\ptmthree}$ by $\pname{\nvar}{(\pap{\ptmthree}{\ptmtwo})}$; and $\ptm\rensub{\nvartwo}{\nvar}$ for the renaming of all free occurrences of $\nvartwo$ in $\ptm$ with $\nvar$.

%%%%%%%%%%%%%%%%%%%%%%%%%%%%%%
\subparagraph*{T-Translation.}
%%%%%%%%%%%%%%%%%%%%%%%%%%%%%%
  The T-translation, first introduced by Danos~\etal~\cite{Danos_Joinet_Schellinx_1995}, embeds classical
  logic into linear logic. It operates on both formulae and sequent-calculus proofs.
  Here we extend it to translate $\CalcParigot$-terms to $\CalcMELL$-terms,
  in such a way that it simulates reduction.
%\begin{definition}[Reduction for $\CalcParigot$]
  Reduction in $\CalcParigot$, denoted $\tomu$, is defined as the contextual closure of the following rules:
  \begin{center}
  \begin{tabular}{p{5cm}@{\hspace{1cm}}p{6.2cm}}
    $\begin{array}{llll}
      \pap{(\lam{\var}{\ptm})}{\ptmtwo} & \toax{\beta} & \ptm\sub{\var}{\ptmtwo} \\
      \pap{(\pmu{\nvar}{\ptm})}{\ptmtwo} & \toax{\mu} & \pmu{\nvar}{(\ptm\psub{\nvar}{\ptmtwo})} \\
      % \pname{\nvar}{(\pmu{\nvartwo}{\ptm})} & \toax{\rho} & \ptm \rensub{\nvartwo}{\nvar}\\
      % \pmu{\nvar}{\pname{\nvar}{\ptm}}  & \toax{\theta} & \ptm &   (\nvar\notin\fv{\ptm})
     \end{array}$
                                                      &
                                                        $\begin{array}{llll}
      % \pap{(\lam{\var}{\ptm})}{\ptmtwo} & \toax{\beta} & \ptm\sub{\var}{\ptmtwo} \\
      % \pap{(\pmu{\nvar}{\ptm})}{\ptmtwo} & \toax{\mu} & \pmu{\nvar}{(\ptm\psub{\nvar}{\ptmtwo})} \\
      \pname{\nvar}{(\pmu{\nvartwo}{\ptm})} & \toax{\rho} & \ptm \rensub{\nvartwo}{\nvar}\\
      \pmu{\nvar}{\pname{\nvar}{\ptm}}  & \toax{\theta} & \ptm &   (\nvar\notin\fv{\ptm})
     \end{array}$
  \end{tabular}
  \end{center}
%\end{definition}
%Recall from \rremark{derivable:rules:for:ofc:and:lam} that $\ofc{\tm}$ is a shorthand for $\iofctwo{\lvar}{(\ibott{\tm}{\lvar})}$.
Let us write $\ofc{\tm} \eqdef \iofctwo{\lvar}{(\ibott{\tm}{\lvar})}$.
The T-translation of a formula and a $\CalcParigot$-term are given by:
\begin{center}
  \begin{tabular}{p{5cm}p{6.2cm}}
    $\begin{array}{rcl}
    \ttra{\bott} & \eqdef & \bott\\
    \ttra{\btyp} & \eqdef & \btyp\\
    \ttra{(\typ\imp\typtwo)} &\eqdef & \why{\ofc{\lneg{(\ttra{\typ})}}}\parr\why{\ttra{\typtwo}} \\
%    \ttra{(\neg\typ)} & \eqdef & \why{\ofc{\lneg{(\ttra{\typ})}}}                                
     \end{array}$                                  
                                  &
  $\begin{array}{rcl}
    \ttrat{\var}{\vark} & \eqdef & \ibott{\var}{\ofc{\vark}} \\
    \ttrat{(\lam{\var}{\ptm})}{\vark} & \eqdef & \ibott{\ipar{\lvar}{\lvartwo}{\ttrat{\ptm}{\varktwo}\eofc{\var}{\lvar}\eofc{\varktwo}{\lvartwo}}}{\vark} \\
    \ttrat{(\ptm\,\ptmtwo)}{\vark} & \eqdef & \ttrat{\ptm}{\varktwo}\sub{\varktwo}{\pair{\ofc{\iwhy{\varkthree}{\ttrat{\ptmtwo}{\varkthree}}}}{\ofc{\vark}}} \\
    \ttrat{(\pname{\nvar}{\ptm})}{\vark} & \eqdef & \ibott{\ttrat{\ptm}{\varktwo}\sub{\varktwo}{\nvar}}{\vark}
    \\
    \ttrat{(\pmu{\nvar}{\ptm})}{\vark} & \eqdef & \ttrat{\ptm}{\varktwo}\sub{\varktwo}{\ione}\sub{\nvar}{\vark}
  \end{array}$
  \end{tabular}
\end{center}

\begin{lemma}[{name=Type preservation~\proofnote{Proof on pg.~\pageref{type_preservation_for_t_translation:proof}}, restate=[name=Type preservation]TypePreservationForTTranslation}]
%  \begin{lemma}
  \llem{type_preservation_for_t_translation}
  If $\pjum{\tenv}{\ptm}{\typ}{\nenv}$ holds in $\CalcParigot$, then $\djum{\why{\ttra{\tenv}},\lneg{(\ttra{\nenv})},\vark:\lneg{(\ttra{\typ})}}{\emptytenv}{\ttrat{\ptm}{\vark}}{\bott}$ holds in $\CalcMELL$, for $\vark$ a fresh linear variable.
  
\end{lemma}

% Substitution, structural substitution and renaming commute with $\ttrat{\bullet}{}$:
% \begin{itemize}
%   \item $\ttrat{\ptm}{\vark}\sub{\var}{\iwhy{\varktwo}{\ttrat{\ptmtwo}{\varktwo}}}
% \rttome
% \ttrat{\ptm\sub{\var}{\ptmtwo}}{\vark}$.
% \item  $\ttrat{\ptm\psub{\nvar}{\ptmtwo}}{\vark} =
%   \ttrat{\ptm}{\vark}\sub{\nvar}{\pair{\ofc{\iwhy{\varktwo}{\ttrat{\ptmtwo}{\varktwo}}}}{\ofc{\nvar}}}$.
% \item $\ttrat{\ptm \rensub{\nvartwo}{\nvar}}{\vark}   = \ttrat{\ptm}{\vark} \sub{\nvartwo}{\nvar}$.
% \end{itemize}

\begin{proposition}[{name=Simulation ~\proofnote{Proof on pg.~\pageref{calcMELL_simulates_calcParigot_via_t_translation:proof}}, restate=[name=Simulation]CalcMELLSimulatesCalcParigotViaTTranslation}]
%  \begin{proposition}
  If $\ptm \tomu \ptmtwo$, then  $\ttrat{\ptm}{\vark} \rttome \ttrat{\ptmtwo}{\vark}$. 
\end{proposition}

%%%%%%%%%%%%%%%%%%%%%%%%%%%
\subparagraph*{Q-Translation.}
%%%%%%%%%%%%%%%%%%%%%%%%%%%
The Q-translation, introduced also in~\cite{Danos_Joinet_Schellinx_1995},
is an alternative embedding of classical logic into linear logic.
While the T-translation is a \emph{call-by-name} translation,
it is well-known that the Q-translation is instead a \emph{call-by-value} translation.
To study the Q-translation, we need to first recall a notion of call-by-value
reduction for the $\CalcParigot$-calculus~\cite{Py1998}.

The types and terms of $\CalcParigotVal$ are the same as those of $\CalcParigot$.
We use the notations $\ptm\psub{\nvar}{\pval}$ for the usual
(right-)structural substitution, $\ptm\psubtwo{\nvar}{\pval}$ for
left-structural substitution.
%and $\ptm \rensub{\nvartwo}{\nvar}$ for renaming.
  Reduction in $\CalcParigotVal$, denoted $\to_{\muv}$, is defined as the contextual closure of the following reduction axioms:
 \begin{center}
  \begin{tabular}{p{5cm}@{\hspace{1cm}}p{6.2cm}}
    $\begin{array}{llll}
      \pap{(\lam{\var}{\ptm})}{\pval} & \toax{\betav} & \ptm\sub{\var}{\pval} \\
      \pap{(\pmu{\nvar}{\ptm})}{\pval} & \toax{\muv} & \pmu{\nvar}{(\ptm\psub{\nvar}{\pval})} \\
      \pap{\pval}{(\pmu{\nvar}{\ptm})} & \toax{\muvprime} &
                                                       \pmu{\nvar}{(\ptm\psubtwo{\nvar}{\pval})} \\
      % \pname{\nvar}{(\pmu{\nvartwo}{\ptm})} & \toax{\rho} & \ptm \rensub{\nvartwo}{\nvar}\\
      % \pmu{\nvar}{\pname{\nvar}{\ptm}}  & \toax{\theta} & \ptm &   (\nvar\notin\fv{\ptm})
     \end{array}$
                                                     &
                                                         $\begin{array}{llll}
      % \pap{(\lam{\var}{\ptm})}{\pval} & \toax{\betav} & \ptm\sub{\var}{\pval} \\
      % \pap{(\pmu{\nvar}{\ptm})}{\pval} & \toax{\muv} & \pmu{\nvar}{(\ptm\psub{\nvar}{\pval})} \\
      % \pap{\pval}{(\pmu{\nvar}{\ptm})} & \toax{\muvprime} &
      %                                                  \pmu{\nvar}{(\ptm\psubtwo{\nvar}{\pval})} \\
      \pname{\nvar}{(\pmu{\nvartwo}{\ptm})} & \toax{\rho} & \ptm \rensub{\nvartwo}{\nvar}\\
      \pmu{\nvar}{\pname{\nvar}{\ptm}}  & \toax{\theta} & \ptm &   (\nvar\notin\fv{\ptm})
     \end{array}$
  \end{tabular}
  \end{center}
The Q-translation of types is defined as follows:
  \[\begin{array}{r@{\hspace{.1cm}}c@{\hspace{.1cm}}l@{\hspace{.5cm}}r@{\hspace{.1cm}}c@{\hspace{.1cm}}l@{\hspace{.5cm}}r@{\hspace{.1cm}}c@{\hspace{.1cm}}l}
    \qtra{\bott} & \eqdef & \ofc{\bott} &
    \qtra{\btyp} & \eqdef & \ofc{\btyp} &
    \qtra{(\typ\imp\typtwo)}   &  \eqdef & \ofc{(\lneg{(\qtra{\typ})}\parr\why{\qtra{\typtwo}})} \\
    %\qtra{(\neg\typ)} & \eqdef & \ofc{\lneg{(\qtra{\typ})}}
  \end{array}\]
We define $\qtratwo{\typ}$ as the function such that $\qtra{\typ}=\ofc{\qtratwo{\typ}}$.   Values are defined as:
 % \[
 %  \begin{array}{rcl}
    $\pval  ::=  \var\,|\,\lam{\var}{\ptm}$.
%   \end{array}
% \]
The $\qtra{\arg}$-translation for values and terms, are given mutually recursively by
\begin{center}
  \begin{tabular}{p{6.2cm}p{5.9cm}}
  $\begin{array}{rcl}
          \qtratwo{\var} & \eqdef & \var \\
     \qtratwo{(\lam{\var}{\ptm})} & \eqdef & \ipar{\lvar}{\lvartwo}{\qtrat{\ptm}{\vark}\eofc{\var}{\lvar}\eofc{\vark}{\lvartwo}}
   \end{array}$
                                             &
                                               $\begin{array}{rcl}
                                                  \qtrat{\pval}{\vark} & \eqdef & \ibott{\ofc{\qtratwo{\pval}}}{\vark} \\
    % \qtra{(\lam{\var}{\ptm})} & \eqdef & (\ibott{(\ipar{\lvar}{\lvartwo}{\qtra{\ptmtwo}\eofc{\var}{\lvar}\eofc{\vark}{\lvartwo}})}{\lvarthree})\ewhy{\lvarthree}{\vark} \\

    %     \qtra{(\lam{\var}{\ptm})} & \eqdef &
    % \ibott{\ofc{(\ipar{\lvar}{\lvartwo}{\qtra{\ptm}\eofc{\var}{\lvar}\eofc{\vark}{\lvartwo}})}}{\vark}\\
    
    % \qtra{(\ptm\,\ptmtwo)} & \eqdef & \qtra{\ptm}\sub{\vark}{\iwhy{\uvartwo}{\qtra{\ptmtwo}\sub{\vark}{\invap{\uvartwo}{\ofc{\iwhy{\uvar}{(\ibott{\lvar}{\uvar})\ewhy{\lvar}{\vark}}}}}}}  \\
    \qtrat{(\ptm\,\ptmtwo)}{\vark} & \eqdef & \qtrat{\ptm}{\varkthree}\sub{\varkthree}{\iwhy{\uvartwo}{\qtrat{\ptmtwo}{\varktwo}\sub{\varktwo}{\invap{\uvartwo}{\ofc{\vark}}}}} \\
    \qtrat{(\pname{\nvar}{\ptm})}{\vark} & \eqdef & \ibott{\qtrat{\ptm}{\varktwo}\sub{\varktwo}{\nvar}}{\vark}
    \\
    \qtrat{(\pmu{\nvar}{\ptm})}{\vark} & \eqdef & \qtrat{\ptm}{\varktwo}\sub{\varktwo}{\ione}\sub{\nvar}{\vark}

                                                \end{array}$
  \end{tabular}
  \end{center}

\begin{lemma}[{name=Type preservation~\proofnote{Proof on pg.~\pageref{type_preservation_for_q_translation:proof}}, restate=[name=Type preservation]TypePreservationForQTranslation}]
%\begin{lemma}
\llem{type_preservation_for_q_translation}
  If $\pjum{\tenv}{\ptm}{\typ}{\nenv}$ holds in $\CalcParigotVal$, then $\djum{\qtratwo{\tenv},\lneg{(\qtra{\nenv})},\vark:\lneg{(\qtra{\typ})}}{\emptytenv}{\qtra{\ptm}}{\bott}$ holds in $\CalcMELL$.
\end{lemma}

% Substitution, structural substitution and renaming commute with $\ttrat{\bullet}{}$:
% \begin{itemize}
%   \item  $\qtrat{\pval}{\vark}\sub{\vark}{\iwhy{\var}{\tm}} \tome \tm\sub{\var}{\qtratwo{\pval}}$.
%   \item  $\qtrat{\ptm}{\vark}\sub{\var}{\qtratwo{\pval}} = \qtrat{\ptm\sub{\var}{\pval}}{\vark}$.
%     \item $\qtrat{\ptm \rensub{\nvartwo}{\nvar}}{\vark}   = \qtrat{\ptm}{\vark} \sub{\nvartwo}{\nvar}$.
% \end{itemize}

\begin{proposition}[{name=Simulation~\proofnote{Proof on pg.~\pageref{calcMELL_simulates_calcParigot_via_q_translation:proof}}, restate=[name=Simulation]CalcMELLSimulatesCalcParigotViaQTranslation}]
% \begin{proposition}
  If $\ptm \to_{\muv} \ptmtwo$, then $\qtrat{\ptm}{\vark} \eqme \qtrat{\ptmtwo}{\vark}$, where $\eqme$ denotes the reflexive, symmetric, and transitive closure of $\tome$.
\end{proposition}

\subsection{Simulation of Curien and Herbelin's $\overline{\lambda}\mu\tilde{\mu}$ via the T-translation}

Curien and Herbelin~\cite{DBLP:conf/icfp/CurienH00} introduced a term assignment for classical logic in sequent calculus style that has since become a cornerstone in the study of computational interpretations of classical logic.  Its merit is that it reveals hidden dualities in the programming language from the natural symmetries of the sequent calculus (\cf~\cite{DOWNEN_ARIOLA_2018} for a tutorial). 

%\frametitle{Translation of $\bar{\lambda}\mu\tilde{\mu}$ into $\CalcMELL$ \hfill (1/2)}

We start by recalling the sets of types and terms of $\bar{\lambda}\mu\tilde{\mu}$,
which are are defined as follows.
Note that terms are split into three syntactic categories (terms, co--terms and commands):
 \begin{center}
  \begin{tabular}{p{5cm}p{6.2cm}}
  $\begin{array}{lrcl}
    \text{(Types)} &
    \typ,\typtwo & ::= & {\btyp} \mid {\typ\imp\typtwo}
   \end{array}$
    &
      $\begin{array}{lrcl}
      \text{(Terms)} &
    \mval,\mval' & ::= &
           \var
%      \mid \UNDER{\mmu{\nvar^{{\neg\typ}}}{\mcmd^{{\bot}}}}{{\typ}}
      \mid \mmu{\mvar}{\mcmd}
%      \mid \UNDER{\lam{\var^{{\typ}}}{\mval^{{\typtwo}}}}{{\typ\imp\typtwo}}
      \mid \lam{\var}{\mval}
    \\
      \text{(Co-terms)} &
    \menv,\menv' & ::= &
           \mvar
%      \mid \UNDER{\mapp{\mval^{{\typ}}}{\menv^{{\neg\typtwo}}}}{{\neg(\typ\imp\typtwo)}}
      \mid \mapp{\mval}{\menv}
    \\
      \text{(Commands)} &
    \mcmd,\mcmd' & ::= &
%           \UNDER{\mabs{\mval^{{\typ}}}{\menv^{{\neg\typ}}}}{{\bot}}
           \mabs{\mval}{\menv}
       \end{array}$
  \end{tabular}
\end{center}
The typing system consists of three judgements $\Gamma \vdash
\mval:\typ \mid \Delta$, $\Gamma \mid \menv:\typ \vdash \Delta$,
and $c: (\Gamma \vdash \Delta)$,
with the following typing rules:

\scalebox{\proofScaleFactor}{\begin{minipage}{\textwidth}
    \[
    \begin{array}{c}
      \indrule{\indrulename{VarR}}
      {}
      {\Gamma, \var:\typ\vdash \var:\typ\mid \Delta}
      \quad
      \indrule{\indrulename{VarL}}
      {}
      {\Gamma \mid \alpha:\typ \vdash \alpha:\typ, \Delta}
      \quad
    \indrule{\indrulename{AppL}}
      {\Gamma \vdash \mval:\typ \mid \Delta
      \quad
      \Gamma \mid \menv:\typtwo \vdash \Delta
      }
      {\Gamma \mid (\mapp{\mval}{\menv}):\typ\imp\typtwo \vdash \Delta}
      \\
      \\
    \indrule{\indrulename{ActivR}}
    {\mcmd:(\Gamma\vdash \mvar:\typtwo, \Delta)}
      {\Gamma \vdash \mmu{\mvar}{\mcmd}:\typtwo, \Delta}
      \quad
      \indrule{\indrulename{Abs}}
      {\Gamma,\var:\typ \vdash \mval:\typtwo \mid \Delta}
      {\Gamma \vdash \lam{\var}{\mval} :\typ\imp\typtwo \mid \Delta}
      \quad
      \indrule{\indrulename{cut}}
      {\Gamma \vdash \mval :\typ \mid \Delta
      \quad
      \Gamma \mid \menv:\typ \vdash \Delta
      }
      {\mabs{\mval}{\menv}:(\Gamma \vdash \Delta)}

\end{array}
\]
\end{minipage}}
\bigskip

\noindent
Reduction in $\bar{\lambda}\mu\tilde{\mu}$ is given by the contextual closure of
the following two reduction axioms, where $\mval_1\sub{\var}{\mval_2}$ and $ \mcmd\sub{\mvar}{\menv}$ are standard substitution:
  \[
    \begin{array}{c@{\hspace{1cm}}c}
  \begin{array}{llll}
    \mabs{\lam{\var}{\mval_1}}{\mapp{\mval_2}{\menv}}
    & \toax{} &
    \mabs{\mval_1\sub{\var}{\mval_2}}{\menv}
  % \\
  %   \mabs{\mmu{\mvar}{\mcmd}}{\menv}
  %   & \toax{}  &
  %   \mcmd\sub{\mvar}{\menv}
  \end{array}
      &
          \begin{array}{llll}
  %   \mabs{\lam{\var}{\mval_1}}{\mapp{\mval_2}{\menv}}
  %   & \toax{} &
  %   \mabs{\mval_1\sub{\var}{\mval_2}}{\menv}
  % \\
    \mabs{\mmu{\mvar}{\mcmd}}{\menv}
    & \toax{}  &
    \mcmd\sub{\mvar}{\menv}
  \end{array}
    \end{array}    
\]

%\begin{definition}{T-translation for $\bar{\lambda}\mu\tilde{\mu}$ (formulae, judgments and terms)}
%\ldef{t_translation_of_barLambdaMuTildeMu}  
The T-translation of Danos~\etal~\cite{Danos_Joinet_Schellinx_1995}
is defined for formulae, judgements and terms of $\bar{\lambda}\mu\tilde{\mu}$,
as follows:
\begin{center}
  \begin{tabular}{p{6.4cm}p{6.2cm}}
 $
  \begin{array}{rcl}
    \ttra{\btyp} & \eqdef & \btyp\\
    \ttra{(\typ\imp\typtwo)} &\eqdef & \why{\ofc{\lneg{(\ttra{\typ})}}}\parr\why{\ttra{\typtwo}} \\
  \end{array}$
    &
  $\begin{array}{rcl}
    \mcmd : (\mjcmd{\mctx}{\mctxtwo})
    & \mapsto &
    \djum{\why{\ttra{\mctx}},{\ttra{\mctxtwo}}^\bot}{\emptytenv}{\ttra{\mcmd}}{\bot}
  \\
    \mjenv{\mctx}{\menv :\typ}{\mctxtwo}
    & \mapsto &
    \djum{\why{\ttra{\mctx}},{\ttra{\mctxtwo}}^\bot}{\emptytenv}{\ttra{\menv}}{\ofc{(\lneg{{\ttra{\typ}}})}}
  \\
    \mjval{\mctx}{\mval : \typ}{\mctxtwo}
    & \mapsto &
    \djum{\why{\ttra{\mctx}},{\ttra{\mctxtwo}}^\bot}{\emptytenv}{\ttra{\mval}}{\why{\ttra{\typ}}}
   \end{array}$
    \\
  $\begin{array}{rcl}
    \ttra{\var} & \eqdef & \var
  \\
    \ttra{(\mmu{\mvar}{\mcmd})} & \eqdef & \iwhy{\mvar}{\ttra{\mcmd}}
  \\
    \ttra{(\lam{\var}{\mval})} & \eqdef &
      \iwhy{\uvar}{
        (\ibott{\uvar}{
          \ipar{\lvar}{\lvartwo}{
            (\ibott{\ttra{\mval}}{\lvartwo}\eofc{\var}{\lvar})
          }
                                          })}
   \end{array}$
                                &
                                  $\begin{array}{rcl}
  \\
    \ttra{\mvar} & \eqdef & \ofc{\mvar}
  \\
    \ttra{(\mapp{\mval}{\menv})} & \eqdef & \ofc{\pair{\ofc{\ttra{\mval}}}{\ttra{\menv}}}
  \\
    \ttra{\mabs{\mval}{\menv}} & \eqdef & \ibott{\ttra{\mval}}{\ttra{\menv}}
  \end{array}
                                          $
  \end{tabular}
\end{center}

\begin{lemma}[{name=Type preservation~\proofnote{Proof on pg.~\pageref{type_preservation_for_t_translation_for_barLambdaMuTildeMu:proof}}, restate=[name=Type preservation]TypePreservationForTTranslationForBarLambdaMuTildeMu}]
%  \begin{lemma}
\llem{type_preservation_for_t_translation_for_barLambdaMuTildeMu}
\quad
\begin{enumerate}
  \item  If $\mcmd : (\mjcmd{\mctx}{\mctxtwo})$ holds in $\bar{\lambda}\mu\tilde{\mu}$, then 
    $\djum{\why{\ttra{\mctx}},{\ttra{\mctxtwo}}^\bot}{\emptytenv}{\ttra{\mcmd}}{\bot}$ holds in $\CalcMELL$.
\item If $ \mjenv{\mctx}{\menv :\typ}{\mctxtwo}$ holds in $\bar{\lambda}\mu\tilde{\mu}$, then   
    $\djum{\why{\ttra{\mctx}},{\ttra{\mctxtwo}}^\bot}{\emptytenv}{\ttra{\menv}}{\ofc{(\lneg{{\ttra{\typ}}})}}$ holds in $\CalcMELL$.
\item  If $ \mjval{\mctx}{\mval:\typ}{\mctxtwo}$ holds in $\bar{\lambda}\mu\tilde{\mu}$, then 
    $\djum{\why{\ttra{\mctx}},{\ttra{\mctxtwo}}^\bot}{\emptytenv}{\ttra{\mval}}{\why{\ttra{\typ}}}$ holds in $\CalcMELL$.
\end{enumerate}
\end{lemma}

\begin{proposition}[{name=Simulation~\proofnote{Proof on pg.~\pageref{calcMELL_simulates_calcParigot_via_t_translation_for_barLambdaMuTildeMu:proof}}, restate=[name=Simulation]CalcMELLSimulatesCalcParigotViaTTranslationForBarLambdaMuTildeMu}]
%  \begin{proposition}
\lprop{calcMELL_simulates_calcParigot_via_t_translation_for_barLambdaMuTildeMu}
  If $\tm \to_{\bar{\lambda}\mu\tilde{\mu}} \tmtwo$, then  $\ttrat{\tm}{\vark} \rttome \ttrat{\tmtwo}{\vark}$.
\end{proposition}

%%% Local Variables:
%%% mode: latex
%%% TeX-master: "main"
%%% End:

% \section{Intuitionistic Fragment}
% \input{06-intuitionistic-fragment}

% \section{Subformula Property}
% \input{07-subformula-property}

\section{Related Work and Conclusions}

\subparagraph*{Proofs-as-Processes.} Although not the focus of our work, there is a large body of work on interpreting cut-elimination in terms of concurrent or parallel programs. This program was perhaps kickstarted by Abramsky~\cite{DBLP:journals/tcs/Abramsky93} and likely influenced by Girard himself~\cite{DBLP:journals/tcs/Girard87} (having mentioned that ``\textit{connectives of linear logic have an obvious meaning in terms of parallel computation...}''). Since then many works have been presented in this area. Relatively recent highlights include the tight correspondence between intuitionistic linear logic and session types by Caires and Pfenning~\cite{10.1007/978-3-642-15375-4_16}, its extension to classical linear logic by Wadler~\cite{Wadler-PropositionsAsSessions,WADLER_2014} and modeling asynchronous communication via a term assignment for a multi-conclusion natural deduction presentation of \MLL~\cite{DBLP:journals/pacmpl/AschieriG20}. Further references are provided in~\cite[Sec.11]{DBLP:journals/pacmpl/AschieriG20}.

\subparagraph*{Proofs-as-Functions.} Troelstra~\cite{Troelstra1992-TROLOL} introduces a one-sided natural deduction presentation of \LL. However, par and why-not are not included, and normalization of proofs is not studied. Martini and Masini~\cite{DBLP:journals/tcs/MartiniM97} study the par connective, arguing that it had not received any attention from a natural deduction perspective because of the ``\textit{apparent impossibility to formulate a suitable introduction-elimination pair, when sticking to single conclusion systems}''. Our rules for par are similar to theirs (the elimination rules are exactly the same, the introduction rule is not but is provably equivalent). A term assignment is not considered in their work nor are the exponentials.
Ramalho Martins et al~\cite{DBLP:journals/igpl/MartinsM04} also introduce a natural deduction presentation for full \LL. All connectives are considered, including the additives. Only weak normalisation is proved and no term calculus or computational interpretation of proof normalisation is discussed.
Albrecht et al~\cite{DBLP:journals/tcs/AlbrechtCJ97} propose a term assignment for \LL. They include an inference rule called \emph{swap} that corresponds to our (derived) notion of contra-substitution but the similarities with our work stop there. There are several shortcomings with their proposed system ${\cal N}$ including a global restriction on well-formedness of derivations  (\cf Sec. 2.4 on pg 221 of \emph{op.cit.}) that requires a complicated notion of term reduction (parallel reduction) to avoid breaking it\footnote{Moreover, we believe ${\cal N}$ is incomplete with respect to linear logic. The issue is in the case that shows that linear logic contraction is captured in ${\cal N}$. Case (?C) of Lemma A.1 does not consider the situation where $n=0$ in the proof of $\vdash \why\alpha,\why\alpha,\beta_1\ldots,\beta_n$. This case does not seem to be provable in ${\cal N}$.}.
Bierman~\cite{DBLP:journals/tcs/Bierman99} presents a Parigot style calculus for full \LL which relies on a multiple-conclusion natural deduction. The term assignment makes use of Parigot's structural substitution. It includes many commuting conversions (most of which are not listed). It also mentions the possibility of defining a rule for par but no discussion is present in the paper (\cite[Sec.3.4.]{DBLP:journals/tcs/Bierman99}).
Mazurak and Zdancewic~\cite{DBLP:conf/icfp/MazurakZ10} propose a term assignment for \MLL. They introduce a functional language with concurrency primitives, the latter motivated by a parallel rendering of Felleisen's control and abort control operators. There are various items that set our work apart from this. We address \MELL which includes the exponentials. Our system $\CalcMELL$ is shown to be complete with respect to $\MELL$; in~\cite{DBLP:conf/icfp/MazurakZ10} this question is left open for their calculus. Finally, our operational semantics arises purely from proof normalization, whereas their calculus includes various term constructors that do not arise from purely proof theoretical means. Examples are the term constructors for the new binder and parallel composition of processes. Finally, there are works that provide a computational intepretation of \LL based on the proof-nets presentation. A recent contribution is Kesner's~\cite{DBLP:journals/pacmpl/Kesner22} where terms are associated to proof-nets. See \emph{op.cit} for further references on the proof-net approach.

\subparagraph*{Future Work.} Extending our approach to include the additives is not immediate since contra-substitution relies crucially on linear terms for its definition. This requires further investigation. Also, simulating our calculus, in particular structural equivalence, via proof-nets would be of interest. Extending $\CalcMELL$ with second-order propositional quantifiers should be unproblematic. However, adding fixed points in the style of Baelde et al~\cite{DBLP:journals/tocl/Baelde12,10.1109/LICS.2012.22} sounds more challenging.

\newpage
\appendix

\section{A Contraposition-Based Calculus for \MLL}
\lsec{app:calcMLL}

%% Appendix for MLL

\CompletenessSoundnessMLL*
  
\begin{proof}\label{compl_sound_MLL:proof}
\textbf{Soundness.}
  By induction on the derivation of the judgment:

    \begin{enumerate}
    \item \rulmAx:
      Then $t = a : A$ and $\tenv = (a : {\typ})$ and the derivation is of the form:
      \[
        \indrule{\rulmAx}{
          \emptyPremise
        }{
          \jum{\lvar:\typ}{\lvar}{\typ}
        }
      \]
      Hence we can take:
      \[\indrule{\rullAx}{
        \emptyPremise
      }{
        \jull{\typ,\lneg{\typ}}
      }\]

    \item \rulmITensor:
      Then $\tm = \pair{\tmtwo}{\tmthree}$ and $\typ = \typtwo\tensor\typthree$
      and the derivation is of the form:
      \[
        \indrule{\rulmITensor}{
          \jum{\tenv}{\tmtwo}{\typtwo}
          \HS
          \jum{\tenvtwo}{\tmthree}{\typthree}
        }{
          \jum{\tenv,\tenvtwo}{\pair{\tmtwo}{\tmthree}}{\typtwo\tensor\typthree}
        }
      \]
      Hence:
      \[
        \indrule{\rullTensor}{
          \jull{\lneg{\tenv},\typtwo}
          \HS
          \jull{\lneg{\tenvtwo},\typthree}
        }{
          \jull{\lneg{\tenvtwo},\lneg{\tenv},\typtwo\tensor\typthree}
        }
      \]
    \item \rulmETensor:
      The derivation is of the form
      \[
      \indrule{\rulmETensor}{
        \jum{\tenvtwo}{\tmtwo}{\typtwo\tensor\typthree}
        \HS
        \jum{\tenvthree,\lvartwo:\typtwo,\lvarthree:\typthree}{\tmthree}{\typ}
      }{
        \jum{\tenvtwo,\tenvthree}{\casepair{\tmtwo}{\lvartwo}{\lvarthree}{\tmthree}}{\typ}
      }
      \]
      Hence by IH:
      
      \[
      \indrule{\rullCut}{
        \indih{
          \jull{\lneg{\tenvtwo},\typtwo \tensor \typthree}
        }
        \indrule{\rullLimp}{
                \indih{
            \jull{\lneg{\tenvthree},\lneg \typtwo,\lneg \typthree,\typ}
          }
        }{
          \jull{
            \lneg{\tenvthree},
            \typtwo \limp \lneg \typthree,
            \typ
          }
        }
      }{
        \jull{
          \lneg{\tenvtwo},\lneg{\tenvthree},
          \typ
        }
      }
      \]

    \item \rulmILimp:
      Then $\tm = \lam{\lvartwo}{\tm'}$ and $\typ = \typtwo \limp \typthree$ and the derivation is of the form:
      \[
        \indrule{\rulmILimp}{
          \jum{\tenv,\lvartwo:\typtwo}{\tm'}{\typthree}
        }{
          \jum{\tenv}{\lam{\lvartwo}{\tm'}}{\typtwo\limp\typthree}
        }
      \]
      Therefore by IH:
      \[
        \indrule{\rullLimp}{
          \indih{
            \jull{\lneg{\tenv},\lneg \typtwo,\typthree}
          }
          
        }{
          \jull{
            \lneg{\tenv},
            \typtwo \limp \typthree
          }
        }
      \]
      
    \item \rulmELimpOne:
    Then $\tm =\ap { \tmtwo }{\tmthree}$ and the derivation is of the form:
    \[
    \indrule{\rulmELimpOne}{
      \jum{\tenvtwo}{\tmtwo}{\typtwo\limp\typ}
      \HS
      \jum{\tenvthree}{\tmthree}{\typtwo}
    }{
      \jum{\tenvtwo,\tenvthree}{\ap{\tmtwo}{\tmthree}}{\typ}
    }
    \]
    Hence by IH:
    \[
    \indrule{\rullCut}{
      \indih{
        \jull{\lneg{\tenvtwo},\typtwo \limp \typ}
      }
      \indrule{\rullTensor}{
        \indih{
          \jull{\lneg{\tenvthree},\typtwo}
        }
        \indrule{\rullAx}{
          \emptyPremise
        }{
          \jull{\typ,\lneg{\typ}}
        }
      }{
        \jull{
          \lneg{\tenvthree},
          \typ,
          \typtwo \tensor \lneg \typ
        }
      }
    }{
      \jull{
        \lneg{\tenvtwo},\lneg{\tenvthree},\typ
      }
    }
    \]
    \item \rulmELimpTwo:
    Then $\tm ={\invap{\tmtwo}{\tmthree}} $ and the derivation is of the form:
    \[
      \indrule{\rulmELimpTwo}{
        \jum{\tenvtwo}{\tmtwo}{\lneg \typ\limp\typtwo}
        \HS
        \jum{\tenvthree}{\tmthree}{\lneg{\typtwo}}
      }{
        \jum{\tenvtwo,\tenvthree}{\invap{\tmtwo}{\tmthree}}{{\typ}}
      }
    \]  
    Therefore by IH:
    \[
      \indrule{\rullCut}{
        \indih{
          \jull{\lneg{\tenvtwo},\lneg \typ \limp \typtwo}
        }
        \indrule{\rullTensor}{
          \indih{
            \jull{\lneg{\tenvthree},\typtwo}
          }
          \indrule{\rullAx}{
            \emptyPremise
          }{
            \jull{\typ,\lneg{\typ}}
          }
        }{
          \jull{
            \lneg{\tenvthree},
            \typ,
            \lneg \typ \tensor \lneg \typtwo
          }
        }
      }{
        \jull{
          \lneg{\tenvtwo},\lneg{\tenvthree},\typ
        }
      }
    \]
    
    \end{enumerate}

\noindent
\textbf{Completeness.}
  By hypothesis, $\jull{\tenv_0}$ holds in \MLL.
  We proceed by induction on a cut-free derivation of $\jull{\tenv_0}$,
  resorting to the (well-known) cut-elimination theorem for \MLL.
  \begin{enumerate}
  \item \rullAx:
  
    Then $\Gamma_{0} = A$ and the derivation is of the form:
    \[
      \indrule{\rullAx}{\emptyPremise}{\jull{\lneg{A}, A}}
    \]
    and we have 
    \[
      \indrule{\rulmAx}{\emptyPremise}{\jum{A}{a}{A}}
    \]

  \item \rullTensor:
    Then $\tenv_0$ is a permutation of
    $\tenvtwo,\tenvthree,\typtwo\tensor\typthree$
    and the derivation is of the form:
    \[
      \indrule{\rullTensor}{
        \jull{\tenvtwo,\typtwo}
        \HS
        \jull{\tenvthree,\typthree}
      }{
        \jull{\tenvtwo,\tenvthree,\typtwo\tensor\typthree}
      }
    \]
    Moreover, we know by hypothesis
    that $\tenvtwo,\tenvthree,\typtwo\tensor\typthree$
    is a permutation of $\tenv,\typ$.
    We consider three subcases,
    depending on whether
       there is an occurrence of $\typ$ in $\tenvtwo$,
    or there is an occurrence of $\typ$ in $\tenvthree$,
    or $\typ = \typtwo\tensor\typthree$:
    \begin{enumerate}
    \item
      If there is an occurrence of $\typ$ in $\tenvtwo$:
      then $\tenvtwo = \tenvtwo',\typ$ and the derivation is of the form:
      \[
        \indrule{\rullTensor}{
          \jull{\tenvtwo',\typ,\typtwo}
          \HS
          \jull{\tenvthree,\typthree}
        }{
          \jull{\tenvtwo',\typ,\tenvthree,\typtwo\tensor\typthree}
        }
      \]
      Hence by \ih there exist terms $\tmtwo,\tmthree$ such that:
      \[
        \indrule{\rulmSub}{
          \indih{
            \jum{\lneg{\tenvtwo'},\lvar:\lneg{\typtwo}}{\tmtwo}{\typ}
          }
          \indrule{\rulmELimpTwo}{
            \indrule{\rulmAx}{
              \emptyPremise
            }{
              \jum{
                \lvartwo:\typtwo\limp\lneg{\typthree}
              }{
                \lvartwo
              }{
                \typtwo\limp\lneg{\typthree}
              }
            }
            \indih{
              \jum{\lneg{\tenvthree}}{\tmthree}{\typthree}
            }
          }{
            \jum{
              \lvartwo:\typtwo\limp\lneg{\typthree},
              \lneg{\tenvthree}
            }{
              \invap{\lvartwo}{\tmthree}
            }{
              \lneg{\typtwo}
            }
          }
        }{
          \jum{
            \lneg{\tenvtwo'},\lneg{\tenvthree},\lvartwo:\typtwo\limp\lneg{\typthree}
          }{
            \tmtwo\sub{\lvar}{\invap{\lvartwo}{\tmthree}}
          }{\typ}
        }
      \]
      so it suffices to take $\tm := \tmtwo\sub{\lvar}{\invap{\lvartwo}{\tmthree}}$.
    \item

        If there is an occurrence of $\typ$ in $\tenvthree$:
        then $\tenvthree = \tenvthree', \typ$ and the derivation is of the form
        \[
          \indrule{\rullTensor}{
            \jull{\tenv, \typtwo} 
            \HS
            \jull{\tenvthree',\typ,\typthree}
            }
            {
              \jull{\tenv,\tenvthree',\typ, \typtwo \tensor \typthree}
            }
        \]

        Hence by IH there exist terms $\tmtwo, \tmthree $ such that:
        \[
          \indrule{\rulmSub}
          {
            \indrule{\rulmELimpOne}
              {
                \indrule{\rulmAx}{\emptyPremise}{\jum{\lvartwo : \typtwo \limp \lneg{\typthree}}{\lvartwo}{\typtwo \limp \lneg{\typthree}}}
                \indih{\jum{\lneg{\tenv}}{\tmtwo}{\typtwo}}
              }
              {
                \jum
                {\lneg{\tenv}, \lvartwo : \typtwo \limp \lneg{\typthree}
                }
                {\ap{\lvartwo}{\tmtwo}}
                {\lneg{\typthree}}
              }
            \HS
            \indih{
            \jum{\lneg{\tenvthree'},\lvarthree:\lneg{\typthree}}{\tmthree}{\typ}  
             
              }
          }
          {\jum{\lneg{\tenv}, \lneg{\tenvthree'}, \typtwo \limp \lneg{\typthree }}{\tmthree \sub{\lvarthree}{\ap{\lvartwo}{\tmtwo}}}{\typ}}
        \]

    \item
      
        If $\typ = \typtwo\tensor\typthree$ and the derivation is of the form:
        \[
          \indrule{\rullTensor}
          {
            \jull{\tenv, \typtwo} 
            \HS
            \jull{\tenvthree,\typthree}
          }
          {\jull{\tenv, \tenvthree, \typtwo \tensor \typthree}}
        \]
        Hence by IH there exists terms $\tmtwo, \tmthree$ such that:
      \[
        \indrule{\rulmITensor}
        {
          \indih{\jum{\lneg{\tenv}}{\tmtwo}{\typtwo}}
          \HS
          \indih{\jum{\lneg{\tenvthree}}{\tmthree}{\typthree}}
        }
        {
          \jum{\lneg{\tenv}, \lneg{\tenvthree}}{\pair{\tmtwo}{\tmthree}}{\typtwo \tensor \typthree}
        }
      \]

    \end{enumerate}

    \item \rullLimp:
      Then $\tenv_{0}$ is a permutation of $\tenv', \typtwo \limp \typthree$ and the derivation is of the form
      \[
        \indrule{\rullLimp}{\jull{\tenv', \lneg{\typtwo}, \typthree}}{\jull{\tenv', \typtwo \limp \typthree}}
      \]
      Moreover, we know by hypothesis that $\tenv', \typtwo \limp \typthree$ is a permutation of $\tenv, \typ$.
      We consider three subcases, depending on  wheter there is an occurrence of $\typ$ in $\tenv'$ or $\typ = \typtwo \limp \typthree$.
  
      \begin{enumerate}
        \item If there is an occurrence of $\typ$ in $\tenv'$: 
        then $\tenv' = \tenv'',\typ$ and the derivation is of the form:
        \[
          \indrule{\rullLimp}{
            \jull{\tenv'',\typ, \lneg{\typtwo},\typthree}
            }{\jull{\tenv'',\typ, \typtwo \limp \typthree}}
        \] 
        and by IH we have:  
  
        \[
          \indrule{\rulmETensor}
          {
            \indih{\jum{\tenv'', \lvartwo : \typtwo, \lvarthree : \lneg{\typthree}}{\tmtwo}{\typ}}
            \HS
            \indrule{\rulmAx}{\emptyPremise}
            {\jum{\tmthree : \typtwo \tensor \lneg{\typthree}}{\tmthree}{\typtwo \tensor \lneg{\typthree}}}
          }
          {
            \jum{\tenv'',\typtwo \tensor \lneg{\typthree}}
            {\casepair{\tmthree}{\lvartwo}{\lvarthree}{\tmtwo}}
            {\typ}
          }
        \]

        \item If $\typ = \typtwo \limp \lneg{\typthree}$ and the derivation is of the form:
        \[
          \indrule{\rullLimp}{
            \jull{\tenv',\lneg{\typtwo},\typthree}
            }{\jull{\tenv',\typ}}
        \] 
        Hence by IH we have that 
        \[
          \indrule{\rulmILimp}{\indih{\jum{\tenv', \lvartwo : \typtwo}{\tm'}{\typthree}}}
          {\jum{\tenv'}{\lam{\lvartwo}{\tm'}}{\typtwo \limp \typthree}}
        \]
      \end{enumerate}

  \end{enumerate}
\end{proof}  

\ContrasubstitutionLemmaMLL*  

\begin{proof}\label{contrasubstitution_lemma_MLL:proof}

We proceed by induction on the derivation of $\jum{\tenv,\lvar:\typ}{\tm}{\typtwo}$.
\begin{enumerate}
\item \rulmAx:
  
    Then $\tenv = \emptyset$, $\tm = \lvar$ and $\typ = \typtwo$.
    % The derivation is of the form:
    % \[
    %   \indrule{\rulmAx}{\emptyPremise}{\jum{\lvar : \typ}{\lvar}{\typ}}
    % \]
    Hence we can take
    \[
      \indrule{}{}{\jum{\tenv'}{\tmtwo}{\lneg{\typ}}}
    \]

\item \rulmITensor:
  Then $\tm = \pair{\tm_1}{\tm_2}$ and $\typtwo = \typtwo_1\tensor\typtwo_2$.
  We consider two subcases, depending on whether
  $\lvar \in \fv{\tm_1}$ or $\lvar \in \fv{\tm_2}$:
  \begin{enumerate}
  \item
    If $\lvar \in \fv{\tm_1}$:
    then the derivation is of the form:
    \[
      \indrule{\rulmITensor}{
        \jum{\tenv_1,\lvar:\typ}{\tm_1}{\typtwo_1}
        \HS
        \jum{\tenv_2}{\tm_2}{\typtwo_2}
      }{
        \jum{\tenv_1,\lvar:\typ,\tenv_2}{\pair{\tm_1}{\tm_2}}{\typtwo}
      }
    \]
    Furthermore, by hypothesis we have that
    $\jum{\tenvtwo}{\tmtwo}{\typtwo_1\limp\lneg{\typtwo_2}}$.
    Recall that
    $\pair{\tm_1}{\tm_2}\cos{\lvar}{\tmtwo} = \tm_1\cos{\lvar}{\invap{\tmtwo}{\tm_1}}$.
    Hence:
    \[
      \indrule{\rulmCos}{
        \jum{\tenv_1,\lvar:\typ}{\tm_1}{\typtwo_1}
        \HS
        \indrule{\rulmELimpTwo}{
          \jum{\tenvtwo}{\tmtwo}{\typtwo_1\limp\lneg{\typtwo_2}}
          \HS
          \jum{\tenv_2}{\tm_2}{\typtwo_2}
        }{
          \jum{\tenv_2,\tenvtwo}{\invap{\tmtwo}{\tm_2}}{\lneg{\typtwo_1}}
        }
      }{
        \jum{\tenv_1,\tenv_2,\tenvtwo}{\tm_1\cos{\lvar}{\invap{\tmtwo}{\tm_2}}}{\lneg{\typ}}
      }
    \]
  \item

      If $\lvar \in \fv{\tm_2}$:
      \[
        \indrule{\rulmITensor}{
          \jum{\tenv_1}{\tm_1}{\typtwo_1}
          \HS
          \jum{\tenv_2, \lvar:\typtwo_{2}}{\tm_2}{\typtwo_2}
        }{
          \jum{\tenv_1,\tenv_2,\lvar:\typ}{\pair{\tm_1}{\tm_2}}{\typtwo_2}
        }
      \]
      Furthermore, by hypothesis we have that $\jum{\tenv'}{\tmtwo}{\typtwo_{1} \limp \lneg{\typtwo_{2}}}$.
      Recall that 
      $\pair{\tm_1}{\tm_2}\cos{\lvar}{\tmtwo} = {\tm_{2} \cos{a}{\ap{\tmtwo}{\tm_{1}}}}$.
      Hence:
      \[
        \indrule{\rulmCos}
        {
          \indrule{\rulmELimpOne}
          {
            \jum{\tenv'}{\tmtwo}{\typtwo_{1}\limp \typtwo_{2}}  
            \HS
            \jum{\tenv_{1}}{\tm_{1}}{\typ_{1}}
          }
          {\jum{\tenv_{1},\tenv'}{\ap{\tmtwo}{\tm_{1}}}{\lneg{\typtwo}}
          }
          \HS
          \jum{\tenv_{2}, \lvar : \typ}{\tm_{2}}{\typtwo_{2}}
        }
        {\jum{\tenv_{1},\tenv_{2},\tenv'}{\tm_{2} \cos{a}{\ap{\tmtwo}{\tm_{1}}}}{\lneg{\typ}}}
      \]

  \end{enumerate}
\item \rulmETensor: Then $\tm=\casepair{\tm_1}{\lvartwo}{\lvarthree}{\tm_2}$. We consider two cases depending on whether $\lvar\in\fv{\tm_1}$ or $\lvar\in\fv{\tm_2$}.
\begin{enumerate}

  \item If $\lvar\in\fv{\tm_1}$, then the derivation is of the form:

  \[
      \indrule{\rulmETensor}{
    \jum{\tenv_1,\lvar:\typ}{\tm_1}{\typthree_1\tensor\typthree_2}
    \HS
    \jum{\tenv_2,\lvartwo:\typthree_1,\lvarthree:\typthree_2}{\tm_2}{\typtwo}
  }{
    \jum{\tenv_1, \lvar:\typ, \tenv_2}{\casepair{\tm_1}{\lvartwo}{\lvarthree}{\tm_2}}{\typtwo}
  }
\]

  Furthermore, by hypothesis we have that $\jum{\tenv'}{\tmtwo}{\lneg{\typtwo}}$.

    \[
    \indrule{\rulmCos}
    {
      \indrule{\rulmILimp}
      {
        \indrule{\rulmCos}
        {
          \jum{
            \tenv_2, \lvartwo:\typthree_1, \lvarthree : \typthree_2
          }
          {\tm_2}
          {\typtwo}
          \HS
          \jum{
            \tenv'
          }{\tmtwo}{\lneg{\typtwo}}
        }
        {
          \jum{\tenv_2, \lvartwo : \typthree_1}{\tm_2 \cos{\lvarthree}{\tmtwo}}{\lneg{\typthree_2}}
        }
      }{
        
        \jum{\tenv_2}{\lam{\lvartwo}{\tmtwo\cos{\lvarthree}{\typtwo}}}{\typthree_1 \limp \lneg{\typthree_2}}
        
      }
      \HS
      {\jum{\tenv_1, \lvar : \typ}{\tm_1}{\typthree_1 \tensor \typthree_2}}

    }
    {
      \jum{\tenv_1, \tenv_2}
      {
        \tm_2\cos{\lvar}{\lam{\lvartwo}{\tm_1\cos{\lvarthree}{\tmtwo}}}
      }
      {\lneg \typ}
    }
    \]

  \item If $\lvar\in\fv{\tm_2}$, then the derivation is of the form:

 \[
      \indrule{\rulmETensor}{
    \jum{\tenv_1}{\tm_1}{\typthree_1\tensor\typthree_2}
    \HS
    \jum{\tenv_2, \lvar:\typ, \lvartwo:\typthree_1,\lvarthree:\typthree_2}{\tm_2}{\typtwo}
  }{
    \jum{\tenv_1, \lvar:\typ, \tenv_2}{\casepair{\tm_1}{\lvartwo}{\lvarthree}{\tm_2}}{\typtwo}
  }
\]
Furthermore, by hypothesis we have that $\jum{\tenv'}{\tmtwo}{\lneg{\typtwo}}$.

\[  
  \indrule{\rulmETensor}
  {
    \jum{\tenv_1}{\tm_1}{\typthree_1 \tensor \typthree_2}
    \HS 
    \indrule{\rulmCos}
    {
      \jum{\tenv_2, \lvartwo : \typtwo, \lvarthree : \typthree, \lvar:\typ}{\tm_2}{\typtwo}
      \HS 
      \jum{\tenv'}{\tmtwo}{\lneg{\typtwo}}
    }
    {
      \jum{\tenv_2, \tenv', \lvartwo : \typthree_1, \lvarthree : \typthree_2 }{\tm_2 \cos{\lvar}{\tmtwo}}{\lneg{\typ}}
    }
  }
  {\jum{\tenv_1, \tenv_2, \tenvtwo}{\casepair{\tm_1}{\lvartwo}{\lvarthree}{\tm_2 \cos{\lvar}{\tmtwo}}}{\lneg{\typ}}}
\]

\end{enumerate}
\item \rulmILimp:

    Then $\tm = \lam{\lvartwo}{\tm'}$ and $\typtwo = \typtwo_1 \limp \typtwo_2$.
    The derivation is of the form:
    \[
      \indrule{\rulmILimp}
      {\jum{\tenv, \lvar: \typ, \lvartwo : \typtwo_1}{\tm'}{\typtwo_2}}
      {\jum{\tenv, \lvar:\typ}{\lam{\lvartwo}{\tm'}}{\typtwo_1 \limp \typtwo_2}}
    \]
    Furthermore by hypothesis we have that $\jum{\tenv'}{\tmtwo}{\lneg{\typtwo}}$.
    \[
      \indrule{\rulmETensor}
      {
        \jum{\tenv'}{\tmtwo}{\typtwo_1 \tensor \lneg{\typtwo_2}}  
        \HS
        \indrule{\rulmCos}
        {
          \jum{\tenv,\lvar:\typ,\lvartwo:\typtwo_1}{\tm'}{\typtwo_2}
          \HS
          \indrule{\rulmAx}{\emptyPremise}
          {          \jum{\lvarthree : \typtwo_2}{\lvarthree}{\typtwo_2}
          }
        }
        {\jum{\tenv, \lvartwo : \typtwo_1, \lvarthree : \lneg{\typtwo_2}}{\tm'\cos{\lvar}{\lvarthree}}{\lneg{\typ}}}

      }
      {
        \jum{\tenv, \tenv'}{\casepair{\tmtwo}{\lvartwo}{\lvarthree}{\tm' \cos{\lvar}{\lvarthree}}}{\lneg{\typ}}
      }
    \]

\item \rulmELimpOne.
  
    Then $\tm=\ap{\tm_1}{\tm_2}$.
    We consider two cases depending on whether $\lvar\in\fv{\tm_1}$ or $\lvar\in\fv{\tm_2$}. 
    Furthermore we know that $\jum{\tenvtwo}{\tmtwo}{\lneg{\typtwo}}$.

    \begin{enumerate}
      \item If $\lvar \in \fv{\tm_1}$ then the derivation is of the following form:
      \[
        \indrule{\rulmELimpOne}{
          \jum{\tenv_1, \lvar : \typ}{\tm_1}{\typ \limp \typtwo}
          \HS
          \jum{\tenv_2}{\tm_2}{\typ}
        }
        {
          \jum{\tenv_1, \tenv_2}{\ap{\tm_1}{\tm_2}}{\typtwo}
        }
      \]
      Hence:
      \[
        \indrule{\rulmCos}
        {
          \jum{\tenv_1, \lvar : \typ}{\tm_1}{\typ \limp \typtwo}
          \HS
          \indrule{\rulmITensor}
          {
            \jum{\tenv_2}{\tm_2}{\typ}
            \HS
            \jum{\tenv'}{\tmtwo}{\lneg{\typtwo}}
          }
          {\jum{\tenv_2, \tenv'}{\pair{\tm_2}{\tmtwo}}{\typ \tensor \lneg{\typtwo}}}
        }
        {
          \jum{\tenv_1,\tenv_2,\tenv'}{\tm_1\cos{\lvar}{\pair{\tm_2}{\tmtwo}}}{\lneg{\typ}}
        }
      \]
    
      \item If $\lvar \in \fv{\tm_2}$ then the derivation is of the following form:
      \[
        \indrule{\rulmELimpOne}{
          \jum{\tenv_1}{\tm_1}{\typ \limp \typtwo}
          \HS
          \jum{\tenv_2, \lvar : \typ}{\tm_2}{\typ}
        }
        {
          \jum{\tenv_1, \tenv_2}{\ap{\tm_1}{\tm_2}}{\typtwo}
        }
      \]
      Hence:

      \[
        \indrule{\rulmCos}
        {
          \indrule{\rulmELimpTwo}
          {
            \jum{\tenv_1}{\tm_1}{\typ \limp \typtwo}
            \HS
            \jum{\tenv'}{\tmtwo}{\lneg{\typtwo}}
          }
          {\jum{\tenv_1,\tenv'}{\invap{\tm_1}{\tmtwo}}{\lneg{\typ}}}
          \HS
          \jum{\tenv_2, \lvar:\typ}{\tm_2}{\typ}
        }
        {\jum{\tenv_1, \tenv_2,\tenv'}{\tm_2 \cos{\lvar}{\invap{\tm_1}{\tmtwo}}}{\lneg{\typ}}}
      \]
    \end{enumerate}

\item \rulmELimpTwo:

  Then $\tm=\invap{\tm_1}{\tm_2}$ and $\typtwo = \lneg{\typthree}$.
  We consider two cases depending on whether $\lvar\in\fv{\tm_1}$ or $\lvar\in\fv{\tm_2$}. 
  Furthermore we know that $\jum{\tenvtwo}{\tmtwo}{\typthree}$.

  \begin{enumerate}
    \item If $\lvar \in \fv{\tm_1}$ then the derivation is of the following form:
    \[
      \indrule{\rulmELimpTwo}{
        \jum{\tenv_1, \lvar : \typ}{\tm_1}{\typthree \limp \typ}
        \HS
        \jum{\tenv_2}{\tm_2}{\lneg{\typ}}
      }
      {
        \jum{\tenv_1, \tenv_2}{\invap{\tm_1}{\tm_2}}{\lneg{\typthree}}
      }
    \]
    Recall that $( \ap{\tm_1}{\tm_2})\cos{\lvar}{\tmtwo} = \tm_1\cos{\lvar}{\pair{\tm_2}{\tmtwo}} $ if $\lvar \in \fv{\tm_1}$.
    Hence:
    \[
      \indrule{\rulmCos}
      {
        \indrule{\rulmILimp}
        {
          \jum{\tenv_1, \lvar : \typ}{\tm_1}{\typthree \limp \typ}
          \HS
          \jum{\tenv'}{\tmtwo}{\typthree}
        }
        {\jum{\tenv_1, \tenv', \lvar : \typ}{\ap{\tm_1}{\tmtwo}}{\typ}}
        \HS
        \jum{\tenv_2}{\tm_2}{\lneg{\typ}}
      }
      {
        \jum{\tenv_1,\tenv_2,\tenv'}{\tm_1\cos{\lvar}{\pair{\tm_2}{\tmtwo}}}{\lneg{\typ}}
      }
    \]
  
     \item If $\lvar \in \fv{\tm_2}$ then the derivation is of the following form:
     \[
      \indrule{\rulmELimpTwo}{
        \jum{\tenv_1}{\tm_1}{\typthree \limp \typ}
        \HS
        \jum{\tenv_2, \lvar : \typ}{\tm_2}{\lneg{\typ}}
      }
      {
        \jum{\tenv_1, \tenv_2}{\invap{\tm_1}{\tm_2}}{\lneg{\typthree}}
      }
    \]
   Hence:

   \[
    \indrule{\rulmCos}
    {
      \indrule{\rulmILimp}
      {
        \jum{\tenv_1}{\tm_1}{\typthree \limp \typ}
        \HS
        \jum{\tenv'}{\tmtwo}{\typthree}
      }
      {\jum{\tenv_1, \tenv'}{\ap{\tm_1}{\tmtwo}}{\typ}}
      \HS
      \jum{\tenv_2, \lvar : \typ}{\tm_2}{\lneg{\typ}}
    }
    {
      \jum{\tenv_1,\tenv_2,\tenv'}{\tm_2\cos{\lvar}{\ap{\tm_1}{\tmtwo}}}{\lneg{\typ}}
    }
  \]
  \end{enumerate}

\end{enumerate}
\end{proof}  

  \begin{definition}[Size]
  The size of a linear term $\tm$, denoted $\size{\tm}$, is defined as
    \[
      \begin{array}{rclcrcl}
        \size{\lvar} & \eqdef &  0 &
        \size{\pair{\tm}{\tmtwo}} & \eqdef & 1 + \size{\tm} + \size{\tmtwo} \\
        \size{\casepair{\tm}{\lvar}{\lvartwo}{\tmtwo}} & \eqdef & 1 + \size{\tm} + \size{\tmtwo} &
        \size{\lam{\lvar}{\tm}} & \eqdef & 1 + \size{\tm} \\
        \size{\ap{\tm}{\tmtwo}} & \eqdef & 1 + \size{\tm} + \size{\tmtwo} &
        \size{\invap{\tm}{\tmtwo}} & \eqdef & 1 + \size{\tm} + \size{\tmtwo} \\
      \end{array}
    \]
    We extend this definition to case contexts as follows:
      \begin{align*}
        \size{\cctxhole}  \eqdef  0 & &  &
        \size{\casepair{\tm}{\lvar}{\lvartwo}{\cctx}}  \eqdef  \size{\cctx} + \size{\tm}
      \end{align*}
  \end{definition}
  
  \begin{lemma}\llem{substitution_size}
    Let $\tm$ be a linear term and $\lvar \in \fv{\tm}$, then 
    $\size{\tm \sub{\lvar}{\tmtwo}} = \size{\tm} + \size{\tmtwo}$.
  \end{lemma}
  \begin{proof}
    By a straight-forward induction on the structure of the term $\tm$. 
  \end{proof}

  \begin{lemma} \llem{contrasubstitution_size}
    Let $\tm$ be a linear term and $\lvar \in \fv{\tm}$, then 
    $\size{\tm \cos{\lvar}{\tmtwo}} = \size{\tm} + \size{\tmtwo}$. 
  \end{lemma}
  \begin{proof}
    By a straight-forward induction on the structure of the term $\tm$. 
    We show the case of the application $\tm = \ap{\tm_{1}}{\tm_{2}}$, where $\size{\tm} = \size{\tm_{1}} + \size{\tm_{2}}$.
    The remaining cases follow by a similar reasoning.
  
      \begin{enumerate}
        \item If $\lvar \in \fv{\tm_{1}}$ then $(\ap{\tm_{1}}{\tm_{2}}) \cos{\lvar}{\tmtwo} = \tm_1\cos{\lvar}{\pair{\tm_2}{\tmtwo}}$.
        By \ih, $\size{\tm_1\cos{\lvar}{\pair{\tm_2}{\tmtwo}}} = \size{\tm_{1}} + \size{\pair{\tm_2}{\tmtwo}}$.
        Since $\size{\pair{\tm_{2}}{\tmtwo}} = \size{\tm_{2}} + \size{\tmtwo} + 1$ then this implies that $\size{(\ap{\tm_{1}}{\tm_{2}}) \cos{\lvar}{\tmtwo}} = \size{\tm_{1}} + \size{\tm_{2}} + \size{\tmtwo} + 1 = \size{\ap{\tm_{1}}{\tm_{2}}} + \size{\tmtwo}$.

        \item If $\lvar \in \fv{\tm_{2}}$ then $(\ap{\tm_{1}}{\tm_{2}}) \cos{\lvar}{\tmtwo} = \tm_2\cos{\lvar}{\invap{\tm_1}{\tmtwo}}$.
        By \ih, $\size{\tm_2\cos{\lvar}{\invap{\tm_1}{\tmtwo}}} = \size{\tm_{2}} + \size{\invap{\tm_1}{\tmtwo}}$.
        Since $\size{\invap{\tm_{1}}{\tmtwo}} = \size{\tm_{2}} + \size{\tmtwo} + 1$ then this implies that $\size{(\ap{\tm_{1}}{\tm_{2}}) \cos{\lvar}{\tmtwo}} = \size{\tm_{2}} + \size{\tm_{1}} + \size{\tmtwo} + 1 = \size{\ap{\tm_{1}}{\tm_{2}}} + \size{\tmtwo}$.
      \end{enumerate}
  \end{proof}
  
  \begin{lemma}
    $\CalcMLL$ is strongly normalizing.
  \end{lemma}
  
  \begin{proof}
    We proceed by showing that each reduction rule strictly decreases the term's size \ie the size of the redex is strictly greater than the size of the contractum.
    Since the size of a term is a non-negative integer, no infinite reduction sequences can exist. 
  
    Let $\tmthree$ denote the redex and $\tmthree'$ the contractum; we proceed by case analysis of the three reductions rules.
    \begin{enumerate}
      \item 
        Case $\ap{(\lam{\lvar}{\tm})\cctx}{\tmtwo} \to \tm\sub{\lvar}{\tmtwo}\cctx$.
        The size of the redex is $\size{\tmthree} = \size{\ap{(\lam{\lvar}{\tm})\cctx}{\tmtwo}} = \size{\lam{\lvar}{\tm}} + \size{\cctx} + \size{\tmtwo} = \size{\tm} + 1 + \size{\cctx} + \size{\tmtwo}$.
        For the contractum we have $\size{\tmthree'}= \size{\tm\sub{\lvar}{\tmtwo}\cctx} = \size{\tm \sub{\lvar}{\tmtwo}} + \size{\cctx}$, by \rlem{substitution_size} we get $\size{\tm \sub{\lvar}{\tmtwo}} + \size{\cctx} = \size{\tm} + \size{\tmtwo} + \size{\cctx}$.
        It follows that $\size{\tmthree} = \size{\tmthree'} + 1$.
  
        \item 
        Case $\invap{(\lam{\lvar}{\tm})\cctx}{\tmtwo} \to \tm\cos{\lvar}{\tmtwo}\cctx$.
        The size of the redex is $\size{\tmthree} = \size{\invap{(\lam{\lvar}{\tm})\cctx}{\tmtwo}} = \size{\lam{\lvar}{\tm}} + \size{\cctx} + \size{\tmtwo} = \size{\tm} + 1 + \size{\cctx} + \size{\tmtwo}$.
        For the contractum we have $\size{\tmthree'}= \size{\tm\cos{\lvar}{\tmtwo}\cctx} = \size{\tm \cos{\lvar}{\tmtwo}} + \size{\cctx}$, by \rlem{contrasubstitution_size} we get $\size{\tm \cos{\lvar}{\tmtwo}} + \size{\cctx} = \size{\tm} + \size{\tmtwo} + \size{\cctx}$.
        It follows that $\size{\tmthree} = \size{\tmthree'} + 1$.
  
        \item 
        Case $\casepair{\pair{\tmtwo}{\tmfour}\cctx}{\lvar}{\lvartwo}{\tm} \to \tm\sub{\lvar}{\tmtwo}\sub{\lvartwo}{\tmfour}\cctx$.
        The size of the redex is $\size{\tmthree} = \size{\casepair{\pair{\tmtwo}{\tmfour}\cctx}{\lvar}{\lvartwo}{\tm}} = \size{\tm} + \size{\pair{\tmtwo}{\tmfour}\cctx} + 1 = \size{\tm} + \size{\pair{\tmtwo}{\tmfour}} + \size{\cctx} + 1  = \size{\tm} + \size{\tmtwo} + \size{\tmfour} + \size{\cctx} + 1$.
        For the contractum we have $\size{\tmthree'} = \size{\tm\sub{\lvar}{\tmtwo}\sub{\lvartwo}{\tmfour}\cctx} = \size{\tm\sub{\lvar}{\tmtwo}\sub{\lvartwo}{\tmfour}} + \size{\cctx}$, by \rlem{substitution_size} we get $\size{\tm\sub{\lvar}{\tmtwo}\sub{\lvartwo}{\tmfour}} + \size{\cctx} = \size{\tm} + \size{\tmtwo} + \size{\tmfour} + \size{\cctx}$.
        It follows that $\size{\tmthree} = \size{\tmthree'} + 1$.  
    \end{enumerate}
  \end{proof}
  
%%% Local Variables:
%%% mode: latex
%%% TeX-master: "../main"
%%% End:
% 
\section{A Contraposition-Based Calculus for \MELL}
\lsec{app:calcMELL}

\begin{lemma}[Weakening and Contraction]
  \llem{weakening_and_contraction}
  \begin{enumerate}

  \item $\djum{\utenv}{\tenv}{\tm}{\typ}$ implies $\djum{\utenv,\uvar:\typtwo}{\tenv}{\tm}{\typ}$.
  \item $\djum{\utenv, \uvar:\typtwo, \uvartwo:\typtwo}{\tenv}{\tm}{\typ}$ implies $\djum{\utenv, \uvar:\typtwo}{\tenv}{\tm\sub{\uvartwo}{\uvar}}{\typ}$.
  \end{enumerate}
  
\end{lemma}

\begin{proof}
By induction on $\djum{\utenv}{\tenv}{\tm}{\typ}$, for the first item, and induction on $\djum{\utenv, \uvar:\typtwo, \uvartwo:\typtwo}{\tenv}{\tm}{\typ}$, for the second.
\end{proof}

\begin{lemma}[Linear and unrestricted substitution are compatible with typing]
The following rules are admissible in \CalcMELL:
\[
  \begin{array}{cc}
  \indrule{\rulmSub}{
    \djum{\utenv}{\tenv,\lvar:\typ}{\tm}{\typtwo}
    \HS
    \djum{\utenv}{\tenvtwo}{\tmtwo}{\typ}
  }{
    \djum{\utenv}{\tenv,\tenvtwo}{\tm\sub{\lvar}{\tmtwo}}{\typtwo}
    }
    &
      \indrule{\rulmUSub}{
    \djum{\utenv, \uvar:\typ}{\tenv}{\tm}{\typtwo}
    \HS
    \djum{\utenv}{\emptytenv}{\tmtwo}{\typ}
  }{
    \djum{\utenv}{\tenv}{\tm\sub{\uvar}{\tmtwo}}{\typtwo}
  }
    \end{array}
\]
\end{lemma}
\begin{proof}
By induction on the derivation of $\djum{\utenv}{\tenv,\lvar:\typ}{\tm}{\typtwo}$, for the first item, and $\djum{\utenv, \uvar:\typ}{\tenv}{\tm}{\typtwo}$ for the second.
\end{proof}

It can be noted that the following two judgments are logically equivalent,
in the sense that the first one is provable if and only if the second
one is provable:
\[
  \djumnt{\utenv,\uvar:\typ}{\tenv}{\typtwo}
  \HS
  \HS
  \djumnt{\utenv}{\lvar:\ofc{\typ},\tenv}{\typtwo}
\]
More precisely, we have the following lemma:

\begin{lemma}
\llem{unrestricted_eq_bang_linear}
\mbox{}
\begin{enumerate}
    \item\label{to_lin} If $\djum{\utenv,\uvar:\typ}{\tenv}{\tm}{\typtwo}$,
      then  $\djum{\utenv}{\lvar:\ofc{\typ},\tenv}{\tm\eofc{\uvar}{\lvar}}{\typtwo}$.

          \item\label{from_lin} If  $\djum{\utenv}{\lvar:\ofc{\typ},\tenv}{\tm}{\typtwo}$, then $\djum{\utenv,\uvar:\typ}{\tenv}{\tm\sub{\lvar}{\iofctwo{\lvartwo}{(\ibott{\uvar}{\lvartwo})}}}{\typtwo}$,

          \end{enumerate}
        \end{lemma}

        \begin{proof}
For the first item we have:
\[
      \indrule{\rulmEOfc}
  {\djum{\utenv}{\lvar:\ofc{\typ}}{\lvar}{\ofc{\typ}}
    \HS
    \djum{\utenv,\uvar:\typ}{\tenv}{\tm}{\typtwo}
    }
    {\djum{\utenv}{\lvar:\ofc{\typ},\tenv}{\tm\eofc{\uvar}{\lvar}}{\typtwo}}
  \]
For the second we have:
  \[
    \indrule{\rulmSub}
    {
      \indruledbl{\rlem{weakening_and_contraction}}
      {\djum{\utenv}{\lvar:\ofc{\typ},\tenv}{\tm}{\typtwo}}
      {\djum{\utenv,\uvar:\typ}{\lvar:\ofc{\typ},\tenv}{\tm}{\typtwo}}
    \HS
    \indrule{\rulmIOfc}
      {\djum{\utenv,\uvar:\typ}{\emptytenv}{\uvar}{\typ}}
      {\djum{\utenv,\uvar:\typ}{\emptytenv}{\iofctwo{\lvartwo}{(\ibott{\uvar}{\lvartwo})}}{\ofc{\typ}}}
  }
  { \djum{\utenv,\uvar:\typ}{\tenv}{\tm\sub{\lvar}{\iofctwo{\lvartwo}{(\ibott{\uvar}{\lvartwo})}}}{\typtwo}}
  %     \indrule{\rulmEOfc}
  % {\djum{\utenv}{\lvar:\ofc{\typ}}{\lvar}{\ofc{\typ}}
  %   \HS
  %   \djum{\utenv,\uvar:\typ}{\tenv}{\tm}{\typtwo}
  %   }
  %   {\djum{\utenv}{\lvar:\ofc{\typ},\tenv}{\eofc{\tm}{\uvar}{\lvar}}{\typtwo}}
\]

\end{proof}

Notation: $\tmtwo\eofc{\seq{\uvar}}{\seq{\tm}}$ to abbreviate
$\tmtwo\eofc{\uvar_1}{\tm_1}\ldots\eofc{\uvar_n}{\tm_n}$. Also,
$\tm\sub{\seq{\lvar}}{\seq{\tmtwo}}$ to abbreviate
$\tm\sub{\lvar_1}{\tmtwo_1}\ldots\sub{\lvar_n}{\tmtwo_n}$.

\SoundnessAndCompletenessOfCalcMELL*

\begin{proof}\label{soundness_and_completeness_of_calc_mell:proof}
We first address soundness. By induction on the derivation of $\djum{\utenv}{\tenv}{\tm}{\typ}$.
  \begin{enumerate}

  \item $\rulmAx$. The derivation is:
    \[
  \indrule{\rulmAx}{
    \emptyPremise
  }{
    \djum{\utenv}{\lvar:\typ}{\lvar}{\typ}
  }
  \]

  We can take:
  \[
    \indruledbl{\rullW}{
    \indrule{\rullAx}{
    \emptyPremise
  }{
    \jull{\lneg{\typ},\typ}
  }
    }{
    \jull{\why{\lneg{\utenv}}, \lneg{\typ},\typ}
  }
  \]

  \item $\rulmAxU$. The derivation is:
  \[
  \indrule{\rulmAxU}{
    \emptyPremise
  }{
    \djum{\utenv, \uvar:\typ}{\emptytenv}{\uvar}{\typ}
  }
\]
We can take:
  \[
    \indruledbl{\rullW}{
        \indrule{\rullD}{
          \indrule{\rullAx}{
              \emptyPremise
            }{
              \jull{\lneg{\typ},\typ}
            }
          }{
            \jull{\why{\lneg{\typ}},\typ}
          }
        }{
          \jull{\why{\lneg{\utenv}}, \why{\lneg{\typ}},\typ}
        }
  \]

  \item $\rulmITensor$. The derivation is:
  \[
  \indrule{\rulmITensor}{
    \djum{\utenv}{\tenv}{\tm}{\typ}
    \HS
    \djum{\utenv}{\tenvtwo}{\tmtwo}{\typtwo}
  }{
    \djum{\utenv}{\tenv,\tenvtwo}{\pair{\tm}{\tmtwo}}{\typ\tensor\typtwo}
  }
\]

We can take:
\[
    \indruledbl{\rullC}{
      \indrule{\rullTensor}{
        \indih{\jull{\why{\lneg{\utenv}},\lneg{\tenv},\typ}}
        \HS
        \indih{\jull{\why{\lneg{\utenv}},\lneg{\tenvtwo},\typtwo}}
  }{
    \jull{\why{\lneg{\utenv}}, \lneg{\tenv}, \why{\lneg{\utenv}},\lneg{\tenvtwo},\typ\tensor\typtwo}
  }
    }{
    \jull{\why{\lneg{\utenv}}, \lneg{\tenv},\lneg{\tenvtwo},\typ\tensor\typtwo}
  }
  \]
  \item $\rulmETensor$. The derivation is:
  \[
  \indrule{\rulmETensor}{
    \djum{\utenv}{\tenv}{\tm}{\typ\tensor\typtwo}
    \HS
    \djum{\utenv}{\tenvtwo,\lvar:\typ,\lvartwo:\typtwo}{\tmtwo}{\typthree}
  }{
    \djum{\utenv}{\tenv,\tenvtwo}{\casepair{\tm}{\lvar}{\lvartwo}{\tmtwo}}{\typthree}
  }
\]
   We take:
   \[
       \indruledbl{\rullC}{
  \indrule{\rullCut}{
    \jull{\why{\lneg{\utenv}},\lneg{\tenv}, \typ\tensor\typtwo}
    \HS
      \indrule{\rullLimp}{
    \jull{\why{\lneg{\utenv}},\lneg{\tenvtwo},\lneg{\typ},\lneg{\typtwo},\typthree}
  }{
     \jull{\why{\lneg{\utenv}},\lneg{\tenvtwo},\typ\limp\lneg{\typtwo},\typthree}
  }
  }{
    \jull{\why{\lneg{\utenv}}, \lneg{\tenv},\why{\lneg{\utenv}},\lneg{\tenvtwo},\typthree}
  }
  }{
        \jull{\why{\lneg{\utenv}}, \lneg{\tenv},\lneg{\tenvtwo},\typthree}
  }
  \]

  \item $\rulmIPar$. The derivation is:

          \[\indrule{\rulmIPar}{
  \djum{\utenv}{\tenv,\lvar:\typ,\lvartwo:\typtwo}{\tm}{\bott}
  }{
    \djum{\utenv}{\tenv}{\ipar{\lvar}{\lvartwo}{\tm}}{\lneg{\typ}\parr\lneg{\typtwo}}
  }
  \]
 We take:
  
  \[
      \indrule{\rullPar}{
        \indrule{\rullCut}{
          \jull{\why{\lneg{\utenv}},\lneg{\tenv},\lneg{\typ},\lneg{\typtwo},\bott}
          \HS
          \jull{\one}
          }{
            \jull{\why{\lneg{\utenv}},\lneg{\tenv}, \lneg{\typ},\lneg{\typtwo}}
            }
          }
          {
            \jull{\why{\lneg{\utenv}}, \lneg{\tenv}, \lneg{\typ}\parr\lneg{\typtwo}}
          }        
\]

  \item $\rulmEParOne$. The derivation is:
  \[
  \indrule{\rulmEParOne}{
    \djum{\utenv}{\tenv}{\tm}{\typ\parr\typtwo}
    \HS
    \djum{\utenv}{\tenvtwo}{\tmtwo}{\lneg{\typ}}
  }{
    \djum{\utenv}{\tenv,\tenvtwo}{\ap{\tm}{\tmtwo}}{\typtwo}
  }
  \]

 We take:
  \[
      \indruledbl{\rullC}{
    \indrule{\rullCut}{
      \indrule{\rullTensor}{
        \jull{\why{\lneg{\utenv}}, \lneg{\tenvtwo},\lneg{\typ}}
        \HS
        \jull{\lneg{\typtwo},\typtwo}
      }{
        \jull{\why{\lneg{\utenv}}, \lneg{\tenvtwo},\typtwo,\lneg{\typ}\tensor\lneg{\typtwo}}
      }
      \HS
      \jull{\why{\lneg{\utenv}},\lneg{\tenv},\typ\parr\typtwo}
    }{
      \jull{\why{\lneg{\utenv}},\lneg{\tenvtwo}, \why{\lneg{\utenv}}, \lneg{\tenv}, \typtwo}
    }
      }{
     \jull{\why{\lneg{\utenv}},\lneg{\tenvtwo}, \lneg{\tenv}, \typtwo}
  }
  \]
  
  \item $\rulmEParTwo$. The derivation is:
  \[
  \indrule{\rulmEParTwo}{
    \djum{\utenv}{\tenv}{\tm}{\typ\parr\typtwo}
    \HS
    \djum{\utenv}{\tenvtwo}{\tmtwo}{\lneg{\typtwo}}
  }{
    \djum{\utenv}{\tenv,\tenvtwo}{\invap{\tm}{\tmtwo}}{\typ}
  }
\]
 We take:
\[
    \indruledbl{\rullC}{
    \indrule{\rullCut}{
    \indrule{\rullTensor}{
  \indrule{\rullAx}{
    \emptyPremise
  }{
    \jull{\typ,\lneg{\typ}}
  }
    \HS
    \jull{\why{\lneg{\utenv}},\lneg{\tenvtwo},\lneg{\typtwo}}
}{
     \jull{\typ,\why{\lneg{\utenv}},\lneg{\tenvtwo},\lneg{\typ}\tensor\lneg{\typtwo}}
  }
      \HS
    \jull{\why{\lneg{\utenv}},\lneg{\tenv},\typ\parr\typtwo}
  }{
    \jull{\typ,\why{\lneg{\utenv}},\lneg{\tenvtwo}, \why{\lneg{\utenv}},\lneg{\tenv}}
  }
    }{
     \jull{\typ,\why{\lneg{\utenv}},\lneg{\tenvtwo}, \lneg{\tenv}}
  }
\]

\item $\rulmIOfc$. The derivation is:
  \[
    \indrule{\rulmIOfc}
  {\djum{\utenv}{\lvar:\typ}{\tm}{\bott}}
  {\djum{\utenv}{\emptytenv}{\iofctwo{\lvar}{\tm}}{\ofc{\lneg{\typ}}}}
\]
 We take:
\[
    \indrule{\rullP}{
      \indrule{\rullCut}{
        \indih{\jull{\why{\lneg{\utenv}},\lneg{\typ},\bott}}
        \HS
        \jull{\one}
      }
      {
         \jull{\why{\lneg{\utenv}},\lneg{\typ}}
      }
  }{
    \jull{\why{\lneg{\utenv}},\ofc{\lneg{\typ}}}
  }
  \]
\item $\rulmEOfc$. The derivation is:
  \[
    \indrule{\rulmEOfc}
    {\djum{\utenv}{\tenv}{\tm}{\ofc{\typ}}
    \HS
    \djum{\utenv,\uvar:\typ}{\tenvtwo}{\tmtwo}{\typthree}
    }
  {\djum{\utenv}{\tenv,\tenvtwo}{\tmtwo\eofc{\uvar}{\tm}}{\typthree}}
\]
 We take:
\[
    \indruledbl{\rullC}{
  \indrule{\rullCut}{
    \jull{\why{\lneg{\utenv}},\lneg{\tenv},\ofc{\typ}}
    \HS
    \jull{\why{\lneg{\utenv}},\why{\lneg{\typ}},\lneg{\tenvtwo},\typthree}
  }{
    \jull{\why{\lneg{\utenv}},\lneg{\tenv},\why{\lneg{\utenv}},\lneg{\tenvtwo},\typthree}
  }
    }{
     \jull{\why{\lneg{\utenv}},\lneg{\tenv},\lneg{\tenvtwo},\typthree}
  }
  \]
\item $\rulmIWhy$.
  The derivation is:
  \[
    \indrule{\rulmIWhy}{
      \djum{\utenv,\uvar:\lneg{\typ}}{\tenv}{\tm}{\bott}
    }{
      \djum{\utenv}{\tenv}{\iwhy{\uvar}{\tm}}{\why{\typ}}
    }
  \]
  We take:
  \[
      \indruledbl{\rullCut}{
        \jull{\why{\lneg{\utenv}},\why{\typ},\lneg{\tenv},\bott}
        \HS
        \jull{\one}
      }{
        \jull{\why{\lneg{\utenv}},\why{\typ},\lneg{\tenv}}
      }
  \]
  
\item $\rulmEWhy$. There are two cases.

  \[
          \indrule{\rulmEWhy}{
    \djum{\utenv}{\tenv}{\tmtwo}{\why{\typ}}
    \HS
  \djum{\utenv}{\lvar:\typ}{\tm}{\bott}
  }{
    \djum{\utenv}{\tenv}{\ewhy{\tm}{\lvar}{\tmtwo}}{\bott}
  }
\]
 We take:
\[
  \indrule{\rullBott}{
  \indruledbl{\rullC}{
  \indrule{\rullCut}{
    \jull{\why{\lneg{\utenv}},\lneg{\tenv}, \why{\typ}}
    \HS
    \indrule{\rullP}{
               \indrule{\rullCut}{
                 \jull{\why{\lneg{\utenv}},\lneg{\typ}, \bott}
                 \HS
                 \jull{\one}
               }
               {
                       \jull{\why{\lneg{\utenv}},\lneg{\typ}}
                     }
                   }
                   {
                       \jull{\why{\lneg{\utenv}},\ofc{\lneg{\typ}}}
                     }
  }  
  {
      \jull{\why{\lneg{\utenv}}, \why{\lneg{\utenv}}, \lneg{\tenv}}
    }
  }
  {
     \jull{\why{\lneg{\utenv}}, \lneg{\tenv}}
   }
 }
 {
   \jull{\why{\lneg{\utenv}}, \lneg{\tenv}, \bott}
   }
  \]

\item $\rulmIBott$. The derivation is:
  \[
        \indrule{\rulmIBott}{
    \djum{\utenv}{\tenv}{\tm}{\typtwo}
    \HS
    \djum{\utenv}{\tenv'}{\tmtwo}{\lneg{\typtwo}}
  }{
    \djum{\utenv}{\tenv,\tenv'}{\ibott{\tm}{\tmtwo}}{\bott}
  }
\]
 We take:
\[
  \indruledbl{\rullC}{
     \indrule{\rullBott}{
      \indrule{\rullCut}{
        \jull{\why{\lneg{\utenv}},\lneg{\tenv},\typtwo}
        \HS
        \jull{\why{\lneg{\utenv}},\lneg{\tenvtwo},\lneg{\typtwo}}
      }{
        \jull{\why{\lneg{\utenv}},\why{\lneg{\utenv}}, \lneg{\tenv},\lneg{\tenvtwo}}
      }
    }
    {
       \jull{\why{\lneg{\utenv}},\why{\lneg{\utenv}}, \lneg{\tenv},\lneg{\tenvtwo},\bott}
      }
    }
    {
        \jull{\why{\lneg{\utenv}}, \lneg{\tenv},\lneg{\tenvtwo},\bott}
    }    
  \]

\item $\rulmIOne$. The derivation is:
  \[
      \indrule{\rulmIOne}{
    \emptyPremise
  }{
    \djum{\utenv}{\emptytenv}{\ione}{\one}
  }
\]
 We take:
\[
  \indruledbl{\rullW}
  {
    \indrule{\rullOne}
    { \emptyPremise}
    {\jull{\one}}
  }
  {
    \jull{\why{\lneg{\utenv}}, \one}
    }
  \]

\item $\rulmEOne$. The derivation is:
  \[
        \indrule{\rulmEOne}{
    \djum{\utenv}{\tenv}{\tm}{\one}
    \HS
    \djum{\utenv}{\tenvtwo}{\tmtwo}{\typ}
  }{
    \djum{\utenv}{\tenv,\tenvtwo}{\tmtwo\eone{\tm}}{\typ}
  }
\]
 We take:
\[
  \indruledbl{\rullC}
  {
  \indrule{\rullCut}
  {
    \jull{\why{\lneg{\utenv}}, \lneg{\tenv}, \one}
    \HS
    \indrule{\rullBott}
    {\jull{\why{\lneg{\utenv}}, \lneg{\tenvtwo},\typ}}
    {\jull{\why{\lneg{\utenv}}, \lneg{\tenvtwo}, \typ,\bott}}
  }
  {
    \jull{\why{\lneg{\utenv}}, \lneg{\tenv}, \why{\lneg{\utenv}}, \lneg{\tenvtwo}, \typ}
  }
}
{
    \jull{\why{\lneg{\utenv}}, \lneg{\tenv}, \lneg{\tenvtwo}, \typ}
  }
\]

  \end{enumerate}
% \end{proof}

% \CompletenessOfCalcMELL*

% \begin{proof}\label{completeness_of_calc_mell:proof}
Next we address completeness.  By induction on the size of the cut-free derivation of $\jull{\tenv_0}$. 
  \begin{enumerate}
  \item $\rullAx$: The derivation ends in
    \[
      \indrule{\rullAx}{
    \emptyPremise
  }{
    \jull{\typtwo,\lneg{\typtwo}}
  }
\]
There are two cases depending on whether $A=\typtwo$ or $\typ=\lneg{\typtwo}$.
\begin{enumerate}

\item  $A=\typtwo$.  Then $\tenv=\lneg{\typtwo}$ and
  \[
      \indrule{\rulmAxU}{
    \emptyPremise
  }{
    \djum{\emptyutenv}{\lvar:\typtwo}{\lvar}{\typtwo}
  }
    \]

  \item $\typ=\lneg{\typtwo}$. Then $\tenv=\typtwo$ and

      \[
      \indrule{\rulmAxU}{
    \emptyPremise
  }{
    \djum{\emptyutenv}{\lvar:\lneg{\typtwo}}{\lvar}{\lneg{\typtwo}}
  }
    \]
\end{enumerate}

\item $\rullTensor$: The derivation ends in
  \[
  \indrule{\rullTensor}{
    \jull{\tenv_1,\typtwo_1}
    \HS
    \jull{\tenv_2,\typtwo_2}
  }{
    \jull{\tenv_1,\tenv_2,\typtwo_1\tensor\typtwo_2}
  }
\]

There are three cases depending on whether $A=\typtwo_1\tensor\typtwo_2$ or $\typ\in \tenv_1$ or $\typ\in\tenv_2$.
\begin{enumerate}

\item  $A=\typtwo_1\tensor\typtwo_2$.  By the \ih twice, $\djum{\emptyutenv}{\lneg{\tenv_1}}{\tmtwo_1}{\typtwo_1}$ and  $\djum{\emptyutenv}{\lneg{\tenv_2}}{\tmtwo_2}{\typtwo_2}$. We conclude from the following derivation
  \[
      \indrule{\rulmITensor}{
   \indih{\djum{\emptyutenv}{\lneg{\tenv_1}}{\tmtwo_1}{\typtwo_1}}
    \HS
    \indih{\djum{\emptyutenv}{\lneg{\tenv_2}}{\tmtwo_2}{\typtwo_2}}
  }{
    \djum{\emptyutenv}{\lneg{\tenv_1},\lneg{\tenv_2}}{\pair{\tmtwo_1}{\tmtwo_2}}{\typtwo_1\tensor\typtwo_2}
  }
  \]  

  Thus it suffices to take $\tm\eqdef \pair{\tmtwo_1}{\tmtwo_2}$.
  
\item $\typ\in \tenv_1$. Then $\tenv_1=\tenv_3,\typ$.
    Hence by \ih there exist terms $\tmtwo,\tmthree$ such that:
      \[
        \indrule{\rulmSub}{
          \indih{
            \djum{\emptyutenv}{\lneg{\tenv_3},\lvar:\lneg{\typtwo_1}}{\tmtwo}{\typ}
          }
          \indrule{\rulmEParTwo}{
            \indrule{\rulmAx}{
              \emptyPremise
            }{
              \djum{\emptyutenv}{
                \lvartwo:\lneg{\typtwo_1}\parr\lneg{\typtwo_2}
              }{
                \lvartwo
              }{
                \lneg{\typtwo_1}\parr\lneg{\typtwo_2}
              }
            }
            \indih{
              \djum{\emptyutenv}{\lneg{\tenv_2}}{\tmthree}{\typtwo_2}
            }
          }{
            \djum{\emptyutenv}{
              \lvartwo:\lneg{\typtwo_1}\parr\lneg{\typtwo_2},
              \lneg{\tenv_2}
            }{
              \invap{\lvartwo}{\tmthree}
            }{
              \lneg{\typtwo_1}
            }
          }
        }{
          \djum{\emptyutenv}{
            \lneg{\tenv_3},\lneg{\tenv_2},\lvartwo:\lneg{\typtwo_1}\parr\lneg{\typtwo_2}
          }{
            \tmtwo\sub{\lvar}{\invap{\lvartwo}{\tmthree}}
          }{\typ}
        }
      \]
      so it suffices to take $\tm := \tmtwo\sub{\lvar}{\invap{\lvartwo}{\tmthree}}$.

    \item $\typ\in \tenv_2$.  Then $\tenv_2=\tenv_3,\typ$.  Hence by \ih there exist terms $\tmtwo,\tmthree$ such that:

      \[
          \indrule{\rulmSub}{
             \indrule{\rulmEParOne}{
    \djum{\emptyutenv}{\lvar: \lneg{\typtwo_1}\parr\lneg{\typtwo_2}}{\lvar}{\lneg{\typtwo_1}\parr\lneg{\typtwo_2}}
    \HS
    \indih{\djum{\emptyutenv}{\lneg{\tenv_1}}{\tmthree}{\typtwo_1}}
  }{
    \djum{\emptyutenv}{\lvar: \lneg{\typtwo_1}\parr\lneg{\typtwo_2}, \lneg{\tenv_1}}{\ap{\lvar}{\tmthree}}{\lneg{\typtwo_2}}
  }
    \HS
    \djum{\emptyutenv}{\lneg{\tenv_3},\lvartwo:\lneg{\typtwo_2}}{\tmtwo}{\typ}
  }{
    \djum{\emptyutenv}{\lneg{\tenv_1},\lneg{\tenv_3}, \lvar: \lneg{\typtwo_1}\parr\lneg{\typtwo_2}}{\tmtwo\sub{\lvartwo}{\ap{\lvar}{\tmthree}}}{\typ}
  }
  \]
\end{enumerate}

\item $\rullPar$: The derivation ends in

  \[
      \indrule{\rullPar}{
    \jull{\tenv_1,\typtwo_1,\typtwo_2}
  }{
    \jull{\tenv_1,\typtwo_1\parr\typtwo_2}
  }
\]

There are two cases depending on whether $\typ=\typtwo_1\parr\typtwo_2$ or $\typ \in \tenv_1$.
\begin{enumerate}

\item $\typ=\typtwo_1\parr\typtwo_2$. By the \ih, $\djum{\emptyutenv}{\seq{a}:\lneg{\tenv_1},\lvartwo:\typtwo_1}{\tmtwo}{\typtwo_2}$. We conclude from the derivation:
  \[
      \indrule{\rulmIPar}{
        \indrule{}{
          \indih{\djum{\emptyutenv}{\seq{a}:\lneg{\tenv_1},\lvartwo: \lneg{\typtwo_1}}{\tmtwo}{\typtwo_2}}
          \HS
          \djum{\emptyutenv}{\lvarthree:\lneg{\typtwo_2}}{\lvarthree}{\lneg{\typtwo_2}}
        }
        {
          \djum{\emptyutenv}{\seq{a}:\lneg{\tenv_1},\lvartwo: \lneg{\typtwo_1}, \lvarthree:\lneg{\typtwo_2}}{\ibott{\tmtwo}{\lvarthree}}{\bott}
        }
  }{
   \djum{\emptyutenv}{\seq{a}:\lneg{\tenv_1}}{\ipar{\lvartwo}{\lvarthree}{\ibott{\tmtwo}{\lvarthree}}}{\typtwo_1\parr\typtwo_2}
  }
\]

Thus it suffices to take $\tm\eqdef \ipar{\lvartwo}{\lvarthree}{\ibott{\tmtwo}{\lvarthree}}$.

\item $\typ \in \tenv_1$.  Thus $\tenv_1=\tenv_2,\typ$.  By the \ih $\djum{\emptyutenv}{\seq{a}:\lneg{\tenv_2},\lvartwo:\lneg{\typtwo_1},\lvarthree:\lneg{\typtwo_2}}{\tmtwo}{\typ}$. Then we construct the following derivation and conclude:
  \[
      \indrule{\rulmETensor}{
    \djum{\emptyutenv}{\lvarfour: \lneg{\typtwo_1}\tensor\lneg{\typtwo_2}}{\lvarfour}{\lneg{\typtwo_1}\tensor\lneg{\typtwo_2}}
    \HS
    \indih{\djum{\emptyutenv}{\seq{a}:\lneg{\tenv_2},\lvartwo:\lneg{\typtwo_1},\lvarthree:\lneg{\typtwo_2}}{\tmtwo}{\typ}}
  }{
    \djum{\emptyutenv}{\lvarfour: \lneg{\typtwo_1}\tensor\lneg{\typtwo_2}, \seq{a}:\lneg{\tenv_2}}{\casepair{\lvarfour}{\lvartwo}{\lvarthree}{\tmtwo}}{\typ}
  }
\]

Thus it suffices to take $\tm\eqdef \casepair{\lvarfour}{\lvartwo}{\lvarthree}{\tmtwo}$.
\end{enumerate}

 \item $\rullP$: The derivation ends in
\[
  \indrule{\rullP}{
    \jull{\why{\tenv_1},\typtwo}
  }{
    \jull{\why{\tenv_1},\ofc{\typtwo}}
  }
\]

There are two cases depending on whether $\typ=\ofc{\typtwo}$ or $\typ \in \why{\tenv_1}$.

\begin{enumerate}

\item $\typ=\ofc{\typtwo}$. By the \ih  $\djum{\emptyutenv}{\seq{\lvar}:\ofc{\lneg{\tenv_1}}}{\tmtwo}{\typtwo}$. By \rlem{unrestricted_eq_bang_linear}(\ref{from_lin}), we have $\djum{\seq{\uvar}:\lneg{\tenv_1}}{\emptytenv}{\tmtwo\sub{\seq{\lvar}}{\seq{\ofc{\uvar}}}}{\typtwo}$. We thus derive:
  \[
  \indrule{\rulmIOfc'}
  {\djum{\seq{\uvar}:\lneg{\tenv_1}}{\emptytenv}{\tmtwo\sub{\seq{\lvar}}{\seq{\ofc{\uvar}}}}{\typtwo}}
  {\djum{\seq{\uvar}:\lneg{\tenv_1}}{\emptytenv}{\ofc{(\tmtwo\sub{\seq{\lvar}}{\seq{\ofc{\uvar}}})}}{\ofc{\typtwo}}}
\]

Finally, from \rlem{unrestricted_eq_bang_linear}(\ref{to_lin}), we obtain $\djum{\emptyutenv}{\seq{\lvartwo}:\ofc{\lneg{\tenv_1}}}{\ofc{(\tmtwo\sub{\seq{\lvar}}{\seq{\ofc{\uvar}}})}\eofc{\seq{\uvar}}{\seq{\lvartwo}}}{\ofc{\typtwo}}$.  Thus it suffices to take $\tm\eqdef \ofc{(\tmtwo\sub{\seq{\lvar}}{\seq{\ofc{\uvar}}})}\eofc{\seq{\uvar}}{\seq{\lvartwo}}$.

\item $\typ \in \why{\tenv_1}$. Thus $\why{\tenv_1}=\why{\tenv_2},\why{\typ_1}$ and $\typ=\why{\typ_1}$.  By the \ih  $\djum{\emptyutenv}{\seq{\lvar}:\ofc{\lneg{\tenv_1}},\lvartwo:\lneg{\typtwo}}{\tmtwo}{\why{\typ_1}}$. By \rlem{unrestricted_eq_bang_linear}(\ref{from_lin}), we have $\djum{\seq{\uvar}:\lneg{\tenv_1}}{\lvartwo:\lneg{\typtwo}}{\tmtwo\sub{\seq{\lvar}}{\seq{\ofc{\uvar}}}}{\why{\typ_1}}$.

    By Weakening we have $\djum{\uvartwo:\lneg{\typ_1},\seq{\uvar}:\lneg{\tenv_1}}{\lvartwo:\lneg{\typtwo}}{\tmtwo\sub{\seq{\lvar}}{\seq{\ofc{\uvar}}}}{\why{\typ_1}}$. We reason as follows
    \[
\indrule{\rulmIWhy}{
      \indrule{\rulmEWhy}{
          \djum{\uvartwo:\lneg{\typ_1},\seq{\uvar}:\lneg{\tenv_1}}{\lvarthree:\why{\lneg{\typtwo}}}{\lvarthree}{\why{\lneg{\typtwo}}}
          \quad
          \indrule{}{
            \begin{array}{l}
              \djum{\uvartwo:\lneg{\typ_1},\seq{\uvar}:\lneg{\tenv_1}}{\lvartwo:\lneg{\typtwo}}{\tmtwo\sub{\seq{\lvar}}{\seq{\ofc{\uvar}}}}{\why{\typ_1}}
              \\
              \djum{\uvartwo:\lneg{\typ_1},\seq{\uvar}:\lneg{\tenv_1}}{\emptytenv}{\iofctwo{\lvarfour}{(\ibott{\lvarfour}{\uvartwo})}}{\ofc{\lneg{\typ_1}}}
             \end{array}
        }
        {
            \djum{\uvartwo:\lneg{\typ_1},\seq{\uvar}:\lneg{\tenv_1}}{\lvartwo:\lneg{\typtwo}}{\ibott{\tmtwo\sub{\seq{\lvar}}{\seq{\ofc{\uvar}}}}{\iofctwo{\lvarfour}{ (\ibott{\lvarfour}{\uvartwo})}}}{\bott}
        }
  }{
    \djum{\uvartwo:\lneg{\typ_1}, \seq{\uvar}:\lneg{\tenv_1}}{\lvarthree:\why{\lneg{\typtwo}}}{\ewhy{(\ibott{\tmtwo\sub{\seq{\lvar}}{\seq{\ofc{\uvar}}}}{\iofctwo{\lvarfour}{ (\ibott{\lvarfour}{\uvartwo})}})}{\lvartwo}{\lvarthree}}{\bott}
  }
}
{
  \djum{\seq{\uvar}:\lneg{\tenv_1}}{\lvarthree:\why{\lneg{\typtwo}}}{\iwhy{\uvartwo}{\ewhy{(\ibott{\tmtwo\sub{\seq{\lvar}}{\seq{\ofc{\uvar}}}}{\iofctwo{\lvarfour}{ (\ibott{\lvarfour}{\uvartwo})}})}{\lvartwo}{\lvarthree}}}{\why{\typ_1}}
  }
  \]

  Finally, from $\djum{\seq{\uvar}:\lneg{\tenv_1}}{\lvarthree:\why{\lneg{\typtwo}}}{\iwhy{\uvartwo}{\ewhy{(\ibott{\tmtwo\sub{\seq{\lvar}}{\seq{\ofc{\uvar}}}}{\iofctwo{\lvarfour}{ (\ibott{\lvarfour}{\uvartwo})}})}{\lvartwo}{\lvarthree}}}{\why{\typ_1}}$ and \rlem{unrestricted_eq_bang_linear}(\ref{from_lin}), we have $\djum{\emptyutenv}{\seq{\lvarfour}:\ofc{\lneg{\tenv_1}}, \lvarthree:\why{\lneg{\typtwo}}}{\iwhy{\uvartwo}{\ewhy{(\ibott{\tmtwo\sub{\seq{\lvar}}{\seq{\ofc{\uvar}}}}{\iofctwo{\lvarfour}{ (\ibott{\lvarfour}{\uvartwo})}})}{\lvartwo}{\lvarthree}}\eofc{\seq{\uvar}}{\seq{\lvarfour}}}{\why{\typ_1}}$. Thus it suffices to take $\tm\eqdef \iwhy{\uvartwo}{\ewhy{(\ibott{\tmtwo\sub{\seq{\lvar}}{\seq{\ofc{\uvar}}}}{\iofctwo{\lvarfour}{ (\ibott{\lvarfour}{\uvartwo})}})}{\lvartwo}{\lvarthree}}\eofc{\seq{\uvar}}{\seq{\lvarfour}}$.
\end{enumerate}

\item $\rullW$: The derivation ends in
\[
  \indrule{\rullW}{
    \jull{\tenv_1}
  }{
    \jull{\tenv_1,\why{\typtwo}}
  }
\]
There are two cases depending on whether $\typ=\why{\typtwo}$ or $\typ \in \tenv_1$.
\begin{enumerate}

\item  $\typ=\why{\typtwo}$. Note that $\tenv_1$ cannot be empty since $\MELL$ is consistent. Let $\tenv_1=\tenv_2,\typthree$. By the \ih $\djum{\emptyutenv}{\seq{\lvar}:\lneg{\tenv_2}}{\tmtwo}{\typthree}$. By weakening (\rlem{weakening_and_contraction}) $\djum{\uvar:\lneg{\typtwo}}{\seq{\lvar}:\lneg{\tenv_2}}{\tmtwo}{\typthree}$. Then we can construct the derivation below and conclude:
  \[
\indrule{\rulmIWhy}{
  \indrule{\rulmIBott}{
    \djum{\uvar:\lneg{\typtwo}}{\seq{\lvar}:\lneg{\tenv_2}}{\tmtwo}{\typthree}
    \HS
    \djum{\uvar:\lneg{\typtwo}}{\lvarthree:\lneg{\typthree}}{\lvarthree}{\lneg{\typthree}}
  }{
    \djum{\uvar:\lneg{\typtwo}}{\seq{\lvar}:\lneg{\tenv_2}, \lvarthree:\lneg{\typthree}}{\ibott{\tmtwo}{\lvarthree}}{\bott}
  }
}
{
      \djum{\emptyutenv}{\seq{\lvar}:\lneg{\tenv_2}, \lvarthree:\lneg{\typthree}}{\iwhy{\uvar}{\ibott{\tmtwo}{\lvarthree}}}{\why{\typtwo}}
  }
\]
  Thus it suffices to take $\tm\eqdef \iwhy{\uvar}{\ibott{\tmtwo}{\lvarthree}}$.

\item $\typ \in \tenv_1$. Then $\tenv_1=\tenv_2,\typ$.  By the \ih $\djum{\emptyutenv}{\seq{\lvar}:\lneg{\tenv_2}}{\tm}{\typ}$. By  weakening \rlem{weakening_and_contraction} $\djum{\uvar:\lneg{\typtwo}}{\seq{\lvar}:\lneg{\tenv_2}}{\tm}{\typ}$. By \rlem{unrestricted_eq_bang_linear}(\ref{to_lin}), $\djum{\emptyutenv}{\lvartwo: \ofc{\lneg{\typtwo}},\seq{\lvar}:\lneg{\tenv_2}}{\tm\eofc{\uvar}{\lvartwo}}{\typ}$. We thus conclude.
\end{enumerate}

\item $\rullD$: The derivation ends in
  \[
  \indrule{\rullD}{
    \jull{\tenv_1,\typtwo}
  }{
    \jull{\tenv_1,\why{\typtwo}}
  }
\]
There are two cases depending on whether $\typ=\why{\typtwo}$ or $\typ \in \tenv_1$.
\begin{enumerate}

\item  $\typ=\why{\typtwo}$. By the \ih $\djum{\emptyutenv}{\seq{\lvar}:\lneg{\tenv_1}}{\tmtwo}{\typtwo}$. By weakening $\djum{\uvar:\lneg{\typtwo}}{\seq{\lvar}:\lneg{\tenv_1}}{\tmtwo}{\typtwo} $. We conclude as follows:
  \[
    \indrule{\rulmIWhy}
    {
  \indrule{\rulmIBott}{
    \djum{\uvar:\lneg{\typtwo}}{\seq{\lvar}:\lneg{\tenv_1}}{\tmtwo}{\typtwo}
    \HS
    \indrule{\rulmAxU}{
    }{
      \djum{\uvar:\lneg{\typtwo}}{\emptytenv}{\uvar}{\lneg{\typtwo}}
    }
  }{
    \djum{\uvar:\lneg{\typtwo}}{\seq{\lvar}:\lneg{\tenv_1}}{\ibott{\tmtwo}{\uvar}}{\bott}
  }
}
{
    \djum{\emptyutenv}{\seq{\lvar}:\lneg{\tenv_1}}{\iwhy{\uvar}{\ibott{\tmtwo}{\uvar}}}{\why{\typtwo}}
 }
  \]
  Thus it suffices to take $\tm \eqdef \iwhy{\uvar}{\ibott{\tmtwo}{\uvar}}$.
  
\item $\typ \in \tenv_1$. Then $\tenv_1=\tenv_2,\typ$.  By the \ih $\djum{\emptyutenv}{\seq{\lvar}:\lneg{\tenv_2},\lvartwo:\lneg{\typtwo}}{\tmtwo}{\typ}$. By weakening \rlem{weakening_and_contraction}, $\djum{\uvar: \lneg{\typtwo}}{\seq{\lvar}:\lneg{\tenv_2},\lvartwo:\lneg{\typtwo}}{\tmtwo}{\typ}$. Then we obtain $\djum{\uvar:\lneg{\typtwo}}{\seq{\lvar}:\lneg{\tenv_2}}{\tmtwo\sub{\lvartwo}{\uvar}}{\typ}$ as follows:
  \[
  \indrule{\rulmSub}{
    \djum{\uvar: \lneg{\typtwo}}{\emptytenv}{\uvar}{\lneg{\typtwo}}
    \HS
   \djum{\uvar: \lneg{\typtwo}}{\seq{\lvar}:\lneg{\tenv_2},\lvartwo:\lneg{\typtwo}}{\tmtwo}{\typ}
  }{
    \djum{\uvar: \lneg{\typtwo}}{\seq{\lvar}:\lneg{\tenv_2},}{\tmtwo\sub{\lvartwo}{\uvar}}{\typ}
  }
\]
  By \rlem{unrestricted_eq_bang_linear}(\ref{to_lin}),  $\djum{\emptyutenv}{\lvartwo:\ofc{\lneg{\typtwo}},\seq{\lvar}:\lneg{\tenv_2}}{\tmtwo\sub{\lvartwo}{\uvar}\eofc{\uvar}{\lvartwo}}{\typ}$.
We thus conclude, taking $\tm\eqdef \tmtwo\sub{\lvartwo}{\uvar}\eofc{\uvar}{\lvartwo}$.

\end{enumerate}

\item $\rullC$: The derivation ends in
  \[
  \indrule{\rullC}{
    \jull{\tenv_1,\why{\typtwo},\why{\typtwo}}
  }{
    \jull{\tenv_1,\why{\typtwo}}
  }
\]
There are two cases depending on whether $\typ=\why{\typtwo}$ or $\typ \in \tenv_1$.
\begin{enumerate}

\item  $\typ=\why{\typtwo}$. By the \ih $\djum{\emptyutenv}{\seq{\lvar}:\lneg{\tenv_1},\lvartwo:\ofc{\lneg{\typtwo}}}{\tmtwo}{\why{\typtwo}}$. By \rlem{unrestricted_eq_bang_linear}(\ref{from_lin}), $\djum{\uvar:\lneg{\typtwo}}{\seq{\lvar}:\lneg{\tenv_1}}{\tmtwo\sub{\lvar}{\ofc{\uvar}}}{\why{\typtwo}}$. We then construct the following derivation and conclude:
  \[
    \indrule{\rulmIWhy}
    {
    \indrule{\rulmIBott}{
      \djum{\uvar:\lneg{\typtwo}}{\seq{\lvar}:\lneg{\tenv_1}}{\tmtwo\sub{\lvar}{\ofc{\uvar}}}{\why{\typtwo}}
      \HS
      \djum{\uvar:\lneg{\typtwo}}{\emptytenv}{\ofc{\uvar}}{\ofc{\lneg{\typtwo}}}
    }{
      \djum{\uvar:\lneg{\typtwo}}{\seq{\lvar}:\lneg{\tenv_1}}{\ibott{(\tmtwo\sub{\lvar}{\ofc{\uvar}})}{\ofc{\uvar}}}{\bott}
    }
  }
  {
          \djum{\emptyutenv}{\seq{\lvar}:\lneg{\tenv_1}}{\iwhy{\uvar}{\ibott{(\tmtwo\sub{\lvar}{\ofc{\uvar}})}{\ofc{\uvar}}}}{\why{\typtwo}}

    }
  \]
  Thus it suffices to take $\tm \eqdef \iwhy{\uvar}{\ibott{(\tmtwo\sub{\lvar}{\ofc{\uvar}})}{\ofc{\uvar}}}$.
    
\item $\typ \in \tenv_1$. Then $\tenv_1=\tenv_2,\typ$.  We reason as follows

  \[ \begin{array}{rl}
       \djum{\emptyutenv}{\seq{\lvar}:\lneg{\tenv_2},\lvartwo_1:\ofc{\lneg{\typtwo}}, \lvartwo_2:\ofc{\lneg{\typtwo}}}{\tmtwo}{\typ} & \ih \\

       \djum{\uvar_1:\lneg{\typtwo},\uvar_2:\lneg{\typtwo}}{\seq{\lvar}:\lneg{\tenv_2}}{\tmtwo\sub{\seq{\lvartwo}}{\seq{\ofc{\uvar}}}}{\typ} &  \rlem{unrestricted_eq_bang_linear}(\ref{from_lin}) \\

       \djum{\uvar_1:\lneg{\typtwo}}{\seq{\lvar}:\lneg{\tenv_2}}{(\tmtwo\sub{\seq{\lvartwo}}{\seq{\ofc{\uvar}}})\sub{\uvar_2}{\uvar_1}}{\typ} & \rlem{weakening_and_contraction} \\
       \djum{\emptyutenv}{\seq{\lvar}:\lneg{\tenv_2}, \lvartwo:\ofc{\lneg{\typtwo}}}{(\tmtwo\sub{\seq{\lvartwo}}{\seq{\ofc{\uvar}}})\sub{\uvar_2}{\uvar_1}\eofc{\uvar_1}{\lvartwo}}{\typ} &  \rlem{unrestricted_eq_bang_linear}(\ref{to_lin})
     \end{array}
     \]

\end{enumerate}

\item $\rullOne$. The derivation ends in:
  \[
       \indrule{\rullOne}{
     \emptyPremise
  }{
     \jull{\one}
   }
 \]
We then take:
 \[
    \indrule{\rulmIOne}{
    \emptyPremise
  }{
    \djum{\emptyutenv}{\emptytenv}{\ione}{\one}
  }
\]
  
\item $\rullBott$. The derivation ends in:
    \[
      \indrule{\rullBott}{
     \jull{\tenv}
  }{
     \jull{\tenv,\bott}
   }
 \]
 We consider two cases depending on whether $\typ=\bott$ or $\typ\in\tenv$.

  \begin{enumerate}

    \item $\typ=\bott$. Since $\MELL$ is consistent, $\tenv$ cannot be empty. Thus $\tenv=\tenv_1,\typtwo$.
 \[
     \indrule{\rulmIBott}{
    \indih{\djum{\emptyutenv}{\lneg{\tenv_1}}{\tmtwo}{\typtwo}}
    \HS
    \djum{\emptyutenv}{\lvar:\lneg{\typtwo}}{\lvar}{\lneg{\typtwo}}
  }{
    \djum{\emptyutenv}{\lneg{\tenv_1},\lvar:\lneg{\typtwo}}{\ibott{\tmtwo}{\lvar}}{\bott}
  }
\]
Thus it suffices to take $\tm\eqdef \ibott{\tmtwo}{\lvar}$.

\item  $\typ\in\tenv$. Then $\tenv=\tenv_1,\typ$.
  \[
    \indrule{\rulmEOne}{
    \djum{\emptyutenv}{\lvar:\one}{\lvar}{\one}
    \HS
     \indih{\djum{\emptyutenv}{\lneg{\tenv_1}}{\tmtwo}{\typ}}
  }{
    \djum{\emptyutenv}{\lneg{\tenv_1},\lvar:\one}{\tmtwo\eone{\lvar}}{\typ}
  }
\]
Thus it suffices to take $\tm\eqdef \tmtwo\eone{\lvar}$.

\end{enumerate}

  \end{enumerate}
  
\end{proof}

\ContrasubstitutionLemma*

\begin{proof}\label{contrasubstitution_lemma:proof}
By induction on the size of the derivation of $\djum{\utenv}{\tenv_1,\lvar:\typ}{\tm}{\typtwo}$. 
We omit the cases already addressed in 
\ref{contrasubstitution_lemma_MLL:proof}  and detail only the new cases below.

\begin{enumerate}

\item Case $\rulmAx$. 
\[
    \indrule{\rulmAx}{
    \emptyPremise
  }{
    \djum{\utenv}{\lvar:\typ}{\lvar}{\typ}
  }
\]
Then $\typtwo=\typ$ and we conclude from the hypothesis.

\item Case $\rulmAxU$. 
%   \[
%   \indrule{\rulmAxU}{
%     \emptyPremise
%   }{
%     \djum{\utenv, \uvar:\typ}{\emptytenv}{\uvar}{\typ}
%   }
% \]
This case is not possible since $\lvar$ does not occur in $\uvar$.

% \item Case $\rulmITensor$. Then $\tm=\pair{\tm_1}{\tm_2}$ and $\typtwo=\typtwo_1\tensor\typtwo_2$ and there are two cases.

%   \begin{enumerate}
%     \item $\lvar:\typ$ is used to type $\tm_1$. \TODO{TODO}
% \[
%     \indrule{\rulmITensor}{
%     \djum{\utenv}{\tenv_1,\lvar:\typ}{\tm_1}{\typtwo_1}
%     \HS
%     \djum{\utenv}{\tenv_2}{\tm_2}{\typtwo_2}
%   }{
%     \djum{\utenv}{\tenv_1,\lvar:\typ,\tenv_2}{\pair{\tm_1}{\tm_2}}{\typtwo_1\tensor\typtwo_2}
%   }
% \]
% \item $\lvar:\typ$ is used to type $\tm_2$. \TODO{TODO}
%   \[
%     \indrule{\rulmITensor}{
%     \djum{\utenv}{\tenv_1}{\tm_1}{\typtwo_1}
%     \HS
%     \djum{\utenv}{\tenv_2, \lvar:\typ}{\tm_2}{\typtwo_2}
%   }{
%     \djum{\utenv}{\tenv_1,\tenv_2,\lvar:\typ}{\pair{\tm_1}{\tm_2}}{\typtwo_1\tensor\typtwo_2}
%   }
% \]
% \end{enumerate}

\item Case $\rulmETensor$. There are two cases.

  \begin{enumerate}

     \item  $\lvar:\typ$ is used to type $\tm_2$. 

  \[
      \indrule{\rulmETensor}{
    \djum{\utenv}{\tenv_{11}, \lvar:\typ}{\tm_2}{\typthree\tensor\typfour}
    \HS
    \djum{\utenv}{\tenv_{12},\lvartwo:\typthree,\lvarthree:\typfour}{\tm_1}{\typtwo}
  }{
    \djum{\utenv}{\tenv_{11},\lvar:\typ,\tenv_{12}}{\tm_1\epair{\lvartwo}{\lvarthree}{\tm_2}}{\typtwo}
  }
\]

\[
  \indrule{}{
        \djum{\utenv}{\tenv_{11}, \lvar:\typ}{\tm_2}{\typthree\tensor\typfour}
        \quad
       \indrule{\rulmIPar}{
         \indrule{\rulmIBott}{
        \djum{\utenv}{\tenv_{12}, \lvartwo:\typthree,\lvarthree:\typfour}{\tm_1}{\typtwo}
        \HS
            \djum{\utenv}{\tenv_2}{\tmtwo}{\lneg{\typtwo}}

      }{
        \djum{\utenv}{\tenv_{12}, \lvartwo:\typthree,\lvarthree:\typfour,\tenv_2}{\ibott{\tm_1}{\tmtwo}}{\bott}
      }
     }{
        \djum{\utenv}{\tenv_{12},\tenv_2}{\ipar{\lvartwo}{\lvarthree}{\ibott{\tm_1}{\tmtwo}}}{\lneg{\typthree}\parr\lneg{\typfour}}
      }
  }
  {
      \djum{\utenv}{\tenv_{11}, \tenv_{12},\tenv_2}{\tm_2\cos{\lvar}{\ipar{\lvartwo}{\lvarthree}{\ibott{\tm_1}{\tmtwo}}}}{\lneg{\typ}}
    }
  \]

\item  $\lvar:\typ$ is used to type $\tm_1$.

    \[
      \indrule{\rulmETensor}{
    \djum{\utenv}{\tenv_{11}}{\tm_2}{\typthree\tensor\typfour}
    \HS
    \djum{\utenv}{\tenv_{12}, \lvar:\typ,\lvartwo:\typthree,\lvarthree:\typfour}{\tm_1}{\typtwo}
  }{
    \djum{\utenv}{\tenv_{11},\tenv_{12},\lvar:\typ}{\tm_1\epair{\lvartwo}{\lvarthree}{\tm_2}}{\typtwo}
  }
\]

    \[
      \indrule{\rulmETensor}{
    \djum{\utenv}{\tenv_{11}}{\tm_2}{\typthree\tensor\typfour}
    \HS
      \indrule{}{
    \djum{\utenv}{\tenv_{12}, \lvar:\typ,\lvartwo:\typthree,\lvarthree:\typfour}{\tm_1}{\typtwo}
    \HS
      \djum{\utenv}{\tenv_2}{\tmtwo}{\lneg{\typtwo}}
  }{
    \djum{\utenv}{\tenv_{12}, \lvartwo:\typthree,\lvarthree:\typfour,\tenv_2}{\tm_1\cos{\lvar}{\tmtwo}}{\lneg{\typ}}
  }
}
{
  \djum{\utenv}{\tenv_{11},\tenv_{12},\lvar:\typ}{\tm_1\cos{\lvar}{\tmtwo}\epair{\lvartwo}{\lvarthree}{\tm_2}}{\lneg{\typ}}
  }
\]

\end{enumerate}

\item Case $\rulmIPar$

   \[       \indrule{\rulmIPar}{
  \djum{\utenv}{\tenv_1,\lvar:\typ,\lvartwo:\typtwo_1,\lvarthree:\typtwo_2}{\tm}{\bott}
  }{
    \djum{\utenv}{\tenv_1,\lvar:\typ}{\ipar{\lvartwo}{\lvarthree}{\tm}}{\lneg{\typtwo_1}\parr\lneg{\typtwo_2}}
  }
\]

\[
  \indrule{\rulmETensor}{
    \djum{\utenv}{\tenv_2}{\tmtwo}{\typtwo_1\tensor \typtwo_2}
    \HS
    \indrule{}{
           \djum{\utenv}{\tenv_1,\lvar:\typ,\lvartwo:\typtwo_1,\lvarthree:\typtwo_2}{\tm}{\bott}
           \HS
           \djum{\utenv}{\emptytenv}{\ione}{\one}
           }{
           \djum{\utenv}{\tenv_1,\lvartwo:\typtwo_1,\lvarthree:\typtwo_2}{\tm\cos{\lvar}{\ione}}{\lneg{\typ}}
           }
  }{
    \djum{\utenv}{\tenv_1,\tenv_2}{\tm\cos{\lvar}{\ione}\epair{\lvartwo}{\lvarthree}{\tmtwo}}{\lneg{\typ}}
  }
  \]

% \item Case $\rulmEParOne$. \TODO{TODO}
%   \[
%     \indrule{\rulmEParOne}{
%     \djum{\utenv}{\tenv_1}{\tm_1}{\typthree\parr\typtwo}
%     \HS
%     \djum{\utenv}{\tenv_2}{\tm_2}{\lneg{\typthree}}
%   }{
%     \djum{\utenv}{\tenv_1,\tenv_2}{\ap{\tm_1}{\tm_2}}{\typtwo}
%   }
% \]

% \item Case $\rulmEParTwo$. \TODO{TODO}
%  \[
%   \indrule{\rulmEParTwo}{
%     \djum{\utenv}{\tenv_1}{\tm_1}{\typtwo\parr\typthree}
%     \HS
%     \djum{\utenv}{\tenv_2}{\tm_2}{\lneg{\typthree}}
%   }{
%     \djum{\utenv}{\tenv_1,\tenv_2}{\invap{\tm_1}{\tm_2}}{\typtwo}
%   }
% \]

  \item Case $\rulmIOfc$. This case is not possible since $\tenv_1,\lvar:\typ$ is not the empty context.

  \item Case $\rulmEOfc$. Then $\tm=\tm_1\eofc{\uvar}{\tm_2}$ for some $\tm_1,\tm_2$ and $\uvar$. There are two cases depending on whether $\lvar:\typ$ is used to type $\tm_1$ or $\tm_2$.
    \begin{enumerate}
      \item  $\lvar:\typ$ is used to type $\tm_1$.  Then $\tenv_1=\tenv_{11},\tenv_{12}$ and the derivation ends in

    \[
  \indrule{\rulmEOfc}
  {\djum{\utenv}{\tenv_{11}}{\tm_2}{\ofc{\typthree}}
    \HS
    \djum{\utenv,\uvar:\typthree}{\tenv_{12},\lvar:\typ}{\tm_1}{\typtwo}
    }
    {\djum{\utenv}{\tenv_{11},\tenv_{12},\lvar:\typ}{\tm_1\eofc{\uvar}{\tm_2}}{\typtwo}}
  \]

  By weakening (\rlem{weakening_and_contraction}), $\djum{\utenv, \uvar:\typthree}{\tenv_2}{\tmtwo}{\lneg{\typtwo}}$. Thus we use the \ih to construct the derivation:
  \[
      \indrule{\rulmEOfc}
      {
        \djum{\utenv}{\tenv_{11}}{\tm_2}{\ofc{\typthree}}
        \HS
  \indrule{\rulmCos}{
    \djum{\utenv,\uvar:\typthree}{\tenv_{12},\lvar:\typ}{\tm_1}{\typtwo}
    \HS
    \djum{\utenv, \uvar:\typthree}{\tenv_1}{\tmtwo}{\lneg{\typtwo}}
  }{
    \djum{\utenv, \uvar:\typthree}{\tenv_{12},\tenv_2}{\tm_1\cos{\lvar}{\tmtwo}}{\lneg{\typ}}
  }
}
  {\djum{\utenv}{\tenv_{11},\tenv_{12},\tenv_2}{\tm_1\cos{\lvar}{\tmtwo}\eofc{\uvar}{\tm_2}}{\lneg{\typ}}}
\]
\item  $\lvar:\typ$ is used to type $\tm_2$.

  \[
  \indrule{\rulmEOfc}
  {\djum{\utenv}{\tenv_{11},\lvar:\typ}{\tm_2}{\ofc{\typthree}}
    \HS
    \djum{\utenv,\uvar:\typthree}{\tenv_{12}}{\tm_1}{\typtwo}
    }
    {\djum{\utenv}{\tenv_{11}, \tenv_{12},\lvar:\typ}{\tm_1\eofc{\uvar}{\tm_2}}{\typtwo}}
  \]

By weakening (\rlem{weakening_and_contraction}), $\djum{\utenv, \uvar:\typthree}{\tenv_2}{\tmtwo}{\lneg{\typtwo}}$. Thus we use the \ih to construct the derivation:
\[
    \indrule{\rulmCos}{
    \djum{\utenv}{\tenv_{11},\lvar:\typ}{\tm_2}{\ofc{\typthree}}
    \HS  
  \indrule{\rulmIWhy}{
       \indrule{\rulmIBott}{
         \djum{\utenv,\uvar:\typthree}{\tenv_{12}}{\tm_1}{\typtwo}
         \HS
         \djum{\utenv, \uvar:\typthree}{\tenv_2}{\tmtwo}{\lneg{\typtwo}}
       }{
         \djum{\utenv, \uvar:\typthree}{\tenv_{12},\tenv_2}{\ibott{\tm_1}{\tmtwo}}{\bott}
       }
}
{
         \djum{\utenv}{\tenv_{12},\tenv_2}{\iwhy{\uvar}{\ibott{\tm_1}{\tmtwo}}}{\why{\lneg{\typthree}}}
  }
}{
      \djum{\utenv}{\tenv_{11},\tenv_{12},\tenv_2}{\tm_2\cos{\lvar}{\iwhy{\uvar}{\ibott{\tm_1}{\tmtwo}}}}{\lneg{\typ}}
  }
\]

\end{enumerate}

\item Case $\rulmIWhy$. Then $\typtwo= \why{\typtwo_1}$, for some type $\typtwo_1$,  and $\tm=\iwhy{\uvar}{\tm_1}$, for some $\tm_1$ and $\uvar$.
  
   \
  \[
  \indrule{\rulmIWhy}{
    \djum{\utenv,\uvar:\lneg{\typtwo_1}}{\tenv_{1},\lvar:\typ}{\tm_1}{\bott}
  }{
    \djum{\utenv}{\tenv_{1},\lvar:\typ}{\iwhy{\uvar}{\tm_1}}{\why{\typtwo_1}}
  }
\] 
We use the \ih and construct the derivation (note that $\lneg{\typtwo} = \ofc{\lneg{\typtwo_1}}$):
\[
    \indrule{\rulmEOfc}
  {\djum{\utenv}{\tenv_2}{\tmtwo}{\lneg{\typtwo}}
    \HS
       \indrule{\rulmIBott}{
        \djum{\utenv,\uvar:\lneg{\typtwo_1}}{\tenv_{1},\lvar:\typ}{\tm_1}{\bott}
         \HS
         \djum{\utenv,\uvar:\lneg{\typtwo_1}}{\emptytenv}{\ione}{\one}
       }{
         \djum{\utenv, \uvar:\lneg{\typtwo_1}}{\tenv_{1}}{\tm_1\cos{\lvar}{\ione}}{\lneg{\typ}}
       }
     }
    {\djum{\utenv}{\tenv_{1},\tenv_2}{\tm_1\cos{\lvar}{\ione}\eofc{\uvar}{\tmtwo}}{\lneg{\typ}}}
\]

\item Case $\rulmEWhy$. Then $\tm=\ewhy{\tm_1}{\lvar}{\tm_2}$ and $\typtwo=\bott$ and thus $\lneg{\typtwo}=\one$.

  \[
     \indrule{\rulmEWhy}{
    \djum{\utenv}{\tenv_1, \lvar:\typ}{\tm_3}{\why{\typthree}}
    \HS
  \djum{\utenv}{\lvarthree:\typthree}{\tm_2}{\bott}
  }{
    \djum{\utenv}{\tenv_1,\lvar:\typ}{\ewhy{\tm_2}{\lvarthree}{\tm_3}}{\bott}
  }
    \]

    \[
      \indrule{\rulmEOne}
      {\djum{\utenv}{\tenv_2}{\tmtwo}{\lneg{\typtwo}}
        \HS
        \indrule{\rulmCos}
        {
          \djum{\utenv}{\tenv_1, \lvar:\typ}{\tm_3}{\why{\typthree}}
          \HS
          \indrule{\rulmIOfc}{
            \djum{\utenv}{\lvarthree:\typthree}{\tm_2}{\bott}
          }        
        {
          \djum{\utenv}{\emptytenv}{\iofctwo{\lvarthree}{\tm_2}}{\ofc{\lneg{\typthree}}}
          }
        }
        {
          \djum{\utenv}{\tenv_1}{\tm_3\cos{\lvar}{\iofctwo{\lvarthree}{\tm_2}}}{\lneg{\typ}}
        }
      } 
      {
        \djum{\utenv}{\tenv_1,\tenv_2}{\tm_3\cos{\lvar}{\iofctwo{\lvarthree}{\tm_2}}\eone{\tmtwo}}{\lneg{\typ}}
      }
    \]

\item Case $\rulmIBott$. Then $\typtwo= \bott$, $\tm=\ibott{\tm_1}{\tm_2}$, and there are two cases depending on whether $\lvar\in\fv{\tm_1}$ or $\lvar\in\fv{\tm_2}$.

  \begin{enumerate}

    \item $\lvar\in\fv{\tm_1}$. Then the derivation ends in
  \[
 \indrule{\rulmIBott}{
    \djum{\utenv}{\tenv_{11},\lvar:\typ}{\tm_1}{\typthree}
    \HS
    \djum{\utenv}{\tenv_{12}}{\tm_2}{\lneg{\typthree}}
  }{
    \djum{\utenv}{\tenv_{11},\tenv_{12},\lvar:\typ}{\ibott{\tm_1}{\tm_2}}{\bott}
  }
\]

\[
  \indrule{}
  {
    \djum{\utenv}{\tenv_{2}}{\tmtwo}{\one}
    \HS
     \indrule{\rulmIBott}{
        \djum{\utenv}{\tenv_{11},\lvar:\typ}{\tm_1}{\typthree}
        \HS
        \djum{\utenv}{\tenv_{12}}{\tm_2}{\lneg{\typthree}}
      }{
        \djum{\utenv}{\tenv_{11},\tenv_{12}}{\tm_1\cos{\lvar}{\tm_2}}{\lneg{\typ}}
      }
}
{
   \djum{\utenv}{\tenv_{11},\tenv_{12},\tenv_2}{\tm_1\cos{\lvar}{\tm_2}\eone{\tmtwo}}{\lneg{\typ}}
  }
\]

    \item $\lvar\in\fv{\tm_2}$. Then the derivation ends in
  \[
 \indrule{\rulmIBott}{
    \djum{\utenv}{\tenv_{11}}{\tm_1}{\typthree}
    \HS
    \djum{\utenv}{\tenv_{12},\lvar:\typ}{\tm_2}{\lneg{\typthree}}
  }{
    \djum{\utenv}{\tenv_{11},\tenv_{12},\lvar:\typ}{\ibott{\tm_1}{\tm_2}}{\bott}
  }
\]

\[
  \indrule{}
  {
    \djum{\utenv}{\tenv_{2}}{\tmtwo}{\one}
    \HS
     \indrule{\rulmIBott}{
        \djum{\utenv}{\tenv_{12},\lvar:\typ}{\tm_2}{\typthree}
        \HS
        \djum{\utenv}{\tenv_{11}}{\tm_1}{\lneg{\typthree}}
      }{
        \djum{\utenv}{\tenv_{11},\tenv_{12}}{\tm_2\cos{\lvar}{\tm_1}}{\lneg{\typ}}
      }
}
{
   \djum{\utenv}{\tenv_{11},\tenv_{12},\tenv_2}{\tm_2\cos{\lvar}{\tm_1}\eone{\tmtwo}}{\lneg{\typ}}
  }
\]

\end{enumerate}

\item Case $\rulmIOne$. This case is not possible since then $\tm=\ione$ and $\fv{\ione}=\emptyset$.

  \[
    \indrule{\rulmIOne}{
    \emptyPremise
  }{
    \djum{\utenv}{\emptytenv}{\ione}{\one}
  }
\]

\item Case $\rulmEOne$. Then $\tm=\tm_1\eone{\tm_2}$, and there are two cases depending on whether $\lvar\in\fv{\tm_1}$ or $\lvar\in\fv{\tm_2}$.

  \begin{enumerate}
   \item $\lvar\in\fv{\tm_1}$. Then the derivation ends in

\[
       \indrule{\rulmEOne}{
    \djum{\utenv}{\tenv_{11}}{\tm_2}{\one}
    \HS
    \djum{\utenv}{\tenv_{12}, \lvar:\typ,}{\tm_1}{\typtwo}
  }{
    \djum{\utenv}{\tenv_{11},\tenv_{12}, \lvar:\typ,}{\tm_1\eone{\tm_2}}{\typtwo}
  }
\]

 \[
       \indrule{\rulmEOne}{
    \djum{\utenv}{\tenv_{11}}{\tm_2}{\one}
    \HS
    \indrule{}
    {
      \djum{\utenv}{\tenv_{12},\lvar:\typ}{\tm_1}{\typtwo}
      \HS
      \djum{\utenv}{\tenv_{2}}{\tmtwo}{\lneg{\typtwo}}
    }
    {
       \djum{\utenv}{\tenv_{12},\tenv_{2}}{\tm_1\cos{\lvar}{\tmtwo}}{\lneg{\typ}}
      }
  }{
    \djum{\utenv}{\tenv_{11},\tenv_{12},\tenv_2}{\tm_1\cos{\lvar}{\tmtwo}\eone{\tm_2}}{\lneg{\typ}}
  }
\]

    \item $\lvar\in\fv{\tm_2}$. Then the derivation ends in
\[
       \indrule{\rulmEOne}{
    \djum{\utenv}{\tenv_{11}, \lvar:\typ,}{\tm_2}{\one}
    \HS
    \djum{\utenv}{\tenv_{12}}{\tm_1}{\typtwo}
  }{
    \djum{\utenv}{\tenv_{11},\lvar:\typ,\tenv_{12}}{\tm_1\eone{\tm_2}}{\typtwo}
  }
\]

 \[
       \indrule{\rulmEOne}{
    \djum{\utenv}{\tenv_{11}, \lvar:\typ,}{\tm_2}{\one}
    \HS
    \indrule{}
    {
      \djum{\utenv}{\tenv_{12}}{\tm_1}{\typtwo}
      \HS
      \djum{\utenv}{\tenv_{2}}{\tmtwo}{\lneg{\typtwo}}
    }
    {
       \djum{\utenv}{\tenv_{12},\tenv_{2}}{\ibott{\tm_1}{\tmtwo}}{\bott}
      }
  }{
    \djum{\utenv}{\tenv_{11},\tenv_{12},\tenv_2}{\tm_2\cos{\lvar}{\ibott{\tm_1}{\tmtwo}}}{\lneg{\typ}}
  }
\]

\end{enumerate}

\end{enumerate}
  
\end{proof}

%   In the following result we assume that the variables in the patterns
% of eliminators are decorated with their types. For example,
% $\epair{\lvar^\typ}{\lvartwo^\typtwo}{\tm}$ rather than
% $\epair{\lvar}{\lvartwo}{\tm}$, so that $\flpv{\cdot}$ returns a typing
% environment. For example, $\flpv{\epair{\lvar^\typ}{\lvartwo^\typtwo}{\tm}}=\{\lvar:\typ,\lvartwo:\typtwo\}$.

% \begin{lemma}
% $\djum{\utenv}{\tenv}{\tm\cctx}{\typ} $ implies $\djum{\utenv\cup\fupv{\cctx}}{\tenv\mid_{\flv{\tm}}\cup\flpv{\cctx}}{\tm}{\typ}$
% \end{lemma}

% Pre-reduction in $\CalcMELL$ preserves types.

% \begin{lemma}[Subject Pre-reduction]
% \llem{subject_prereduction}
% $\djum{\utenv}{\tenv}{\tm}{\typ} $ and $\tm  \pretome \tmtwo$ implies $\djum{\utenv}{\tenv}{\tmtwo}{\typ}$.
% \end{lemma}

% \begin{proof}
%  \TODO{TODO}
  
% \end{proof}

%%% Local Variables:
%%% mode: latex
%%% TeX-master: "../main"
%%% End:

\section{Equivalence}
\lsec{app:calcMELL:equiv}
\subsection{Substitution and structural equivalence}

\begin{lemma}
\llem{equivalent_terms_have_same_free_variables}
  $\tm\cceq\tmtwo$ implies $\fv{\tm}=\fv{\tmtwo}$
\end{lemma}

\begin{proof}
By induction on the derivation of  $\tm\cceq\tmtwo$.
\end{proof}

\begin{corollary}
$\cceq$ is consistent.
\end{corollary}

\begin{proof}
From~\rlem{equivalent_terms_have_same_free_variables}, we deduce that $\lvar\cceq\lvartwo$ is not derivable, for arbitrary linear variables $\lvar$ and $\lvartwo$ such that $\lvar\neq\lvartwo$.
\end{proof}

\begin{lemma}[Contrasubstitution is compatible with Equivalence (Right Case)]
\llem{cos_compatible_with_equivalence:right}
Let $\tm$ be a linear term and $\lvar\in\fv{\tm}$. Suppose $\tmtwo\cceq\tmthree$. Then
$\tm\cos{\lvar}{\tmtwo}\cceq\tm\cos{\lvar}{\tmthree}$.
\end{lemma}

%\ContrasubstitutionCompatibleWithEquivalenceRightCase*

\begin{proof}%\label{cos_compatible_with_equivalence:right:proof}
  By induction on $\tm$.

  \end{proof}

  We say that  $\tmthree$ is \emph{contra-free for $\lvar$ in $\tm$}, if the
  free variables of $\tmthree$ are not bound in
  $\tm\cos{\lvar}{\tmthree}$.  
  
% \begin{lemma}[{name=Eliminator flotation for
%     contrasubstitution~\proofnote{Proof on
%       pg.~\pageref{elim_cos_cceq:proof}}, restate=[name=Eliminator
%     flotation for contrasubstitution]EliminatorFlotation}]

  \begin{lemma}[Eliminator flotation]
  \llem{elim_cos_cceq}
  Let $\tm$ be a linear term and $\lvar\in\fv{\tm}$. Suppose $\fv{\patt}\cap\fv{\tm}=\emptyset$ and $\tmthree$ is contra-free for $\lvar$ in $\tm$. Then
  \[\tm\cos{\lvar}{\tmtwo}\pelim{\patt}{\tmthree} \cceq   \tm\cos{\lvar}{\tmtwo\pelim{\patt}{\tmthree}}
    \]
\end{lemma}

%\EliminatorFlotation*

  \begin{proof}%\label{elim_cos_cceq:proof}
  By induction on $\tm$. 

\end{proof}

% \begin{lemma}[{name=Eliminator flotation for
%     substitution~\proofnote{Proof on
%       pg.~\pageref{elim_sub_cceq:proof}}, restate=[name=Eliminator
%     flotation for substitution]EliminatorFlotationForSubstitution}]
\begin{lemma}[Eliminator flotation]
  \llem{elim_sub_cceq}
 Let $\tm$ be a linear term and $\lvar\in\fv{\tm}$.  Suppose $\fv{\patt}\cap\fv{\tm}=\emptyset$ and $\tmthree$ is free for $\lvar$ in $\tm$. Then
  \[\tm\sub{\lvar}{\tmtwo}\pelim{\patt}{\tmthree} \cceq   \tm\sub{\lvar}{\tmtwo\pelim{\patt}{\tmthree}}
    \]
  \end{lemma}
  
%  \EliminatorFlotationForSubstitution*

  \begin{proof}%\label{elim_sub_cceq:proof}
  By induction on $\tm$.  

\end{proof}

% \TODO{En \rlem{simple_admissible}, hay casos que no son tipables pero que se consideran igual. Ej. $\tm-=\pair{\tm_1}{\tm_2}$. Puede que suceda lo mismo con otros lemas}

\begin{lemma}
\llem{sub_compatible_with_equivalence}
Suppose $\tmtwo\cceq\tmthree$. Then
\begin{enumerate}
\item $\tm\sub{\var}{\tmtwo}\cceq\tmtwo\sub{\var}{\tmthree}$; and
\item $\tmtwo\sub{\var}{\tm}\cceq\tmthree\sub{\var}{\tm}$; and
  
  \end{enumerate}
\end{lemma}

\begin{proof}
  By induction on $\tm$.
\end{proof}

\begin{remark}
  \lremark{eone_propagation}
  Since $\fv{\ione}=\emptyset$ we have:
  \begin{enumerate}
  \item\label{eone_propagation:pair} $\pair{\tm}{\tmtwo}\eone{\tmthree} \cceq \pair{\tm\eone{\tmthree}}{\tmtwo}\cceq \pair{\tm}{\tmtwo\eone{\tmthree}}$.
  \item\label{eone_propagation:ap} $(\ap{\tm}{\tmtwo})\eone{\tmthree} \cceq \ap{\tm\eone{\tmthree}}{\tmtwo}\cceq \ap{\tm}{\tmtwo\eone{\tmthree}}$.
  \item\label{eone_propagation:invap} $(\invap{\tm}{\tmtwo})\eone{\tmthree} \cceq \invap{\tm\eone{\tmthree}}{\tmtwo}\cceq \invap{\tm}{\tmtwo\eone{\tmthree}}$.
      \item\label{eone_propagation:ibott} $(\ibott{\tm}{\tmtwo})\eone{\tmthree} \cceq \ibott{\tm\eone{\tmthree}}{\tmtwo}\cceq \ibott{\tm}{\tmtwo\eone{\tmthree}}$.
      \item\label{eone_propagation:par} $(\ipar{\lvar}{\lvartwo}{\tm})\eone{\tmthree}  \cceq \ipar{\lvar}{\lvartwo}{\tm\eone{\tmthree}}$, if $\lvar,\lvartwo\notin\flv{\tmthree}$.
      \item\label{eone_propagation:pelim} $\tm\pelim{\patt}{\tmtwo}\eone{\tmthree} \cceq \tm\eone{\tmthree}\pelim{\patt}{\tmtwo} \cceq \tm\pelim{\patt}{\tmtwo\eone{\tmthree}}$.
      \item\label{eone_propagation:iwhy} $(\iwhy{\uvar}{\tm})\eone{\tmthree}  \cceq \iwhy{\uvar}{\tm\eone{\tmthree}}$, if $\uvar\notin\fv{\tmthree}$.
  \end{enumerate}
  
\end{remark}

%\GeneralizedSimpl*

\begin{lemma}[Generalized Simpl]
  \llem{simple_admissible}
Let $\tm$ be a linear term and $\lvar\in\fv{\tm}$. Then $\tm\cos{\lvar}{\ione}\eone{\tmthree} \cceq   \tm\cos{\lvar}{\tmthree}$
\end{lemma}

\begin{proof}
  By induction on $\tm$, using \rlem{elim_cos_cceq}  and \rlem{cos_compatible_with_equivalence:right} and $\ruleCCCosOneSimplLeft$.

\end{proof}

% \begin{lemma}[{name=Contrasubstitution Lemma I~\proofnote{Proof on pg.~\pageref{mell_sub_contra:proof}}, restate=[name=Contrasubstitution Lemma I]LinearSubContraLemma}]
\begin{lemma}[Contrasubstitution Lemma I]
  \llem{mell_sub_contra}
  Let $\set{\tm,\tmtwo,\tmthree}$ be a set of linear terms
such that $\lvar \in \fv{\tm}$ and let $\lvartwo \neq \lvar$.
Then
\begin{enumerate}
  \item $
  \tm\cos{\lvar}{\tmtwo}\sub{\lvartwo}{\tmthree}
  =
  \tm\sub{\lvartwo}{\tmthree}\cos{\lvar}{\tmtwo\sub{\lvartwo}{\tmthree}}
$. 
Note that $\lvartwo$ occurs at most once in $\tm,\tmtwo$,
so the right-hand side
is either of the form
  $\tm\sub{\lvartwo}{\tmthree}\cos{\lvar}{\tmtwo}$
or of the form
$\tm\cos{\lvar}{\tmtwo\sub{\lvartwo}{\tmthree}}$.
  \item $
  \tm\cos{\lvar}{\tmtwo}\sub{\uvar}{\tmthree}
  =
  \tm\sub{\uvar}{\tmthree}\cos{\lvar}{\tmtwo\sub{\uvar}{\tmthree}}
$. 
\end{enumerate}
\end{lemma}

%\LinearSubContraLemma*

\begin{proof}%\label{mell_sub_contra:proof}
By induction on $\tm$.
\end{proof}

% \begin{lemma}[{name=Contra/contra lemma~\proofnote{Proof on pg.~\pageref{mell_contra_contra:proof}}, restate=[name=Contra/contra lemma]ContraContraLemma}]
\begin{lemma}[Contra/contra lemma]
\llem{contra_contra_lemma} 
Let $\set{\tm,\tmtwo,\tmthree}$ be a linear set of terms
such that $\lvar \in \fv{\tm}$
and let $\lvartwo \neq \lvar$
be such that $\lvartwo \in \fv{\tm}\cup\fv{\tmtwo}$.
Then:
\begin{enumerate}
\item
  If $\lvartwo \in \fv{\tm}$, then:
  \[
    \tm\cos{\lvar}{\tmtwo}\cos{\lvartwo}{\tmthree}
    \cceq
    \tm\cos{\lvartwo}{\tmtwo}\sub{\lvar}{\tmthree}
  \]
\item
  If $\lvartwo \in \fv{\tmtwo}$, then:
  \[
    \tm\cos{\lvar}{\tmtwo}\cos{\lvartwo}{\tmthree}
    \cceq
    \tmtwo\cos{\lvartwo}{\tm\sub{\lvar}{\tmthree}}
  \]
\end{enumerate}
\end{lemma}

% \ContraContraLemma*

\begin{proof}%\label{mell_contra_contra:proof}
By simulateneous induction on $\tm$.
\end{proof}

% \begin{lemma}[{name=Nested Contra/contra lemma~\proofnote{Proof on pg.~\pageref{mell_nested_contra_contra:proof}}, restate=[name=Nested Contra/contra lemma]NestedContraContraLemma}]
\begin{lemma}[Nested Contra/contra lemma]
\llem{nested_contra_contra}
Let $\set{\tm,\tmtwo,\tmthree}$ be a linear set of terms
such that $\lvar \in \fv{\tm}$
and let $\lvartwo \neq \lvar$
be such that $\lvartwo \in \fv{\tmtwo}$.
Then:
  \[
    \tmtwo\sub{\lvartwo}{\tm}\cos{\lvar}{\tmthree}
    \cceq \tm\cos{\lvar}{\tmtwo\cos{\lvartwo}{\tmthree}}
  \]
\end{lemma}

% \NestedContraContraLemma*

\begin{proof}%\label{mell_nested_contra_contra:proof}
  By induction  on $\tmtwo$.
  \end{proof}

%%%%%%%%%%%%
%\subsection{Compatibility of contrasubstitution with equivalence}
%%%%%%%%%%%%

% The next result relies on the following lemmas: 
% \rlem{cos_compatible_with_equivalence:right},
% \rlem{elim_cos_cceq},
% \rlem{cos_with_target_under_context},
%  \rlem{mell_sub_contra},
%  \rlem{contra_contra_lemma}, and
%   \rlem{nested_contra_contra}.

  % \begin{lemma}[{name=Contrasubstitution is compatible with equivalence~\proofnote{Proof on pg.~\pageref{cos_compatible_with_equivalence:proof}}, restate=[name=Contrasubstitution is compatible with equivalence]ContrasubstitutionCompatibleWithEquivalence}]
\begin{lemma}[Contrasubstitution is compatible with equivalence]
\llem{cos_compatible_with_equivalence}
Suppose $\tmtwo\cceq\tmthree$. Then both $\tm\cos{\lvar}{\tmtwo}\cceq\tm\cos{\lvar}{\tmthree}$ and
$\tmtwo\cos{\lvar}{\tm}\cceq\tmthree\cos{\lvar}{\tm}$.
\end{lemma}

%\ContrasubstitutionCompatibleWithEquivalence*

  \begin{proof}\label{cos_compatible_with_equivalence:proof}
  The first item was proved in \rlem{cos_compatible_with_equivalence:right}. The second one can be proved by induction on the derivation of  $\tmtwo\cceq\tmthree$. 

\end{proof}

  %%%%%%%%%%%%
\subsection{An alternative presentation of $\cceq$}
%%%%%%%%%%%%

This subsection provides an equivalent characterization of the structural equivalence, established in order to facilitate the proof of the strong bisimulation theorem (\rthm{cceqalt_is_a_strong_bisimulation}).

\begin{definition}[Alternative-structural equivalence]
  \ldef{alt_equivalence}
 The structural equivalence relation $\cceqalt$ is defined as the reflexive, symmetric, transitive and contextual closure of the following axioms and assumes that both sides are typed:
  \[
    \begin{array}{llll}
            \of{\fctx}{\tm}\pelim{\patt}{\tmthree} & \cceqalt_{\ruleCCEqCtx} & \of{\fctx}{\tm \pelim{\patt}{\tmthree}} & (\fv{\patt}\cap\fv{\fctx}=\emptyset \mbox{ and } \tmthree\mbox{ free for }\fctx)
      \\
      \tm\eone{\ione} & \cceqalt_{\ruleCCEqBetaOne} & \tm
      \\      
      \tm\eone{\tmtwo}  & \cceqalt_{\ruleCCEqEOneSymm} & \tmtwo\eone{\tm}
      \\
      \ibott{\ione}{\tm} &
                           \cceqalt_{\ruleCCEqOneLUnitForBott} & \tm  &
      \\
      \ibott{(\invap{\tmtwo}{\tm})}{\tmthree} & \cceqalt_{\ruleCCEqInvapDualToAp} &
                                                            \ibott{\tm}{\ap{\tmtwo}{\tmthree}} &
      \\  
      \ibott{\pair{\tmthree}{\tm}}{\tmtwo} & \cceqalt_{\ruleCCEqPairDualToParElimTwo} & \ibott{\tmthree}{(\invap{\tmtwo}{\tm})}
     %  \\
     %  % \ibott{\lam{\lvartwo}{\tmtwo\cos{\lvarthree}{\tmthree}}}{\tm} & \cceqalt_{\ruleCCEqLamDualToPairElim} & \ibott {\tmthree}{\tmtwo \epair{\lvartwo}{\lvarthree}{\tm}} & 
     % \cancel{\ibott{\ipar{\lvar}{\lvartwo}{(\ibott{\tmtwo}{\tmthree})}}{\tm}}
     %                  & \cceqalt_{\ruleCCEqIParDualToEPair} & \ibott
     %                                                       {\tmthree}{\tmtwo
     %                                                       \epair{\lvar}{\lvartwo}{\tm}}
     %                    & \lvar,\lvartwo\notin\flv{\tmthree}
      \\
      \ibott{\ipar{\lvar}{\lvartwo}{\tmtwo}}{\tm}
                      & \cceqalt_{\ruleCCEqIParDualToEPair} & \tmtwo
                                                           \epair{\lvar}{\lvartwo}{\tm}
                        &
      % \\
      % \cancel{\ibott{(\iwhy{\uvar}{(\ibott{\tmtwo}{\tmthree})})}{\tm}} & \cceqalt_{\ruleCCEqIwhyDualToEofc}
      % & \ibott{\tmthree}{\tmtwo\eofc{\uvar}{\tm}}   &  \uvar\notin\fv{\tmthree}
      \\
      \ibott{(\iwhy{\uvar}{\tmtwo})}{\tm} & \cceqalt_{\ruleCCEqIwhyDualToEofc}
                                                                             & \tmtwo\eofc{\uvar}{\tm}
      % \\
      % \cancel{\ibott{(\iofctwo{\lvar}{\tm})\eone{\tmtwo}}{\tmthree}} & \cceqalt_{\ruleCCEqIofcDualToEwhy} & \ibott{\tmtwo}{\tm\ewhynew{\lvar}{\tmthree}{\tmtwo}} &  
      \\
     \ibott{(\iofctwo{\lvar}{\tm})}{\tmtwo} & \cceqalt_{\ruleCCEqIofcDualToEwhy} & \ewhy{\tm}{\lvar}{\tmtwo} &  
      \\
      % \ibott{(\ibott{\tmthree}{\tm})}{\tmtwo} & \cceqalt_{\ruleCCEqBottDualToEone}  & \ibott{\tmthree}{\tm\eone{\tmtwo}} & 
      \ibott{\tm\eone{\tmtwo}}{\tmthree} & \cceqalt_{\ruleCCEqBottDualToEone}  & \ibott{\tmtwo}{(\ibott{\tmthree}{\tm})} & 
    
    \end{array}
  \]
\end{definition}

\begin{lemma}
  \llem{cceq_derived}
  The following equations are in $\cceqalt$
  \[
    \begin{array}{llll}
      \ibott{\ione}{\tm} &
                           \cceqalt_{\ruleCCEqOneLUnitForBott}  & \tm  & 
      \\
      \of{\fctx}{\tm}\pelim{\patt}{\tmthree} & \cceqalt_{\ruleCCEqCtx} & \of{\fctx}{\tm \pelim{\patt}{\tmthree}} & \mbox{ if }\fv{\patt}\cap\fv{\fctx}=\emptyset \mbox{ and } \tmthree \mbox{ free for }\fctx
      \\
      \ibott{\tm}{\tmtwo} & \cceqalt_{\ruleCCEqiBottSym}    & \ibott{\tmtwo}{\tm} & 
      \\
      \ibott{(\invap{\tmtwo}{\tm})}{\tmthree} & \cceqalt_{\ruleCCEqInvapDualToPair} & \ibott{\tmtwo}{\pair{\tmthree}{\tm}} & 
      \\
      \ibott{(\invap{\tmtwo}{\tm})}{\tmthree} & \cceqalt_{\ruleCCEqInvapDualToAp} & \ibott{\tm}{\ap{\tmtwo}{\tmthree}} & 
      \\
      \ibott{(\ap{\tmtwo}{\tm})}{\tmthree} & \cceqalt_{\ruleCCEqApDualToParElimTwo} & \ibott{\tm}{\invap{\tmtwo}{\tmthree}} & 
      \\
      \ibott{\pair{\tm}{\tmtwo}}{\tmthree} & \cceqalt_{\ruleCCEqEPairDualToAp} & \ibott{\tmtwo}{(\ap{\tmthree}{\tm})} & 
      \\
      \ibott{\pair{\tm}{\tmtwo}}{\tmthree} & \cceqalt_{\ruleCCEqPairDualToParElimTwo} & \ibott{\tm}{(\invap{\tmthree}{\tmtwo})} & 
      \\
      \ibott{\ipar{\lvar}{\lvartwo}{\tmtwo}}{\tm}  & \cceq_{\ruleCCEqIParDualToEPair}  &  \tmtwo \epair{\lvar}{\lvartwo}{\tm} & 
      \\
      \ibott{(\iwhy{\uvar}{\tmtwo})}{\tm} & \cceqalt_{\ruleCCEqIwhyDualToEofc}
      & \tmtwo\eofc{\uvar}{\tm}   & 
      \\
       \ibott{(\iofctwo{\lvar}{\tm})\eone{\tmfour}}{\tmthree} & \cceqalt_{\ruleCCEqIofcDualToEwhy} & \ibott{\tmfour}{\ewhy{\tm}{\lvar}{\tmthree}} & 
      \\
     \ibott{\tm\eone{\tmtwo}}{\tmthree} & \cceqalt_{\ruleCCEqBottDualToEone}  & \ibott{\tmtwo}{(\ibott{\tmthree}{\tm})} & 
    \end{array}
  \]
\end{lemma}

\begin{proof}
  We give the proof for the rule $ruleCCEqOneLUnitForBott$, the remaining rules follow from similar reasonings.
      \[\begin{array}{llll}
          & \ibott{\tm}{\tmtwo} \\
          \cceqalt & \ibott{\tm\eone{\ione}}{\tmtwo}  & \ruleCCEqBetaOne\\
          \cceqalt & \ibott{\ione}{(\ibott{\tmtwo}{\tm})} & \ruleCCEqBottDualToEone \\
          \cceqalt & \ibott{\tmtwo}{\tm} & \ruleCCEqOneLUnitForBott\\          
        \end{array}\]

\end{proof}

We establish several auxiliar results to show that the relation $\cceqalt$ coincides with $\cceq$.
The only challenge lies in proving that $\cceqalt$ captures $\ruleCCEqCosAndSub$, the fourth axiom of $\cceq$ (\rdef{cceq}).
The remaining axioms $\ruleCCEqCtx$, $\ruleCCEqBetaOne$, and $\ruleCCEqEOneSymm$ are common to both relations.
To this end, we begin by defining the notions of a deep context and the dual of a deep context.

\begin{definition}
\emph{Deep contexts} $\dctx$ are defined by the grammar:
\[
  \begin{array}{lll}
    \dctx & ::= & \ctxhole\,|\,\ipar{\lvar}{\lvartwo}{\dctx}\,|\, \invap{\dctx}{\tm} \,|\,\invap{\tm}{\dctx}\,|\,\ap{\dctx}{\tm} \,|\,\ap{\tm}{\dctx}\,|\, \pair{\dctx}{\tm} \,\\
    &  & |\,\pair{\tm}{\dctx}\,|\,
                  \ibott{\dctx}{\tm} \,|\,\ibott{\tm}{\dctx}\,|\, \iwhy{\uvar}{\dctx}\,|\,
                  \dctx\pelim{\patt}{\tm} \,|\,\tm\pelim{\patt}{\dctx}\,|\,\ewhy{\tm}{\lvar}{\dctx} 
   % \dctx\pelim{\patt}{\tm} \,|\,\tm\pelim{\patt}{\dctx}\\
%         & & 
  \end{array}
\]
\end{definition}

We say $\tm$ is free for $\dctx$ if the free variables of $\tm$ are not captured in $\of{\dctx}{\tm}$.
Similarly, a \emph{context $\dctx$ is free for a deep context $\dctxtwo$} if $\of{\dctxtwo}{\dctx}$ does not bind free occurrences of variables of $\dctx$.

% \begin{lemma}
% Let $\dctx$ and $\dctxtwo$ be deep contexts. Then $\dctx\cos{\ctxhole}{\dctxtwo}$ is a deep context. 
% \end{lemma}

% \begin{proof}
% By induction on $\dctx$. 
% \end{proof}

% \begin{lemma}[Deep $\ruleCCEqCtx$]
% $      \of{\dctx}{\tm}\pelim{\patt}{\tmthree} \cceq  \of{\dctx}{\tm \pelim{\patt}{\tmthree}}$, if $\fv{\patt}\cap\fv{\dctx}=\emptyset$ and $\dctx$ is free for $\tmthree$.
% \end{lemma}

% \begin{proof}
% By induction  on $\dctx$ using $\ruleCCEqCtx$. 
% \end{proof}

\begin{definition}[Dual of a deep context]
 \ldef{dual_of_a_context}
 The dual of a deep context $\dctx$, denoted $\lneg{\dctx}$, is defined as the deep context $\dctx\cos{\ctxhole}{\ctxhole}$. 
\end{definition}

%\begin{lemma}[{name=Contrasubstitution and Dual of a Deep Context~\proofnote{Proof on pg.~\pageref{cos_as_negation_of_context:proof}}, restate=[name=Contrasubstitution and Dual of a Context]ContrasubstitutionAndDualOfAContext}]
\begin{lemma}[Contrasubstitution lemma]
\llem{cos_as_negation_of_context}
Suppose $\tm$ is free for $\dctx$. Then $\of{\dctx}{\lvar}\cos{\lvar}{\tm}\cceq\of{\lneg{\dctx}}{\tm}$.
\end{lemma}

% \ContrasubstitutionAndDualOfAContext*

\begin{proof}%\label{cos_as_negation_of_context:proof}
  By induction on $\dctx$.
\end{proof}

% \begin{lemma}[{name=Dual Contravariant on Deep Context Composition~\proofnote{Proof on pg.~\pageref{dual_of_nested_contexts:proof}}, restate=[name=Dual Contravariant on Context Composition]DualContravariantOnContextComposition}]
\begin{lemma}[Dual Contravariant on Deep Context Composition]
  \llem{dual_of_nested_contexts}
Suppose $\dctxtwo$ is free for $\dctx$. Then $\lneg{\of{\dctx}{\dctxtwo}}\cceq \of{\lneg{\dctxtwo}}{\lneg{\dctx}}$
\end{lemma}

  The following example illustrates why the condition that $\dctxtwo$ be free for $\dctx$ is required.
Suppose $\dctx= \ipar{\lvarthree}{\lvarfour}{\ctxhole}$ and $\dctxtwo=\ctxhole\epair{\lvar}{\lvartwo}{\pair{\lvarthree}{\lvarfour}}$. Then
\[
  \begin{array}{l}
    \of{\dctx}{\dctxtwo}=\ipar{\lvarthree}{\lvarfour}{\ctxhole\epair{\lvar}{\lvartwo}{\pair{\lvarthree}{\lvarfour}}} \\ 
    \lneg{\of{\dctx}{\dctxtwo}} = \ione\epair{\lvar}{\lvartwo}{\pair{\lvarthree}{\lvarfour}} \epair{\lvarthree}{\lvarfour}{\ctxhole}\\
    \lneg{\dctx} = \ione\epair{\lvarthree}{\lvarfour}{\ctxhole}\\
    \lneg{\dctxtwo} = \ctxhole\epair{\lvar}{\lvartwo}{\pair{\lvarthree}{\lvarfour}}\\
    \of{\lneg{\dctxtwo}}{\lneg{\dctx}} = \ione\epair{\lvarthree}{\lvarfour}{\ctxhole}\epair{\lvar}{\lvartwo}{\pair{\lvarthree}{\lvarfour}}\\
  \end{array}
\]
Notice that
\[
  \begin{array}{l}
    \lneg{\of{\dctx}{\dctxtwo}} = \ione\epair{\lvar}{\lvartwo}{\pair{\lvarthree}{\lvarfour}} \epair{\lvarthree}{\lvarfour}{\ctxhole}  \not\cceq
    \ione\epair{\lvarthree}{\lvarfour}{\ctxhole}\epair{\lvar}{\lvartwo}{\pair{\lvarthree}{\lvarfour}}
    =\of{\lneg{\dctxtwo}}{\lneg{\dctx}}
  \end{array}
\]

% \DualContravariantOnContextComposition*

\begin{proof}%\label{dual_of_nested_contexts:proof}
  By induction on $\dctx$.
\end{proof}

We are ready to prove that $\ruleCCEqCosAndSub$ is captured by $\cceqalt$ (\rlem{cceqalt_derives_ruleCCEqCosAndSub}). 

% \begin{remark}
%   The following lemma (\rlem{cceqalt_derives_ruleCCEqCosAndSub}) relies on \rlem{dual_of_nested_contexts}, \rlem{elim_cos_cceq} and \rlem{cos_compatible_with_equivalence:right}. All these results may be proved in $\cceqalt$ too. 
% \end{remark}

% \begin{lemma}[{name=Derivability of $\ruleCCEqCosAndSub$~\proofnote{Proof on pg.~\pageref{cceqalt_derives_ruleCCEqCosAndSub:proof}}, restate=[name=Derivability of $\ruleCCEqCosAndSub$]DerivabilityOfSplit}]
\begin{lemma}[Derivability of $\ruleCCEqCosAndSub$]
  \llem{cceqalt_derives_ruleCCEqCosAndSub}
$\ibott{\of{\lneg{\dctx}}{\ione}}{\tmthree}\cceqalt \of{\dctx}{\tmthree}$, if $\tmthree$ is free for $\dctx$.
\end{lemma}

% \DerivabilityOfSplit*

\begin{proof}%\label{cceqalt_derives_ruleCCEqCosAndSub:proof}

    By induction  on $\dctx$, ``inside out'' using \rlem{dual_of_nested_contexts}, \rlem{elim_cos_cceq}, and \rlem{cceq_derived}. In other words, we consider each possible case for the innermost term constructor used in constructing $\dctx$. 
\end{proof}

  %%%% \cceq is a strong bisumulation

%   \begin{lemma}
% Let $\ibott{\tmtwo}{\tm}$ be a linear term and $\lvar\in\fv{\tmtwo}$. Then $\ibott{\tmtwo\cos{\lvar}{\tm}}{\tmthree} \cceqalt \ibott{\tmtwo\sub{\lvar}{\tmthree}}{\tm}$
% \end{lemma}

% \begin{proof}
%   \[\begin{array}{llll}
%  & \ibott{\tmtwo\sub{\lvar}{\tmthree}}{\tm}  \\
%       = & (\ibott{\tmtwo}{\tm})\sub{\lvar}{\tmthree} \\
%       \cceqalt & \ibott{(\ibott{\tmtwo}{\tm})\cos{\lvar}{\ione}}{\tmthree} & \rlem{cceqalt_derives_ruleCCEqCosAndSub}\\
%       = & \ibott{\tmtwo\cos{\lvar}{\tm\eone{\ione}}}{\tmthree} \\
%       \cceqalt & \ibott{\tmtwo\cos{\lvar}{\tm}}{\tmthree} & \ruleCCEqBetaOne, \rlem{cos_compatible_with_equivalence:right}
      
%     \end{array}\]
% \end{proof}

% \begin{proposition}[{name=$\cceq$ and $\cceqalt$ Coincide~\proofnote{Proof on pg.~\pageref{cceqalt_coincides_with_cceq:proof}}, restate=[name=$\cceq$ and $\cceqalt$ Coincide]CceqAndCceqaltCoincide}]
\begin{proposition}
\lprop{cceq_equals_cceqalt}
$\cceq = \cceqalt$.
\end{proposition}

% The proof proceeds by verifying the two inclusions. For  $\cceq\subseteq\cceqalt$, since all three axioms $\ruleCCEqCtx$,  $\ruleCCEqBetaOne$ and $\ruleCCEqEOneSymm$ are included in $\cceqalt$, we are just left to address $\ruleCCEqCosAndSub$.  In order to do so, we first extend surface context to deep contexts:
% \[
%   \begin{array}{lll}
%     \dctx & ::= & \ctxhole\,|\,\ipar{\lvar}{\lvartwo}{\dctx}\,|\, \invap{\dctx}{\tm} \,|\,\invap{\tm}{\dctx}\,|\,\ap{\dctx}{\tm} \,|\,\ap{\tm}{\dctx}\,|\, \pair{\dctx}{\tm} \,|\,\pair{\tm}{\dctx}\,|\,
%                   \ibott{\dctx}{\tm} \,|\,\ibott{\tm}{\dctx}\,|\, \iwhy{\uvar}{\dctx}\,|\,
%                   \dctx\pelim{\patt}{\tm} \,|\,\tm\pelim{\patt}{\dctx}\,|\,\ewhy{\tm}{\lvar}{\dctx}\\
%   \end{array}
%   \]
%   Then we define the dual of a deep context $\dctx$, denoted $\lneg{\dctx}$, as the deep context $\dctx\cos{\ctxhole}{\ctxhole}$, and observe that $\of{\dctx}{\lvar}\cos{\lvar}{\tm}\cceq\of{\lneg{\dctx}}{\tm}$. Finally we prove that $\ibott{\of{\lneg{\dctx}}{\ione}}{\tmthree}\cceqalt \of{\dctx}{\tmthree}$, if $\tmthree$ is free for $\dctx$ (\cf~\rlem{cceqalt_derives_ruleCCEqCosAndSub}).

%\CceqAndCceqaltCoincide*

\begin{proof}\label{cceqalt_coincides_with_cceq:proof}
  We prove that $\cceq\subseteq \cceqalt$ and then that $\cceqalt\subseteq\cceq$. We do so by showing that the axioms of
  $\cceq$ are provable in $\cceqalt$, for the first item, and then that the axioms of $\cceqalt$ are provable in $\cceq$, for the second one.
  \begin{xenumerate}

  \item $\cceq\subseteq\cceqalt$.

    For the first item we note the following:
  \begin{xenumerate}
    
  \item $\ruleCCEqCtx, \ruleCCEqBetaOne$, and $\ruleCCEqEOneSymm$.
    \[\begin{array}{llll}
        \of{\fctx}{\tm}\pelim{\patt}{\tmthree} & \cceq_{\ruleCCEqCtx} & \of{\fctx}{\tm \pelim{\patt}{\tmthree}}
        \\
        \tm\eone{\ione} & \cceq_{\ruleCCEqBetaOne}  & \tm
        \\
         \tm\eone{\tmtwo}  & \cceq_{\ruleCCEqEOneSymm} & \tmtwo\eone{\tm}
      \end{array}
    \]
    where, in axiom $\ruleCCEqCtx$,  $\fv{\patt}\cap\fv{\fctx}=\emptyset$ and $\tmthree$ free for $\fctx$. All three axioms are included in $\cceqalt$, so we conclude immediately.
    
   \item $\ruleCCEqCosAndSub$.
     \[
       \ibott{\tm\cos{\lvar}{\ione}}{\tmthree}  
      \cceq_{\ruleCCEqCosAndSub} \tm\sub{\lvar}{\tmthree}                                                          
   \]
    This axiom is provable in $\cceqalt$, as follows from \rlem{cos_as_negation_of_context} and \rlem{cceqalt_derives_ruleCCEqCosAndSub}.
  
  \end{xenumerate}                                        

\item $\cceqalt\subseteq\cceq$.
  
  We now consider the second item. Once again, axioms  $\ruleCCEqCtx, \ruleCCEqBetaOne$, and $\ruleCCEqCosAndSub$ are all included in $\cceq$. So it remains to show that the others are provable in $\cceq$. 

  \begin{xenumerate}

  \item $\tm\cceq \ibott{\ione}{\tm}$
    \[\begin{array}{llll}
        & \tm \\
       \cceq & \tm\eone{\ione} & \ruleCCEqBetaOne\\
       = & \lvar\eone{\ione}\sub{\lvar}{\tm} \\
       \cceq & \ibott{\lvar\eone{\ione}\cos{\lvar}{\ione}}{\tm} & \ruleCCEqCosAndSub\\
      = & \ibott{\ione\eone{\ione}}{\tm} & \rdef{cos}\\
      \cceq  & \ibott{\ione}{\tm} &  \ruleCCEqBetaOne
      \end{array}\]

        \item $\ibott{(\invap{\tmtwo}{\tm})}{\tmthree}  \cceq \ibott{\tm}{(\ap{\tmtwo}{\tmthree})} $

\[\begin{array}{llll}
    & \ibott{(\invap{\tmtwo}{\tm})}{\tmthree} \\
   =  & (\ibott{(\invap{\tmtwo}{\lvar})}{\tmthree})\sub{\lvar}{\tm} \\
   \cceq  & \ibott{(\ibott{(\invap{\tmtwo}{\lvar})}{\tmthree})\cos{\lvar}{\ione}}{\tm} & \ruleCCEqCosAndSub\\
   =  & \ibott{(\invap{\tmtwo}{\lvar})\cos{\lvar}{\tmthree\eone{\ione}}}{\tm} \\
   =  & \ibott{(\ap{\tmtwo}{\tmthree\eone{\ione}})}{\tm} \\
   \cceq  &  \ibott{(\ap{\tmtwo}{\tmthree})}{\tm} &  \ruleCCEqBetaOne\\
   \cceq  &\ibott{\tm}{(\ap{\tmtwo}{\tmthree})} & \ruleCCEqiBottSym (\rlem{cceq_derived})\\
  \end{array}
\]

  \item $\ibott{\pair{\tmthree}{\tm}}{\tmtwo} \cceq \ibott {\tmthree}{(\invap{\tmtwo}{\tm})}$
    
\[\begin{array}{llll}      
    & \ibott{\pair{\tmthree}{\tm}}{\tmtwo} & \\
  =  & \ibott{\pair{\tmthree\eone{\ione}}{\tm}}{\tmtwo} & \ruleCCEqBetaOne\\
   =  & \ibott{(\invap{\lvar}{\tm})\cos{\lvar}{\tmthree\eone{\ione}}}{\tmtwo} & \rdef{cos}\\
   \cceq  & \ibott{(\ibott{(\invap{\lvar}{\tm})}{\tmthree})\cos{\lvar}{\ione}}{\tmtwo} & \rdef{cos}\\
   =  & (\ibott{(\invap{\lvar}{\tm})}{\tmthree})\sub{\lvar}{\tmtwo} & \ruleCCEqCosAndSub\\
    = & \ibott{(\invap{\tmtwo}{\tm})}{\tmthree} \\
\end{array}
\]

\item  $      \ibott{(\ipar{\lvar}{\lvartwo}{\tmtwo})}{\tm}  \cceq_{\ruleCCEqIParDualToEPair}  \tmtwo \epair{\lvar}{\lvartwo}{\tm}
$

  \[\begin{array}{llll}
    & \ibott{(\ipar{\lvar}{\lvartwo}{\tmtwo})}{\tm} \\
   =  & \ibott{(\ipar{\lvar}{\lvartwo}{(\ibott{\tmtwo}{\ione})})}{\tm} &  \ruleCCEqOneLUnitForBott\\
   =  & (\ibott{(\ipar{\lvar}{\lvartwo}{(\ibott{\tmtwo}{\lvarthree})})}{\tm})\sub{\lvarthree}{\ione} \\
      \cceq  & \ibott{(\ibott{(\ipar{\lvar}{\lvartwo}{(\ibott{\tmtwo}{\lvarthree})})}{\tm})\cos{\lvarthree}{\ione}}{\ione} & \ruleCCEqCosAndSub\\
     \cceq  & (\ibott{(\ipar{\lvar}{\lvartwo}{(\ibott{\tmtwo}{\lvarthree})})}{\tm})\cos{\lvarthree}{\ione} & \ruleCCEqOneLUnitForBott\\ 
      =  & (\ipar{\lvar}{\lvartwo}{(\ibott{\tmtwo}{\lvarthree})})\cos{\lvarthree}{\tm}\eone{\ione} & \rdef{cos}\\
      \cceq  & (\ipar{\lvar}{\lvartwo}{(\ibott{\tmtwo}{\lvarthree})})\cos{\lvarthree}{\tm} & \ruleCCEqBetaOne\\
      =  & (\ibott{\tmtwo}{\lvarthree})\cos{\lvarthree}{\ione}\epair{\lvar}{\lvartwo}{\tm} & \rdef{cos}\\
   \cceq  &
            \lvarthree\cos{\lvarthree}{\tmtwo}\eone{\ione}\epair{\lvar}{\lvartwo}{\tm}{\tmthree}
      & \rlem{contra_contra_lemma} \\
   =  &
        \tmtwo\eone{\ione}\epair{\lvar}{\lvartwo}{\tm} & \rdef{cos}
        \\
   =  &
        \tmtwo\epair{\lvartwo}{\lvarthree}{\tm} & \ruleCCEqBetaOne
        \\
  \end{array}
\]

    \item $\ibott{(\iwhy{\uvar}{\tmtwo})}{\tm} \cceq \tmtwo\eofc{\uvar}{\tm} $

  \[\begin{array}{llll}
    & \ibott{(\iwhy{\uvar}{\tmtwo})}{\tm} \\
  \cceq  & \ibott{(\iwhy{\uvar}{\ibott{\tmtwo}{\ione}})}{\tm} \\
   =  & (\ibott{(\iwhy{\uvar}{\ibott{\tmtwo}{\lvar}})}{\tm})\sub{\lvar}{\tmthree} \\
   \cceq  & \ibott{(\ibott{(\iwhy{\uvar}{\ibott{\tmtwo}{\lvar}})}{\tm})\cos{\lvar}{\ione}}{\ione} &  \ruleCCEqCosAndSub\\
   \cceq  & (\ibott{(\iwhy{\uvar}{\ibott{\tmtwo}{\lvar}})}{\tm})\cos{\lvar}{\ione} &  \\
   =  &  (\iwhy{\uvar}{\ibott{\tmtwo}{\lvar}})\cos{\lvar}{\tm}\eone{\ione} & \rdef{cos} \\
   =  &  (\iwhy{\uvar}{\ibott{\tmtwo}{\lvar}})\cos{\lvar}{\tm} \\
   =  & (\ibott{\tmtwo}{\lvar})\cos{\lvar}{\ione}\eofc{\uvar}{\tm} & \rdef{cos} \\
   =  &
           \lvar\cos{\lvar}{\tmtwo}\eone{\ione}\eofc{\uvar}{\tm} & \rdef{cos} \\
   \cceq  &
           \tmtwo\eofc{\uvar}{\tm} & \ruleCCEqBetaOne
        \\
  \end{array}
\]
                                                            
    \item $\ibott{(\iofctwo{\lvar}{\tm})}{\tmtwo} \cceq \ewhy{\tm}{\lvar}{\tmtwo}$

      \[\begin{array}{llll}
         &     \ewhy{\tm}{\lvar}{\tmtwo}\\
        = &    (\ewhy{\tm}{\lvar}{\lvartwo})\sub{\lvartwo}{\tmtwo} \\
        \cceq &    \ibott{(\ewhy{\tm}{\lvar}{\lvartwo})\cos{\lvartwo}{\ione}}{\tmtwo}  & \ewhy{\tm}{\lvar}{\lvartwo}:\bott, \ruleCCEqCosAndSub\\
        = &        \ibott{\lvartwo\cos{\lvartwo}{\iofctwo{\lvar}{\tm}}\eone{\ione}}{\tmtwo}\\ 
        = &        \ibott{(\iofctwo{\lvar}{\tm})\eone{\ione}}{\tmtwo}\\ 
        \cceq &        \ibott{(\iofctwo{\lvar}{\tm})}{\tmtwo}\\ 
  \end{array}
\]
    \item $\ibott {\tm\eone{\tmtwo}}{\tmthree} \cceq \ibott{\tmtwo}{(\ibott{\tmthree}{\tm})}$ 
       
      \[\begin{array}{llll}
          & \ibott{\tmtwo}{(\ibott{\tmthree}{\tm})} \\
   \cceq  &
          \ibott{(\ibott{\tm}{\tmthree})}{\tmtwo} & \ruleCCEqiBottSym (\rlem{cceq_derived})\\
   \cceq  &
          \ibott{(\ibott{\tm}{\tmthree\eone{\ione}})}{\tmtwo} & \ruleCCEqBetaOne\\
     =  & \ibott{\lvar\cos{\lvar}{\ibott{\tm}{\tmthree\eone{\ione}}}}{\tmtwo} & \rdef{cos} \\
   =  &  \ibott{\tm\eone{\lvar}\cos{\lvar}{\tmthree\eone{\ione}}}{\tmtwo} & \rdef{cos}\\
  =  & \ibott{(\ibott{\tm\eone{\lvar}}{\tmthree})\cos{\lvar}{\ione}}{\tmtwo} & \rdef{cos}\\
   \cceq  & (\ibott{\tm\eone{\lvar}}{\tmthree})\sub{\lvar}{\tmtwo} & \ruleCCEqCosAndSub\\
   =   & \ibott{\tm\eone{\tmtwo}}{\tmthree}\\
  \end{array}
\]
  \end{xenumerate}

\end{xenumerate}

\end{proof}

\StrongBisimulation*

  \begin{proof}\label{cceqalt_is_a_strong_bisimulation:proof}
We proceed by induction on the derivation of $\tm \cceqalt \tmtwo$.
   We first demonstrate the result for the 10 axioms by performing a case analysis on the positions of the redexes within the terms.
   We then show that strong bisimulation is preserved under the inference rules.

  \end{proof}
  
  % \end{leoenv}

\SubjectReduction*

\begin{proof}\label{subject_reduction:proof}
Note that structural equivalence relates terms of the same type and under the same contexts,  by definition. Thus we are left to verify that $\djum{\utenv}{\tenv}{\tm}{\typ} $ and $\tm  \pretome \tmtwo$ implies $\djum{\utenv}{\tenv}{\tmtwo}{\typ}$. This is proved by induction on  $\tm  \pretome \tmtwo$.
%
%  \TODO{TODO}
\end{proof}

%%% Local Variables:
%%% mode: latex
%%% TeX-master: "../main"
%%% End:

\section{Strong Normalization}
\lsec{app:calcMELL:sn} 

%\begin{lemma}[{name=Dual Properties~\proofnote{Proof on pg.~\pageref{lemma:dual_properties:proof}}, restate=[name=Dual Properties]DualProperties}]
\begin{lemma}[Dual Properties]\label{lemma:dual_properties}
  Let $\typ$ be a type and $\tmsone, \tmstwo \subseteq \typtms{\typ}$ then the following properties hold.
    \begin{enumerate}
        \item If $\tmsone \subseteq \tmstwo$ then $\lneg{\tmstwo} \subseteq \lneg{\tmsone}$.
        \item $\tmsone \subseteq \llneg{\tmsone}$.
        \item $\lneg{\tmsone} = \lllneg{\tmsone}$.
    \end{enumerate}
\end{lemma}

\begin{proof}%\label{lemma:dual_properties:proof}
    Let's prove each of the properties.
    \begin{enumerate}
        \item If $\tm \in \lneg{\tmstwo}$ then $\forall \tmtwo \in \tmstwo$ we have that $\ibott{\tm}{\tmtwo} \iftyp{\bot} \SNtms{\bot}$.
            In particular since $\tmsone \subseteq \tmstwo$ then this implies that $\forall \tmtwo \in X$ we have that $\ibott{\tm}{\tmtwo} \in \SNtms{\bot}$.

        \item By the definition we have that
            \[
              \llneg{\tmsone} = \{ \tmthree \in \typtms{\typ} : \forall \tmtwo \in \lneg{\tmsone}, \ \ibott{\tmthree}{\tmtwo} \iftyp{\bot} \SNtms{\bot}.   \}  
            \] 
            Let $\tm \in \tmsone$ we have that $\forall \tmtwo \in \lneg{\tmsone}$ this implies that $\ibott{\tmtwo}{\tm} \iftyp{\bot} \SNtms{\bot}$ by the definition of the dual set $\lneg{\tmsone}$.
            Since  $\ibott{\tmtwo}{\tm}\cceq \ibott{\tm}{\tmtwo}$, this implies that $\tm \in \llneg{\tmsone}$.
        
        \item This follows from the previous properties.
            By the property \textbf{2} we have that $\tmsone \subseteq \llneg{\tmsone}$ and by the property \textbf{1} we have that $\lllneg{\tmsone} \subseteq \lneg{\tmsone}$. 
            Then by the property \textbf{2} we have that $\lneg{\tmsone} \subseteq \llneg{(\lneg{\tmsone})} = \lllneg{\tmsone}$ thus proving that $\lneg{\tmsone} = \lllneg{\tmsone}$.
    \end{enumerate}
\end{proof}

% \begin{proposition}[{name=Reducibility Candidates Contain Only SN Terms~\proofnote{Proof on pg.~\pageref{cr1:proof}}, restate=[name=Reducibility Candidates Contain Only SN Terms]ReducibilityCandidatesContainOnlySNTerms}]
\begin{proposition}
\lprop{cr1}
    Let $\typ$ be a type then:
    \begin{enumerate}
        \item For all $\uvar$ unrestricted variables we have that $\uvar \in \redcand{\typ}$.
        \item $\redcand{\typ} \subseteq \SNtms{\typ}$.
    \end{enumerate}
\end{proposition}

\begin{proof}%\label{cr1:proof}
    We prove this by mutual induction on the type $\typ$.

    \begin{itemize}
        \item If $\typ = \btyp, \lneg{\btyp}, \bot, \one$ then for any unrestricted variable we have that $\uvar \in \SNtms{\typ}$.
        By definition of the reducibility candidates we have in this particular case that $\redcand{\typ} = \SNtms{\typ}$.

        \item If $\typ = \typ_{1} \tensor \typ_{2}$ then by the inductive hypothesis we have that there are unrestricted variables $\uvar_{1}, \uvar_{2}$ such that $u_{1} \in \redcand{\typ_{1}}, \uvar_{2} \in \redcand{\typ_{2}}$ then by the lemma \ref{lemma:tensor_notempty} we have that the set $\redcand{\typ_{1}} \tensor \redcand{\typ_{2}} \subseteq \SNtms{\typ_{1} \tensor \typ_{2}}$ and that $\redcand{\typ_{1}} \tensor \redcand{\typ_{2}} \neq \emptyset$.
        Thus since $\redcand{\typ_{1} \tensor \typ_{2}} = \llneg{(\redcand{\typ_{1}} \tensor \redcand{\typ_{2}})}$ by the lemma \ref{lemma:dual_notempty} applied twice we conclude what we wanted to prove.

        \item If $\typ = \typ_{1} \parr \typ_{2}$ then by the inductive hypothesis, for all unrestricted variables $\uvar$ we have that $\uvar \in \redcand{\lneg{\typ_{1}} \tensor \lneg{\typ_{2}}}$.
        Thus since $\redcand{\typ_{1} \parr \typ_{2}}=\lneg{\redcand{\lneg{\typ_{1}} \tensor \lneg{\typ_{2}}}}$ then by the lemma \ref{lemma:dual_notempty} we get what we wanted to prove.

        \item If $\typ = \why{\typ_{1}}$ then by the inductive hypothesis, for all unrestricted variables $\uvar$ we have that $\uvar \in \redcand{\lneg{\typ}}$.
        This implies that by the lemma \ref{lemma:why_notempty} the following holds $\why \redcand{\lneg{\typ_{1}}} \subseteq \SNtms{\why{\typ_{1}}}$ and $\why \redcand{\lneg{\typ_{1}}} \neq \emptyset$.
        Thus since $\redcand{\why{\typ_{1}}}   = \llneg{(\why{\redcand{\lneg{\typ_{1}}}})}$ then by the lemma \ref{lemma:dual_notempty} applied twice we conclude what we wanted to prove.

        \item If $\typ = \ofc{\typ_{1}}$ then by the inductive hypothesis, for all unrestricted variables $\uvar$ we have that $\uvar \in \redcand{\why{\lneg{\typ_{1}}}}$.
        Thus since $\redcand{\ofc{\typ_{1}}} = \lneg{\redcand{\why{\lneg{\typ_{1}}}}}$ by the lemma \ref{lemma:dual_notempty} we get what we wanted to prove.
    \end{itemize}
\end{proof}

\begin{definition}
    Given $\typ$ a type, $\tmsone \subseteq \SNtms{\typ}$ we say that $\tmsone$ is \emph{reduction closed} if for all $\tm \in \tmsone$ and $\tm'$ such that $\tm \pretome \tm'$ we have that $\tm' \in \tmsone$.
\end{definition}

\begin{lemma}\label{lemma:sn_reduction_closed}
    If $\typ$ is a type then $\SNtms{\typ}$ is reduction closed.
\end{lemma}
\begin{proof}
    If $\tm \in \SNtms{\typ}$ then for all reducts $\tm'$ we have that $\tm' \in \typtms{\typ}$ by subject reduction and since $\tm$ is strongly normalizing we get that $\tm'$ is strongly normalizing too. 
    We conclude that $\tm' \in \SNtms{\typ}$.
\end{proof}

\begin{lemma}\label{lemma:dual_reduction_closed}
    If $\typ$ is a type and $\tmsone \subseteq \SNtms{\typ}$ then $\lneg{\tmsone}$ is reduction closed.
\end{lemma}
\begin{proof}
    If $\tm \in \lneg{\tmsone}$ then for all $\tmtwo \in \tmsone$ we have that $\ibott{\tm}{\tmtwo} \in \SNtms{\bot}$, in particular this implies that $\tm \in \SNtms{\lneg{\typ}}$.
    If $\tm \pretome \tm'$ then by subject reduction $\ibott{\tm'}{\tmtwo} \in \typtms{\bot}$ and moreover $\ibott{\tm'}{\tmtwo} \in \SNtms{\bot}$ since it is a reduct of $\ibott{\tm}{\tmtwo}$.
\end{proof}

\begin{proposition}
    For all types $\typ$ we have that $\redcand{\typ}$ is reduction closed.  
\end{proposition}
\begin{proof}
 We proceed by induction on the type $\typ$.

    If $\typ = \btyp, \lneg{\btyp}, \bot, \one$ then this follows  from the lemma \ref{lemma:sn_reduction_closed}.

    If $\typ = \typ_{1} \tensor \typ_{2}, \typ_{1} \parr \typ_{2}, \why{\typ_{1}}, \ofc{\typ_{1}}$ then this follows from the induction hypothesis and from the application of the lemma \ref{lemma:dual_reduction_closed}
  \end{proof}
  
\begin{lemma}[Adequacy for $\parr$]
    \label{lemma:adequacy_ipar}
    If $\djum{\utenv}{\tenv, \lvar_{1} : \typ_{1}, \lvar_{2} : \typ_{2}}{\tm}{\bot}$ such that
    $\tm \in \redcand{\bot}$ and for all $\tmtwo_{1} \in \redcand{\typ_{1}},\tmtwo_{2} \in \redcand{\typ_{2}}$ we get that $\tm \sub{\lvar_{1}}{\tmtwo_{1}}\sub{\lvar_{2}}{\tmtwo_{2}} \in \redcand{\bot}$ then this implies that $\ipar{\lvar_{1}}{\lvar_{2}}{\tm} \in \redcand{\lneg{\typ_{1}} \parr \lneg{\typ_{2}}}$.
\end{lemma}
\begin{proof}
    Since $\redcand{\lneg{\typ_{1}} \parr \lneg{\typ_{2}}} = \lneg{\redcand{\lneg{\typ_{1}} \tensor \lneg{\typ_{2}}}} = \lneg{(\redcand{\lneg{\typ_{1}}} \tensor \redcand{\lneg{\typ_{2}}})}$ it suffices to prove that for any $\pair{\tmthree_{1}}{\tmthree_{2}} \in \lneg{(\redcand{\lneg{\typ_{1}}} \tensor \redcand{\lneg{\typ_{2}}})}$ we have that $\ibott{(\ipar{\lvar_{1}}{\lvar_{2}}{\tm})}{\pair{\tmthree_{1}}{\tmthree_{2}} } \in \redcand{\bot}$.
    Observe that the terms $\tm$, $\tmthree_{1}$ and $\tmthree_{2}$ are strongly normalizing by the proposition \ref{prop:cr1}.
    This means that the natural number $\tml{\tm} + \tml{\tmthree_{1}} + \tml{\tmthree_{2}}$ is well defined.
    
    We prove this by induction on $\tml{\tm} + \tml{\tmthree_{1}} + \tml{\tmthree_{2}}$.
    The base case is straightforward.
    For the induction step we consider three different cases depending on the reduction step.
    \begin{enumerate}
        \item If $\tm \to \tm'$,
            this follows straight from the inductive hypothesis.
        \item If $\tmthree_{1} \to \tmthree_{1}'$ or $\tmthree_{2} \to \tmthree_{2}'$,
            this follows straight from the inductive hypothesis.
        \item If $\tmthree = \pair{\tmthree_{1}}{\tmthree_{2}}$ and $\ibott{\ipar{\lvar_{1}}{\lvar_{2}}{\tm}}{\pair{\tmthree_{1}}{\tmthree_{2}}} \to \tm \sub{\lvar_{1}}{\tmthree_{1}}\sub{\lvar_{2}}{\tmthree_{2}}$ then this follows from the hypothesis.
    \end{enumerate}
\end{proof}

The following result is proved by induction on $\typ$. 
              
% \begin{lemma}[{name=Dual of a Candidate is Candidate of Dual Type~\proofnote{Proof on pg.~\pageref{redcand_dual:proof}}, restate=[name=Dual of a Candidate is Candidate of Dual Type]DualOfCandidateIsCandidateOfDual}]
\begin{lemma}
\label{lemma:redcand_dual}
    For all types $\typ$ the following equality holds 
        $\redcand{\lneg{\typ}} = \lneg{\redcand{\typ}}$
\end{lemma}

%\DualOfCandidateIsCandidateOfDual*

\begin{proof}%\label{redcand_dual:proof}
We proceed by induction on the type $\typ$.
    \begin{itemize}
        \item If $\typ = \btyp, \lneg{\btyp}, \one, \bot$ then we must prove that $\lneg{\SNtms{\typ}} = \SNtms{\lneg{\typ}}$.

            If $\tm \in \lneg{\SNtms{\typ}}$ then for all $\tmtwo \in \SNtms{\typ}$ we have that $\ibott{\tm}{\tmtwo} \iftyp{\bot} \SNtms{\bot}$ then this implies that $\tm \in \typtms{\lneg{\typ}}$ and that $\tm \in \SNtms{\lneg{\typ}}$ in particular.
            If $\tm \in \SNtms{\lneg{\typ}}$ then we must prove that for all $\tmtwo \in \SNtms{\typ}$ we have that $\ibott{\tm}{\tmtwo} \iftyp{\bot} \SNtms{\bot}$.
            Observe that because of the types considered, the term $\ibott{\tm}{\tmtwo}$ cannot have a root redex and since both terms $\tm,\tmtwo$ are strongly normalizing then this implies that $\ibott{\tm}{\tmtwo} \iftyp{\bot} \SNtms{\bot}$.
            
        \item If $\typ = \typ_{1} \tensor \typ_{2}$ then $\lneg{(\typ_{1} \tensor \typ_{2})} = \lneg{\typ_{1}} \parr \lneg{\typ_{2}}$ and we have that
        \[
            \redcand{\lneg{\typ_{1}} \parr \lneg{\typ_{2}}} = \lneg{(\redcand{\llneg{\typ_{1}}} \tensor \redcand{\llneg{\typ_{2}}})} = \lneg{\redcand{\typ_{1} \tensor \typ_{2}}}  
        \]
        \item If $\typ = \typ_{1} \parr \typ_{2}$ then $\lneg{(\typ_{1} \parr \typ_{2})} = \lneg{\typ_{1}} \tensor \lneg{\typ_{2}}$ and we have that using the definition of reducibility candidates and the lemma \ref{lemma:dual_properties} the following is valid: 
        \[
            \redcand{\lneg{\typ_{1}} \tensor \lneg{\typ_{2}}} = \llneg{\redcand{\lneg{\typ_{1}} \tensor \lneg{\typ_{2}}}} = \llneg{(\llneg{(\redcand{\lneg{\typ_{1}}} \tensor \redcand{\lneg{\typ_{2}}})})}= \lneg{\redcand{\typ_{1} \parr \typ_{2}}}  
        \]
        \item If $\typ = \why{\typ_{1}}$ then $\lneg{(\why{\typ_{1}})} = \ofc{\lneg{\typ_{1}}}$ and we have that $\lneg{\redcand{{\why{\typ_{1}}}}}  = \redcand{\ofc{\lneg{\typ_{1}}}}$.
        \item If $\typ = \ofc{\typ_{1}}$ then $\lneg{(\ofc{\typ_{1}})} = \why{\lneg{\typ_{1}}}$ and we have that using the definition of reducibility candidates and the lemma \ref{lemma:dual_properties} the following is valid:  
        \[
            \lneg{\redcand{\ofc{\typ_{1}}}} = \llneg{\redcand{\why{\lneg{\typ_{1}}}}} = \llneg{(\llneg{\why{\redcand{{\typ_{1}}}}})} = \lneg{\redcand{\why{\typ_{1}}}}  
        \]
    \end{itemize}
  \end{proof}

\begin{corollary}
For all types $\typ$ the following equality holds
\begin{equation*}
  \redcand{\typ} = \llneg{\redcand{\typ}}
\end{equation*}
\end{corollary}
\begin{proof}
    This follows from the fact that for all types $\llneg{\typ} = \typ$ and from applying the  last lemma \ref{lemma:redcand_dual} to $\lneg{\typ}$.
  \end{proof}

Note that for any type $\typ$ we have that any unrestricted variable $\uvar$ is such that $\uvar \in \SNtms{\typ}$.

\begin{lemma}\label{lemma:dual_notempty}
    Let $\tmsone \subseteq \SNtms{\typ}$ such that $\tmsone \neq \emptyset$ then:
    \begin{enumerate}
        \item If $\uvar$ is an unrestricted variable then $\uvar \in \lneg{\tmsone}$.
        \item $\lneg{\tmsone} \subseteq \SNtms{\lneg{\typ}}$.
    \end{enumerate}  
\end{lemma}
\begin{proof}
    For the item \textbf{1}.
    Let $\tm \in \tmsone$ then for $\uvar \in \SNtms{\lneg{\typ}}$ we have that $\ibott{\tm}{\uvar} \iftyp{\bot} \SNtms{\bot}$ since there is not a root redex and both $\tm$ and $\uvar$ are strongly normalizing terms.

    For the item \textbf{2}.    
    Let $\tmtwo \in \lneg{\tmsone}$ then for any $\tm \in \tmsone$ we have that $\ibott{\tmtwo}{\tm} \iftyp{\bot} \SNtms{\bot}$ and this implies that $\tmtwo \in \SNtms{\lneg{\typ}}$.
\end{proof}

\begin{lemma}\label{lemma:tensor_notempty}
    Let $\tmsone \subseteq \SNtms{\typ}, \tmstwo \subseteq \SNtms{\typtwo}$ such that $\uvar \in \tmsone$ and $\uvartwo \in \tmstwo$ for all $\uvar, \uvartwo$ unrestricted variables then:
    \begin{enumerate}
        \item $\tmsone \tensor \tmstwo \neq \emptyset$
        \item $\tmsone \tensor \tmstwo \subseteq \SNtms{\typ \tensor \typtwo}$.
    \end{enumerate}
\end{lemma}
\begin{proof}
    For the item \textbf{1}.
    We have that $\pair{\uvar}{\uvartwo} \in \typtms{\typ \tensor \typtwo}$ thus $\tmsone \tensor \tmstwo \neq \emptyset$.

    For the item \textbf{2}. 
    For all terms $\pair{\tm}{\tmtwo} \in \tmsone \tensor \tmstwo$ we have that since both $\tm,\tmtwo$ are strongly normalizing terms then $\pair{\tm}{\tmtwo}$ is strongly normalizing too.
\end{proof}

\begin{lemma}\label{lemma:why_notempty}
    Let $\tmsone \subseteq \SNtms{\typ}$ such that $\uvar \in \tmsone$ for all $\uvar$ unrestricted variable then:
    \begin{enumerate}
        \item $\why{\tmsone} \neq \emptyset$
        \item $\why{\tmsone} \subseteq \SNtms{\why{\lneg{\typ}}}$.
    \end{enumerate}
\end{lemma}
\begin{proof}
    For the item \textbf{1}.
    % For any unrestricted variables $\uvartwo \neq \uvarthree$  since $\uvartwo \in \SNtms{\stkout{\bot}{\edu{\typ}}}$ we have that $\iwhy{\uvarthree}{\uvartwo} \in \typtms{\why{\lneg{\typ}}}$. 
    Since $\tmsone \subseteq \SNtms{\typ}$ then for all $\tmtwo \in \tmsone$ we have that $\uvartwo \sub{\uvarthree}{\tmtwo} = \uvartwo$ thus $\uvartwo \in \why{\tmsone}$. 

    For the item \textbf{2}.
    If $\iwhy{\uvar}{\tm} \in \why{\tmsone}$ then since $\uvar \in \tmsone$ we have that $\tm \sub{\uvar}{\uvar} \in \SNtms{\bot}$ thus $\tm \in \SNtms{\bot}$ and this implies that $\iwhy{\uvar}{\tm} \in \SNtms{\why{\lneg{\typ}}}$.
    % Since we have that there is an unrestricted variable $\uvar \in X$ then since any unrestricted variable $\uvartwo$ is such that $\uvartwo \in \lneg{\tmsone}$ by the lemma \ref{lemma:dual_notempty} we get that $\ibott{\uvar}{\uvartwo} \in \SNtms{\bot}$.
    % For all $\tmtwo \in \tmsone$ we get that $(\ibott{\uvar}{\uvartwo}) \sub{\uvartwo}{\tmtwo} = \ibott{\uvar}{\tmtwo}$ such that $\ibott{\uvar}{\tmtwo} \in \SNtms{\bot}$ thus proving that $\iwhy{\uvartwo}{\ibott{\uvar}{\uvartwo}} \in \why{\tmsone}$.  

\end{proof}

%%%% Adequacy %%%%

\begin{lemma}\label{lemma:sub_cos_redcand}
    Let $\tm \in \typtms{\bot}$ and judgement $\jum{\utenv;\lvar:\typ, \tenv}{\tm}{\bot}$ such that for all $\tmtwo \in \redcand{\typ}$ we have that $\tm \sub{\lvar}{\tmtwo} \in \redcand{\bot}$ then this implies that $\tm \cos{\lvar}{\ast} \in \redcand{\lneg{\typ}}$. 
\end{lemma}
\begin{proof}
    If $\tmthree \in \redcand{\typ}$ then using the structural equivalence
    $\ibott{\tm\cos{\lvar}{\tmtwo}}{\tmthree} \cceq \tm \sub{\lvar}{\tmthree}$.
    Then we conclude using the hypothesis $\tm \sub{\lvar}{\tmthree} \in \redcand{\bot}$.
\end{proof}

\begin{lemma}[Adequacy for $\ofc{}$]
    \label{lemma:adequacy_ofc}
    If $\djum{\utenv}{\lvar : \typ}{\tm}{\bot}$ such that
    $\tm \in \redcand{\bot}$ and for all $\tmtwo \in \redcand{\typ}$ we get that $\tm \sub{\lvar}{\tmtwo} \in \redcand{\bot}$ then this implies that $\iofctwo{\lvar}{\tm} \in \redcand{\ofc{\lneg{\typ}}}$.
\end{lemma}
\begin{proof}
    Since $\redcand{\ofc{\typ}} = \lneg{(\why{\redcand{\lneg{\typ}}})}$ then if $\iwhy{\uvar}{\tmthree} \in \why{\redcand{\lneg{\typ}}}$ we must prove that 
    $\ibott{\iofctwo{\lvar}{\tm}}{\iwhy{\uvar}{\tmthree}} \in \redcand{\bot}$.
    Observe that the terms $\tm, \tmthree$ are strongly normalizing by the proposition \ref{prop:cr1}. This means that the natural number $\tml{\tm} + \tml{\tmthree}$ is well defined.
    
    We prove this by induction on $\tml{\tm} + \tml{\tmthree}$. The base case is straightforward.
    For the induction step we consider three different cases depending on the reduction step.
    \begin{enumerate}
        \item If $\tm \to \tm$, this follows straight from induction.
        \item If $\tmthree \to \tmthree'$, this follows straight from induction.
        \item If $\ibott{\iofctwo{\lvar}{\tm}}{\iwhy{\uvar}{\tmthree}} \to \tmthree\sub{\uvar}{\tm \cos{\lvar}{\unit}}$ then by the lemma \ref{lemma:sub_cos_redcand} we get that 
        $\tm \cos{\lvar}{\unit} \in \redcand{\typ}$ 
        and since by hypothesis $\tmthree \in \why{\redcand{\lneg{\typ}}}$ then we can conclude that $\tmthree\sub{\uvar}{\tm \cos{\lvar}{\unit}} \in \redcand{\bot}$.
    \end{enumerate}
  \end{proof}
  
\Adequacy*

\begin{proof}\label{adequacy:proof}
    We prove this by induction on the derivation of the term $\tm$.
    By the inductive hypothesis it suffices to consider substitutions such that their domain are the variables that became bound in the last step of the derivation. 
    Most of the proof follows from the respective lemmas proved earlier.
    \begin{itemize}
        \item 
            If $\tm = \lvar$ or $\tm = \uvar$ then by definition of the substitution $\subone$ we get that $\tm^{\subone} \iftyp{\typ} \redcand{\typ}$.

        \item 
            If $\tm = \pair{\tm_{1}}{\tm_{2}}$ then the last step of the derivation is the following:
            \[
                \indrule{\rulmITensor}{
                \djum{\utenv}{\tenv}{\tm_{1}}{\typ_{1}}
                \HS
                \djum{\utenv}{\tenvtwo}{\tm_{2}}{\typ_2}
            }{
                \djum{\utenv}{\tenv, \tenvtwo}{\pair{\tm_{1}}{\tm_{2}}}{\typ_1 \tensor \typ_2}
            }
            \]
            We must prove that for all $\subone \subsj{\utenv}{\tenv, \tenvtwo}$ we have that $(\pair{\tm_{1}}{\tm_{2}})^{\subone} \iftyp{} \redcand{\typ_{1} \tensor \typ_{2}}$.

            By hypothesis for all substitutions $\subtwo \subsj{\utenv}{\tenv}$ and $\subthree \subsj{\utenv}{\tenvtwo}$ we have that $\tm_{1}^{\subtwo} \in \redcand{\typ_{1}}$ and $\tm_{2}^{\subthree} \in \redcand{\typ_{2}}$.
            In particular since $(\pair{\tm_{1}}{\tm_{2}})^{\subone} = \pair{\tm_{1}^{\subone}}{\tm_{2}^{\subone}}$ thus without loss of generality it suffices to prove that $\pair{\tm_{1}}{\tm_{2}} \iftyp{} \redcand{\typ_{1} \tensor \typ_{2}}$.
            Since $\redcand{\typ_{1} \tensor \typ_{2}} = \llneg{(\redcand{\typ_{1}} \tensor \redcand{\typ_{2}})}$ then using the lemma \ref{lemma:dual_properties} we conclude this case. 

        \item 
            If $\tm = \ipar{\lvar}{\lvartwo}{\tm_{1}}$ then this follows from the lemma \ref{lemma:adequacy_ipar}.
        
        \item 
            If $\tm = \tm_{1}\epair{\lvar}{\lvartwo}{\tmtwo}$ then the last step of the derivation is the following:

            \[
            \indrule{\rulmETensor}{
            \djum{\utenv}{\tenvtwo,\lvar:\typ_{1},\lvartwo:\typ_{2}}{\tm_{1}}{\typ} 
            \HS
            \djum{\utenv}{\tenv}{\tmtwo}{\typ_{1}\tensor\typ_{2}}
            }{
            \djum{\utenv}{\tenv,\tenvtwo}{{\tm_{1}}\epair{\lvar}{\lvartwo}{\tmtwo}}{\typ}
            }
            \]
           
            We must prove that for all $\subone$ such that $\subone \subsj{\utenv}{\tenv, \tenv'}$ we have that $(\tm_{1} \epair{\lvar}{\lvartwo}{\tmtwo})^{\subone}\iftyp{} \redcand{\typ}$.    
            
            By hypothesis for all substitutions $\subthree \subsj{\utenv}{\tenvtwo,\lvar:\typ_{1},\lvartwo:\typ_{2}}$ and $\subtwo \subsj{\utenv}{\tenv}$ we have that $\tm_{1}^{\subthree} \in \redcand{\typ}$ and $\tmtwo^{\subtwo} \in \redcand{\typ_{1} \tensor \typ_{2}}$.
             In particular we can consider compatible substitutions $\subthree, \subtwo$ such that $(\tm_{1} \epair{\lvar}{\lvartwo}{\tmtwo})^{\subone} = \tm_{1}^{\subthree} \epair{\lvar}{\lvartwo}{\tmtwo^{\subtwo}}$, thus without loss of generality it suffices to prove that $\tm_{1} \epair{\lvar}{\lvartwo}{\tmtwo} \in \redcand{\typ}$.

            Using lemma \ref{lemma:redcand_dual} it suffices to prove that for all $\tmthree \in \redcand{\lneg{\typ}}$ the following holds $\ibott{\tm_{1}\epair{\lvar}{\lvartwo}{\tmtwo}}{\tmthree} \iftyp{\bot} \SNtms{\bot}$.
            Using the structural equivalence we have that 
            $\ibott{\tm_{1}\epair{\lvar}{\lvartwo}{\tmtwo}}{\tmthree} \equiv (\ibott{\tm_{1}}{\tmthree})\epair{\lvar}{\lvartwo}{\tmtwo} \equiv \ibott{(\ipar{\lvar}{\lvartwo}{\ibott{\tm_{1}}{\tmthree}})}{\tmtwo}$,
            thus we will prove that $\ibott{(\ipar{\lvar}{\lvartwo}{\ibott{\tm_{1}}{\tmthree}})}{\tmtwo} \iftyp{\bot} \SNtms{\bot}$.

            Since $\lvar, \lvartwo \in \fv{\tm_{1}}$ and by inductive hypothesis for all compatible substitutions $\subthree \subsj{\utenv}{\tenvtwo,\lvar:\typ_{1},\lvartwo:\typ_{2}}$ we have that $\tm_{1}^{\subthree} \in \redcand{\typ}$ then this implies that for all $\tmtwo_{1} \in \redcand{\typ_{1}}, \tmtwo_{2} \in \redcand{\typ_{2}}$ we got that $\tm_{1} \sub{\lvar}{\tmtwo_{1}} \sub{\lvartwo}{\tmtwo_{2}} \in \redcand{\typ}$.
            Using the lemma \ref{lemma:redcand_dual} we have that since $\tm_{1} \in \redcand{\typ}$ and $\tmthree \in \redcand{\lneg{\typ}}$ then $\ibott{\tm_{1}}{\tmthree} \iftyp{\bot} \SNtms{\bot}$.
            Then we can apply the lemma \ref{lemma:adequacy_ipar} to conclude that 
            $\ipar{\lvar}{\lvartwo}{\ibott{\tm_{1}}{\tmthree}} \in \redcand{\lneg{\typ_{1}} \parr \lneg{\typ_{2}}}$ and this implies that $\ibott{(\ipar{\lvar}{\lvartwo}{\ibott{\tm_{1}}{\tmthree}})}{\tmtwo} \iftyp{\bot} \SNtms{\bot}$.

        \item 
            If $\tm = \invap{\tm_{1}}{\tm_{2}}$ then the last step of the derivation is the following: 
            \[
                \indrule{\rulmEParTwo}{
                \djum{\utenv}{\tenv}{\tm_{1}}{\typ_{1}\parr\typ_{2}}
                \HS
                \djum{\utenv}{\tenvtwo}{\tm_{2}}{\lneg{\typ_{2}}}
                }{
                \djum{\utenv}{\tenv,\tenvtwo}{\invap{\tm_{1}}{\tm_{2}}}{\typ_{1}}
                }
            \]
            We must prove that for all $\subone$ such that $\subone \subsj{\utenv}{\tenv, \tenv'}$ we have that $(\invap{\tm_{1}}{\tm_{2}})^{\subone}\iftyp{} \redcand{\typ}$.    
            
            By hypothesis for all substitutions $\subthree \subsj{\utenv}{\tenvtwo}$ and $\subtwo \subsj{\utenv}{\tenv}$ we have that $\tm_{1}^{\subthree} \in \redcand{\typ_{1} \parr \typ_{2}}$ and $\tm_{2}^{\subtwo} \in \redcand{\lneg{\typ_{2}}}$.
            In particular we can consider compatible substitutions $\subthree, \subtwo$ such that $(\invap{\tm_{1}}{\tm_{2}})^{\subone} = \invap{\tm_{1}^{\subthree}}{\tm_{2}^{\subtwo}}$, thus without loss of generality it suffices to prove that $\invap{\tm_{1}}{\tm_{2}} \in \redcand{\typ}$.

            Using lemma \ref{lemma:redcand_dual} it suffices to prove that for all $\tmtwo \in \redcand{\lneg{\typ_{1}}}$ that $\ibott{(\invap{\tm_{1}}{\tm_{2}})}{\tmtwo} \iftyp{\bot} \SNtms{\bot}$.
            Using the structural equivalence we have that $\ibott{(\invap{\tm_{1}}{\tm_{2}})}{\tmtwo} \cceq \ibott{\tm_{1}}{(\pair{\tmtwo}{\tm_{2}})}$, thus it suffices to prove that $\ibott{\tm_{1}}{(\pair{\tmtwo}{\tm_{2}})} \iftyp{\bot} \SNtms{\bot}$.
            Since by hypothesis $\tm_{2} \in \redcand{\lneg{\typ_{2}}}$ and $\tmtwo \in \redcand{\lneg{\typ_{1}}}$ then this implies that $\pair{\tmtwo}{\tm_{2}} \iftyp{} \redcand{\lneg{\typ_{1}}} \tensor \redcand{\lneg{\typ_{2}}}$ and by the lemma \ref{lemma:dual_properties} that $\pair{\tmtwo}{\tm_{2}} \iftyp{}  \redcand{\lneg{\typ_{1}} \tensor \lneg{\typ_{2}}}$. 
            Since by hypothesis $\tm_{1} \in \redcand{\typ_{1}\parr \typ_{2}}$ then we can conclude that $\ibott{\tm_{1}}{(\pair{\tmtwo}{\tm_{2}})} \iftyp{\bot} \SNtms{\bot}$.

        \item 
            If $\tm = \ap{\tm_{1}}{\tm_{2}}$ then the last step of the derivation is the following:
            \[
                \indrule{\rulmEParOne}{
                \djum{\utenv}{\tenv}{\tm_{1}}{\typ_{1}\parr\typ_{2}}
                \HS
                \djum{\utenv}{\tenvtwo}{\tm_{2}}{\lneg{\typ_{1}}}
                }{
                \djum{\utenv}{\tenv,\tenvtwo}{\ap{\tm_{1}}{\tm_{2}}}{\typ_2}}
            \]
            We must prove that for all $\subone$ such that $\subone \subsj{\utenv}{\tenv, \tenv'}$ we have that $(\ap{\tm_{1}}{\tm_{2}})^{\subone}\iftyp{} \redcand{\typ}$.    
            
            By hypothesis for all substitutions $\subthree \subsj{\utenv}{\tenvtwo}$ and $\subtwo \subsj{\utenv}{\tenv}$ we have that $\tm_{1}^{\subthree} \in \redcand{\typ_{1} \parr \typ_{2}}$ and $\tm_{2}^{\subtwo} \in \redcand{\lneg{\typ_{1}}}$.
            In particular we can consider compatible substitutions $\subthree, \subtwo$ such that $(\ap{\tm_{1}}{\tm_{2}})^{\subone} = \ap{\tm_{1}^{\subthree}}{\tm_{2}^{\subtwo}}$, thus without loss of generality it suffices to prove that $\ap{\tm_{1}}{\tm_{2}} \in \redcand{\typ}$.

            Using lemma \ref{lemma:redcand_dual} it suffices to prove that for all $\tmtwo \in \redcand{\lneg{\typ_{1}}}$ that $\ibott{(\ap{\tm_{1}}{\tm_{2}})}{\tmtwo} \iftyp{\bot} \SNtms{\bot}$ 
            Using the structural equivalence we have that $\ibott{(\ap{\tm_{1}}{\tm_{2}})}{\tmtwo} \cceq \ibott{(\invap{\tm_{1}}{\tmtwo})}{\tm_{2}}$, so that this case follows from the last one. 
        
        \item 
            If $\tm = \iwhy{\uvar}{\tm_{1}}$ then the last step of the derivation is:
            \[
            \indrule{\rulmIWhy}{
            \djum{\utenv,\uvar:\lneg{\typ}}{\tenv}{\tm_{1}}{\bott}
            }{
            \djum{\utenv}{\tenv}{\iwhy{\uvar}{\tm_{1}}}{\why{\typ}}
            }
            \]
            We must prove that for all $\subone$ such that $\subone \subsj{\utenv}{\tenv}$ we have that $(\iwhy{\uvar}{\tm})^{\subone} \iftyp{} \redcand{\why{\typ}}$.

            By hypothesis for all substitutions $\subtwo \subsj{\utenv}{\tenv, \uvar:\lneg{\typ}}$ we have that $ \tm_{1}^{\subtwo} \iftyp{\bot}  \SNtms{\bot}$, this implies that $(\iwhy{\uvar}{\tm_{1}})^{\subone} \in \why{\redcand{\lneg{\typ}}}$.
            Using the lemma \ref{lemma:dual_properties} we conclude that $(\iwhy{\uvar}{\tm_{1}})^{\subone} \in \redcand{\why{\typ}}$.

        \item 
            If $\tm = \iofctwo{\lvar}{\tm_{1}}$ then this follows from the lemma \ref{lemma:adequacy_ofc}.
            % \[
            %     \indrule{\rulmIOfc}
            %     {\djum{\utenv}{\lvar:\lneg{\typ}}{\tm_{1}}{\bott}}
            %     {\djum{\utenv}{\emptytenv}{\iofctwo{\lvar}{\tm_{1}}}{\ofc{{\typ}}}}
            % \]
            % Since $\redcand{\ofc{\typ}} = \lneg{(\why{\redcand{\lneg{\typ}}})}$ then if $\iwhy{\uvar}{\tmthree} \in \why{\redcand{\lneg{\typ}}}$ we must prove that 
            % $\ibott{\iofctwo{\lvar}{\tm_{1}}}{\iwhy{\uvar}{\tmthree}} \in \SNtms{\bot}$.
            
            % We prove this by induction on $\tml{\tm_{1}} + \tml{\tmthree}$.
            % For the induction step we consider three different cases depending on the reduction step.
            % \begin{enumerate}
            %     \item If $\tm_{1} \to \tm_{1}$, this follows straight from induction.
            %     \item If $\tmthree \to \tmthree'$, this follows straight from induction.
            %     \item If $\ibott{\iofctwo{\lvar}{\tm_{1}}}{\iwhy{\uvar}{\tmthree}} \to \tmthree\sub{\uvar}{\tm_{1} \cos{\lvar}{\unit}}$ then by the lemma \ref{lemma:sub_cos_redcand} we get that 
            %     $\tm_{1} \cos{\lvar}{\unit} \in \redcand{\typ}$ 
            %     and since by hypothesis $\tmthree \in \why{\redcand{\lneg{\typ}}}$ then we can conclude that $\tmthree\sub{\uvar}{\tm_{1} \cos{\lvar}{\unit}} \in \SNtms{\bot}$.
            % \end{enumerate}
        
        \item 
            If $\tm = \ibott{\tm_{1}}{\tm_{2}}$ then the last step of the derivation is the following:
            \[
            \indrule{\rulmIBott}{
                \djum{\utenv}{\tenv}{\tm_{1}}{\typ}
                \HS
              \djum{\utenv}{\tenv'}{\tm_{2}}{\lneg{\typ}}
              }{
                \djum{\utenv}{\tenv,\tenv'}{\ibott{\tm_{1}}{\tm_{2}}}{\bott}
            }
            \]
            We must prove that for all $\subone \subsj{\utenv}{\tenv, \tenvtwo}$ we have that $(\ibott{\tm_{1}}{\tm_{2}})^{\subone} \iftyp{\bot} \redcand{\bot} = \SNtms{\bot}$.

            By hypothesis for all substitutions $\subtwo \subsj{\utenv}{\tenv}$ and $\subthree \subsj{\utenv}{\tenvtwo}$ we have that $\tm_{1}^{\subtwo} \in \redcand{\typ}$ and $\tm_{2}^{\subthree} \in \redcand{\lneg{\typ}}$.
            In particular we can consider compatible substitutions such that $(\ibott{\tm_{1}}{\tm_{2}})^{\subone} = \ibott{\tm_{1}^{\subtwo}}{\tm_{2}^{\subthree}}$ thus using the lemma \ref{lemma:redcand_dual} we can conclude that $(\ibott{\tm_{1}}{\tm_{2}})^{\subone} \iftyp{\bot} \SNtms{\bot}$.
            
        \item 
            If $\tm = \ewhy{\tm_{1}}{\lvar}{\tm_{2}}$ then the last step of the derivation is the following:
            \[
                \indrule{\rulmEWhy}{
                \djum{\utenv}{\tenv}{\tm_{2}}{\why{\typ}}
                \HS
                \djum{\utenv}{\lvar:\typ}{\tm_{1}}{\bott}
                }{
                \djum{\utenv}{\tenv}{\ewhy{\tm_{1}}{\lvar}{\tm_{2}}}{\bott}
                }
            \]
            We must prove that for all $\subone \subsj{\utenv}{\tenv,\tenv'}$ we have that $(\ewhy{\tm_{1}}{\lvar}{\tm_{2}})^{\subone} \in \SNtms{\bot}$.

            By hypothesis for all substitutions $\subtwo \subsj{\utenv}{\tenv}$ and $\subthree \subsj{\utenv}{\lvar:\typ}$ we have that $\tm_{2}^{\subtwo} \iftyp{} \redcand{\why{\typ}}$ and $\tm_{1}^{\subthree} \iftyp{} \redcand{\bot}$. In particular we can consider compatible substitutions such that $(\ewhy{\tm_{1}}{\lvar}{\tm_{2}})^{\subone} = \ewhy{\tm_{1}^{\subthree}}{\lvar}{\tm_{2}^{\subtwo}}$ thus without loss of generality it suffices to prove that $\ewhy{\tm_{1}}{\lvar}{\tm_{2}} \in \redcand{\why{\typ}}$.

            Using lemma \ref{lemma:redcand_dual} it suffices to prove that for all $\tmthree \in \redcand{\one}$ that $\ibott{\ewhy{\tm_{1}}{\lvar}{\tm_{2}}}{\tmthree} \iftyp{\bot} \SNtms{\bot}$.
            By the structural equivalence we get $\ibott{\ewhy{\tm_{1}}{\lvar}{\tm_{2}}}{\tmthree} \cceq (\ibott{\iofctwo{\lvar}{\tm_{1}}}{\tmtwo})\eone{\tmthree}$ and since $\tmthree \in \redcand{\one} = \SNtms{\one}$ then $\tmthree$ is strongly normalizing, 
            thus it suffices to prove that $\ibott{\iofctwo{\lvar}{\tm_{1}}}{\tmtwo} \iftyp{\bot} \SNtms{\bot}$ and this follows from the lemma \ref{lemma:adequacy_ofc}.
            
        \item 
            If $\tm = \tm_{1}\eofc{\uvar}{\tm_{2}}$ then the last step of the derivation is the following:
            \[
                \indrule{\rulmEOfc} 
                {\djum{\utenv}{\tenv}{\tm_{2}}{\ofc{\typ}}   
                \HS 
                \djum{\utenv,\uvar:\typ}{\tenvtwo}{\tm_{1}}{\typthree}  
                }   
                {\djum{\utenv}{\tenv,\tenvtwo}{\tm_{1}\eofc{\uvar}{\tm_{2}}}{\typthree}} 
            \]
            We must prove that for all $\subone \subsj{\utenv}{\tenv, \tenvtwo}$ we have that $(\tm_{1}\eofc{\uvar}{\tm_{2}})^{\subone} \in \redcand{\typthree}$.

            By hypothesis for all substitutions $\subtwo \subsj{\utenv}{\tenv}$ and $\subthree \subsj{\utenv, \uvar:\typ}{\tenvtwo}$ we have that $\tm_{2}^{\subtwo} \in \redcand{\ofc{\typ}}$ and $\tm_{1}^{\subthree} \in \redcand{\typthree}$.
            In particular we can consider compatible substitutions such that $(\tm_{1} \eofc{\uvar}{\tm_{2}})^{\subone} = \tm_{1}^{\subthree} \eofc{\uvar}{\tm_{2}^{\subtwo}}$ thus it suffices to prove that $\tm_{1}\eofc{\uvar}{\tm_{2}} \in \redcand{\typthree}$.

            Using lemma \ref{lemma:redcand_dual} it suffices to prove that for all $\tmthree \in \redcand{\lneg{\typthree}}$ that $\ibott{\tm_{1}\eofc{\uvar}{\tm_{2}}}{\tmthree} \iftyp{\bot} \SNtms{\bot}$.
            Using the structural equivalence we get that $\ibott{\tm_{1}\eofc{\uvar}{\tm_{2}}}{\tmthree} \cceq \ibott{\iwhy{\uvar}{(\ibott{\tm_{1}}{\tmthree})}}{\tm_{2}}$ thus it suffices to prove that $\ibott{\iwhy{\uvar}{(\ibott{\tm_{1}}{\tmthree})}}{\tm_{2}} \iftyp{\bot} \SNtms{\bot}$.

            Since $\tmthree \in \redcand{\lneg{\typthree}}$ then we get that $\ibott{\tm_{1}}{\tmthree} \iftyp{\bot} \SNtms{\bot}$ and by hypothesis for all $\tmtwo \in \redcand{\typ}$ we have that $(\ibott{\tm_{1}}{\tmthree} )\sub{\lvar}{\tmtwo} \iftyp{\bot} \SNtms{\bot}$.
            This implies that $\iwhy{\uvar}{(\ibott{\tm_{1}}{\tmthree})} \in \why{\redcand{\lneg{\typ}}}$ and by the lemma \ref{lemma:dual_properties} then $\iwhy{\uvar}{(\ibott{\tm_{1}}{\tmthree})} \in \redcand{\why{\lneg{\typ}}}$.
            Since $\tm_{2} \in \redcand{\ofc{\typ}}$ then we conclude that $\ibott{\iwhy{\uvar}{(\ibott{\tm_{1}}{\tmthree})}}{\tm_{2}} \in \SNtms{\bot}$.
        
        \item 
            If $\tm = \tm_{1}\eone{\tm_{2}}$ then the last step of the derivation is the following:
            \[
                \indrule{\rulmEOne}{
                \djum{\utenv}{\tenv}{\tm_{2}}{\one}
                \HS
                \djum{\utenv}{\tenvtwo}{\tm_{1}}{\typ}
                }{
                \djum{\utenv}{\tenv,\tenvtwo}{\tm_{1}\eone{\tm_{2}}}{\typ}
                }
            \]
            We must prove that for all $\subone \subsj{\utenv}{\tenv, \tenvtwo}$ we have that $(\tm_{1} \eone{\tm_{2}})^{\subone} \in \redcand{\typ}$.

            By hypothesis for all substitutions $\subtwo \subsj{\utenv}{\tenv}$ and $\subthree \subsj{\utenv}{\tenvtwo}$ we have that $\tm_{2}^{\subtwo} \in \redcand{\one}$ and $\tm_{1}^{\subthree} \in \redcand{\typ}$.
            In particular we can consider compatible substitutions such that $(\tm_{1} \eofc{\uvar}{\tm_{2}})^{\subone} = \tm_{1}^{\subthree} \eone{\tm_{2}^{\subtwo}}$ thus it suffices to prove that $\tm_{1}\eone{\tm_{2}} \in \redcand{\typ}$.

            Using lemma \ref{lemma:redcand_dual} it suffices to prove that for all $\tmthree \in \redcand{\lneg{\typ}}$ that $\ibott{\tm_{1}\eone{\tm_{2}}}{\tmthree} \iftyp{\bot} \SNtms{\bot}$.
            Using the structural equivalence we get that $\ibott{\tm_{1}\eone{\tm_{2}}}{\tmthree} \cceq (\ibott{\tm_{1}}{\tmthree}) \eone{\tm_{2}}$, thus it suffices to prove that $(\ibott{\tm_{1}}{\tmthree}) \eone{\tm_{2}} \iftyp{\bot} \SNtms{\bot}$.
            Since by the inductive hypothesis $\ibott{\tm_{1}}{\tmthree} \iftyp{\bot} \SNtms{\bot}$ and since $\tm_{2} \in \redcand{\one} = \SNtms{\one}$ then we conclude that $ (\ibott{\tm_{1}}{\tmthree}) \eone{\tm_{2}} \iftyp{\bot} \SNtms{\bot}$.

    \end{itemize}
\end{proof}

%%% Local Variables:
%%% mode: latex
%%% TeX-master: "../main"
%%% End:

\section{Weak Church Rosser}
\lsec{app:calcMELL:wcr}

\begin{lemma}
Let $\cctx$ be a positive eliminator context. Then the following items hold:
  
  \begin{enumerate}

  \item$(\ipar{\lvar}{\lvartwo}{\tm})\cctx$ is well-typed iff $\ipar{\lvar}{\lvartwo}{\tm\cctx}$ is well-typed.
  \item $(\iwhy{\uvar}{\tm})\cctx$ is well-typed iff $\iwhy{\uvar}{\tm\cctx}$ is well-typed.
  \item If $\dom{\cctx}\cap \fv{\tmtwo}=\emptyset$, then  $(\ibott{\tm}{\tmtwo})\cctx$ is well-typed iff $ \ibott{\tm\cctx}{\tmtwo}$ is well-typed. Similarly, $(\ibott{\tmtwo}{\tm})\cctx$ is well-typed iff $ \ibott{\tmtwo}{\tm\cctx}$ is well-typed.
  \end{enumerate}
  
\end{lemma}

\begin{proof}
This follows from the fact that in terms of the form $\tm\epair{\lvar}{\lvartwo}{\tmtwo}$, $\tm\eofc{\uvar}{\tmtwo}$ and $\tm\eone{\tmtwo}$, the typing rules  $\rulmETensor$, $\rulmEOfc$, and $\rulmEOne$, resp., place no requirement on the type of $\tm$ and allow any number of linear and unrestricted variables.
\end{proof}

\begin{lemma}
\llem{ctx_propagation}\mbox{}
Let $\cctx$ be a positive eliminator context. Then the following items hold:
  \begin{enumerate}

  \item\label{ctx_propagation:par} $(\ipar{\lvar}{\lvartwo}{\tm})\cctx \cceq \ipar{\lvar}{\lvartwo}{\tm\cctx}$.
  \item\label{ctx_propagation:iwhy} $(\iwhy{\uvar}{\tm})\cctx \cceq \iwhy{\uvar}{\tm\cctx}$.
  \item\label{ctx_propagation:ibott} $(\ibott{\tm}{\tmtwo})\cctx \cceq \ibott{\tm\cctx}{\tmtwo}$ and $(\ibott{\tmtwo}{\tm})\cctx \cceq \ibott{\tmtwo}{\tm\cctx}$, if $\dom{\cctx}\cap \fv{\tmtwo}=\emptyset$.
  % \item\label{ctx_propagation:ibott:why} $(\ibott{\tm}{\iwhy{\uvar}{\tmtwo}})\cctx \cceq \ibott{\tm}{(\iwhy{\uvar}{\tmtwo})\cctx}$ and $(\ibott{\iwhy{\uvar}{\tmtwo}}{\tm})\cctx \cceq \ibott{(\iwhy{\uvar}{\tmtwo})\cctx}{\tm}$, if $\dom{\cctx}\cap \fv{\tm}=\emptyset$.

    \item\label{ctx_propagation:epair} $\tm\epair{\lvar}{\lvartwo}{\tmtwo}\cctx \cceq \tm\epair{\lvar}{\lvartwo}{\tmtwo\cctx}$, if $\dom{\cctx}\cap\fv{\tmtwo}\emptyset$.

          \item\label{ctx_propagation:eofc} $\tm\eofc{\uvar}{\tmtwo}\cctx \cceq \tm\eofc{\uvar}{\tmtwo\cctx}$,  if $\dom{\cctx}\cap\fv{\tmtwo}=\emptyset$.

          \item\label{ctx_propagation:el_ctx} $ \tm\pelim{\patt}{\tmtwo}\cctx\cceq \tm\cctx\pelim{\patt}{\tmtwo}$,  if $\fv{\cctx}\cap\fv{\patt}=\emptyset$.

  \end{enumerate}
  
\end{lemma}

\begin{proof}
  All the items can be proved by induction on $\cctx$.

\end{proof}

\begin{lemma}[Reduction is compatible with substitution of linear variables]
\llem{linear_subst_compat_with_reduction}
Suppose $\tm\pretome \tmtwo$. Then 
  \begin{xenumerate}
  \item\label{linear_subst_compat_with_reduction:left} $\ltm\sub{\lvar}{\ltmthree}\pretome \ltmtwo\sub{\lvar}{\ltmthree}$.
  \item\label{linear_subst_compat_with_reduction:right} $\ltmthree\sub{\lvar}{\ltm}\pretome \ltmthree\sub{\lvar}{\ltmtwo}$. 

  \end{xenumerate}
\end{lemma}

\begin{proof}
By induction  on $\tm$.
\end{proof}

\begin{lemma}[Reduction is compatible with substitution of unrestricted variables]
\llem{subst_compat_with_reduction}
Suppose  $\tm\pretome \tmtwo$. Then
  \begin{xenumerate}
  \item\label{subst_compat_with_reduction:left} $\ltm\sub{\uvar}{\ltmthree}\pretome \ltmtwo\sub{\uvar}{\ltmthree}$.
  \item\label{subst_compat_with_reduction:right} $\ltmthree\sub{\uvar}{\ltm}\rtpretome \ltmthree\sub{\uvar}{\ltmtwo}$.  

  \end{xenumerate}
\end{lemma}

\begin{proof}
By induction  on $\tm$.
\end{proof}

% \begin{lemma}[{name=Reduction is compatible with contra-substitution~\proofnote{Proof on pg.~\pageref{cosubst_compat_with_reduction:proof}}, restate=[name=Reduction is compatible with contra-substitution]ReductionCompatibleWithContraSubstitution}]
\begin{lemma}
\llem{cosubst_compat_with_reduction}
Let $\tm$ and  $\ltm\cos{\lvar}{\ltmthree}$ be typable terms.
  \begin{xenumerate}
  \item\label{cosubst_compat_with_reduction:left} $\ltm\pretome \ltmtwo$ implies $\ltm\cos{\lvar}{\ltmthree}\tome \ltmtwo\cos{\lvar}{\ltmthree}$.
  \item\label{cosubst_compat_with_reduction:right} $\ltmthree\pretome \ltmtwo$ implies $\ltm\cos{\lvar}{\ltmthree}\tome \ltm\cos{\lvar}{\ltmtwo}$. 

  \end{xenumerate}
\end{lemma}

%\ReductionCompatibleWithContraSubstitution*

\begin{proof}%\label{cosubst_compat_with_reduction:proof}
  Both items are proved simultaneously by induction on $\ltm$.
\end{proof}

% \begin{lemma}[{name=Weak Church-Rosser~\proofnote{Proof on pg.~\pageref{wcr:proof}}, restate=[name=Weak Church-Rosser]WeakChurchRosser}]
\begin{lemma}[name=Weak Church-Rosser]
\llem{wcr}
Let $\tm$ be typed. Then $\ltm\tome{}\ltmtwo$ and $\ltm\tome{}\ltmthree$, implies there exists $\ltmfour$ such that $\ltmtwo\rttome\ltmfour$ and $\ltmthree\rttome\ltmfour$.
\end{lemma}

%\WeakChurchRosser*

\begin{proof}%\label{wcr:proof}
Since $\cceq$ is a strong-bisimulation, it suffices to consider the case where  $\ltm\pretome{}\ltmtwo$ through a step $\step$ and $\ltm\pretome{}\ltmthree$ through a step $\steptwo$.  By induction on $\ltm$.

\end{proof}

%%% Local Variables:
%%% mode: latex
%%% TeX-master: "../main"
%%% End:

\section{Translations}
\lsec{app:translations}

%%%%%%%%%%%%%%%%%%%%%%%%%%%%%%%%%%
%        T-Translation for Parigot
%%%%%%%%%%%%%%%%%%%%%%%%%%%%%%%%%%

\subsection{Parigot's $\CalcParigot$}

\subsubsection{T-Translation}
\TypePreservationForTTranslation*

\begin{proof}\label{type_preservation_for_t_translation:proof}
  By induction on the derivation. 
  \begin{enumerate}
  \item $\rulpAx$. The derivation ends in
\[
  \indrule{\rulpAx}
  {}
  {\pjum{\tenv,\var:\typ}{\var}{\typ}{\nenv}}
  \]

Note that $\ttrat{\var}{\vark}=\ibott{\var}{\ofc{\vark}}$ and thus we conclude from:
\[
        \indrule{\rulpAx}
        {}
        {\djum{\why{\ttra{\tenv}},\var:\why{\ttra{\typ}},
            \lneg{(\ttra{\nenv})}, \vark:\lneg{(\ttra{\typ})}}{\emptytenv}{\ibott{\var}{\ofc{\vark}}}{\bott}
      }
\]

\item $\rulpLam$. The derivation ends in
  \[
  \indrule{\rulpLam}{
    \pjum{\tenv,\var:\typ}{\ptmtwo}{\typtwo}{\nenv}
  }{
    \pjum{\tenv}{\lam{\var}{\ptmtwo}}{\typ\imp\typtwo}{\nenv}
  }
\]

\[
  \indrule{\rulmIBott}
  {
  \indrule{\rulmIPar}{
    \indrule{\rlem{unrestricted_eq_bang_linear}(\ref{to_lin})}{
         \indrule{\rlem{unrestricted_eq_bang_linear}(\ref{to_lin})}
         {
           \indih{
             \djum{\why{\ttra{\tenv}},\var:\why{\ttra{\typ}},\lneg{(\ttra{\nenv})},\varktwo:\lneg{(\ttra{\typtwo})}}{\emptytenv}{\ttrat{\ptmtwo}{\varktwo}}{\bott}}
           }
         {\djum{\why{\ttra{\tenv}},\lneg{(\ttra{\nenv})}, \varktwo:\lneg{(\ttra{\typtwo})}}{\lvar:\ofc{\why{\ttra{\typ}}}}{\ttrat{\ptmtwo}{\varktwo}\eofc{\var}{\lvar}}{\bott}
         }
  }
  {
    \djum{\why{\ttra{\tenv}},\lneg{(\ttra{\nenv})}}{\lvar:\ofc{\why{\ttra{\typ}}},\lvartwo:\ofc{\lneg{(\ttra{\typtwo})}}}{\ttrat{\ptmtwo}{\varktwo}\eofc{\var}{\lvar}\eofc{\varktwo}{\lvartwo}}{\bott}
  }
    }{
        \djum{\why{\ttra{\tenv}},\lneg{(\ttra{\nenv})}}{\emptytenv}{\ipar{\lvar}{\lvartwo}{\ttrat{\ptmtwo}{\varktwo}\eofc{\var}{\lvar}\eofc{\varktwo}{\lvartwo}}}{\why{\ofc{\lneg{(\ttra{\typ})}}}\parr \why{\ttra{\typtwo}}}
      } 
    }
          {
                          \djum{\why{\ttra{\tenv}},\lneg{(\ttra{\nenv})}, \vark: \ofc{\why{\ttra{\typ}}}\tensor \ofc{\lneg{(\ttra{\typtwo})}}}{\emptytenv}{\ibott{\ipar{\lvar}{\lvartwo}{\ttrat{\ptmtwo}{\varktwo}\eofc{\var}{\lvar}\eofc{\varktwo}{\lvartwo}}}{\vark}}{\bott}
            }
  \]

\item $\rulpApp$. The derivation ends in
  \[
  \indrule{\rulpApp}{
    \pjum{\tenv}{\ptmfour}{\typ\imp\typtwo}{\nenv}
    \HS
    \pjum{\tenv}{\ptmfive}{\typ}{\nenv}
  }{
    \pjum{\tenv}{\pap{\ptmfour}{\ptmfive}}{\typtwo}{\nenv}
  }
\]

Let $\pi$ be the derivation:
\[
  \indrule{\rulmITensor}
  {
       \indrule{\rulmIOfc}
       {
             \indrule{\rulmIWhy}
             {
               \indih{
                 \djum{\why{\ttra{\tenv}},\lneg{(\ttra{\nenv})},\varkthree:\lneg{(\ttra{\typ})}}{\emptytenv}{\ttrat{\ptmfive}{\varkthree}}{\bott}
                 }
             }
             {
                   \djum{\why{\ttra{\tenv}},\lneg{(\ttra{\nenv})}}{\emptytenv}{\iwhy{\varkthree}{\ttrat{\ptmfive}{\varkthree}}}{\why{\ttra{\typ}}}
                 }
         }
         {
                         \djum{\why{\ttra{\tenv}},\lneg{(\ttra{\nenv})}}{\emptytenv}{\ofc{\iwhy{\varkthree}{\ttrat{\ptmfive}{\varkthree}}}}{\ofc{\why{\ttra{\typ}}}}
                       }
    }
    {
       \djum{\why{\ttra{\tenv}},\lneg{(\ttra{\nenv})},\vark:\lneg{(\ttra{\typtwo})}}{\emptytenv}{\pair{\ofc{\iwhy{\varkthree}{\ttrat{\ptmfive}{\varkthree}}}}{\ofc{\vark}}}{\ofc{\why{\ttra{\typ}}}\tensor\ofc{\lneg{(\ttra{\typtwo})}}}
    }
    \]

    \[
\indrule{\rulmSub}
             {
               \indih{
                 \djum{\why{\ttra{\tenv}},\lneg{(\ttra{\nenv})},\varktwo:\ofc{\why{\ttra{\typ}}}\tensor\ofc{\lneg{(\ttra{\typtwo})}}}{\emptytenv}{\ttrat{\ptmfour}{\varktwo}}{\bott}
                 }
             }
             {
                    \djum{\why{\ttra{\tenv}},\lneg{(\ttra{\nenv})},\varktwo:\ofc{\lneg{(\ttra{\typtwo})}}}{\emptytenv}{\ttrat{\ptmfour}{\varktwo}\sub{\varktwo}{\pair{\ofc{\iwhy{\varkthree}{\ttrat{\ptmfive}{\varkthree}}}}{\ofc{\vark}}}}{\bott}
               }
      \]
\item $\rulpName$. The derivation ends in
  \[
      \indrule{\rulpName}{
    \pjum{\tenv}{\ptmtwo}{\typ}{\nenv,\nvar:\typ}
  }{
    \pjum{\tenv}{\pname{\nvar}{\ptmtwo}}{\bott}{\nenv,\nvar:\typ}
    }
\]

\[
  \indrule{\rulmIBott}{
      \indrule{\rulmSub}{
        \indih{
                \djum{\why{\ttra{\tenv}},\lneg{(\ttra{\nenv})},\nvar:
          \lneg{(\ttra{\typ})},\varktwo:\lneg{(\ttra{\typ})}}{\emptytenv}{\ttrat{\ptmtwo}{\varktwo}}{\bott}
           }
        }
        {
                      \djum{\why{\ttra{\tenv}},\lneg{(\ttra{\nenv})},\nvar:
          \lneg{(\ttra{\typ})}}{\emptytenv}{\ttrat{\ptmtwo}{\varktwo}\sub{\varktwo}{\nvar}}{\bott}
      }
}
{
                        \djum{\why{\ttra{\tenv}},\lneg{(\ttra{\nenv})},\nvar:
          \lneg{(\ttra{\typ})},\vark:\one}{\emptytenv}{\ibott{\ttrat{\ptmtwo}{\varktwo}\sub{\varktwo}{\nvar}}{\vark}}{\bott}

  }
\]

\item $\rulpMu$. The derivation ends in
  \[
          \indrule{\rulpMu}{
    \pjum{\tenv}{\ptmtwo}{\bott}{\nenv,\nvar:\typ}
  }{
    \pjum{\tenv}{\pmu{\nvar}{\ptmtwo}}{\typ}{\nenv}
  }
\]

\[
  \indrule{\rulmSub}{
      \indrule{\rulmSub}{
        \indih{
                \djum{\why{\ttra{\tenv}},\lneg{(\ttra{\nenv})},\nvar:
          \lneg{(\ttra{\typ})},\varktwo:\one}{\emptytenv}{\ttrat{\ptmtwo}{\varktwo}}{\bott}
           }
        }
        {
                      \djum{\why{\ttra{\tenv}},\lneg{(\ttra{\nenv})},\nvar:
          \lneg{(\ttra{\typ})}}{\emptytenv}{\ttrat{\ptmtwo}{\varktwo}\sub{\varktwo}{\ione}}{\bott}
      }
}
{
                        \djum{\why{\ttra{\tenv}},\lneg{(\ttra{\nenv})},\vark:\lneg{(\ttra{\typ})}}{\emptytenv}{\ttrat{\ptmtwo}{\varktwo}\sub{\varktwo}{\ione}\sub{\nvar}{\vark}}{\bott}
  }
\]

% \item $\rulpCL$. The derivation ends in
%   \[
%   \indrule{\rulpCL}{
%     \pjum{\tenv,\var:\typ,\vartwo:\typ}{\ptm}{\typthree}{\nenv}
%   }{
%     \pjum{\tenv, \varthree:\typ}{\ptm\sub{\var,\vartwo}{\varthree}}{\typthree}{\nenv}
%   }
% \]

% \item $\rulpCR$. The derivation ends in
%   \[
%      \indrule{\rulpCR}{
%     \pjum{\tenv}{\ptm}{\typthree}{\nenv, \nvar:\typ,\nvartwo:\typ}
%   }{
%     \pjum{\tenv}{\ptm\sub{\nvar,\nvartwo}{\nvarthree}}{\typthree}{\nenv, \nvarthree:\typ}
%   }
% \]

% \item $\rulpWL$. The derivation ends in
%   \[
%    \indrule{\rulpWL}{
%     \pjum{\tenv}{\ptm}{\typ}{\nenv}
%   }{
%     \pjum{\tenv, \var:\typtwo}{\ptm}{\typ}{\nenv}
%   }
% \]

% \item $\rulpWR$. The derivation ends in
%   \[
%      \indrule{\rulpWR}{
%     \pjum{\tenv}{\ptm}{\typ}{\nenv}
%   }{
%     \pjum{\tenv}{\ptm}{\typ}{\nenv, \nvar:\typtwo}
%   }
%   \]
  \end{enumerate}
  
\end{proof}

\begin{lemma}
\llem{calcMELL_simulates_beta_via_t_translation}
  Suppose $\ptm \to_{\beta} \ptmtwo$, then $\ttrat{\ptm}{\vark} \tome \ttrat{\ptmtwo}{\vark}$.
% $\ttra{(\pap{(\lam{\var}{\ptm})}{\ptmtwo})} \rtto \ttra{\ptm\sub{\var}{\ptmtwo}}$
\end{lemma}

\begin{proof}
By induction on $\ptm$.

\end{proof}

\begin{lemma}
\llem{calcMELL_simulates_mu_via_t_translation}  Suppose $\ptm \to_{\mu} \ptmtwo$, then $\ttrat{\ptm}{\vark} = \ttrat{\ptmtwo}{\vark}$.
%  $\ttra{(\pap{(\pmu{\nvar}{\ptm})}{\ptmtwo})} =\ttra{(\pmu{\nvar}{(\ptm\psub{\nvar}{\ptmtwo})})}$
\end{lemma}

\begin{proof}
By induction on $\ptm$.

\end{proof}

\begin{lemma}
  \llem{calcMELL_simulates_rho_via_t_translation}
  Suppose $\ptm \to_{\rho} \ptmtwo$, then $\ttrat{\ptm}{\vark} \cceq \ttrat{\ptmtwo}{\vark}$.
%  $ \ttra{(\pname{\nvar}{(\pmu{\nvartwo}{\ptm})})} \cceq\ttra{(\ptm \rensub{\nvartwo}{\nvar})}$
\end{lemma}

\begin{proof}
By induction on $\ptm$.

\end{proof}

% \begin{lemma}
%   \llem{t_translation_and_free_variables}
%   $\nvar\notin\fv{\ptm}$ implies $\nvar\notin\fv{\ttrat{\ptm}{\vark}}$.
% \end{lemma}

% \begin{proof}
% By induction on $\ptm$.
% \end{proof}

\begin{lemma}
  \llem{calcMELL_simulates_theta_via_t_translation}
  Suppose $\ptm \to_{\theta} \ptmtwo$, then $\ttrat{\ptm}{\vark} \cceq \ttrat{\ptmtwo}{\vark}$.
%  $\ttra{(\pmu{\nvar}{\pname{\nvar}{\ptm}})}\cceq \ttra{\ptm}$, if $\nvar\notin\fv{\ptm}$.
\end{lemma}

\begin{proof}
By induction on $\ptm$.

\end{proof}

\CalcMELLSimulatesCalcParigotViaTTranslation*

\begin{proof}\label{calcMELL_simulates_calcParigot_via_t_translation:proof}
By induction on $\ptm \tomu \ptmtwo$ using~\rlem{calcMELL_simulates_beta_via_t_translation},~\rlem{calcMELL_simulates_mu_via_t_translation},~\rlem{calcMELL_simulates_rho_via_t_translation} and~\rlem{calcMELL_simulates_theta_via_t_translation}. 
% In summary,
% \begin{enumerate}
% \item  If $\ptm \tome^{\beta} \ptmtwo$, then $\ttrat{\ptm}{\vark} \rttome \ttrat{\ptmtwo}{\vark}$.
% \item  If $\ptm \tome^{\mu} \ptmtwo$, then $\ttrat{\ptm}{\vark} = \ttrat{\ptmtwo}{\vark}$.
% \item  If $\ptm \tome^{\rho} \ptmtwo$, then $\ttrat{\ptm}{\vark} \cceq \ttrat{\ptmtwo}{\vark}$.  
% \item If $\ptm \tome^{\theta} \ptmtwo$, then $\ttrat{\ptm}{\vark} \cceq \ttrat{\ptmtwo}{\vark}$.
%   \end{enumerate}
\end{proof}

% \edu{Completeness de la traducci'on en el rango de la traducci'on?  $\{ \tm\in \CalcMELL\,|\,\ttra{\ptm}\tome \tm, \mbox{ for some }\ptm\in\lambda\mu\}$}

%%%%%%%%%%%%%%%%%%%%%%%%%%%%%%%%%%
%        Q-Translation for Parigot
%%%%%%%%%%%%%%%%%%%%%%%%%%%%%%%%%%

% \begin{lemma}
% \llem{qtra_cont_var_renaming}
% $\qtrat{\ptm}{\vark}\sub{\vark}{\varktwo} = 
%   \qtrat{\ptm}{\varktwo}$.
% \end{lemma}

% \begin{proof}
% By induction on $\ptm$.
% \end{proof}

% \begin{remark}
% \lremark{qtra_and_qtratwo}
% $ \ibott{\ofc{\qtratwo{\pval}}}{\vark} = \qtrat{\pval}{\vark} $.
% \end{remark}
\subsubsection{Q-Translation}

\TypePreservationForQTranslation*

\begin{proof}\label{type_preservation_for_q_translation:proof}
  By induction on the derivation. 
  \begin{enumerate}
  \item $\rulpAx$. The derivation ends in
\[
  \indrule{\rulpAx}
  {}
  {\pjum{\tenv,\var:\typ}{\var}{\typ}{\nenv}}
  \]

\[
        \indrule{\rulpAx}
        {}
        {\djum{\qtratwo{\tenv},\var:\qtratwo{\typ},
            \lneg{(\qtra{\nenv})}, \vark:\lneg{(\qtra{\typ})}}{\emptytenv}{\ibott{\ofc{\var}}{\vark}}{\bott}
      }
\]

    \item $\rulpLam$. The derivation ends in
  \[
  \indrule{\rulpLam}{
    \pjum{\tenv,\var:\typ}{\ptmtwo}{\typtwo}{\nenv}
  }{
    \pjum{\tenv}{\lam{\var}{\ptmtwo}}{\typ\imp\typtwo}{\nenv}
  }
\]

\[
  \indrule{\rulmIBott}
  {
  \indrule{\rulmIOfc}
  {
  \indrule{\rulmIPar}{
    \indrule{\rlem{unrestricted_eq_bang_linear}(\ref{to_lin})}{
         \indrule{\rlem{unrestricted_eq_bang_linear}(\ref{to_lin})}
         {
           \indih{
             \djum{\qtratwo{\tenv},\var:\qtratwo{\typ},\lneg{(\qtra{\nenv})},\varktwo:\lneg{(\qtra{\typtwo})}}{\emptytenv}{\qtrat{\ptmtwo}{\varktwo}}{\bott}}
           }
         {\djum{\qtratwo{\tenv},\lneg{(\qtra{\nenv})}, \varktwo:\lneg{(\qtra{\typtwo})}}{\lvar:\ofc{\qtratwo{\typ}}}{\qtrat{\ptmtwo}{\varktwo}\eofc{\var}{\lvar}}{\bott}
         }
  }
  {
    \djum{\qtratwo{\tenv},\lneg{(\qtra{\nenv})}}{\lvar:\ofc{\qtratwo{\typ}},\lvartwo:\ofc{\lneg{(\qtrat{\typtwo}{\varktwo})}}}{\qtrat{\ptmtwo}{\varktwo}\eofc{\var}{\lvar}\eofc{\varktwo}{\lvartwo}}{\bott}
  }
    }{
        \djum{\qtratwo{\tenv},\lneg{(\qtra{\nenv})}}{\emptytenv}{\ipar{\lvar}{\lvartwo}{\qtrat{\ptmtwo}{\varktwo}\eofc{\var}{\lvar}\eofc{\varktwo}{\lvartwo}}}{\lneg{(\qtra{\typ})}\parr \why{\qtra{\typtwo}}}
      } 
    }
          {
                          \djum{\qtratwo{\tenv},\lneg{(\qtra{\nenv})}}{\emptytenv}{\ofc{(\ipar{\lvar}{\lvartwo}{\qtrat{\ptmtwo}{\varktwo}\eofc{\var}{\lvar}\eofc{\varktwo}{\lvartwo}})}}{\ofc{(\lneg{(\qtra{\typ})}\parr \why{\qtra{\typtwo}})}}
                        }
   }
   {
            \djum{\qtratwo{\tenv},\lneg{(\qtra{\nenv})}, \vark: \why{(\qtra{\typ}\tensor \ofc{\lneg{(\qtrat{\typtwo}{\varktwo})}})}}{\emptytenv}{\ibott{\ofc{(\ipar{\lvar}{\lvartwo}{\qtra{\ptmtwo}\eofc{\var}{\lvar}\eofc{\varktwo}{\lvartwo}})}}{\vark}}{\bott}
     }
  \]

\item $\rulpApp$. The derivation ends in
  \[
  \indrule{\rulpApp}{
    \pjum{\tenv}{\ptmtwo}{\typ\imp\typtwo}{\nenv}
    \HS
    \pjum{\tenv}{\ptmthree}{\typ}{\nenv}
  }{
    \pjum{\tenv}{\pap{\ptmtwo}{\ptmthree}}{\typtwo}{\nenv}
  }
\]

Let $\pi$ be the derivation:
    \[
      \indrule{\rulmEParTwo}
      {
      \indrule{\rulmIOfc}
    {
      \djum{\uvartwo:\lneg{(\qtra{\typ})}\parr\why{\qtra{\typtwo}}, \qtratwo{\tenv}, \lneg{(\qtra{\nenv})},\vark:\lneg{(\qtra{\typtwo})}}{\emptytenv}{\vark}{\lneg{(\qtra{\typtwo})}}
            }          
          {
              \djum{\uvartwo:\lneg{(\qtra{\typ})}\parr\why{\qtra{\typtwo}}, \qtratwo{\tenv}, \lneg{(\qtra{\nenv})},\vark:\lneg{(\qtra{\typtwo})}}{\emptytenv}{\ofc{\vark}}{\ofc{\lneg{(\qtra{\typtwo})}}}
          }
        }
        {
            \djum{\uvartwo:\lneg{(\qtra{\typ})}\parr\why{\qtra{\typtwo}}, \qtratwo{\tenv}, \lneg{(\qtra{\nenv})},\vark:\lneg{(\qtra{\typtwo})}}{\emptytenv}{\invap{\uvartwo}{\ofc{\vark}}}{\lneg{(\qtra{\typ})}}
          }
    \]

    \[
      \indrule{\rulmIWhy}
      {
             \indrule{\rulmSub}
             {               
                 \begin{array}{l}
                 \djum{\uvartwo:\lneg{(\qtra{\typ})}\parr\why{\qtra{\typtwo}}, \qtratwo{\tenv}, \lneg{(\qtra{\nenv})},\vark:\lneg{(\qtra{\typtwo})}}{\emptytenv}{\invap{\uvartwo}{\ofc{\vark}}}{\lneg{(\qtra{\typ})}}
                 \\
                   \djum{\uvartwo:\lneg{(\qtra{\typ})}\parr\why{\qtra{\typtwo}},\qtratwo{\tenv},\lneg{(\qtra{\nenv})}, \vark:\lneg{(\qtra{\typtwo})}, \varktwo:\lneg{(\qtra{\typ})}}{\emptytenv}{\qtrat{\ptmthree}{\varktwo}}{\bott}
                   \end{array}                                      
             }
             {
                   \djum{\uvartwo:\lneg{(\qtra{\typ})}\parr\why{\qtra{\typtwo}},\qtratwo{\tenv},\lneg{(\qtra{\nenv})}, \vark:\lneg{(\qtra{\typtwo})}}{\emptytenv}{\qtrat{\ptmthree}{\varktwo}\sub{\varktwo}{\invap{\uvartwo}{\ofc{\vark}}}}{\bott}
                 }
               }
               {
                 \djum{\qtratwo{\tenv},\lneg{(\qtra{\nenv})}, \vark:\lneg{(\qtra{\typtwo})}}{\emptytenv}{\iwhy{\uvartwo}{\qtrat{\ptmthree}{\varktwo}\sub{\varktwo}{\invap{\uvartwo}{\ofc{\vark}}}}}{\why{(\qtra{\typ}\tensor\ofc{\lneg{(\qtra{\typtwo})}})}}
                 }
               \]
               
    \[
\indrule{\rulmSub}
             {
                 \begin{array}{l}
                   \djum{\qtratwo{\tenv},\lneg{(\qtra{\nenv})}, \vark:\lneg{(\qtra{\typtwo})}}{\emptytenv}{\iwhy{\uvartwo}{\qtrat{\ptmthree}{\varktwo}\sub{\varktwo}{\invap{\uvartwo}{\ofc{\vark}}}}}{\why{(\qtra{\typ}\tensor\ofc{\lneg{(\qtra{\typtwo})}})}}
                   \\
                 \djum{\qtratwo{\tenv},\lneg{(\qtra{\nenv})}, \varkthree:\lneg{(\qtra{\typtwo})}, \vark':\why{(\qtra{\typ}\tensor\ofc{\lneg{(\qtra{\typtwo})}})}}{\emptytenv}{\qtrat{\ptmtwo}{\varkthree}}{\bott}                   
                   \end{array}
             }
             {
                           \djum{\qtratwo{\tenv},\lneg{(\qtra{\nenv})}, \vark:\lneg{(\qtra{\typtwo})}}{\emptytenv}{\qtrat{\ptmtwo}{\varkthree}\sub{\varkthree}{\iwhy{\uvartwo}{\qtrat{\ptmthree}{\varktwo}\sub{\varktwo}{\invap{\uvartwo}{\ofc{\vark}}}}}}{\bott}
               }
             \]

\item $\rulpName$. The derivation ends in
  \[
      \indrule{\rulpName}{
    \pjum{\tenv}{\ptmtwo}{\typ}{\nenv,\nvar:\typ}
  }{
    \pjum{\tenv}{\pname{\nvar}{\ptmtwo}}{\bott}{\nenv,\nvar:\typ}
    }
\]

\[
  \indrule{\rulmIBott}{
      \indrule{\rulmSub}{
        \indih{
                \djum{\qtratwo{\tenv},\lneg{(\qtra{\nenv})},\nvar:
          \lneg{(\qtra{\typ})},\varktwo:\lneg{(\qtra{\typ})}}{\emptytenv}{\qtrat{\ptmtwo}{\varktwo}}{\bott}
           }
        }
        {
                      \djum{\qtratwo{\tenv},\lneg{(\qtra{\nenv})},\nvar:
          \lneg{(\qtra{\typ})}}{\emptytenv}{\qtrat{\ptmtwo}{\varktwo}\sub{\varktwo}{\nvar}}{\bott}
      }
}
{
                        \djum{\qtratwo{\tenv},\lneg{(\qtra{\nenv})},\nvar:
          \lneg{(\qtra{\typ})},\vark:\one}{\emptytenv}{\ibott{\qtrat{\ptmtwo}{\varktwo}\sub{\varktwo}{\nvar}}{\vark}}{\bott}

  }
\]

\item $\rulpMu$. The derivation ends in
  \[
          \indrule{\rulpMu}{
    \pjum{\tenv}{\ptmtwo}{\bott}{\nenv,\nvar:\typ}
  }{
    \pjum{\tenv}{\pmu{\nvar}{\ptmtwo}}{\typ}{\nenv}
  }
\]

\[
  \indrule{\rulmSub}{
      \indrule{\rulmSub}{
        \indih{
                \djum{\qtratwo{\tenv},\lneg{(\qtra{\nenv})},\nvar:
          \lneg{(\qtra{\typ})},\varktwo:\one}{\emptytenv}{\qtrat{\ptmtwo}{\varktwo}}{\bott}
           }
        }
        {
                      \djum{\qtratwo{\tenv},\lneg{(\qtra{\nenv})},\nvar:
          \lneg{(\qtra{\typ})}}{\emptytenv}{\qtrat{\ptmtwo}{\varktwo}\sub{\varktwo}{\ione}}{\bott}
      }
}
{
                        \djum{\qtratwo{\tenv},\lneg{(\qtra{\nenv})},\vark:\lneg{(\qtra{\typ})}}{\emptytenv}{\qtrat{\ptmtwo}{\varktwo}\sub{\varktwo}{\ione}\sub{\nvar}{\vark}}{\bott}
  }
\]

% \item $\rulpCL$. The derivation ends in
%   \[
%   \indrule{\rulpCL}{
%     \pjum{\tenv,\var:\typ,\vartwo:\typ}{\ptm}{\typthree}{\nenv}
%   }{
%     \pjum{\tenv, \varthree:\typ}{\ptm\sub{\var,\vartwo}{\varthree}}{\typthree}{\nenv}
%   }
% \]

% \item $\rulpCR$. The derivation ends in
%   \[
%      \indrule{\rulpCR}{
%     \pjum{\tenv}{\ptm}{\typthree}{\nenv, \nvar:\typ,\nvartwo:\typ}
%   }{
%     \pjum{\tenv}{\ptm\sub{\nvar,\nvartwo}{\nvarthree}}{\typthree}{\nenv, \nvarthree:\typ}
%   }
% \]

% \item $\rulpWL$. The derivation ends in
%   \[
%    \indrule{\rulpWL}{
%     \pjum{\tenv}{\ptm}{\typ}{\nenv}
%   }{
%     \pjum{\tenv, \var:\typtwo}{\ptm}{\typ}{\nenv}
%   }
% \]

% \item $\rulpWR$. The derivation ends in
%   \[
%      \indrule{\rulpWR}{
%     \pjum{\tenv}{\ptm}{\typ}{\nenv}
%   }{
%     \pjum{\tenv}{\ptm}{\typ}{\nenv, \nvar:\typtwo}
%   }
%   \]
  \end{enumerate}
\end{proof}

\begin{lemma}
\llem{calcMELL_simulates_betav_via_q_translation}
  Suppose $\ptm \to_{\betav} \ptmtwo$, then $\qtrat{\ptm}{\vark} \tome \qtrat{\ptmtwo}{\vark}$.
% $\ttra{(\pap{(\lam{\var}{\ptm})}{\ptmtwo})} \rtto \ttra{\ptm\sub{\var}{\ptmtwo}}$
\end{lemma}

\begin{proof}
By induction on $\ptm$.

\end{proof}

\begin{lemma}
\llem{calcMELL_simulates_muv_via_q_translation}
  Suppose $\ptm \to_{\muv} \ptmtwo$, then $\qtrat{\ptm}{\vark} = \qtrat{\ptmtwo}{\vark}$.
%  $\ttra{(\pap{(\pmu{\nvar}{\ptm})}{\ptmtwo})} =\ttra{(\pmu{\nvar}{(\ptm\psub{\nvar}{\ptmtwo})})}$
\end{lemma}

\begin{proof}
By induction on $\ptm$.

\end{proof}

\begin{lemma}
\llem{calcMELL_simulates_muvprime_via_q_translation}
Suppose $\ptm \to_{\muvprime} \ptmtwo$, then $\qtrat{\ptm}{\vark} \rttofrommod{\cceq} \qtrat{\ptmtwo}{\vark}$.
%  $\ttra{(\pap{(\pmu{\nvar}{\ptm})}{\ptmtwo})} =\ttra{(\pmu{\nvar}{(\ptm\psub{\nvar}{\ptmtwo})})}$
\end{lemma}

\begin{proof}
By induction on $\ptm$.

\end{proof}

\begin{lemma}
\llem{calcMELL_simulates_rho_via_q_translation}
  Suppose $\ptm \to_{\rho} \ptmtwo$, then $\qtrat{\ptm}{\vark} \cceq \qtrat{\ptmtwo}{\vark}$.
%  $ \qtra{(\pname{\nvar}{(\pmu{\nvartwo}{\ptm})})} \cceq\qtra{(\ptm \rensub{\nvartwo}{\nvar})}$
\end{lemma}

\begin{proof}
By induction on $\ptm$.

\end{proof}

% \begin{lemma}
%   \llem{q_translation_and_free_variables}
%   $\nvar\notin\fv{\ptm}$ implies $\nvar\notin\fv{\qtrat{\ptm}{\vark}}$.
% \end{lemma}

% \begin{proof}
% By induction on $\ptm$.
% \end{proof}

\begin{lemma}
\llem{calcMELL_simulates_theta_via_q_translation}
  Suppose $\ptm \to_{\theta} \ptmtwo$, then $\qtrat{\ptm}{\vark} \cceq \qtrat{\ptmtwo}{\vark}$.
%  $\qtra{(\pmu{\nvar}{\pname{\nvar}{\ptm}})}\cceq \qtra{\ptm}$, if $\nvar\notin\fv{\ptm}$.
\end{lemma}

\begin{proof}
By induction on $\ptm$.

\end{proof}

\CalcMELLSimulatesCalcParigotViaQTranslation*

\begin{proof}\label{calcMELL_simulates_calcParigot_via_q_translation:proof}
By induction on $\ptm \to_{\muv} \ptmtwo$ using~\rlem{calcMELL_simulates_betav_via_q_translation},~\rlem{calcMELL_simulates_muv_via_q_translation},~\rlem{calcMELL_simulates_muvprime_via_q_translation},~\rlem{calcMELL_simulates_rho_via_q_translation} and~\rlem{calcMELL_simulates_theta_via_q_translation}. 
% In summary,
%   \begin{enumerate}
%     \item 
%   Suppose $\ptm \to_{\betav} \ptmtwo$, then $\qtrat{\ptm}{\vark} \rttome \qtrat{\ptmtwo}{\vark}$.
% \item   Suppose $\ptm \to_{\muv} \ptmtwo$, then $\qtrat{\ptm}{\vark} = \qtrat{\ptmtwo}{\vark}$.
% \item   Suppose $\ptm \to_{\muvprime} \ptmtwo$, then $\qtrat{\ptm}{\vark} \eqme \qtrat{\ptmtwo}{\vark}$.
%  \item   Suppose $\ptm \to_{\rho} \ptmtwo$, then $\qtrat{\ptm}{\vark} \cceq \qtrat{\ptmtwo}{\vark}$.
% \item   Suppose $\ptm \to_{\theta} \ptmtwo$, then $\qtrat{\ptm}{\vark} \cceq \qtrat{\ptmtwo}{\vark}$.
%   \end{enumerate}
\end{proof}

%%%%%%%%%%%%%%%%%%%%%%%%%%%%%%%%%%
%        T-Translation for barLambdaMuTildeMu
%%%%%%%%%%%%%%%%%%%%%%%%%%%%%%%%%%

\subsection{Curien and Herbelin's $\overline{\lambda}\mu\tilde{\mu}$}

\TypePreservationForTTranslationForBarLambdaMuTildeMu*

\begin{proof}\label{type_preservation_for_t_translation_for_barLambdaMuTildeMu:proof}
  By induction on the derivation.
\end{proof}

\CalcMELLSimulatesCalcParigotViaTTranslationForBarLambdaMuTildeMu*

\begin{proof}\label{calcMELL_simulates_calcParigot_via_t_translation_for_barLambdaMuTildeMu:proof}
  By induction on the reduction step $\tm \to_{\bar{\lambda}\mu\tilde{\mu}} \tmtwo$.
\end{proof}

%%%%%%%%%%%%%%%%%%%%%%%%%%%%%%%%%%
%    Hasegawa's $\muDCLL$
%%%%%%%%%%%%%%%%%%%%%%%%%%%%%%%%%%

\subsection{Hasegawa's $\muDCLL$}

Hasegawa~\cite{DBLP:conf/csl/Hasegawa02,DBLP:journals/mscs/Hasegawa05} introduces $\muDCLL$, an extension of the Dual Intuitionistic Linear Logic of Barber and Plotkin~\cite{Barber:PhDThesis:1997,BarberPlotkin:1997} to \MELL. The system features both the intuitionistic (non-linear) arrow type $\imp$ and the linear arrow type $\limp$.

\begin{definition}[Types and terms of $\muDCLL$]
\ldef{types_and_terms_for_calc_parigot}
They are defined as follows:
\[
  \begin{array}{r@{\hspace{.1cm}}c@{\hspace{.1cm}}l@{\hspace{.5cm}}r@{\hspace{.1cm}}c@{\hspace{.1cm}}l}
    \typ,\typtwo & ::= & \bott\,|\, \btyp\,|\,\typ\imp\typtwo\,|\,\typ\limp\typtwo &
    \ptm,\ptmtwo & ::= & \var\,|\,\uvar\,|\,\ilam{\uvar}{\ptm}\,|\,\iap{\ptm}{\ptmtwo}\,|\,\lam{\var}{\ptm}\,|\,\ap{\ptm}{\ptmtwo}\,|\, \pname{\nvar}{\ptm}\,|\, \pmu{\nvar}{\ptm}
  \end{array}
\]
\end{definition}

Typing rules for $\muDCLL$ are given below, where in $\tenv,\tenv'$ one assumes that $\dom{\tenv}\cap\dom{\tenv'}=\emptyset$ and similarly for $\nenv,\nenvtwo$. 

\scalebox{\proofScaleFactor}{\begin{minipage}{\textwidth}
    \[
  \begin{array}{c}
  \indrule{\rulhiAx}
  {}
    {\hjum{\utenv,\uvar:\typ}{\emptytenv}{\uvar}{\typ}{\emptynenv}}
    \quad
    \indrule{\rulhlAx}
  {}
  {\hjum{\utenv}{\var:\typ}{\var}{\typ}{\emptynenv}}
  
  \\
  \\
  \indrule{\rulhiLam}{
    \hjum{\utenv,\uvar:\typ}{\tenv}{\ptm}{\typtwo}{\nenv}
  }{
    \hjum{\utenv}{\tenv}{\ilam{\uvar}{\ptm}}{\typ\imp\typtwo}{\nenv}
  }
  \indrule{\rulhiApp}{
    \hjum{\utenv}{\tenv}{\ptm}{\typ\imp\typtwo}{\nenv}
    \HS
    \hjum{\utenv}{\emptytenv}{\ptmtwo}{\typ}{\emptynenv}
  }{
    \hjum{\utenv}{\tenv}{\iap{\ptm}{\ptmtwo}}{\typtwo}{\nenv}
    }
    \\
    \\
      \indrule{\rulhlLam}{
    \hjum{\utenv}{\tenv,\var:\typ}{\ptm}{\typtwo}{\nenv}
  }{
    \hjum{\utenv}{\tenv}{\lam{\var}{\ptm}}{\typ\limp\typtwo}{\nenv}
  }
  \indrule{\rulhlApp}{
    \hjum{\utenv}{\tenv}{\ptm}{\typ\limp\typtwo}{\nenv}
    \HS
    \hjum{\utenv}{\tenv'}{\ptmtwo}{\typ}{\nenvtwo}
  }{
    \hjum{\utenv}{\tenv,\tenv'}{\ap{\ptm}{\ptmtwo}}{\typtwo}{\nenv,\nenvtwo}
  }
  \\
    \\
      \indrule{\rulhName}{
    \hjum{\utenv}{\tenv}{\ptm}{\typ}{\nenv}
  }{
    \hjum{\utenv}{\tenv}{\pname{\nvar}{\ptm}}{\bott}{\nenv,\nvar:\typ}
    }
    \quad
    \indrule{\rulhMu}{
    \hjum{\utenv}{\tenv}{\ptm}{\bott}{\nenv,\nvar:\typ}
  }{
    \hjum{\utenv}{\tenv}{\pmu{\nvar}{\ptm}}{\typ}{\nenv}
    }
    \end{array}
  \]
\end{minipage}}
\bigskip

Our presentation of equality in $\muDCLL$ differs slightly from that of~\cite{DBLP:conf/csl/Hasegawa02,DBLP:journals/mscs/Hasegawa05} in three ways. First, we distinguish linear variables (\eg $\var$) from non-linear variables (\eg $\uvar$). Second, the two extensionality axioms have been removed (\ie~$\ilam{\var}{\iap{\ptm}{\var}}\doteq \ptm$ and $\lam{\var}{\ap{\ptm}{\var}}\doteq \ptm$, with $\var\notin\fv{\ptm}$). In addition, \textit{op cit} uses a ``one-shot'' $\mu$-binder in their presentation  of ($\mu-R$).  That is, in the right-hand side of  ($\mu-R$),  \textit{op cit} removes both the occurrence of the $\mu$-binder and the (unique) occurrence of $\pname{\nvar}{}$ in $\ptm$, whereas we retain both\footnote{More precisely, the equation is $\pap{\ptmtwo}{(\pmu{\nvar}{\ptm})} \doteq \ptm[^{\ptmtwo(-)}/\pname{\nvar}{(-)}]$, where $\ptm[^{\ptmtwo(-)}/\pname{\nvar}{(-)}]$ is obtained by replacing the unique occurrence of $\pname{\nvar}{\ptmthree}$ by $\pap{\ptmtwo}{\ptmthree}$ in a capture-free way.}. Our approach is also the one taken in the classical linear $\lambda\mu$-calculus of Bierman~\cite{DBLP:journals/tcs/Bierman99}.

  \begin{definition}[Axioms for $\muDCLL$]
    \ldef{axioms_for_muDCLL}
    $\muDCLL$-equivalence is defined as the contextual closure of the following equivalence axioms:
   \begin{center}
  \begin{tabular}{p{5cm}@{\hspace{1cm}}p{6.2cm}}
    $\begin{array}{llll}
      \iap{(\ilam{\uvar}{\ptm})}{\ptmtwo} &\heq_{lin\beta} & \ptm\sub{\uvar}{\ptmtwo} \\
      \ap{(\lam{\var}{\ptm})}{\ptmtwo} &\heq_{\beta} & \ptm\sub{\var}{\ptmtwo} \\
      %       \pap{\ptmtwo}{(\pmu{\nvar}{\ptm})} & \heq_{\mu-R} &
      %                                                  \pmu{\nvar}{(\ptm\psubtwo{\nvar}{\ptmtwo})} \\
      % \pmu{\nvar}{\pname{\nvar}{\ptm}}  & \heq_{\theta} & \ptm &   (\nvar\notin\fv{\ptm})
     \end{array}$
                                                     &
                        $\begin{array}{llll}
      % \iap{(\ilam{\var}{\ptm})}{\ptmtwo} &\heq_{lin\beta} & \ptm\sub{\var}{\ptmtwo} \\
      % \ap{(\lam{\var}{\ptm})}{\ptmtwo} &\heq_{\beta} & \ptm\sub{\var}{\ptmtwo} \\
            \pap{\ptmtwo}{(\pmu{\nvar}{\ptm})} & \heq_{\mu-R} &
                                                       \pmu{\nvar}{(\ptm\psubtwo{\nvar}{\ptmtwo})} \\
      \pmu{\nvar}{\pname{\nvar}{\ptm}}  & \heq_{\theta} & \ptm &   (\nvar\notin\fv{\ptm})
    \end{array}$
  \end{tabular}
\end{center}
\end{definition}

\begin{definition}[Translation]
\ldef{hasegawa_translation}  
The translation of a formula and $\muDCLL$-term are given by
\begin{center}
  \begin{tabular}{p{5cm}p{6.2cm}}    
  $\begin{array}{llll}
    \htra{\bott} & \eqdef & \bott\\
    \htra{\btyp} & \eqdef & \btyp\\
    \htra{(\typ\imp\typtwo)} & \eqdef & \why{\lneg{(\htra{\typ})}}\parr\htra{\typtwo} \\
    \htra{(\typ\limp\typtwo)} & \eqdef & \lneg{(\htra{\typ})}\parr\htra{\typtwo} 
  \end{array}$
&
$\begin{array}{llll}
    \htra{\var} & \eqdef & \var\\
    \htra{(\ilam{\uvar}{\ptm})} & \eqdef & \ipar{\lvar}{\vark}{(\ibott{\htra{\ptm}\eofc{\uvar}{\lvar}}{\vark})}\\
    \htra{(\iap{\ptm}{\ptmtwo})} & \eqdef & \ap{\htra{\ptm}}{\ofc{\htra{\ptmtwo}}}\\
    \htra{(\lam{\var}{\ptm})} & \eqdef & \ipar{\var}{\vark}{(\ibott{\htra{\ptm}}{\vark})}\\
    \htra{(\ap{\ptm}{\ptmtwo})} & \eqdef & \ap{\htra{\ptm}}{\htra{\ptmtwo}}\\
    \htra{(\pname{\nvar}{\ptm})} & \eqdef & \ibott{\htra{\ptm}}{\nvar}\\
    \htra{(\pmu{\nvar}{\ptm})} & \eqdef & \htra{\ptm}\cos{\nvar}{\ione}

\end{array}$
  \end{tabular}
\end{center}
\end{definition}

% \begin{lemma}[{name=Type preservation~\proofnote{Proof on pg.~\pageref{type_preservation_for_hasegawa:proof}}, restate=[name=Type preservation]TypePreservationForHasegawa}]
\begin{lemma}
\llem{type_preservation_for_hasegawa} If $\hjum{\utenv}{\tenv}{\ptm}{\typ}{\nenv}$ holds in $\muDCLL$, then $\djum{\htra{\utenv}}{\htra{\tenv},\lneg{(\htra{\nenv})}}{\htra{\ptm}}{\htra{\typ}}$ holds in $\CalcMELL$.
\end{lemma}

% Linear and unrestricted substitution commute with $\htra{\bullet}{}$:
% \begin{itemize}
%   \item
%     $\htra{\ptm\sub{\var}{\ptmtwo}}=\htra{\ptm}\sub{\lvar}{\htra{\ptmtwo}}$
%     with $\var$ unrestricted.
% \item
%   $\htra{\ptm\sub{\var}{\ptmtwo}}=\htra{\ptm}\sub{\uvar}{\htra{\ptmtwo}}$
%   with $\var$ linear.
% \end{itemize}

%\TypePreservationForHasegawa*

\begin{proof}%\label{type_preservation_for_hasegawa:proof}
  By induction on the derivation of $\hjum{\utenv}{\tenv}{\ptm}{\typ}{\nenv}$.
  \begin{enumerate}
  \item $\rulhiAx$
\[
  \indrule{\rulhiAx}
  {}
  {\hjum{\utenv,\uvar:\typ}{\emptytenv}{\uvar}{\typ}{\emptynenv}}
\]

\[
    \indrule{\rulmAxU}{
    \emptyPremise
  }{
   \djum{\htra{\utenv},\uvar:\htra{\typ}}{\emptytenv}{\uvar}{\htra{\typ}}
  }
\]

   \item $\rulhlAx$ 
     \[
       \indrule{\rulhlAx}
  {}
  {\hjum{\utenv}{\var:\typ}{\var}{\typ}{\emptynenv}}
\]

\[
    \indrule{\rulmAx}{
    \emptyPremise
  }{
   \djum{\htra{\utenv}}{\var:\htra{\typ}}{\var}{\htra{\typ}}
  }
\]

\item $\rulhiLam$

  \[
    \indrule{\rulhiLam}{
    \hjum{\utenv,\uvar:\typ}{\tenv}{\ptm}{\typtwo}{\nenv}
  }{
    \hjum{\utenv}{\tenv}{\ilam{\uvar}{\ptm}}{\typ\imp\typtwo}{\nenv}
  }
\]

\[
  \indrule{}
  {
  \indrule{}
  {
  \indrule{}
  {\djum{\htra{\utenv}, \uvar:\htra{\typ}}{\htra{\tenv},\lneg{(\htra{\nenv})}}{\htra{\ptm}}{\htra{\typtwo}}}
  {\djum{\htra{\utenv}}{\htra{\tenv}, \lvar:\ofc{\htra{\typ}}, \lneg{(\htra{\nenv})}}{\htra{\ptm}\eofc{\uvar}{\lvar}}{\htra{\typtwo}}}
  }
{
  \djum{\htra{\utenv}}{\htra{\tenv}, \lvar:\ofc{\htra{\typ}}, \lneg{(\htra{\nenv})}, \vark:\lneg{(\htra{\typtwo})}}{\ibott{\htra{\ptm}\eofc{\uvar}{\lvar}}{\vark}}{\bott}
}
}
{
  \djum{\htra{\utenv}}{\htra{\tenv}, \lneg{(\htra{\nenv})} }{\ipar{\lvar}{\vark}{\ibott{\htra{\ptm}\eofc{\uvar}{\lvar}}{\vark}}}{\why{\lneg{(\htra{\typ})}}\parr \htra{\typtwo}}
  }
\]

\item $\rulhiApp$
  \[
    \indrule{\rulhiApp}{
    \hjum{\utenv}{\tenv}{\ptm}{\typ\imp\typtwo}{\nenv}
    \HS
    \hjum{\utenv}{\emptytenv}{\ptmtwo}{\typ}{\emptynenv}
  }{
    \hjum{\utenv}{\tenv}{\iap{\ptm}{\ptmtwo}}{\typtwo}{\nenv}
  }
\]

\[
  \indrule{}
  {\djum{\htra{\utenv}}{\htra{\tenv},\lneg{(\htra{\nenv})}}{\htra{\ptm}}{\why{\lneg{(\htra{\typ})}}\parr\htra{\typtwo}}
    \quad
    \indrule{}
    {\djum{\htra{\utenv}}{\emptytenv}{\htra{\ptmtwo}}{\htra{\typ}}}
    {
      \djum{\htra{\utenv}}{\emptytenv}{\ofc{\htra{\ptmtwo}}}{\ofc{\htra{\typ}}}
    }
  }
  {
    \djum{\htra{\utenv}}{\htra{\tenv},\lneg{(\htra{\nenv})}}{\ap{\htra{\ptm}}{\ofc{\htra{\ptmtwo}}}}{\htra{\typtwo}}
  }
\]
\item $\rulhlLam$
  \[
      \indrule{\rulhlLam}{
    \hjum{\utenv}{\tenv,\var:\typ}{\ptm}{\typtwo}{\nenv}
  }{
    \hjum{\utenv}{\tenv}{\lam{\var}{\ptm}}{\typ\limp\typtwo}{\nenv}
  }
\]

\[
  \indrule{}
  {
  \indrule{}
  {\djum{\htra{\utenv}}{\htra{\tenv},\var:\htra{\typ},\lneg{(\htra{\nenv})}}{\htra{\ptm}}{\htra{\typtwo}}}
  {\djum{\htra{\utenv}}{\htra{\tenv}, \var:\htra{\typ}, \lneg{(\htra{\nenv})}, \vark:\lneg{(\htra{\typtwo})}}{\ibott{\htra{\ptm}}{\vark}}{\bott}}
}
{
  \djum{\htra{\utenv}}{\htra{\tenv}, \var:\htra{\typ}, \lneg{(\htra{\nenv})}, \vark:\lneg{(\htra{\typtwo})}}{\ipar{\var}{\vark}{\ibott{\htra{\ptm}}{\vark}}}{\lneg{(\htra{\typ})}\parr\htra{\typtwo}}
  }
\]

\item $\rulhlApp$
  \[
  \indrule{\rulhlApp}{
    \hjum{\utenv}{\tenv_1}{\ptm}{\typ\limp\typtwo}{\nenv_1}
    \HS
    \hjum{\utenv}{\tenv_2}{\ptmtwo}{\typ}{\nenv_2}
  }{
    \hjum{\utenv}{\tenv_1,\tenv_2}{\ap{\ptm}{\ptmtwo}}{\typtwo}{\nenv_1,\nenv_2}
  }
\]

\[
  \indrule{}
  {\djum{\htra{\utenv}}{\htra{\tenv_1}, \lneg{(\htra{\nenv_1})}}{\htra{\ptm}}{\lneg{(\htra{\typ})}\parr\htra{\typtwo}}
    \quad
    \djum{\htra{\utenv}}{\htra{\tenv_2}, \lneg{(\htra{\nenv_2})}}{\htra{\ptmtwo}}{\htra{\typ}}
  }
  {
    \djum{\htra{\utenv}}{\htra{\tenv_1},\htra{\tenv_2}, \lneg{(\htra{\nenv_1})}, \lneg{(\htra{\nenv_2})}}{\ap{\htra{\ptm}}{\htra{\ptmtwo}}}{\htra{\typtwo}}
  }
  \]
\item $\rulhName$
  \[
      \indrule{\rulhName}{
    \hjum{\utenv}{\tenv}{\ptm}{\typ}{\nenv}
  }{
    \hjum{\utenv}{\tenv}{\pname{\nvar}{\ptm}}{\bott}{\nenv,\nvar:\typ}
  }
\]

\[
  \indrule{}
  {
    \djum{\htra{\utenv}}{\htra{\tenv},\lneg{(\htra{\nenv})}}{\htra{\ptm}}{\htra{\typ}}
  }
  {
    \djum{\htra{\utenv}}{\htra{\tenv},\lneg{(\htra{\nenv})},\nvar:\lneg{(\htra{\typ})}}{\ibott{\htra{\ptm}}{\nvar}}{\bott}
  }
\]

\item $\rulhMu$
  \[    \indrule{\rulhMu}{
    \hjum{\utenv}{\tenv}{\ptm}{\bott}{\nenv,\nvar:\typ}
  }{
    \hjum{\utenv}{\tenv}{\pmu{\nvar}{\ptm}}{\typ}{\nenv}
    }
\]

\[
  \indrule{}
  {
    \djum{\htra{\utenv}}{\htra{\tenv}, \lneg{(\htra{\nenv})},\nvar:\lneg{(\htra{\typ})}}{\htra{\ptm}}{\bott}
  }
  {
    \djum{\htra{\utenv},\lneg{(\htra{\nenv})}}{\htra{\tenv}}{\htra{\ptm}\cos{\nvar}{\ione}}{\htra{\typ}}
  }
\]

  \end{enumerate}
  
\end{proof}

\begin{lemma}[$\htra{\arg}$ commutes with linear and unrestricted substitution]
  \llem{htra_commutes_with_substitution}
  $\htra{\ptm\sub{\var}{\ptmtwo}}=\htra{\ptm}\sub{\var}{\htra{\ptmtwo}}$ and  $\htra{\ptm\sub{\uvar}{\ptmtwo}}=\htra{\ptm}\sub{\uvar}{\htra{\ptmtwo}}$.
\end{lemma}

\begin{proof}
By induction on $\ptm$ using~\rlem{mell_sub_contra}.
\end{proof}

\begin{lemma}
  \llem{calcMELL_simulates_beta_of_muDCLL}
  $\htra{(\iap{(\ilam{\var}{\ptm})}{\ptmtwo})} \rttome \htra{(\ptm\sub{\var}{\ptmtwo})}$
\end{lemma}

\begin{proof}
\[\begin{array}{llll }
        & \htra{(\iap{(\ilam{\var}{\ptm})}{\ptmtwo})}  \\
        = &  \ap{\htra{(\ilam{\var}{\ptm})}}{\ofc{\htra{\ptmtwo}}}\\
        = &  \ap{(\ipar{\lvar}{\vark}{(\ibott{\htra{\ptm}\eofc{\var}{\lvar}}{\vark})})}{\ofc{\htra{\ptmtwo}}}\\
        \tome &  (\ibott{\htra{\ptm}\eofc{\var}{\lvar}}{\vark})\sub{\lvar}{\ofc{\htra{\ptmtwo}}}\cos{\vark}{\ione} \\
        = &  (\ibott{\htra{\ptm}\eofc{\var}{\ofc{\htra{\ptmtwo}}}}{\vark})\cos{\vark}{\ione} \\
    \cceq &  \htra{\ptm}\eofc{\var}{\ofc{\htra{\ptmtwo}}}\\
    \tome & \htra{\ptm}\sub{\var}{\htra{\ptmtwo}} \\
    = & \htra{(\ptm\sub{\var}{\ptmtwo})} & \rlem{htra_commutes_with_substitution}
      \end{array}
    \]
\end{proof}

\begin{lemma}
  \llem{calcMELL_simulates_linearBeta_of_muDCLL}
$\htra{(\ap{(\lam{\var}{\ptm})}{\ptmtwo})} \tome \htra{(\ptm\sub{\var}{\ptmtwo})}$
  \end{lemma}

  \begin{proof}
\[\begin{array}{llll }
        & \htra{(\ap{(\lam{\var}{\ptm})}{\ptmtwo})}  \\
        = &  \ap{\htra{(\lam{\var}{\ptm})}}{\htra{\ptmtwo}}\\
        = &  \ap{(\ipar{\var}{\vark}{(\ibott{\htra{\ptm}}{\vark})})}{\htra{\ptmtwo}}\\
        \tome &  (\ibott{\htra{\ptm}}{\vark})\sub{\var}{\htra{\ptmtwo}}\cos{\vark}{\ione} \\
        = &  (\ibott{\htra{\ptm}\sub{\var}{\htra{\ptmtwo}}}{\vark})\cos{\vark}{\ione} \\
    \cceq &  \htra{\ptm}\sub{\var}{\htra{\ptmtwo}} \\
        = & \htra{(\ptm\sub{\var}{\ptmtwo})} & \rlem{htra_commutes_with_substitution}\\
     
      \end{array}
    \]
  \end{proof}

% \begin{lemma}[{name=Contrasubstitution and Refocusing~\proofnote{Proof on pg.~\ref{contrasubstitution_and_refocusing:proof}}, restate=[name=Contrasubstitution and Refocusing]ContrasubstitutionAndRefocusing}]
\begin{lemma}
\llem{contrasubstitution_and_refocussing}   $\of{\ctx}{\ibott{\tm}{\lvar}}\cos{\lvar}{\tmtwo}=\tm\cctx$, for some $\cctx$ which depends on $\ctx$ and $\tmtwo$
  \end{lemma}

    %\ContrasubstitutionAndRefocusing*

   \begin{proof}
     By induction on $\ctx$. 
   \end{proof}

\begin{lemma}
  \llem{calcMELL_simulates_mur_of_muDCLL}
  $\htra{(\pap{\ptmtwo}{(\pmu{\nvar}{\ptm})})} \cceq \htra{(\pmu{\nvar}{(\ptm\psubtwo{\nvar}{\ptmtwo})})}$
\end{lemma}

  \begin{proof}

    Since $\ptm$ has a unique occurrence of a subterm of the form $\pname{\nvar}{\ptmthree}$, it must have the form $\of{\ctx}{\pname{\nvar}{\ptmthree}}$. Thus $\htra{\ptm}=\of{\htra{\ctx}}{\ibott{\htra{\ptmthree}}{\nvar}}$.
    \[\begin{array}{llll }
        & \htra{(\pap{\ptmtwo}{(\pmu{\nvar}{\ptm})})}  \\
        = &  \ap{\htra{\ptmtwo}}{\htra{(\pmu{\nvar}{\ptm})}}\\
        = &  \ap{\htra{\ptmtwo}}{\htra{\ptm}\cos{\nvar}{\ione}}\\
        = &  \ap{\htra{\ptmtwo}}{\of{\htra{\ctx}}{\ibott{\htra{\ptmthree}}{\nvar}}\cos{\nvar}{\ione}}\\
        = &  \ap{\htra{\ptmtwo}}{\htra{\ptmthree}\cctx}  & \rlem{contrasubstitution_and_refocussing}\\
        \cceq &  (\ap{\htra{\ptmtwo}}{\htra{\ptmthree}}) \cctx  & \rlem{elim_cos_cceq}\\
        = & (\of{\htra{\ctx}}{\ibott{(\ap{\htra{\ptmtwo}}{\htra{\ptmthree}})}{\nvar}})\cos{\nvar}{\ione} & \rlem{contrasubstitution_and_refocussing}\\
        = & \htra{(\ptm\psubtwo{\nvar}{\ptmtwo})}\cos{\nvar}{\ione} & \\
        = & \htra{(\pmu{\nvar}{(\ptm\psubtwo{\nvar}{\ptmtwo})})}
      \end{array}
    \]
  \end{proof}

  \begin{lemma}
    \llem{calcMELL_simulates_theta_of_muDCLL}
    $\htra{(\pmu{\nvar}{\pname{\nvar}{\ptm}})} \cceq \htra{\ptm}$, if $\nvar\notin\fv{\ptm}$
  \end{lemma}
  \begin{proof}

    \[\begin{array}{lll}
        & \htra{(\pmu{\nvar}{\pname{\nvar}{\ptm}})}  \\
        = & \htra{(\pname{\nvar}{\ptm})}\cos{\nvar}{\ione} \\
        = & (\ibott{\htra{\ptm}}{\nvar})\cos{\nvar}{\ione} \\
        = & \nvar\cos{\nvar}{\htra{\ptm}\eone{\ione}} \\
        = & \htra{\ptm}\eone{\ione} \\
        \cceq & \htra{\ptm} \\
      \end{array}
    \]

  \end{proof}

  % \begin{proposition}[{name=Simulation~\proofnote{Proof on pg.~\pageref{calcMELL_simulates_muDCLL:proof}}, restate=[name=Simulation]CalcMELLSimulatesMuDCLL}]
\begin{proposition}
  \lprop{calcMELL_simulates_muDCLL}
  If $\ptm \heq \ptmtwo$, then $\qtrat{\ptm}{\vark} \eqme \qtrat{\ptmtwo}{\vark}$. 
\end{proposition}

%\CalcMELLSimulatesMuDCLL*

\begin{proof}%\label{calcMELL_simulates_muDCLL:proof}
  By induction on $\ptm \heq \ptmtwo$, using~\rlem{calcMELL_simulates_mur_of_muDCLL},~\rlem{calcMELL_simulates_linearBeta_of_muDCLL},~\rlem{calcMELL_simulates_mur_of_muDCLL}, and~\rlem{calcMELL_simulates_theta_of_muDCLL}. In summary,
  \begin{enumerate}
\item  $\htra{(\iap{(\ilam{\uvar}{\ptm})}{\ptmtwo})} \rttome \htra{(\ptm\sub{\uvar}{\ptmtwo})}$.
\item   $\htra{(\ap{(\lam{\var}{\ptm})}{\ptmtwo})} \tome \htra{(\ptm\sub{\var}{\ptmtwo})}$.
 \item     $\htra{(\pap{\ptmtwo}{(\pmu{\nvar}{\ptm})})} \cceq \htra{(\pmu{\nvar}{(\ptm\psubtwo{\nvar}{\ptmtwo})})}$.
\item  $\htra{(\pmu{\nvar}{\pname{\nvar}{\ptm}})} \cceq \htra{\ptm}$, if $\nvar\notin\fv{\ptm}$.
  \end{enumerate}
\end{proof}

%%% Local Variables:
%%% mode: latex
%%% TeX-master: "../main"
%%% End:

%
%\section{Intuitionistic Fragment} 
%\input{appendix/11a-intuitionistic-fragment}

\bibliography{main}

\end{document}

%%% Local Variables:
%%% mode: latex
%%% TeX-master: t
%%% End: